\documentclass[11pt]{article}
\usepackage{fullpage, amssymb}
\usepackage{caption}
\usepackage{subcaption}
\usepackage{graphicx}
\usepackage{amsmath}
\usepackage{float}
\usepackage{listings}
\usepackage{xcolor}
\usepackage{amsthm}
\usepackage{algorithmic}
\usepackage{algorithm}
\usepackage[colorlinks,
linkcolor=blue,
anchorcolor=blue,
citecolor=blue]{hyperref}
\usepackage{natbib}
\usepackage{comment}
\usepackage{bbm}

\usepackage{epstopdf}
\setcitestyle{authoryear,round}

\usepackage{bm,url,mathrsfs}
\usepackage{enumerate}
\usepackage{enumitem}
\usepackage{subcaption}

\usepackage{url} 
\usepackage{makecell}

\DeclareMathOperator*{\argmin}{arg\,min}
\usepackage{stmaryrd}

\def\mfd{\MM_{\br}}
\def\X{\boldsymbol{\mathcal{X}}}
\def\C{\boldsymbol{\mathcal{C}}}
\def\S{\boldsymbol{\mathcal{S}}}
\def\T{\boldsymbol{\mathcal{T}}}
\def\G{\boldsymbol{\mathcal{G}}}
\def\W{\boldsymbol{\mathcal{W}}}
\def\Y{\boldsymbol{\mathcal{Y}}}
\def\Z{\boldsymbol{\mathcal{Z}}}
\def\T{\boldsymbol{\mathcal{T}}}

\def\D{\boldsymbol{\mathcal{D}}}

\def\A{\boldsymbol{\mathcal{A}}}
\def\M{\boldsymbol{\mathcal{M}}}
\def\r{\boldsymbol{r}}
\def\errorE{\boldsymbol{\mathcal{E}}}

\def\frakD{\mathfrak{D}}
\def\frakE{\mathfrak{E}}

\def\kprune{\textsf{k}_{\textsf{pr}}}
\def\frakl{\mathfrak{l}}
\def\kdinf{\textsf{k}_{\infty}}

\def\errinf{\textsf{Err}_{\infty}}
\def\errrank{\textsf{Err}_{2\br}}

\def\opM{\mathcal{M}}
\def\opH{\mathscr{H}^{{\rm HO}}_{\br}}

\def\E{\mathbb{E}}

\def\frakL{\mathfrak{L}}

\def\calE{{\cal  E}}

\def\calM{{\cal  M}}

\def\calP{{\cal  P}}

\def\calS{{\cal  S}}

\newcommand{\bfm}[1]{\ensuremath{\mathbf{#1}}}

   \def\bA{\bfm A}  
     \def\BB{\mathbb{B}}
   \def\bC{\bfm C}  
     
\def\be{\bfm e}   \def\bE{\bfm E}  \def\EE{\mathbb{E}}

   \def\bI{\bfm I}  
     \def\JJ{\mathbb{J}}
     
   \def\bL{\bfm L}  
   \def\bM{\bfm M}  \def\MM{\mathbb{M}}
     \def\NN{\mathbb{N}}
     \def\OO{\mathbb{O}}
     \def\PP{\mathbb{P}}
     
\def\br{\bfm r}   \def\bR{\bfm R}  \def\RR{\mathbb{R}}
   \def\bS{\bfm S}  \def\SS{\mathbb{S}}
\def\bt{\bfm t}   \def\bT{\bfm T}  \def\TT{\mathbb{T}}
   \def\bU{\bfm U}  \def\UU{\mathbb{U}}
\def\bv{\bfm v}   \def\bV{\bfm V}  
   \def\bW{\bfm W}  
\def\bx{\bfm x}   \def\bX{\bfm X}  
     
   \def\bZ{\bfm Z}

\def\bSigma{\bfm \Sigma}
\def\bLambda{\bfm \Lambda}

\def\R{\mathbb{R}}

\def\pro{\mathcal{P}}
\def\cS{\mathcal{S}}

\def\usigma{\underline{\lambda}}
\def\bsigma{\bar{\lambda}}

\def\rmax{\bar{r}}

\def\dmax{\bar{d}}
\def\dmin{\underline{d}}

\def\hat{\widehat}
\def\wt{\widetilde}
\def\no{\notag}

\def\reshape{\textsf{reshape}}

\newcommand\inp[2]{\langle #1, #2 \rangle}
\newcommand{\RN}[1]{%
  \textup{\uppercase\expandafter{\romannumeral#1}}%
}
\newcommand\lr[1]{^{\langle #1 \rangle}}
\newcommand\fro[1]{\| #1 \|_{\rm F}}
\newcommand\op[1]{\|#1\|}
\newcommand\nuc[1]{\|#1\|_{\ast}}
\newcommand\subg[1]{\| #1 \|_{\psi_2}}
\newcommand\ps[1]{#1^{\dagger}}

\newcommand{\ijk}[1]{[#1]_{\omega}}
\newcommand{\pol}[1]{\calP_{\Omega_l}(#1)}
\newcommand{\poll}[1]{\calP_{\Omega^*\backslash\Omega_l}(#1)}
\newcommand{\subw}[1]{[#1]_{\omega}}

\newtheorem{theorem}{Theorem}

\newtheorem{assumption}{Assumption}
\newtheorem{lemma}[theorem]{Lemma}
\newtheorem{corollary}[theorem]{Corollary}

\numberwithin{theorem}{section}
\numberwithin{equation}{section}

\def\trunc{\textsf{Trunc}}
\def\maxl{\bar{\lambda}}
\def\minl{\underline{\lambda}}

\newtheorem{exmp}{Example}[section]

\begin{document}

\title{Generalized Low-rank plus Sparse Tensor Estimation by Fast Riemannian Optimization}

\author{Jian-Feng Cai, Jingyang Li and Dong Xia\thanks{
    Jian-Feng Cai's research was partially supported by Hong Kong RGC Grant GRF 16310620 and GRF 16309219. Dong Xia's research was partially supported by Hong Kong RGC Grant ECS 26302019 and GRF 16303320, 16300121.}\hspace{.2cm}\\
    {\small  Hong Kong University of Science and Technology}}

\date{(\today)}

\maketitle

\begin{abstract}

We investigate a generalized framework to estimate a latent low-rank plus sparse tensor, where the low-rank tensor often captures the multi-way principal components and the sparse tensor accounts for potential model mis-specifications or heterogeneous signals that are unexplainable by the low-rank part. The framework flexibly covers both linear and generalized linear models, and can easily handle continuous or categorical variables. We propose a fast algorithm by integrating the Riemannian gradient descent and a novel gradient pruning procedure. Under suitable conditions, the algorithm converges linearly and can simultaneously estimate both the low-rank and sparse tensors. The statistical error bounds of final estimates are established in terms of the gradient of loss function. The error bounds are generally sharp under specific statistical models, e.g., the sub-Gaussian robust PCA and Bernoulli tensor model. Moreover, our method achieves non-trivial error bounds for heavy-tailed tensor PCA whenever the noise has a finite $2+\varepsilon$ moment. We apply our method to analyze the international trade flow dataset and the statistician hypergraph co-authorship network, both yielding new and interesting findings.

\end{abstract}

\section{Introduction}
In recent years, massive {\it multi-way} datasets, often called {\it tensor} data, have routinely arisen in diverse fields. An $m$th-order tensor is a multilinear array with $m$ ways, e.g., matrices are second order tensors. These multi-way structures often emerge when, to name a few, information features are collected from distinct domains \citep{bi2020tensors,liu2017characterizing,han2020optimal,bi2018multilayer,zhang2020denoising,wang2019multiway}, the multi-relational interactions or higher-order interactions of entities are present \citep{ke2019community,jing2020community, luo2020tensor,wang2020learning,pensky2019spectral}, or the higher-order moments of data are explored \citep{anandkumar2014tensor, sun2017provable,hao2020sparse}. There is an increasing demand for effective methods to analyze large and complex tensorial datasets. 
Low-rank tensor models are a class of statistical models for describing and analyzing tensor datasets. At its core is the assumption that the observed data obeys a distribution that is characterized by a {\it latent} low-rank tensor $\T^{\ast}$. Oftentimes, analyzing tensor datasets boils down to estimating the low-rank $\T^{\ast}$. This procedure is usually referred to as the {\it low-rank tensor estimation}. Together with specifically designed algorithms, low-rank tensor methods have demonstrated encouraging performances on many real-world applications and datasets such as the spatial and temporal pattern analysis of human brain developments \citep{liu2017characterizing}, community detection on multi-layer networks and hypergraph networks \citep{jing2020community,ke2019community,wang2020learning}, multi-dimensional recommender system \citep{bi2018multilayer},  learning the hidden components of mixture models \citep{anandkumar2014tensor}, analysis of brain dynamic functional connectivity\citep{sun2019dynamic}, image denoising and recovery \citep{xia2017statistically} and etc. 

However, the exact low-rank assumption is stringent and sometimes untrue, making low-rank tensor methods vulnerable under model misspecification or in the existence of outliers or heterogeneous signals. While low-rank structure underscores the multi-way principal  components, it fails to capture the dimension-specific outliers or heterogeneous signals that often carry distinctive and useful information.  Consider the international trade flow dataset (see Section~\ref{sec:real_app}) that forms a third-order tensor by the dimensions ${\rm countries}\times {\rm countries}\times {\rm commodities}$. On the one hand we observe that the low-rank structure is capable to reflect the shared similarities among countries such as their geographical locations and economic structures, but on the other hand the low-rank structure tends to disregard the heterogeneity in the trading flows of different countries. This vital heterogeneity often reveals distinctive trading patterns of certain commodities for some countries. Moreover, the heterogeneous signals are usually full-rank and strong that can deteriorate the estimates of the multi-way low-rank principal components. We indeed observe that by filtering out these outliers or heterogeneous signals, the resultant low-rank estimates become more insightful. It is therefore advantageous to decouple the low-rank signal and the heterogeneous one in the procedure of low-rank tensor estimation. Fortunately, these outliers or heterogeneous signals are usually representable by a {\it sparse} tensor, which, is identifiable in generalized low-rank tensor models under suitable conditions. 

In this paper, we propose a generalized low-rank plus sparse tensor model to analyze tensorial datasets. Our fundamental assumption is that the observed data is sampled from a statistical model characterized by the latent tensor $\T^{\ast}+\S^{\ast}$. We assume $\T^{\ast}$ to be low-rank capturing the multi-way principal components, and $\S^{\ast}$ to be sparse (the precise definition of being ``sparse" can be found in Section~\ref{sec:model}) addressing potential model mis-specifications, outliers or heterogeneous signals that are unexplainable by the low-rank part.  Our framework is very flexible which covers both linear and generalized linear models, and can easily handle both quantitative and categorical data.  Compared with existing literature on low-rank tensor methods \citep{gu2014robust, xia2017statistically, zhang2018tensor, yuan2017incoherent, sun2017provable,hao2020sparse,xia2019normal}, our framework and method are more robust, particularly when the latent tensor is only {\it approximately} low-rank or when the noise have heavy tails.   A special case of our model, robust sub-Gaussian tensor PCA, was proposed in \cite{gu2014robust}. Their method was based on matrix unfolding and is thus statistically sub-optimal. The robust tensor PCA model studied by \cite{zhou2017outlier} was based on low tubal-rank and their method is sub-optimal for treating low Tucker-rank tensors. 
The Bernoulli tensor model introduced in \cite{wang2020learning,yu2016learning}, which cannot handle sparse corruptions, is also a special case of our model. See Table~\ref{table:comparison} for the comparison with related works. Compared with the aforementioned works on robust tensor estimation, our model is more general covering a much wider spectrum of tensor-related applications and our method deals with nearly all kinds of tensor data -- Poisson (if intensity is strong), Bernoulli, heavy-tailed data, to name but a few. For instance, to our best knowledge, we derive the first non-trivial convergence rate for heavy-tailed tensor PCA. Meanwhile, our method is robust to model mis-specification up to sparse corruptions.   See numerical comparison results in Section~\ref{sec:numerical}.  We note that the generalized low-rank plus sparse matrix model has been investigated by \cite{zhang2018unified} and \cite{robin2020main}.  However, the estimating procedure is more involved for tensors,  the technical proofs are more challenging,  and treating tensors by matrix unfolding is generally statistically sub-optimal. 

With a properly chosen loss function $\frakL(\cdot)$,  our estimating procedure is formulated into an optimization framework, which aims at minimizing $\frakL(\T+\S)$ subject to the low-rank and sparse constraints on $\T$ and $\S$, respectively.  We propose a new and fast algorithm to solve for the underlying tensors of interest. The algorithm is iterative and consists of two main ingredients: the {\it Riemannian gradient descent} and the {\it gradient pruning}. By viewing the low-rank solution as a point on the Riemannian manifold, we adopt Riemannian gradient descent to update the low-rank estimate. Basically, the Riemannian gradient is the projection of the {\it vanilla gradient} $\nabla \frakL$ onto 
the tangent space of a Riemannian manifold. Unlike the vanilla gradient that is usually full-rank, the Riemannian gradient is often low-rank which can significantly boost up the speed of updating the low-rank estimate. Provided with a reliable estimate of the low-rank tensor $\T^{\ast}$, the gradient pruning is a fast procedure to update our estimate of the sparse tensor $\S^{\ast}$. It is based on the belief that, under suitable conditions, if the current estimate $\hat \T$ is close to $\T^{\ast}$ entry-wisely, the entries of the gradient $\nabla \frakL(\hat\T)$ should have small magnitudes on the complement of the support of $\S^{\ast}$. Then it suffices to run a screening of the entries of $\nabla \frakL(\hat\T)$, locate its entries with large magnitudes and choose $\hat\S$ to minimize the magnitudes of those entries of $\nabla \frakL(\hat\T+\hat\S)$. The procedure looks like pruning the gradient $\nabla \frakL(\hat\T)$ -- thus the name gradient pruning. The algorithm alternates between Riemannian gradient descent and gradient pruning until reaching a locally optimal solution.  

\subsection{Our Contributions}
We propose a novel and generalized framework to analyze tensor datasets.  Our framework allows the latent tensor to be high-rank, as long as it is within the sparse perturbations of a low-rank tensor. The sparse tensor can account for potential mis-specifications of the exact low-rank tensor models,  making our framework robust to outliers, heavy-tailed distributions, heterogeneous signals, etc. Meanwhile, our framework flexibly covers both linear and generalized linear models, and is applicable to both continuous and categorical variables. 
Details are summarize in Table \ref{table:comparison}.
\begin{table}[h!]
	\small
	\centering
	\begin{tabular}{c||c|c|c|c|c} 
		\hline
		Methods & Model & \thead{Sparse outlier} & \thead{SG-RPCA error rate} & \thead{Heavy-tailed PCA}& Poisson RPCA \\ [0.5ex] 
		\hline\hline
		TBM \cite{wang2020learning}& Binary&No&N/A&No&No\\
		\hline
		tubal-tRPCA \cite{lu2016tensor}&RPCA&Yes&Only noiseless&No&No\\
		\hline
		Convex \cite{gu2014robust}&RPCA&Yes&$O(rd^{m-1}+|\Omega^{\ast}|)$&No&No\\
		\hline
		Projected GD \cite{chen2019non}&GLM&No&N/A & No&No\\
		\hline
		Jointly GD \cite{han2020optimal}&GLM&No& $O(mrd)$ when $|\Omega^{\ast}|=0$ & No&No\\
		\hline
		{\bf RGrad (this paper)} & {\bf GLM} & {\bf Yes} & $O(mrd+|\Omega^{\ast}|)$ & {\bf Yes}&{\bf Yes}\\
		\hline 
	\end{tabular}
	\caption{Comparison with related literature. Here, GLM stands for generalized linear model. For the sub-Gaussian robust PCA (SG-RPCA) error rate, we assume $d_j\asymp d, j\in[m]$ for simplicity.  We remark that a strong Poisson intensity is required by RGrad.  See Lemma 8.2 in the supplement.  }
	\label{table:comparison}
\end{table}

We develop a new and fast algorithm which can simultaneously estimate both the low-rank and the sparse tensors. The algorithm is based on the integration of Riemannian gradient descent and a novel gradient pruning procedure.  Our proposed method works for both linear and generalized linear models, adapts to additional sparse perturbations, and is reliable in the existence of stochastic noise. We prove, in a general framework, that our algorithm converges fast even with fixed step-sizes, and establish the statistical error bounds of final estimates. The error bounds are sharp and proportional to the intrinsic degrees of freedom under many specific statistical models.

To showcase the superiority of our methods, we consider applying our framework to interesting examples. The first application is on the sub-Gaussian robust tensor principal component analysis (SG-RPCA) where the observation is simply $\T^{\ast}+\S^{\ast}$ with additive sub-Gaussian noise. We show that our method can recover both $\T^{\ast}$ and $\S^{\ast}$ with sharp error bounds, and recover the support of $\S^{\ast}$ under fairly weak conditions. 
The second example is on the tensor PCA when the noise has heavy tails. We show that our framework is naturally immune to the potential outliers caused by the heavy-tailed noise, and demonstrate that our method achieves non-trivial error bounds as long as the noise have a finite $2+\varepsilon$ moment. This bridges a fundamental gap in the understanding of tensor PCA since the existing methods are usually effective only under sub-Gaussian or sub-Exponential noise. We then apply our framework to learn the latent low-rank structure $\T^{\ast}$ from a binary tensorial observation, assuming the Bernoulli tensor model with a general link function, e.g., the logistic and probit link. Compared with the existing literature, our method is robust and allows an arbitrary but sparse corruption. Finally, our method is applied to Poisson tensor RPCA under a strong intensity condition.  To our best knowledge, our results are the first in these three applications. We also provide computationally fast methods to obtain good initializations.  

Lastly, we employ our method to analyze two real-world datasets: the international commodity trade flow network (continuous variables) and the statistician hypergraph co-authorship network (binary variables). We observe that the low-rank plus sparse tensor framework yields intriguing and new findings that are unseen by the exact low-rank tensor methods. The sparse tensor can nicely capture informative patterns which are overlooked by the multi-way principal components. 

\subsection{Notations and Preliminaries of Tensor}\label{sec:notation}
We use calligraphic-font bold-face letters (e.g. $\T, \X, \T_1$) to denote tensors, bold-face capital letters (e.g. $\bT,\bX,\bT_1$) for matrices, bold-face lower-case letters (e.g. $\bt,\bx,\bt_1$) for vectors and blackboard bold-faced letters (e.g. $\RR, \MM, \UU, \TT$) for sets. We use square brackets with subscripts (e,g. $[\T]_{i_1,i_2,i_3}, [\bT]_{i_1,i_2}, [\bt]_{i_1}$) to represent corresponding entries of tensors, matrices and vectors, respectively. Denote $[\T]_{i_1,:,:}$ and $[\bT]_{i_1,:}$ the $i_1$-th frontal-face and $i_1$-th row-vector of $\T$ and $\bT$, respectively. 
Denote $\|\cdot\|_{\rm F}$ the Frobenius norm of matrices and tensors, and denote $\|\cdot\|_{\ell_p}$ the $\ell_p$-norm of vectors or vectorized tensors for $0\leq p\leq \infty$. Thus, $\|\bv\|_{\ell_0}$ represents the number of non-zero entries of $\bv$, and $\|\bv\|_{\ell_\infty}$ denotes the largest magnitude of the entries of $\bv$. The $j$-th canonical basis vector is written as $\be_j$ whose actual dimension might vary at different appearances. We denote $C, C_1, C_2, c, c_1, c_2\cdots$ some absolute constants whose actual values can change at different lines.

An $m$-th order tensor is an $m$-way array, e.g., $\T\in \RR^{d_1\times\cdots\times d_m}$ means that its $j$-th dimension has size $d_j$. Thus, $\T$ has in total $d_1\cdots d_m$ entries. 
The $j$-th matricization (also called unfolding) $\calM_j(\cdot):\RR^{d_1\times \cdots\times d_m}\mapsto \RR^{d_j\times d_j^{-}}$ with $d_j^-=(d_1\cdots d_m)/d_j$ is a linear mapping so that, for example if $m=3$, $[\calM_1(\T)\big]_{i_1, (i_2-1)d_3+i_3}=[\T]_{i_1,i_2,i_3}$ for $\forall i_j\in[d_j]$. Then, the collection ${\rm rank}(\T):=\big({\rm rank}(\calM_1(\T)),\cdots, {\rm rank}(\calM_m(\T))\big)^{\top}$ is called the multi-linear ranks or {\it Tucker ranks} of $\T$. Given a matrix $\bW_j\in \RR^{p_j\times d_j}$ for any $j\in[m]$, the multi-linear product, denoted by $\times_j$,  between $\T$ and $\bW_j$ is defined by 
$
[\T\times_j \bW_j]_{i_1,\cdots,i_m}:=\sum\nolimits_{k=1}^{d_j}[\T]_{i_1,\cdots,i_{j-1},k,i_{j+1},\cdots,i_m}\cdot [\bW_j]_{i_j,k}, \ \forall i_{j'}\in [d_{j'}] \textrm{ for } j'\neq j; \forall i_j\in[p_j]. 
$
If $\T$ has Tucker ranks $\br=(r_1,\cdots, r_m)^{\top}$, there exist $\C\in\RR^{r_1\times\cdots\times r_m}$ and $\bU_j\in\RR^{d_j\times r_j}$ satisfying $\bU_j^{\top}\bU_j=\bI_{r_j}$ for all $j\in[m]$ such that $
\T=\bC\cdot\llbracket\bU_1,\cdots,\bU_m\rrbracket:=\C\times_1\bU_1\times_2\cdots\times_m \bU_m.
$
This  is referred to as the Tucker decomposition of a low-rank tensor. Tucker ranks and decomposition are well-defined. Readers are suggested to refer \citep{kolda2009tensor} for more details and examples on tensor decomposition and tensor algebra.

\section{General Low-rank plus Sparse Tensor Model} \label{sec:model}
Suppose that we observe data $\frakD$, which can be, for instance, simply a tensorial observation such as the binary adjacency tensor of a hypergraph network or multi-layer network \citep{ke2019community, jing2020community,luo2020tensor,wang2020learning,jin2015fast,ji2016coauthorship}; a real-valued tensor describing multi-dimensional observations \citep{han2020optimal,  sun2017provable, sun2019dynamic,liu2017characterizing}; or a collection of pairs of tensor covariate and real-valued response \citep{hao2020sparse, zhang2020islet, xia2020inference, raskutti2019convex, chen2019non}. 
At the core of our model is the assumption that the observed $\frakD$ is sampled from a distribution characterized by a latent large tensor, denoted by $\T^{\ast}+\S^{\ast}$, where $\T^{\ast}$ has small multi-linear ranks and $\S^{\ast}$ is sparse. Unlike the exact low-rank tensor models, the additional sparse tensor $\S^{\ast}$ can account for potential model mis-specifications and outliers. Consider that $\T^{\ast}$ has multi-linear ranks $\br=(r_1,\cdots,r_m)^{\top}$ with $r_j\ll d_j$ so that $\T^{\ast}\in\MM_{\br}$ where 
$\MM_{\br}:=\big\{\W\in\RR^{d_1\times\cdots\times d_m}: {\rm rank}\big(\calM_j(\W)\big)\leq r_j,\ \forall j\in[m]\big\}.$ 
As for the sparse tensor, we assume that each slice of $\S^{\ast}$ has at most $\alpha$-portion of entries being non-zero for some $\alpha\in(0,1)$. We write $\S^{\ast}\in\SS_{\alpha}$ where the latter is defined by 
$\SS_{\alpha}:=\big\{\S\in\RR^{d_1\times\cdots\times d_m}: \|\be_i^{\top}\calM_j(\S)\|_{\ell_0}\leq \alpha d_j^-,\ \forall j\in[m], i\in[d_j]\big\},$
 where $\be_i$ denotes the $i$-th canonical basis vector whose dimension varies at different appearances.

When the low-rank tensor $\T^{\ast}$ is also sparse, it is generally impossible to distinguish between $\T^{\ast}$ and its sparse counterpart $\S^{\ast}$. To make $\T^{\ast}$ and $\S^{\ast}$ identifiable, we assume that $\T^{\ast}$ satisfies the {\it spikiness condition} meaning that the information it carries spreads fairly across nearly all its entries. Put differently, the spikiness condition enforces $\T^{\ast}$ to be dense -- thus distinguishable from the sparse $\S^{\ast}$. 
This is a typical condition in robust matrix estimation \citep{candes2011robust, chen2020bridging} and tensor completion \citep{xia2019polynomial,xia2017statistically,cai2019nonconvex}. 
For exact low-rank tensor models where $\S^{\ast}$ is absent, this assumption is generally not required. See Section~\ref{sec:exact_lowrank} in the supplementary file for more details. 

\begin{assumption}\label{assump:spikiness}
	Let $\T^{\star}\in \MM_{\br}$, and suppose there exists $\mu_1 > 0$ such that the following holds:
	$
	\textsf{Spiki}(\T^{\ast}):=(d^*)^{1/2}\|\T^*\|_{\ell_{\infty}} /\fro{\T^*} \leq \mu_1,
	$
	where $d^{\ast}=d_1\cdots d_m$.
\end{assumption}
We denote $\UU_{\br, \mu_1}:=\big\{ \T\in \MM_{\br}: \textsf{Spiki}(\T)\leq \mu_1\big\}$ the set of low-rank tensors with spikiness bounded by $\mu_1$.

\textit{Relation between spikiness condition and incoherence condition.}
Let $\T^{\ast}\in \MM_{\br}$ admit a Tucker decomposition $\T^{\ast}=\C^{\ast}\cdot\llbracket \bU_1^{\ast},\cdots,\bU_m^{\ast}\rrbracket$ with $\C^{\ast}\in\RR^{r_1\times\cdots\times r_m}$ and $\bU_j^{\ast}\in\RR^{d_j\times r_j}$ satisfying $\bU_j^{\ast\top}\bU_j^{\ast}=\bI_{r_j}$ for all $j\in [m]$. Suppose that there exists $\mu_0>0$ so that 
$
\mu(\T^{\ast}):=\max_{j\in[m]}\ \max_{i\in [d_j]}\ \|\be_i^{\top}\bU_j^{\ast}\|_{\ell_2}\cdot (d_j/r_j)^{1/2}\leq \sqrt{\mu_0}. 
$ 
Then, $\T^{\ast}$ is said to satisfy the incoherence condition with constant $\mu_0$. 
 The spikiness condition implies the incoherence condition {\it and vice versa}. See Lemma~\ref{lemma:spikiness_incoherence}. 

After observing data $\frakD$, our goal is to estimate the underlying $(\T^{\ast},\S^{\ast})\in (\UU_{\br,\mu_1}, \SS_{\alpha})$. Oftentimes, the problem is formulated as an optimization program equipped with a properly chosen loss function. More specifically, let $\frakL(\cdot):=\frakL_{\frakD}(\cdot): \RR^{d_1\times\cdots\times d_m}\mapsto \RR$ be a smooth (see Assumption~\ref{assump:lowrank}) loss function whose actual form depends on the particular applications. The estimators of $(\T^{\ast}, \S^{\ast})$ are then defined by
$(\hat \T_{\gamma},\ \hat \S_{\gamma}):=\argmin_{\T\in \UU_{\br,\mu_1}, \S\in \SS_{\gamma\alpha}}\ \frakL(\T+\S),$ 
where $\gamma>1$ is a tuning parameter determining the desired sparsity level of $\hat \S_{\gamma}$. For ease of exposition, we tentatively assume that the true ranks are known. In real-world applications, $\br, \alpha$ can be selected by a BIC-type criterion (\ref{eq:BIC}). See Section~\ref{sec:numerical} for more details.  
This generalized framework covers many interesting and important examples as special cases. These examples are investigated more closely in Section~\ref{sec:app}.

\begin{exmp}(Tensor robust principal component analysis)\label{exp:rpca}
For tensor RPCA, the data observed is simply a tensor $\A\in\RR^{d_1\times\cdots\times d_m}$. The basic assumption of tensor PCA is the existence of a low-rank tensor $\T^{\ast}$, called the ``signal", planted inside of $\A$. See, e.g. \citep{zhang2018tensor, richard2014statistical} and references therein. The exact low-rank condition on the ``signal" is sometimes stringent. Tensor robust PCA \citep{lu2016tensor, robin2020main} relaxes this condition by assuming that the ``signal" is the sum of a low-rank tensor $\T^{\ast}$ and a sparse tensor $\S^{\ast}$. With additional additive stochastic noise, the Sub-Gaussian RPCA (SG-RPCA) model assumes $\A=\T^{\ast}+\S^{\ast}+\Z$ with $(\T^{\ast},\S^{\ast})\in(\UU_{\br,\mu_1}, \SS_{\alpha})$ and $\Z$ being a noise tensor having i.i.d. random centered sub-Gaussian entries. We reserve RPCA exclusively for SG-RPCA in the subsequent chapters. 
Given $\A$, the goal is to estimate $\T^{\ast}$ and $\S^{\ast}$. A suitable loss function is 
$\frakL(\T+\S):= \frac{1}{2}\|\T+\S-\A\|_{\rm F}^2$, which measures the goodness-of-fit by $\T+\S$ to data. The estimator $(\hat\T_{\gamma}, \hat \S_{\gamma})$ is thus defined by
\begin{align}\label{eq:rpca_loss}
(\hat\T_{\gamma}, \hat \S_{\gamma}):=\argmin_{\T\in \UU_{\br,\mu_1}, \S\in \SS_{\gamma\alpha}} \frac{1}{2}\|\T+\S-\A\|_{\rm F}^2.
\end{align}
\end{exmp}

\begin{exmp}(Learning low-rank structure from binary tensor)\label{exp:binary_tensor}
In many applications, the observed data $\A$ is merely a binary tensor. Examples include the adjacency tensor in multi-layer networks \citep{jing2020community, paul2020spectral}, brain structural connectivity networks \citep{wang2019common,wang2020learning} and etc. Following the Bernoulli tensor model proposed in \citep{wang2020learning} or generalizing the $1$-bit matrix completion model \citep{davenport20141}, we assume that there exist $(\T^{\ast},\S^{\ast})\in (\UU_{\br, \mu_1}, \SS_{\alpha})$ satisfying
$[\A]_{\omega}\stackrel{{\rm ind.}}{\sim} {\rm Bernoulli}\big(p([\T^{\ast}+\S^{\ast}]_{\omega})\big), \forall \omega\in[d_1]\times\cdots\times [d_m],$ 
where $p(\cdot):\RR\mapsto [0,1]$ is a suitable inverse link function. Popular choices of $p(\cdot)$ include the logistic link $p(x)=(1+e^{-x/\sigma})^{-1}$ and probit link $p(x)=1-\Phi(-x/\sigma)$ where $\sigma>0$ is a scaling parameter. We note that, due to potential symmetry in networks, the entry independence statement might only hold for a subset of its entries (e.g., upper-triangular entries in a single-layer undirected network). Compared with the exact low-rank Bernoulli tensor model \citep{wang2020learning}, ours is more robust to model mis-specifications and outliers. For any pair $(\T, \S)\in (\UU_{\br,\mu_1}, \SS_{\gamma\alpha})$, a suitable loss function is the negative log-likelihood. By maximizing the log-likelihood, we define
\begin{align}\label{eq:hatT_gamma_binary_tensor}
(\hat \T_{\gamma}, \hat \S_{\gamma}):=\argmin_{\T\in\UU_{\br,\mu_1}, \S\in\SS_{\gamma\alpha}} -\sum_{\omega}\big([\A]_{\omega}\log p([\T+\S]_{\omega})+\big(1-[\A]_{\omega}\big)\log \big(1-p([\T+\S]_{\omega})\big)\big).
\end{align}
\end{exmp}

\section{Estimating by Non-convex Optimization}\label{sec:method}

Suppose that a pair\footnote{We will show, in Section~\ref{sec:theory}, that obtaining a good initialization for $\S$ is, under suitable conditions, easy once a good initialization for $\T$ is available.} of initializations near the ground truth is provided. Our estimating procedure adopts a gradient-based iterative algorithm to search for a local minimum of the loss. Since the problem is a constrained optimization, the major difficulty is on the enforcement of constraints during gradient descent updates. To ensure low-rankness, we apply the Riemannian gradient descent algorithm that is fast and simple to implement. Meanwhile, we enforce the sparsity constraint via a gradient-based pruning algorithm. 

\subsection{Riemannian Gradient Descent}\label{sec:rgd}
Provided with $(\hat \T_l, \hat\S_l)$ at the $l$-th iteration, the {\it vanilla} gradient of the loss function is $ \G_l=\nabla \frakL(\hat \T_l+\hat \S_l)$. The naive gradient descent updates the low-rank part to  $\hat \T_l -\beta \G_l$ with a carefully chosen stepsize $\beta>0$, and then projects it back into the set $\MM_{\br}$. This procedure is sometimes referred to as the projected gradient descent (PGD) \citep{chen2019non}. 
Oftentimes, the gradient $\G_l$ has full ranks and thus the subsequent low-rank projection is computationally expensive. Observe that $\hat\T_{l}$ is an element in the smooth manifold $\MM_{\br}$. Meanwhile, due to the smoothness of loss function, it is well recognized that the optimization problem can be solved by Riemannian optimization \citep{edelman1998geometry,kressner2014low} on the respective smooth manifold. Therefore, instead of using the vanilla gradient $\G_l$, it suffices to take the {\it Riemannian gradient}, which corresponds to the steepest descent of the loss but is restricted to the tangent space of $\MM_{\br}$ at the point $\hat \T_l$. The Riemannian gradient is low-rank rendering amazing computational speed-up. See numerical comparison in Section~\ref{sec:numerical}.

 An essential ingredient of Riemannian gradient descent is to project the vanilla gradient onto the tangent space of $\MM_{\br}$. 
 Let $\TT_l$ denote the tangent space of $\MM_{\br}$ at $\hat \T_l$. Suppose that $\hat\T_l$ admits a Tucker decomposition $\hat\T_l=\hat\C_l\cdot \llbracket \hat\bU_{l,1},\cdots,\hat\bU_{l,m}\rrbracket$. The tangent space $\TT_l$ \citep{cai2020provable} has an explicit form written as
$\TT_{l}=\big\{\D_l\times_{i\in[m]}\hat\bU_{l, i}+\sum\nolimits_{i=1}^m \hat \C_l\times_{j\in[m]\backslash i} \hat \bU_{l,j}\times_{i}\bW_{i}: \D_l\in\RR^{\br}, \bW_i\in\RR^{d_i\times r_i}, \bW_i^{\top}\hat\bU_{l,i}={\bf 0} \big\} $. 
Clearly, all elements in $\TT_l$ has their multi-linear ranks upper bounded by $2\br$. Given the vanilla gradient $\G_l$, its projection onto $\TT_l$ is defined by $\calP_{\TT_l}(\G_l):=\argmin_{\X\in \TT_l} \|\G_l-\X\|_{\rm F}^2$. The summands in $\TT_l$ are all orthogonal to each other, allowing fast computation for $\calP_{\TT_l}(\G_l)$.

By choosing a suitable stepsize $\beta>0$, the update by Riemannian gradient descent yields
$
\W_l:=\hat\T_l-\beta \calP_{\TT_l} \G_l.
$ But $\W_l$ may fail to be an element in $\MM_{\br}$. To enforce the low-rank constraint, another key step in Riemannian optimization is the so-called {\it retraction},  which projects a general tensor $\W_l$ back to the smooth manifold $\MM_{\br}$. This procedure amounts to a low-rank approximation of the tensor $\W_l$. In addition, we also need to enforce the spikiness (or incoherent) condition on the low-rank estimate.  Towards that end, we first truncate $\W_l$ entry-wisely by $\zeta_{l+1}/2$ for some easily chosen threshold $\zeta_{l+1}$ and obtain $\wt{\W}_l$, and then retract the truncated tensor $\wt\W_l$ back to the manifold $\MM_{\br}$. We show that a low-rank approximation of $\wt\W_l$ by a simple higher order singular value decomposition (HOSVD) guarantees the convergence of Riemannian gradient descent algorithm. More specifically, for all $j\in[m]$, compute $\bV_{l,j}$ which is the top-$r_j$ left singular vectors of $\calM_j(\wt\W_l)$. The HOSVD approximation of $\wt\W_l$ with multi-linear ranks $\br$ is obtained by 
$	\opH(\wt\W_l):=(\wt\W_l\times_{j=1}^m \bV_{l,j}^{\top})\cdot \llbracket\bV_{l,1},\cdots,\bV_{l,m}\rrbracket $. 
Basically, retraction by HOSVD is the generalization of low-rank matrix approximation by singular value thresholding, although HOSVD is generally not the optimal low-rank approximation of $\wt\W_l$. See, e.g. \citep{zhang2018tensor, xia2019sup, liu2017characterizing, richard2014statistical} for more explanations. Now put these two steps together and we define a trimming operator $\textsf{Trim}_{\zeta,\br}$.
\vspace{-0.2cm}
\begin{align}\label{eq:trim2}
	\textsf{Trim}_{\zeta, \br}(\W):=\opH(\wt\W),\quad \textrm{ where } [\wt\W]_{\omega}=\begin{cases}
		(\zeta/2)\cdot {\rm Sign}([\W]_{\omega}),& \textrm{ if }|[\W]_{\omega}|>\zeta/2\\
		[\W]_{\omega},& \textrm{otherwise}
	\end{cases}
\end{align}
Equipped by the retraction and the entry-wise truncation, the Riemannian gradient descent algorithm updates the low-rank estimate by 
$\hat\T_{l+1}=\textsf{Trim}_{\zeta_{l+1}, \br}(\W_l),$ 
with a properly chosen $\zeta_{l+1}$. 

\subsection{Gradient Pruning}\label{sec:ghd}
The next step is to update the estimate of sparse tensor $\S^{\ast}$. 
Provided with the updated $\hat\T_l$ at the $l$-th iteration, an ideal estimator of the sparse tensor $\S^{\ast}$ is to find
$\argmin_{\S\in \SS_{\gamma\alpha}}\ \frakL(\hat\T_l+\S).$ 
Solving this problem is NP-hard for a general loss function. 
Interestingly, if the loss function is entry-wise meaning that $\frakL(\T)=\sum_{\omega}\frakl_{\omega}([\T]_{\omega})$ where $\frakl_{\omega}(\cdot):\RR\mapsto \RR$ for each $\omega\in[d_1]\times\cdots\times [d_m]$, the computation of sparse estimate becomes tractable. 
More exactly, given a tensor $\G\in \RR^{d_1\times\cdots\times d_m}$, we denote $|\G|^{(n)}$ the value of its $n$-th largest entry in absolute value for $\forall n\in[d_1\cdots d_m]$. Thus, $|\G|^{(1)}$ denotes its largest entry in absolute value. The {\it level-$\alpha$ active indices} of $\G$ is defined by 
$
\textsf{Level-$\alpha$ AInd}(\G):=\big\{\omega=(i_1,\cdots,i_m): \big|[\G]_{\omega}\big|\geq \max\nolimits_{j\in[m]}\big|\be_{i_j}^{\top}\calM_j(\G)\big|^{(\lfloor\alpha d_j^{-} \rfloor)} \big\}. 
$
By definition, the {\it level-$\alpha$ active indices} of $\G$ are those entries whose absolute value is no smaller than the $(1-\alpha)$-th percentile in absolute value on each of its corresponding slices. Clearly, for any $\S\in\SS_{\alpha}$, the support of $\S$ belongs to the $\textsf{Level-$\alpha$ AInd}(\S)$.

We compute the vanilla gradient $\hat\G_l=\nabla \frakL(\hat \T_l)$ so that $[\hat\G_l]_{\omega}=\frakl'_{\omega}([\hat\T_l]_{\omega})$ and find $\JJ=\textsf{Level-$\alpha$ AInd}(\hat\G_l)$. The indices in $\JJ$ have the greatest potential in decreasing the value of loss function. The gradient pruning algorithm sets $[\hat\S_l]_{\omega}=0$ if $\omega\notin \JJ$. On the other hand, for $\omega\in \JJ$, ideally, the entry $[\hat\S_l]_\omega$ is chosen to vanish the gradient in that $\frakl_{\omega}'\big([\hat \T_l+\hat \S_l]_{\omega}\big)=0$.  However, for functions with always-positive gradient (e.g. $e^{x}$), it is impossible to vanish the gradient. Generally, we choose a pruning parameter $\kprune>0$ and set 
\begin{align}\label{eq:prune}
[\hat\S_l]_{\omega}:=\argmin\nolimits_{s:|s+[\hat\T_l]_{\omega}|\leq \kprune}\ \big| \frakl_{\omega}'([\hat\T_l]_{\omega}+s)\big|,\quad \forall \omega\in \JJ.
\end{align}
Basically, eq. (\ref{eq:prune}) chooses $[\hat \S_l]_{\omega}$ from the closed interval $\big[-\kprune-[\hat\T_l]_{\omega}, \kprune-[\hat\T_l]_{\omega}\big]$ to minimize the gradient. 
For a properly selected loss function $\frakl_{\omega}(\cdot)$, searching for the solution $[\hat\S_{l}]_{\omega}$ is usually fast. 
Moreover, for entry-wise square loss, the pruning parameter $\kprune$ can be $\infty$ and $[\hat \S_l]_{\omega}$ has a closed-form solution. See Section~\ref{sec:app} for more details. 
The procedure of gradient pruning is summarized in Algorithm~\ref{algo:hd_thre}.  

\begin{algorithm}
\caption{Gradient Pruning for Sparse Estimate}\label{algo:hd_thre}
\begin{algorithmic}
\STATE {\bf Input: } $\hat\T_l$ and parameters $\gamma>1, \alpha, \kprune>0$
\STATE {Calculate the gradient} $\hat\G_l=\nabla \frakL(\hat \T_l)$ and find  $\JJ=\textsf{Level-$\gamma\alpha$ AInd}(\hat\G_l)$
\FOR{$\omega\in [d_1]\times\cdots\times [d_m]$}
\STATE{
 $$
 [\hat\S_l]_{\omega}=\begin{cases} \textrm{by (\ref{eq:prune})}, &\textrm{ if }\omega\in\JJ\\
 0,&\textrm{ if }\omega\notin \JJ
 \end{cases}
 $$}
\ENDFOR
\STATE {\bf Output: } $\hat\S_l$
\end{algorithmic}
\end{algorithm}

\textit{Final algorithm}. 
Putting together the Riemannian gradient descent and the gradient pruning algorithm, 
we propose the following Algorithm~\ref{algo:lowrank+sparse}. The algorithm alternatingly updates the low-rank estimate and the sparse estimate. We emphasize that the notations $\alpha$ and $\mu_1$ in Algorithm~\ref{algo:lowrank+sparse} do not have to be exactly the model parameters $\alpha$ and $\mu_1$. In theory, we only require them to be larger than the true model parameters $\alpha$ and $\mu_1$, respectively. See Section~\ref{sec:numerical} for more details.


\begin{algorithm}
\caption{Riemannian Gradient Descent and Gradient Pruning}\label{algo:lowrank+sparse}
\begin{algorithmic}
\STATE{\textbf{Initialization: } $\hat\T_0\in\mfd$}, stepsize $\beta$ and parameters $\alpha,\gamma, \mu_1,\kprune>0$
\STATE{Apply Algorithm~\ref{algo:hd_thre} with input $\hat\T_0$ and parameters $\alpha,\gamma,\kprune$  to obtain $\hat\S_0$}
\FOR{$l=0,1,\cdots, l_{\max}-1$}
\STATE{$\G_l = \nabla \frakL(\hat\T_l+\hat\S_l)$}
\STATE{$\W_l = \hat \T_l - \beta \pro_{\TT_l}\G_l$}
\STATE{$\zeta_{l+1} = \frac{16}{7}\mu_1\frac{\fro{\W_l}}{\sqrt{d^*}}$}
\STATE{$\hat\T_{l+1} = \textsf{Trim}_{\zeta_{l+1},\br}(\W_l)$}
\STATE{Apply Algorithm~\ref{algo:hd_thre} with input $\hat\T_{l+1}$ and parameters $\alpha, \gamma,\kprune$  to obtain $\hat\S_{l+1}$}
\ENDFOR
\STATE{\bf Output:} $\hat\T_{l_{\max}}$ and $\hat\S_{l_{\max}}$
\end{algorithmic}
\end{algorithm}

\textit{Rank, sparsity and algorithmic parameters selection}. 
For applications where the true ranks are small, we can simply run Algorithm~\ref{algo:lowrank+sparse} for multiple times with distinct choices of these ranks and decide the best ones according to certain criterion, e.g., interpretability if no ground truth \citep{jing2020community,fan2021alma} or the mis-clustering rate if ground truth is available \citep{ke2019community, zhou2013tensor, wang2020learning}. Sometimes, it suffices to take the singular values of the matricizations and decide the cut-off point by the famous {\it scree plot} \citep{cattell1966scree}. For generalized linear models, selecting the best $\br$ and $\alpha$ is challenging. Nevertheless, we suggest to minimize the following BIC-type criterion: 
\begin{equation}\label{eq:BIC}
{\rm BIC}(\br,\alpha):=\big(\|\hat \S_{\br,\alpha}\|_{\ell_0}+\sum\nolimits_{i=1}^m r_id_i\big)\cdot \ln(d^{\ast})-2\ln(\hat L_{\br,\alpha})
\end{equation}
where $\hat\S_{\br,\alpha}$ is the estimated sparse tensor and $\hat L_{\br,\alpha}$ denotes the respective value of likelihood function, i.e., $-2\ln (\hat L_{\br, \alpha})=d^{\ast}\log(\|\A-\hat \S_{\br,\alpha}-\hat \T_{\br,\alpha}\|_{\rm F}^2)$ for Example~\ref{exp:rpca} (assuming Gaussian noise with unknown variance) and $-\ln(\hat L_{\br,\alpha})$ is the RHS of (\ref{eq:hatT_gamma_binary_tensor}) for Example \ref{exp:binary_tensor}.   Criterion (\ref{eq:BIC}) works reasonably well for Example~\ref{exp:rpca}  and \ref{exp:binary_tensor}, and yields interesting outcomes on international commodity trade flows data. We also propose practical guideline on choosing the algorithmic parameters $\gamma, \mu_1$ and $\kprune$. See Section~\ref{sec:numerical} and the supplement for more details.

\section{General Convergence and Statistical Guarantees}\label{sec:theory}
In this section, we investigate the local convergence of Algorithm~\ref{algo:lowrank+sparse} in a general framework, and characterize the error of final estimates in terms of the gradient of loss function. Their applications on more specific examples are collected in Section~\ref{sec:app}. Our theory relies crucially on the regularity of loss function. Recall that Algorithm~\ref{algo:lowrank+sparse} involves: routine 1. Riemannian gradient descent for the low-rank estimate; and routine 2. gradient pruning for the sparse estimate. It turns out that these two routines generally require different regularity conditions on the loss function, although these conditions can be equivalent in special cases (e.g. see Section~\ref{sec:rpca}). Recall that $\gamma>1$ is the tuning parameter in Algorithm~\ref{algo:lowrank+sparse} which {\it only} plays a role in $\gamma\alpha$, i.e., the desired sparsity.

\begin{assumption}\label{assump:lowrank} (Needed for Low-rank Estimate)
There exist $b_l, b_u>0$ such that $\frakL(\cdot)$ is $b_l$-strongly convex and $b_u$-smooth in a subset $\BB_2^{\ast}\subset \{\T+\S: \T\in\MM_{\br}, \S\in\SS_{\gamma\alpha}\}$ meaning that
\begin{align}
    \inp{\X - (\T^*+\S^{\ast})}{\nabla \frakL(\X) - \nabla \frakL(\T^*+\S^{\ast})} \geq b_l\|\X-\T^*-\S^{\ast}\|_{\rm F}^2\label{eq:lowrank_strcvx}\\
    \|\nabla \frakL(\X) - \nabla \frakL(\T^*+\S^{\ast})\|_{\rm F} \leq b_u \|\X-\T^*-\S^{\ast}\|_{\rm F}\label{eq:lowrank_smooth}
\end{align}
for all $\X \in \BB_2^{\ast}$. Note that $b_l$ and $b_u$ may depend on $\BB_{2}^{\ast}$. 
\end{assumption}
Note that the explicit form of subset $\BB_2^{\ast}$ in Assumption~\ref{assump:lowrank} is usually determined by the actual problems (see examples in Section~\ref{sec:app}). For the main theorem in this section (Theorem~\ref{thm:lowrank+sparse}), we consider $\BB_2^{\ast}$ to be a small neighbour around the truth $\T^{\ast}+\S^{\ast}$. In this case, 
Assumption~\ref{assump:lowrank} requires the loss function to be {\it locally} strongly convex and smooth.  

\begin{assumption}(Needed for Sparse Estimate)\label{assump:sparse}
Suppose that $\frakL$ is an entry-wise loss meaning $\frakL(\T)=\sum_{\omega}\frakl_{\omega}([\T]_{\omega})$ where $\frakl_{\omega}(\cdot): \RR\mapsto \RR$ for any $\omega\in[d_1]\times\cdots\times [d_m]$. There exist a subset $\BB_{\infty}^{\ast}\subset \{\T+\S: \T\in\MM_{\br}, \S\in\SS_{\gamma\alpha}\}$  and $b_l, b_u>0$ such that 
\begin{align}
\langle [\X]_{\omega}-[\Z]_{\omega}, \nabla \frakl_{\omega}([\X]_{\omega})&-\nabla\frakl_{\omega}([\Z]_{\omega})\rangle \geq b_l |[\X-\Z]_{\omega}|^2\label{eq:lowrank+sparse_cond1}\\
\big| \nabla \frakl_{\omega}([\X]_{\omega})-\nabla\frakl_{\omega}([\Z]_{\omega})\big|& \leq b_u\big|[\X-\Z]_{\omega}\big|
\label{eq:lowrank+sparse_cond2}
\end{align}
for  $\forall \omega\in[d_1]\times\cdots\times[d_m]$ and any $\X,\Z\in\BB^{\ast}_{\infty}$. Similarly, $b_l$ and $b_u$ may depend on $\BB^{\ast}_{\infty}$.
\end{assumption}

The gradient pruning Algorithm~\ref{algo:hd_thre} operates on entries of the gradients. Intuitively, entry-wise loss not only simplifies the computation but also helps characterize the performance of gradient pruning algorithm. If the sparse component is absent in our model, i.e. the underlying tensor is exactly low-rank, Assumption~\ref{assump:sparse} will be unnecessary. See Section~\ref{sec:exact_lowrank} in the supplement for more details. 
Notice that the same parameters $b_l, b_u$ are both used in Assumption~\ref{assump:lowrank} and Assumption~\ref{assump:sparse}. This slightly abuse of notations is for the ease of exposition. These parameters are not necessarily equal. 

For an entry-wise loss, condition (\ref{eq:lowrank+sparse_cond1}) and (\ref{eq:lowrank+sparse_cond2}) imply the condition (\ref{eq:lowrank_strcvx}) and (\ref{eq:lowrank_smooth}), respectively. Therefore, Assumption~\ref{assump:lowrank} can be a by-product of Assumption~\ref{assump:sparse}, if we ignore the possible differences between the two neighbours $\BB_2^{\ast}$ and $\BB_{\infty}^{\ast}$. In this way, these two assumptions can be merged into one single assumption. However, we state them separately for several purposes. First, they highlight the differences of theoretical requirements between Riemannian gradient descent and gradient pruning algorithms. Second, the neighbours in these assumptions ($\BB_2^{\ast}$ and $\BB_{\infty}^{\ast}$) can be drastically different. Third, keeping them separate eases subsequent applications for special cases (e.g., for exact low-rank estimate in Section~\ref{sec:exact_lowrank}).

The signal strength $\usigma$ is defined by $\usigma:=\lambda_{\min}(\T^{\ast}):=\min_{j\in[m]} \lambda_{r_j}\big(\calM_j(\T^{\ast})\big)$. Here $\lambda_r(\cdot)$ denotes the $r$-th largest singular value of a matrix. Thus, $\usigma$ represents the smallest non-zero singular value among all the matricizations of $\T^{\ast}$. Similarly, denote $\overline{\lambda}:=\lambda_{\max}(\T^{\ast}):=\max_{j\in[m]} \lambda_1\big(\calM_j(\T^{\ast})\big)$ and define $\kappa_0:=\overline{\lambda}\usigma^{-1}$ to be the condition number of $\T^{\ast}$.

Define 
\begin{align}\label{eq:errrank}
\textsf{Err}_{2\br}:=\sup\nolimits_{\M\in\MM_{2\br}, \|\M\|_{\rm F}\leq 1} \big<\nabla \frakL(\T^{\ast}+\S^{\ast}), \M\big>
\end{align}
and
$
\textsf{Err}_{\infty}:= \max\big\{\|\nabla \frakL(\T^{\ast}+\S^{\ast})\|_{\ell_\infty}, \min\nolimits_{\|\X\|_{\ell_{\infty}}\leq \kprune} \|\nabla \frakL(\X)\|_{\ell_\infty}\big\},
$
where $\kprune$ is the tuning parameter in gradient pruning Algorithm~\ref{algo:hd_thre}.  
The quantity $\textsf{Err}_{2\br}$ is typical in the aforementioned literature in exact low-rank matrix and tensor estimation. But the special quantity $\textsf{Err}_{\infty}$ appears in our paper for investigating the performance of gradient pruning algorithm. The first term in $\errinf$ comes from the gradient of loss function at the ground truth characterizing the stochastic error in many statistical models, while the second term is due to the setting of tuning parameter $\kprune$. 

Theorem~\ref{thm:lowrank+sparse} displays the general performance bounds of Algorithm~\ref{algo:lowrank+sparse}. For simplicity, we denote $\Omega^{\ast}$ the support of $\S^{\ast}$, $\rmax=\max_j r_j, \dmax=\max_j d_j, \dmin=\min_j d_j, r^{\ast}=r_1\cdots r_m$ and $d^{\ast}:=d_jd_j^{-}=d_1\cdots d_m$.  
Let $\|\cdot\|_{\ell_\infty}$ denote the vectorized $\ell_\infty$-norm of tensors. 
\begin{theorem}\label{thm:lowrank+sparse}
Let $\gamma>1, \kprune>0$ be the parameters used in Algorithm~\ref{algo:lowrank+sparse}. 
Suppose that Assumptions \ref{assump:spikiness}, \ref{assump:lowrank} and \ref{assump:sparse} hold with $\BB_{2}^{\ast}=\{\T+\S: \|\T+\S-\T^{\ast}-\S^{\ast}\|_{\rm{F}}\leq C_{0,m}\usigma, \T\in\MM_{\br}, \S\in\SS_{\gamma\alpha}\}$, $\BB_{\infty}^{\ast}=\big\{\T+\S: \|\T+\S-\T^{\ast}-\S^{\ast}\|_{\ell_\infty}\leq \kdinf, \T\in\MM_{\br}, \S\in\SS_{\gamma\alpha}\big\}$ where $\kdinf = C_{1,m}\mu_1^{2m}(\rmax^{m-1}/\dmin^{m-1})^{1/2}\usigma+\kprune+\|\S^{\ast}\|_{\ell_\infty}$, $0.36b_l(b_u^2)^{-1}\leq 1$ and  $b_ub_l^{-1} \leq 0.4(\sqrt{\delta})^{-1}$ for some $\delta\in(0, 1]$ and large absolute constants $C_{0,m}, C_{1,m}>0$ depending only on $m$. Assume that
\vspace{-0.2cm}
\begin{enumerate}[label=(\alph*)]
\item {\it Initialization}: $\|\hat\T_0-\T^*\|_{\rm F} \leq c_{1,m}\usigma\cdot\min\big\{\delta^2\rmax^{-1/2}, (\kappa_0^{2m}\rmax^{1/2})^{-1}\big\}$, $\hat\T_0\in\BB_{\infty}^{\ast}$ and $\hat\T_0$ is $(2\mu_1\kappa_0)^2$-incoherent
\vspace{-0.3cm}
\item {\it Signal-to-noise ratio}:
$
\textsf{Err}_{2\br}/\usigma+\textsf{Err}_{\infty}(b_u+1)(|\Omega^{\ast}|+\gamma\alpha d^{\ast})^{1/2}/(b_l\usigma)\leq c_{2,m}\cdot\min\Big\{\delta^2\rmax^{-1/2}, (\kappa_0^{2m}\rmax^{1/2})^{-1}\Big\}
$
\vspace{-0.8cm}
\item {\it Sparsity condition}: $\alpha\leq c_{3,m}(\kappa_0^{4m}\mu_1^{4m}\rmax^mb_u^4b_l^{-4})^{-1}$ and $\gamma \geq 1+(4m)^{-1}b_u^4b_l^{-4}$
\end{enumerate}
\vspace{-0.2cm}
where $c_{1,m},c_{2,m},c_{3,m}>0$ are small constants depending only on $m$. If the stepsize $\beta$ is between $ [0.005b_l/(b_u^2), 0.36b_l/(b_u^2)]$,  we have
\begin{align}
\fro{\hat\T_{l+1} - \T^*}^2 &\leq (1-\delta^2)\fro{\hat\T_l - \T^*}^2 + C_{1,\delta}\textsf{Err}_{2\br}^2+ C_{1,b_u,b_l}\left(|\Omega^*| + \gamma\alpha d^{\ast} \right)\textsf{Err}_{\infty}^2 \label{eq:hatS+1_err} \\
\fro{\hat\S_{l+1}-\S^*}^2 &\leq \frac{b_u^2}{b_l^2}\left(C_{2,m}\frac{1}{\gamma-1} +C_{3,m} (\mu_1\kappa_0)^{4m}\rmax^m\alpha \right)\fro{\hat\T_{l+1}-\T^*}^2 + \frac{C_1}{b_l^2}(|\Omega^{\ast}| + \gamma\alpha d^{\ast})\errinf^2 \notag
\end{align}
where $C_{1,\delta} = 6\delta^{-1}$ and $C_{1,b_u,b_l} = (C_2+ C_3 b_u + C_4 b_u^2)b_l^{-2}$ for absolute constants $C_1,\cdots,C_4>0$ and $C_{2,m}, C_{3,m}>0$ depending only on $m$. 
Therefore, for all $l\in[l_{\max}]$, we have
\begin{align}\label{eq:hatTl_err}
\fro{\hat\T_l - \T^*}^2 \leq (1-\delta^2)^l \fro{\hat\T_0 - \T^*}^2 + \frac{C_{1,\delta}\textsf{Err}_{2\br}^2+ C_{1,b_u,b_l}\left(|\Omega^*| + \gamma\alpha d^{\ast} \right)\textsf{Err}_{\infty}^2}{\delta^2}.
\end{align} 
\end{theorem}

By eq. (\ref{eq:hatTl_err}), after suitably chosen $l_{\max}$ iterations and treating $b_l, b_u,\delta$ as constants, we conclude with the following error bounds:
\begin{align}\label{eq:hatTlmax}
\fro{\hat\T_{l_{\max}}-\T^{\ast}}^2\leq C_1\textsf{Err}_{2\br}^2+C_2(|\Omega^{\ast}|+\gamma\alpha d^{\ast})\textsf{Err}^2_{\infty}
\end{align}
and
\begin{align*}
\fro{\hat\S_{l_{\max}}-\S^{\ast}}\leq \frac{\alpha(\mu_1\kappa_0)^{4m}\rmax^m(\gamma-1)+1}{\gamma-1}\cdot\big(C_5\textsf{Err}_{2\br}^2+C_6(|\Omega^{\ast}|+\gamma\alpha d^{\ast})\textsf{Err}^2_{\infty}\big)+C_7(|\Omega^{\ast}|+\gamma\alpha d^{\ast})\textsf{Err}^2_{\infty}.
\end{align*}
There exist two types of error as illustrated on the RHS of (\ref{eq:hatTlmax}). The first term $\textsf{Err}_{2\br}^2$ comes from the model complexity of low-rank $\T^{\ast}$, and the term $|\Omega^{\ast}|\textsf{Err}_{\infty}^2$ is related to the model complexity of sparse $\S^{\ast}$. These two terms both reflect the intrinsic complexity of our model. On the other hand, the last term $\gamma\alpha d^{\ast}\textsf{Err}_{\infty}^2$ is a human-intervened complexity which originates from the tuning parameter $\gamma$ in the algorithm design. If the cardinality of $\Omega^{\ast}$ happens to be of the same order as $\alpha d^{\ast}$ (it is the worse-case cardinality of $\Omega^{\ast}$ for $\S^{\ast}\in\SS_{\alpha}$), the error bound is simplified into the following corollary. It is an immediate result from Theorem~\ref{thm:lowrank+sparse} and we hence omit the proof.

\begin{corollary}\label{cor:lowrank+sparse}
Suppose that the conditions of Theorem~\ref{thm:lowrank+sparse} hold and assume that  $|\Omega^*| \asymp \alpha d^{\ast}$. Then for all $l=1,\cdots,l_{\max}$,
$$
\fro{\hat\T_l - \T^*}^2 \leq (1-\delta^2)^l \fro{\hat\T_0 - \T^*}^2 + \frac{C_{1,\delta}\textsf{Err}_{2\br}^2+ C_{1,b_u,b_l}|\Omega^*| \textsf{Err}_{\infty}^2}{\delta^2}.
$$
\end{corollary}

\textit{Remarks on the conditions of Theorem~\ref{thm:lowrank+sparse}}. 
The initialization is required to be as close to $\T^{\ast}$ as $o(\usigma)$, if $b_l ,b_u, \kappa_0$ and $\rmax$ are all $O(1)$ constants. It is a common condition for non-convex methods for low-rank matrix and tensor related problems.  Concerning the signal-to-noise ratio condition, Theorem~\ref{thm:lowrank+sparse} requires $\usigma$ to dominate $\textsf{Err}_{2\br}$ and $(|\Omega^{\ast}|+\gamma \alpha d^{\ast})^{1/2}\textsf{Err}_{\infty}$ if $b_l ,b_u, \kappa_0,\rmax=O(1)$. This condition is mild and perhaps minimal.  The sparsity requirement on $\S^{\ast}$ is also mild. Assuming $b_l, b_u, \kappa_0, \rmax, \mu_1=O(1)$, Theorem~\ref{thm:lowrank+sparse} merely requires $\alpha\leq c$ for a sufficiently small $c>0$ which depends only on $m$,  implying that Algorithm~\ref{algo:lowrank+sparse} allows a wide range of sparsity on $\S^{\ast}$. Similarly, Theorem~\ref{thm:lowrank+sparse} only requires $\gamma\geq C$ for a sufficiently large $C>0$ which depends on $m$ only.

We now investigate the recovery of the support of $\S^{\ast}$.  Algorithm~\ref{algo:lowrank+sparse} usually over-estimates the size of the support of $\S^{\ast}$ since the Level-${\gamma\alpha}$ active indices are used for a $\gamma$ strictly greater than $1$.  

\begin{theorem}\label{thm:hatS_infty}
Suppose conditions of Theorem~\ref{thm:lowrank+sparse} hold, $b_l, b_u=O(1)$, $|\Omega^{\ast}|\asymp \alpha d^{\ast}$ and $l_{\max}$ is chosen such that (\ref{eq:hatTlmax}) holds. Then, 
\begin{equation}\label{eq:hatS-Sstar-supnorm}
\|\hat\S_{l_{\max}}-\S^{\ast}\|_{\ell_\infty}\leq C_{1,m}\kappa_0^{2m}\mu_1^{2m}\Big(\frac{\rmax^m}{\dmin^{m-1}}\Big)^{1/2}\cdot \big(\textsf{Err}_{2\br}+(\gamma|\Omega^{\ast}|)^{1/2}\textsf{Err}_{\infty}\big)+C_{2,m}\textsf{Err}_{\infty},
\end{equation}
where $C_{1,m},C_{2,m}>0$ only depend on $m$.
\end{theorem}
If the non-zero entries of $\S^{\ast}$ satisfy $|[\S^{\ast}]_{\omega}|> 2\delta^{\ast}$ for all $\omega\in\Omega^{\ast}$ where $\delta^{\ast}$ is the RHS of (\ref{eq:hatS-Sstar-supnorm}),  we obtain $\hat\S$ by a final-stage hard thresholding on $\hat\S_{l_{\max}}$ so that
$
[\hat \S]_{\omega}:=[\hat \S_{l_{\max}}]_{\omega}\cdot \mathbbm{1}\big(|[\hat \S_{l_{\max}}]_{\omega}|>\delta^{\ast}\big). $ 
By Theorem~\ref{thm:hatS_infty}, we get ${\rm supp}(\hat \S)=\Omega^{\ast}$ and thus recovering the support of $\S^{\ast}$. The lower bound on the outliers is necessary  for distinguishing the noise and outliers. If an entry in the outliers is of small magnitude, then it might be considered as noise.  When the noise does not exist,  we can set $\kprune = \infty$ implying $\delta^{\ast} = 0$.

\section{Applications}\label{sec:app}
We now apply the established results in Section~\ref{sec:theory} to more specific examples and elaborate the respective statistical performances.  Our framework certainly covers many other interesting examples but we do not intend to exhaust them.

\subsection{Sub-Gaussian Tensor Robust PCA with i.i.d.  Noise}\label{sec:rpca}
As introduced in Example~\ref{exp:rpca}, the goal of SG-RPCA is to extract low-rank \textit{signal} from a noisy tensor observation $\A\in\RR^{d_1\times\cdots\times d_m}$.  Due to the linearity, we use the loss function $\frakL(\T+\S):= \frac{1}{2}\|\T+\S-\A\|_{\rm F}^2$.  Clearly, this loss is an entry-wise loss function, and satisfies the strongly-convex and smoothness conditions of Assumptions~\ref{assump:lowrank} and \ref{assump:sparse} with constants $b_l=b_u=1$ within any subsets $\BB_2^{\ast}$ and $\BB_{\infty}^{\ast}$, or simply $\BB_{2}^{\ast}=\BB_{\infty}^{\ast}=\RR^{d_1\times\cdots\times d_m}$. 
 As a result, Theorem~\ref{thm:lowrank+sparse} and Theorem~\ref{thm:hatS_infty} are readily applicable by choosing $\delta=0.15$, and setting the tuning parameter $\kprune=\infty$.  

\begin{theorem}\label{thm:rpca}
Suppose Assumption~\ref{assump:spikiness} holds and there exists $\sigma_z>0$ such that $\EE\exp\{t[\Z]_{\omega}\}\leq \exp\{t^2\sigma_z^2/2\}$ for $\forall t\in\RR$ and $\forall \omega\in[d_1]\times\cdots\times [d_m]$. Let $r^{\ast}=r_1\cdots r_m$ and $\gamma>1$ be the tuning parameter in Algorithm~\ref{algo:lowrank+sparse}. Assume $|\Omega^{\ast}|\asymp \alpha d^{\ast}$ and 
\vspace{-0.2cm}
\begin{enumerate}[label=(\alph*)]
\item {\it Initialization}: $\|\hat\T_0-\T^*\|_{\rm F} \leq c_{1,m}\usigma\cdot(\kappa_0^{2m}\rmax^{1/2})^{-1}$ and $\hat\T_0$ is $(2\mu_1\kappa_0)^2$-incoherent
\vspace{-0.3cm}
\item {\it Signal-to-noise ratio}: $\usigma/\sigma_z\geq C_{1,m}\kappa_0^{2m}\rmax^{1/2}\cdot (\dmax\rmax+r^{\ast}+\gamma|\Omega^{\ast}|\log\dmax)^{1/2}$
\vspace{-0.3cm}
\item {\it Sparsity condition}: $\alpha\leq c_{2,m}(\mu_1^{4m}\kappa_0^{4m} \rmax^m)^{-1}$ and $\gamma\geq 1+4m$
\end{enumerate}
\vspace{-0.2cm}
where $c_{1,m},c_{2,m},C_{1,m}>0$ are constants depending only on $m$. If the step size $\beta\in[0.005, 0.36]$, then after $l_{\max}>1$ iterations, with probability at least $1-\dmax^{-2}$, we have 
\begin{align}\label{eq:rpca_hatTlmax}
\|\hat\T_{l_{\max}}-\T^{\ast}\|_{\rm F}^2\leq& 0.98^{l_{\max}} \|\hat\T_0-\T^{\ast}\|_{\rm F}^2+C_{2,m}(\dmax\rmax+r^{\ast}+\gamma |\Omega^{\ast}|\log\dmax)\sigma_z^2\\
\|\hat\S_{l_{\max}}-\S^{\ast}\|_{\rm F}^2\leq& \left(C_{3,m}\alpha\rmax^m\mu_1^{4m}\kappa_0^{4m}+C_{4,m}(\gamma-1)^{-1}\right)\cdot \|\hat \T_{l_{\max}}-\T^{\ast}\|_{\rm F}^2+C_{5,m}\sigma_z^2\cdot\gamma|\Omega^{\ast}|\log\dmax\notag
\end{align}
where $C_{2,m},C_{3,m},C_{4,m},C_{5,m}>0$ are constants depending only on $m$. Moreover, If $l_{\max}$ is chosen large enough such that the second term on RHS of (\ref{eq:rpca_hatTlmax}) dominates and assume $\mu_1^{4m}\kappa_0^{4m}\rmax^{m}(\rmax\dmax+r^{\ast})\leq C_{9,m}\dmin^{m-1}$, we get with probability at least $1-\dmax^{-2}$ that
\begin{align*}
\|\hat \S_{l_{\max}}-\S^{\ast}\|_{\ell_\infty}\leq& \left(C_{6,m} \kappa_0^{2m}\mu_1^{2m}  \rmax^{m/2}(\gamma|\Omega^{\ast}|)^{1/2}/\dmin^{(m-1)/2} + C_{7,m}\right)\cdot \sigma_z\log^{1/2}\dmax
\end{align*}
where $C_{6,m},C_{7,m}>0$ are constants depending only on $m$.
\end{theorem}

Theorem~\ref{thm:rpca} has several interesting implications. If the noise is absent meaning $\sigma_z=0$, eq. (\ref{eq:rpca_hatTlmax}) implies that, for an arbitrary $\varepsilon>0$, after $l_{\max}\asymp \log(\varepsilon^{-1})$ iterations, Algorithm~\ref{algo:lowrank+sparse} outputs a $\hat\T_{l_{\max}}$ satisfying $\|\hat \T_{l_{\max}}-\T^{\ast}\|_{\rm F}=O(\varepsilon)$. Therefore, Algorithm~\ref{algo:lowrank+sparse} can exactly recover the low-rank and sparse component, separately.  
On the other hand, if $\sigma_z>0$ and $l_{\max}\asymp \log (\usigma \sigma_z^{-1})$, eq. (\ref{eq:rpca_hatTlmax}) implies that Algorithm~\ref{algo:lowrank+sparse} produces, with probability at least $1-\dmax^{-2}$, 
\begin{align}\label{eq:rpca_hatTlmax_opt}
\|\hat \T_{l_{\max}}-\T^{\ast}\|_{\rm F}^2\leq C_{2,m}\sigma_z^2(\dmax\rmax+r^{\ast}+\gamma|\Omega^{\ast}|\log \dmax).
\end{align}
Since the intrinsic model complexity is of order $\dmax\rmax+r^{\ast}+|\Omega^{\ast}|$, the bound (\ref{eq:rpca_hatTlmax_opt}) is sharp up to logarithmic factors. Similar bounds also hold for $\|\hat \S_{l_{\max}}-\S^{\ast}\|_{\rm F}^2$. In addition, if $\mu_1^{4m}\kappa_0^{4m}\rmax^m(\dmax\rmax+r^{\ast}+\gamma|\Omega^{\ast}|)\leq C_{9,m}\dmin^{m-1}$, we get with probability at least $1-\dmax^{-2}$,
\begin{align}\label{eq:rpca_hatSlmax_opt}
\|\hat \S_{l_{\max}}-\S^{\ast}\|_{\ell_\infty}\leq C_{7,m}\sigma_z\log^{1/2}\dmax. 
\end{align}
Bound (\ref{eq:rpca_hatSlmax_opt}) is nearly optimal. To see it, consider the simpler model that $\T^{\ast}={\bf 0}$.  Then, it is equivalent to estimate a sparse tensor from the data $\S^{\ast}+\Z$. Without further information, $O(\sigma_z\log^{1/2}\dmax)$ is the best sup-norm performance one can expect in general. 

\textit{Comparison with existing literature.}
In \cite{lu2016tensor,zhou2017outlier}, the authors studied noiseless RPCA assuming low tubal rank and proved that a convex program can exactly recover the underlying parameters. Their method works only for third order tensor and is not applicable to low Tucker-rank tensors, and there exists no statistical guarantee for the noisy setting.  
 In \cite{gu2014robust}, the authors proposed the convex relaxation by unfolding a tensor into matrices. Their method is statistically sub-optimal and  computationally more demanding. See Table~\ref{table:comparison} in Introduction.   
Interestingly, by setting $|\Omega^{\ast}|=0$, our Theorem~\ref{thm:rpca} degrades to well-established results for tensor PCA in the literature, e.g.,  higher order orthogonal iteration in \cite{zhang2018tensor} and regularized jointly gradient descent in \cite{han2020optimal}.

\textit{Initialization.}
We verify that the initialization conditions required by Theorem \ref{thm:rpca} can be satisfied under mild conditions.  Recall that we denote the $\lfloor pd^*\rfloor $-th largest entry of $\A$ in absolute value by $|\A|^{(\lfloor pd^*\rfloor)}$.
The idea of the initialization process is to first truncate the observed tensor $\A$ defined by $[\trunc_{\tau}(\A)]_{\omega}=\tau\cdot {\rm Sign}([\A]_{\omega})\cdot \mathbbm{1}(|[\A]_{\omega}|> \tau)+[\A]_{\omega}\cdot \mathbbm{1}(|[\A]_{\omega}|\leq \tau)$.  
We then apply the higher-order orthogonal iteration (HOOI) to $[\trunc_{\tau}(\A)]_{\omega}$. The choice of $\tau$ is given in Algorithm \ref{alg:init:rpca}.  The details of HOOI can be found in the supplement.  We remark that, in practice, the spikiness-related parameter $\mu_1$ can be set to $2^m+\log \bar d$ and gradually double it if Algorithm~\ref{algo:lowrank+sparse} fails to converge.  The theoretical guarantee is provided by the following lemma.

\begin{lemma}\label{lemma:init:rpca}
	Suppose Assumption \ref{assump:spikiness} holds, the support of $\S^*$ is $\Omega^*$ with cardinality $|\Omega^*|$, $\Z$ has i.i.d entries with $\EE[\Z]^2_{\omega}=\sigma_z^2$ and $\EE\exp\{t[\Z]_{\omega}\}\leq \exp\{c_1t^2\sigma_z^2\}$ for $\forall t\in \RR$ and some absolute constant $c_1>0$. There exist absolute constants $C_{1,m}, C_{2,m}, C_{3,m}, c_{2,m}>0$ such that if the maximum iteration of {\rm HOOI} $t_{\max} \geq C_{1,m}\big(\log(\dmax\kappa_0)\vee 1\big)$ and
\vspace{-0.3cm}
\begin{enumerate}[label=(\alph*)]
\item {\it Sparsity condition}: $|\Omega^*| \leq  c_{2,m}\kappa_0^{-4m-2}\rmax^{-2}\mu_1^{-4}\log^{-2}(\dmax)\cdot\min\{(\minl/\sigma_z)^2,d^*\}$,
\vspace{-0.3cm}
\item {\it Signal-to-noise ratio}: $\minl/\sigma_z \geq C_{2,m}\max\{\kappa_0^{2m}\rmax^{1/2}\mu_1[(r^*)^{1/2} + (\dmax\rmax)^{1/2}]\log(\dmax), (d^*)^{1/4}\}$.,
\end{enumerate}
\vspace{-0.3cm}
the output of Algorithm~\ref{alg:init:rpca} satisfies the initialization condition in Theorem \ref{thm:rpca} with probability at least $1-C_{3,m}\dmax^{-2}$.
\end{lemma}

\begin{algorithm}
	\caption{Initialization for SG-RPCA}\label{alg:init:rpca}
	\begin{algorithmic}
		\STATE{Take $p = \min\{(8\mu_1^2)^{-1},(64m\log\dmax)^{-1}\}$, and set $\tau_0 = |\A|^{(\lfloor pd^* \rfloor)}$.}
		\STATE{Define $\A_0$ by $[\A_0]_{\omega} = \begin{cases}
				&[\A]_{\omega}, \text{~if~}|[\A]_{\omega}|\leq \tau_0\\
				&0,\text{~otherwise.}
			\end{cases}$}
		\STATE{Set $\wt\A = \trunc_{\tau}(\A)$ with $\tau = 10\sqrt{m\log(\dmax)}\mu_1\fro{\A_0}/\sqrt{d^*}$.}
		\STATE{$(\hat\T; \hat\bU_1,\cdots,\hat\bU_m)= \textrm{HOOI}(\wt\A)$.}
		\STATE{$\hat\T_0 = \textsf{Trim}_{\eta,\br}(\hat\T)$ with $\eta = 16\mu_1\fro{\hat\T}/(7\sqrt{d^*})$.}
	\end{algorithmic}
\end{algorithm}

\subsection{Tensor PCA under Heavy-tailed Noise}\label{sec:heavy_tail}
Most aforementioned literature in Section~\ref{sec:rpca} on tensor PCA focus on sub-Gaussian \citep{vershynin2011spectral,vershynin2018high,pan2018covariate,li2018tucker}.  Nowadays, heavy-tailed noise routinely arise  in diverse fields.   However, the performances of most existing approaches for tensor PCA significantly deteriorate when noise have heavy tails. 
Interestingly, tensor PCA under heavy-tailed noise can be regarded as a special case of SG-RPCA. Suppose that the observed tensorial data $\A$ satisfies $\A=\T^{\ast}+\Z$ with $\T^{\ast}\in\UU_{\br,\mu_1}$. The noise tensor $\Z$ satisfies the following tail assumption.

\begin{assumption}($\theta$-tailed noise)\label{assump:heavy_tail}
The entries of $\Z$ are $i.i.d.$ with $\EE [\Z]_{\omega}=0$ and ${\rm Var}([\Z]_{\omega})\leq \sigma_z^2$. There exist $\theta>2$  such that $\PP(|[\Z]_{\omega}|\geq \sigma_z\cdot t )\leq t^{-\theta}$ for all $t>1$. 
\end{assumption}
If $\theta$ is only moderately large (e.g., $\theta= 3$), many entries of $\Z$ can have large magnitudes such that the typical concentration properties of $\Z$ (e.g., the bounds of $\errrank$ and $\errinf$ in Section~\ref{sec:rpca}) disappear.  
Fix any $\alpha>1$, we decompose $[\Z]_{\omega}=[\S_{\alpha}]_{\omega}+[\wt\Z]_{\omega}$ such that 
$$
[\wt \Z]_{\omega}=\mathbbm{1}(|[\Z]_{\omega}|\geq \alpha \sigma_z)\cdot (\alpha \sigma_z)\textrm{sign}([\Z]_{\omega})+\mathbbm{1}(|[\Z]_{\omega}|< \alpha \sigma_z)\cdot [\Z]_{\omega},\quad \forall \omega\in[d_1]\times\cdots\times [d_m].
$$
By definition, the entry $[\S_{\alpha}]_{\omega}\neq 0$ if and only if $|[\Z]_{\omega}|> \alpha \sigma_z$.  Now, we write 
\begin{align}\label{eq:heavy_tailed_PCA}
\A=\T^{\ast}+\S_{\alpha}+\wt \Z
\end{align}

\begin{lemma}\label{lem:heavy_tail}
Suppose Assumption~\ref{assump:heavy_tail} holds and the distribution of $[\Z]_{\omega}$ is symmetric. For any $\alpha>1$, we have $\EE \wt\Z={\bf 0}$. There exists an event $\frakE_1$ with $\PP(\frakE_1)\geq 1-\dmax^{-2}$ such that $\S_{\alpha}\in \SS_{\alpha'}$ in the event $\frakE_1$ where $\alpha'=\max\{2\alpha^{-\theta}, 10(\dmax/d^{\ast})\log(m\dmax^3)\}$.
\end{lemma}

By Lemma~\ref{lem:heavy_tail}, in the event $\frakE_1$, model (\ref{eq:heavy_tailed_PCA}) satisfies the SG-RPCA model  such that $\T^{\ast}\in\UU_{\br,\mu_1}, \S_{\alpha}\in \SS_{\alpha'}$. Meanwhile, the entries of $\wt \Z$ are sub-Gaussian for being uniformly bounded by $\alpha \sigma_z$. Therefore, conditioned on $\frakE_1$, Theorem~\ref{thm:rpca} is readily applicable.

\begin{theorem}\label{thm:heavy_tail}
Suppose Assumption~\ref{assump:spikiness} and the conditions of Lemma~\ref{lem:heavy_tail} hold. Choose $\alpha\asymp (d^{\ast}/\dmax)^{1/\theta}, \gamma>1$ as the tuning parameters in Algorithm~\ref{algo:lowrank+sparse}. Assume $(\dmax/d^{\ast})\log(m\dmax^3)\leq c_{2,m}(\mu_1^{4m}\kappa_0^{4m} \rmax^m)^{-1}$, $\gamma\geq 1+4m$ and
\vspace{-0.3cm}
\begin{enumerate}[label=(\alph*)]
\item {\it Initialization}: $\|\hat\T_0-\T^*\|_{\rm F} \leq c_{1,m}\usigma\cdot(\kappa_0^{2m}\rmax^{1/2})^{-1}$ and $\hat\T_0$ is $(2\mu_1\kappa_0)^2$-incoherent
\vspace{-0.3cm}
\item {\it Signal-to-noise ratio}: $\usigma/(\alpha\sigma_z)\geq C_{1,m}\kappa_0^{2m}\rmax^{1/2}\cdot \big(\dmax\rmax+r^{\ast}+\gamma\dmax\log(m\dmax)\big)^{1/2}$
\vspace{-0.3cm}
\end{enumerate}
where $c_{1,m},c_{2,m},C_{1,m}>0$ are constants depending only on $m$. If the step size $\beta\in[0.005, 0.36]$, then after $l_{\max}>1$ iterations, with probability at least $1-2\dmax^{-2}$, we have 
\begin{align}\label{eq:heavy_tail_hatTlmax}
\|\hat\T_{l_{\max}}-\T^{\ast}\|_{\rm F}^2\leq& 0.98^{l_{\max}} \|\hat\T_0-\T^{\ast}\|_{\rm F}^2+C_{2,m}\gamma\big(\dmax\rmax+r^{\ast}+\dmax \log(m\dmax)\big)\alpha^2\sigma_z^2,
\end{align}
where $C_{2,m}>0$ is a constant depending only on $m$. 
\end{theorem}

By Theorem~\ref{thm:heavy_tail} and (\ref{eq:heavy_tail_hatTlmax}), if $\kappa_0,\gamma=O(1)$, $d_j\asymp d$ with $\alpha\asymp d^{(m-1)/\theta}$ and $l_{\max}$ is properly chosen, we get with probability at least $1-2d^{-2}$ that 
\begin{align}\label{eq:heavy_tail_hatTlmax_eq1}
\|\hat\T_{l_{\max}}-\T^{\ast}\|_{\rm F}^2\leq C_{2,m}\big(d\rmax+r^{\ast}+d\log(md)\big) d^{2(m-1)/\theta}\cdot \sigma_z^2.
\end{align}
The bound (\ref{eq:heavy_tail_hatTlmax_eq1}) decreases as $\theta$ increases implying that the final estimate $\hat\T_{l_{\max}}$ becomes more accurate as the noise tail gets lighter. 
Moreover, if Assumption~\ref{assump:heavy_tail} holds with $\theta=2(m-1)\log d$ so that $d^{2(m-1)/\theta}=O(1)$, bound (\ref{eq:heavy_tail_hatTlmax_eq1}) implies $\|\hat\T_{l_{\max}}-\T^{\ast}\|_{\rm F}^2/\sigma_z^2=O\big(d\rmax+r^{\ast}+d\log(md)\big)$ which is sharp up to logarithmic factors. Similarly as Theorem~\ref{thm:rpca}, under Assumption~\ref{assump:spikiness}, it is possible to derive an $\ell_\infty$-norm bound for $\hat \T_{l_{\max}}-\T^{\ast}$. 

\textit{Initialization}. 
The initialization can be attained by Algorithm \ref{alg:init:rpca}. Under bounded spikiness condition, we initially take $\mu_1=2^m+\log(\dmax)$ and gradually double it if Algorithm~\ref{algo:lowrank+sparse} fails to converge. Similarly by Lemma \ref{lemma:init:rpca}, we have the following initialization guarantee for heavy-tailed tensor PCA.
\begin{lemma}\label{lemma:init:heavytail}
Under the setting of Theorem \ref{thm:heavy_tail}, there exist absolute constants $C_{1,m}, C_{2,m}, C_{3,m}>0$ such that if  the signal-noise-ratio 
	$$
	\minl/(\alpha\sigma_z) \geq C_{1,m}\max\left\{\kappa_0^{2m+1}\mu_1^2\rmax^{1/2}[(r^*)^{1/2} + (\dmax\rmax)^{1/2}]\log^2(\dmax),(d^*)^{1/4}\right\},
	$$
after $t_{\max} \geq C_{2,m}\big(\log(\dmax\kappa_0)\vee 1\big)$ iterations, the output of Algorithm \ref{alg:init:rpca} satisfies the initialization condition in Theorem \ref{thm:heavy_tail} with probability at least $1-C_{3,m}\dmax^{-2}$. 
\end{lemma}

\subsection{Bernoulli Tensor Robust PCA}\label{sec:binary_tensor}
Following the Bernoulli tensor model in Example~\ref{exp:binary_tensor},  based on a binary tensorial observation $\A\in\{0,1\}^{d_1\times\cdots\times d_m}$,  we choose the loss function to be the {\it negative} log-likelihood (without loss of generality, we set the scale parameter $\sigma=1$ for ease of exposition)
\begin{align}\label{eq:binary_loss}
\frakL(\T+\S)=-\sum\nolimits_{\omega} \Big([\A]_{\omega}\log p([\T+\S]_{\omega})+(1-[\A]_{\omega})\log\big(1-p([\T+\S]_{\omega})\big)\Big).
\end{align}
The RHS of (\ref{eq:binary_loss}) is an entry-wise loss, and Assumptions~\ref{assump:lowrank} and \ref{assump:sparse} are determined by the entry-wise second order derivatives. For $\forall\zeta>0$, define
$$
b_{u,\zeta}:=\max\bigg\{ \sup_{|x|\leq \zeta} \frac{(p'(x))^2}{p^2(x)}-\frac{p''(x)}{p(x)},\   \sup_{|x|\leq \zeta} \frac{(p'(x))^2}{(1-p(x))^2}+\frac{p''(x)}{1-p(x)}   \bigg\}
$$
$$
b_{l,\zeta}:=\min\bigg\{ \inf_{|x|\leq \zeta} \frac{(p'(x))^2}{p^2(x)}-\frac{p''(x)}{p(x)},\   \inf_{|x|\leq \zeta} \frac{(p'(x))^2}{(1-p(x))^2}+\frac{p''(x)}{1-p(x)}   \bigg\}
$$
Assuming $b_{l,\zeta}, b_{u,\zeta}>0$, then the loss function (\ref{eq:binary_loss}) satisfies Assumptions~\ref{assump:lowrank} and \ref{assump:sparse} with constants $b_{l,\zeta}$ and $b_{u,\zeta}$ for $\BB_{2}^{\ast}=\BB_{\infty}^{\ast}=\{\T+\S: \|\T+\S\|_{\ell_\infty}\leq \zeta, \T\in\MM_{\br}, \S\in\SS_{\gamma\alpha}\}$ (more precisely, the low-rank and sparse conditions are unnecessary). 

Notice that $b_{l,\zeta}$ and $b_{u,\zeta}$ can be extremely sensitive to large $\zeta$. For instance \citep{wang2020learning}, we have
\begin{align*}
b_{l,\zeta}=\begin{cases}
\frac{e^{\zeta}}{(1+e^{\zeta})^2},& \textrm{ if } p(x)=(1+e^{-x})^{-1}\\
\gtrsim \frac{\zeta+1/6}{\sqrt{2\pi}}e^{-\zeta^2},& \textrm{ if } p(x)=\Phi(x)
\end{cases}
\quad {\rm and}
\quad
b_{u,\zeta}=\begin{cases}
\frac{1}{4},& \textrm{ if } p(x)=(1+e^{-x})^{-1}\\
\geq 0.6,& \textrm{ if } p(x)=\Phi(x) 
\end{cases} 
\end{align*}
implying that $b_{u,\zeta}b_{l,\zeta}^{-1}$ increases very fast as $\zeta$ becomes larger. Toward that end, we impose the following assumption which implies $\|\T^{\ast}\|_{\ell_\infty}\leq \zeta/2$ so that $\|\S^{\ast}+\T^{\ast}\|_{\ell_\infty}\leq \zeta$. 
\begin{assumption}\label{assump:binary_tensor}
There exists a small $\zeta>0$ such that $\|\S^{\ast}\|_{\ell_\infty}\leq \zeta/2$, $\T^*$ satisfies Assumption~\ref{assump:spikiness} and its largest singular value
$\overline{\lambda}\leq c_m(\mu_1\kappa_0)^{-m}(\sqrt{d^*/r^*})\cdot\zeta$
where $r^{\ast}=r_1\cdots r_m$ and $d^{\ast}=d_1\cdots d_m$.
\end{assumption}
Meanwhile, we shall guarantee that the iterates $(\hat \T_l, \hat \S_l)$ produced by our algorithm satisfy $\|\hat\T_l\|_{\ell_\infty}, \|\hat\S_l\|_{\ell_\infty}=O(\zeta)$. The infinity norm bound for the sparse part is ensured by the choice of $\kprune$. 
And the low rank part is guaranteed by the following lemma.

\begin{lemma}\label{lem:trim2}
Suppose that Assumptions~\ref{assump:spikiness} and \ref{assump:binary_tensor} hold. Given any $\W$ such that $\|\W-\T^{\ast}\|_{\rm F}\leq \usigma/8$, if we choose $\eta = 16\mu_1\fro{\W}/(7\sqrt{d^*})$, then $\|\textsf{Trim}_{\eta,\br}(\W)\|_{\ell_{\infty}}\leq (9\zeta/16)\cdot (\kappa_0\mu_1)^m$.
\end{lemma}

By Lemma~\ref{lem:trim2}, if $\kappa_0\mu_1, m=O(1)$ and $\big\|\W_l-\T^{\ast}\|_{\rm F}\leq \usigma/8$, the trimming operator guarantees $\|\hat\T_{l+1}\|_{\ell_{\infty}} 
=O(\zeta)$. Equipped with (\ref{eq:trim2}) and by setting $\kprune=C_1\zeta$ for some absolute $C_1>1$ depending only on $\kappa_0\mu_1$ and $m$, we apply Algorithm~\ref{algo:lowrank+sparse} to minimize the RHS of (\ref{eq:binary_loss}).

Similarly, the error of final estimate relies on $\errrank$ and $\errinf$, both of which are related to the gradient of loss (\ref{eq:binary_loss}). Denote $
	L_{\zeta}:=\sup\nolimits_{|x|\leq \zeta} |p'(x)|/|p(x)\big(1-p(x)\big)| .
$ 
Since $\kprune=C_1 \zeta>\zeta$, by definition of $\errinf$, we have 
\begin{align}\label{est:bionomial:errinf}
\errinf\leq \max\bigg\{L_{\zeta},  \min_{|x|\leq \kprune}\bigg|\frac{p'(x)}{p(x)\big(1-p(x)\big)} \bigg|\bigg\}\leq L_{\zeta}. 
\end{align}
In practice, due to sparsity, the value $\zeta$ is often small and it suffices to set $\kprune=1$ in Algorithm \ref{algo:lowrank+sparse} and for a cleaner bound of Theorem~\ref{thm:binary_tensor}. 
\begin{theorem}\label{thm:binary_tensor}
Let $\gamma>1, \kprune:=C_1\zeta$ be the parameters used in Algorithm~\ref{algo:lowrank+sparse} for a constant $C_1>1$ depending only on $\kappa_0\mu_1$ and $m$ via Lemma~\ref{lem:trim2}.  
Suppose Assumptions~\ref{assump:spikiness} and \ref{assump:binary_tensor} hold. Assume $|\Omega^{\ast}|\asymp \alpha d^{\ast}$, $0.36b_{l,\zeta'}b_{u,\zeta'}^{-2}\leq 1$, $b_{u,\zeta'}b_{l,\zeta'}^{-1}\leq 0.4(\sqrt{\delta})^{-1}$ for some $\delta\in(0,1]$ and $\zeta' = (2C_1+1)\zeta$, and 
\vspace{-0.3cm}
\begin{enumerate}[label=(\alph*)]
\item {\it Initialization}: $\|\hat\T_0-\T^*\|_{\rm F} \leq c_{1,m}\usigma\cdot\min\big\{\delta^2\rmax^{-1/2}, (\kappa_0^{2m}\rmax^{1/2})^{-1}\big\}$,
$\|\hat\T_0\|_{\ell_\infty}\leq c_{2,m}\zeta$
and $\hat\T_0$ is $(2\mu_1\kappa_0)^2$-incoherent
\vspace{-0.3cm}
\item {\it Signal-to-noise ratio}: 
\vspace{-.2cm}
$$
\usigma\cdot \min\big\{\delta^2\rmax^{-1/2}, (\kappa_0^{2m}\rmax^{1/2})^{-1}\big\}\geq C_{2,m}\Big(\sqrt{\dmax\rmax+r^{\ast}}+\gamma|\Omega^{\ast}|\frac{1+b_{u,\zeta'}}{b_{l,\zeta'}}\Big)\cdot L_{\zeta}
$$
\vspace{-1cm}
\item {\it Sparsity condition}: $\alpha\leq c_{3,m}b_{l,\zeta'}^4(b_{u,\zeta'}^4\kappa_0^{4m}\mu_1^{4m} \rmax^m)^{-1}$ and $\gamma\geq 1+4m\cdot b_{u,\zeta'}^4b_{l,\zeta'}^{-4}$
\end{enumerate}
\vspace{-0.2cm}
where $c_{1,m},c_{2,m},c_{3,m}, C_{2,m}>0$ are some constants depending on $m$ only. If the stepsize $\beta\in[0.005b_{l,\zeta'}/(b_{u,\zeta'})^2, 0.36b_{l,\zeta'}/(b_{u,\zeta'})^2]$, after $l_{\max}$ iterations, with probability at least $1-\dmax^{-2}$,
\begin{align}\label{eq:binary_hatTlmax_fro}
\|\hat\T_{l_{\max}}-\T^{\ast}\|_{\rm F}^2&\leq (1-\delta^2)^{l_{\max}}\cdot \|\hat\T_0-\T^{\ast}\|_{\rm F}^2+C_{3}L_{\zeta}^2\cdot(\dmax\rmax+r^{\ast}+\gamma |\Omega^{\ast}|)\\
\|\hat\S_{l_{\max}}-\S^{\ast}\|_{\rm F}^2&\leq \frac{b_{u,\zeta'}^2}{b_{l,\zeta'}^2}\big(C_{4,m}\alpha\rmax^m(\mu_1\kappa_0)^{4m}+C_{5,m}(\gamma-1)^{-1}\big)\|\hat\T_{l_{\max}}-\T^{\ast}\|_{\rm F}^2+\frac{C_{6,m}}{b_{l,\zeta'}^2}L_{\zeta}^2\cdot\gamma|\Omega^{\ast}|\notag
\end{align}
where $C_{3}>0$ depends only on $\delta,b_{l,\zeta'}, b_{u,\zeta'},m$ , and $C_{4,m},C_{5,m},C_{6,m}>0$ are constants depending only on $m$. Moreover, if $l_{\max}$ is chosen large enough such that the second term on RHS of (\ref{eq:binary_hatTlmax_fro}) dominates and assume $\kappa_0^{4m}\mu_1^{4m}\rmax^{m}(\rmax\dmax+r^{\ast})\lesssim_{m}O(\dmin^{m-1})$, we get with probability at least $1-\dmax^{-2}$ that
\begin{align*}
\|\hat\T_{l_{\max}}-\T^{\ast}\|_{\ell_\infty}\leq& C_{6}\kappa_0^{2m}\mu_1^{2m}(\rmax^m/\dmin^{m-1})^{1/2} (\dmax\rmax+r^{\ast}+\gamma |\Omega^{\ast}|)^{1/2}\cdot L_{\zeta}\\
\|\hat\S_{l_{\max}}-\S^{\ast}\|_{\ell_\infty}\leq& \Big(C_{7}\kappa_0^{2m}\mu_1^{2m}\rmax^{m/2}|\Omega^{\ast}|^{1/2}/\dmin^{(m-1)/2}+C_{8}\Big)\cdot L_{\zeta}
\end{align*}
where $C_{6}, C_{7}, C_{8}>0$ depend only on $\delta, b_{l,\zeta'}, b_{u,\zeta'},m$. 
\end{theorem}

By Theorem~\ref{thm:binary_tensor}, after a properly chosen $l_{\max}$ iterations and treating $\gamma$ as a bounded constant,  bound (\ref{eq:binary_hatTlmax_fro}) implies $\|\hat \T_{l_{\max}}-\T^{\ast}\|_{\rm F}^2=O\big(L_{\zeta}^2\cdot (\dmax \rmax+r^{\ast}+|\Omega^{\ast}|)\big)$. Note that the term $\dmax\rmax+r^{\ast}+|\Omega^{\ast}|$ is the model complexity and thus this rate is sharp in general. If $\S^{\ast}={\bf 0}$ so that $|\Omega^{\ast}|=0$, this rate is comparable to the existing ones in generalized low-rank tensor estimation \citep{wang2020learning, han2020optimal}. Additionally, for the matrix case ($m=2$), this rate matches the well-known results in \citep{davenport20141, robin2020main}. 

\textit{Initialization}.
The initialization for treating binary tensorial data can be obtained by unfolding $\A$ into a $(d_1\cdots d_{m_0})\times (d_{m_0+1}\cdots d_{m})$ matrix and applying $1$ bit matrix estimation \citep{davenport20141}. Here $m_0=\lfloor m/2\rfloor$. It is based on the solution to a convex program. See  the supplement for more details. Its theoretical performance is guaranteed by Lemma~\ref{lemma:init:binary}. We write $d_1^{\ast}=d_1\cdots d_{m_0}$, $d_2^{\ast}=d_{m_0+1}\cdots d_{m}$, and $\beta_{\zeta} = \sup_{|x|\leq \zeta} |p(x)(1-p(x))| / (p'(x))^2$.

\begin{lemma}\label{lemma:init:binary}
Suppose that Assumptions~\ref{assump:spikiness} and \ref{assump:binary_tensor} hold and $\S^*\in\SS_{\alpha}$. There exist absolute constants $C_{1,m}, C_{2,m}, C_{3,m}>0$ such that if 
\begin{enumerate}[label=(\alph*)]
		\item {\it Sparsity of $\S^*$}: $|\Omega^*| \leq \min\bigg\{\frac{d^*r}{\min(d_1^*,d_2^*)},\ C_{1,m}\zeta^{-2}\bsigma^2\cdot\min\big\{\delta^4\rmax^{-1}, \kappa_0^{-4m}\rmax^{-1}\big\}\bigg\}$ ,
		\item {\it Signal-to-noise ratio}: $\minl^2\cdot\min\big\{\delta^4\rmax^{-1}, \kappa_0^{-4m}\rmax^{-1}\big\}\geq C_{2,m}\zeta L_{\zeta}\beta_{\zeta}[r(d_1^{\ast}+d_2^{\ast})d^{\ast}]^{1/2}$, 
\end{enumerate}	
the output of Algorithm \ref{alg:init:binary} satisfies the initialization in Theorem \ref{thm:binary_tensor} with probability at least $1 - C_{3,m}(d^{\ast})^{-1}$. 
\end{lemma}

Compared with Theorem~\ref{thm:binary_tensor}, the required sparsity of $\S^{\ast}$ and signal-to-noise ratio are more stringent to guarantee a warm initialization. It is typical that, oftentimes, the signal-to-noise ratio condition required by warm initialization is the primary bottleneck in tensor-related problems. See, for instance, \cite{xia2017statistically, zhang2018tensor}.

\subsection{Poisson Tensor Robust PCA}
In this section, we consider the Poisson tensor RPCA model. Suppose we observe $\Y\in\NN^{d_1\times\cdots \times d_m}$ that satisfies
$$\forall \omega \in [d_1]\times\cdots\times[d_m], ~[\Y]_{\omega}\sim \text{Poisson}(I\exp([\T^*]_{\omega} + [\S^*]_{\omega})) \text{~independently},$$
where $(\T^*,\S^*)\in(\UU_{\br, \mu_1},\SS_{\alpha})$ are the low rank part and sparse part respectively and $I>0$ is the intensity parameter that is revealed as in \cite{han2020optimal}. We choose the loss function to be the negative log-likelihood with scaling
$$\frakL(\T+\S) = \frac{1}{I}\sum_{\omega} \left(-[\Y]_{\omega}[\T+\S]_{\omega} + I\exp([\T+\S]_{\omega})\right).$$
This is an entry-wise loss, and simple calculation shows Assumptions \ref{assump:lowrank} and \ref{assump:sparse} are satisfied with $\BB_{2}^* = \BB_{\infty}^* = \{\T+\S:\|\T+\S\|_{\ell_{\infty}} \leq \zeta, \T\in\MM_{\br},\S\in\SS_{\gamma\alpha}\}$ with $b_{l,\zeta} = e^{-\zeta}, b_{u,\zeta} = e^{\zeta}$. Since the parameter will become trivial in an unbounded set, we impose the following assumption which implies $\|\T^*\|_{\ell_{\infty}} \leq \frac{\zeta}{2}$ and thus $\|\T^*+\S^*\|_{\ell_{\infty}}\leq\zeta$.
\begin{assumption}\label{assump:poisson}
	There exists a small $\zeta > 0$ such that $\|\S^*\|_{\infty}\leq\frac{\zeta}{2}$, $\T^*$ satisfies Assumption \ref{assump:spikiness} with its largest singular value $\maxl\leq c_m(\kappa_0\mu_1)^{-m}\sqrt{\frac{d^*}{r^*}}\zeta$ where $d^* = d_1\cdots d_m$ and $r^* = r_1\cdots r_m$.
\end{assumption}
Similar with the binary case, we also need to show $\|\hat\T_l\|_{\ell_{\infty}},\|\hat\S_l\|_{\ell_{\infty}} = O(\zeta)$. These are guaranteed by choosing $\kprune = C_1\zeta$ for some $C_1>1$ depending only on $\kappa_0\mu_1,m$ and from Lemma \ref{lem:trim2}, when $\kappa_0\mu_1,m = O(1)$ and $\fro{\hat \T_l-\T^*}\leq \minl/8$, we have $\|\hat\T_{l+1}\|_{\ell_{\infty}} = O(\zeta)$. We summarize the result in the following Theorem.
\begin{theorem}\label{thm:poisson}
	Let $\gamma>1, \kprune:=C_1\zeta$ be the parameters used in Algorithm~\ref{algo:lowrank+sparse} for a constant $C_1>1$ depending only on $\kappa_0\mu_1$ and $m$ via Lemma~\ref{lem:trim2}.  
	Suppose Assumptions~\ref{assump:spikiness} and \ref{assump:poisson} hold. Assume $|\Omega^{\ast}|\asymp \alpha d^{\ast}$, $e^{2\zeta'}\leq 0.4(\sqrt{\delta})^{-1}$ for some $\delta\in(0,1]$ and $\zeta' = (2C_1+1)\zeta$, and 
	\vspace{-0.3cm}
	\begin{enumerate}[label=(\alph*)]
		\item {\it Initialization}: $\|\hat\T_0-\T^*\|_{\rm F} \leq c_{1,m}\usigma\cdot\min\big\{\delta^2\rmax^{-1/2}, (\kappa_0^{2m}\rmax^{1/2})^{-1}\big\}$,
		$\|\hat\T_0\|_{\ell_\infty}\leq c_{2,m}\zeta$
		and $\hat\T_0$ is $(2\mu_1\kappa_0)^2$-incoherent
		\vspace{-0.3cm}
		\item {\it Signal-to-noise ratio}: 
		\vspace{-.2cm}
		\begin{align*}
			\usigma\cdot \min\big\{\delta^2\rmax^{-1/2}, (\kappa_0^{2m}\rmax^{1/2})^{-1}\big\}\geq C_{2,m}\Big(\gamma|\Omega^{\ast}|\frac{1+e^{\zeta'}}{e^{-\zeta'}}\cdot e^{\zeta} + \sqrt{(r^*+\dmax\rmax)e^{\zeta}/I}\Big),
			\text{~and~}I\geq Ce^{\zeta}\log(d^*)
		\end{align*}
		\vspace{-1cm}
		\item {\it Sparsity condition}: $\alpha\leq c_{3,m}e^{-8\zeta'}(\kappa_0^{4m}\mu_1^{4m} \rmax^m)^{-1}$ and $\gamma\geq 1+(4m)^{-1}\cdot e^{8\zeta'}$
	\end{enumerate}
	\vspace{-0.2cm}
	where $c_{1,m},c_{2,m},c_{3,m}, C_{2,m}>0$ are some constants depending on $m$ only. If the stepsize $\beta\in[0.005e^{-3\zeta'}, 0.36e^{-3\zeta'}]$, after $l_{\max}$ iterations, with probability at least $1-\frac{2}{d^*}$,
	\begin{align*}
		\|\hat\T_{l_{\max}}-\T^{\ast}\|_{\rm F}^2&\leq (1-\delta^2)^{l_{\max}}\cdot \|\hat\T_0-\T^{\ast}\|_{\rm F}^2+ C_{1,\delta}\frac{r^* + \dmax\rmax}{I/e^{\zeta}}+ C_{3}e^{2\zeta}\cdot\gamma |\Omega^{\ast}|\\
		\|\hat\S_{l_{\max}}-\S^{\ast}\|_{\rm F}^2&\leq e^{4\zeta'}\big(C_{4,m}\alpha\rmax^m(\mu_1\kappa_0)^{4m}+C_{5,m}(\gamma-1)^{-1}\big)\|\hat\T_{l_{\max}}-\T^{\ast}\|_{\rm F}^2+C_{6,m}e^{2\zeta+2\zeta'}\cdot\gamma|\Omega^{\ast}|
	\end{align*}
	where $C_{3}>0$ depends only on $\delta,\zeta,m$ , and $C_{4,m},C_{5,m},C_{6,m}>0$ are constants depending only on $m$. Moreover, if $l_{\max}$ is chosen large enough such that the second term on RHS of (\ref{eq:binary_hatTlmax_fro}) dominates and assume $\kappa_0^{4m}\mu_1^{4m}\rmax^{m}(\rmax\dmax+r^{\ast})\lesssim_{m}O(\dmin^{m-1})$, we get with probability at least $1-\frac{2}{d^*}$ that
	\begin{align*}
		\|\hat\T_{l_{\max}}-\T^{\ast}\|_{\ell_\infty}\leq& C_{6}\kappa_0^{2m}\mu_1^{2m}(\rmax^m/\dmin^{m-1})^{1/2}\Big(\frac{r^*+\dmax\rmax}{I}+\gamma|\Omega^*|\Big)^{1/2}\\
		\|\hat\S_{l_{\max}}-\S^{\ast}\|_{\ell_\infty}\leq& C_{7}\kappa_0^{2m}\mu_1^{2m}\rmax^{m/2}/\dmin^{(m-1)/2}\cdot\Big(\sqrt{(r^*+\dmax\rmax)/I} + |\Omega^*|^{1/2}\Big)+C_{8}
	\end{align*}
	where $C_{6}, C_{7}, C_{8}>0$ depend only on $\gamma, \delta, \zeta,m$. 
\end{theorem}
From Theorem \ref{thm:poisson}, after a properly chosen $l_{\max}$ iterations, we will obtain $\fro{\hat\T_{\l_{\max}} - \T^*}^2 = O(\frac{r^* + \dmax\rmax}{I/e^{\zeta}}+e^{2\zeta}\cdot\gamma |\Omega^{\ast}|)$.
As a special case when $|\Omega^*| = 0$, our result matches the previous result in Poisson tensor PCA in \cite{han2020optimal} that is rate optimal under the same requirements on the intensity parameter $I$. When there are outliers, the error for the estimation of $\T^*$ is further influenced by the outliers.

\textit{Initialization.} We shall adopt the initialization proposed in \cite{han2020optimal} with slight modification. The theoretical guarantee is summarized in the following lemma.
\begin{lemma}\label{lemma:init:poisson}
Suppose that Assumptions~\ref{assump:spikiness} and \ref{assump:poisson} hold. There exist absolute constants $c, C>0$ such that if $I\geq C\max\{\dmax,\minl^{-2}\sum_{i=1}^m(d_ir_i+d_i^-r_i)\rmax\}$, and the sparsity of $\S^*$ satisfies $|\Omega^*|\leq c\zeta^{-2}\minl^2\rmax^{-1}$, then
the output of Algorithm \ref{alg:poisson} satisfies the initialization requirement in Theorem \ref{thm:poisson} with probability at least $1 - 1/d^*$. 
\end{lemma}

\begin{algorithm}
	\caption{Initialization for Poisson RPCA}\label{alg:poisson}
	\begin{algorithmic}
		\STATE{Set $\wt\T = \log(\frac{\Y+1/2}{I})$.}
		\STATE{Let $\wt\T_0 = \opH(\wt\T)$.}
		\STATE{Return $\hat\T_0 = \textsf{Trim}_{\eta,\br}(\wt\T_0)$ with $\eta = 16\mu_1\fro{\wt\T_0}/(7\sqrt{d^*})$.}
	\end{algorithmic}
\end{algorithm}

\section{When Sparse Component is Absent}\label{sec:exact_lowrank}
In this section, we consider the special case when the sparse component is absent, i.e., $\S^{\ast}={\bf 0}$. For the exact low-rank tensor model, we observe that many conditions in Section~\ref{sec:theory} can be relaxed. A major difference is that the spikiness condition is generally not required for exact low-rank model. Consequently, the trimming step in Algorithm~\ref{algo:lowrank+sparse} is unnecessary.  Therefore, it suffices to simply apply the Riemannian gradient descent algorithm to solve for the underlying low-rank tensor $\T^{\ast}$. For ease of exposition, the procedure is summarized in Algorithm~\ref{algo:lowrank} (largely the same as Algorithm~\ref{algo:lowrank+sparse}).

\begin{algorithm}
\caption{Riemannian Gradient Descent for Exact Low-rank Estimate}\label{algo:lowrank}
\begin{algorithmic}
\STATE \textbf{Initialization: } $\hat\T_0$ $\in$ $\mfd$ and  stepsize $\beta>0$
\FOR{$l=0,1,\cdots, l_{\max}$}
\STATE{$\G_l = \nabla \frakL(\hat\T_l)$}
\STATE{$\W_l = \hat \T_l - \beta \pro_{\TT_l}\G_l$}
\STATE{$\hat\T_{l+1} = \opH(\W_l)$}
\ENDFOR
\STATE{\bf  Output:} $\hat\T_{l_{\max}}$
\end{algorithmic}
\end{algorithm}

Algorithm~\ref{algo:lowrank} runs fast and guarantees favourable convergence performances under weaker conditions than Theorem~\ref{thm:lowrank+sparse}.  Indeed, since there is no sparse component, only Assumption~\ref{assump:lowrank} is required to guarantee the convergence of Algorithm~\ref{algo:lowrank}. 
Similarly as Section~\ref{sec:theory}, the error of final estimate produced by Algorithm~\ref{algo:lowrank} is characterized by the gradient at $\T^{\ast}$. With a slightly abuse of notation, denote
$
    \errrank = \sup\nolimits_{\X\in\MM_{2\br},\|\X\|_{\rm F}\leq 1} \inp{\nabla \frakL(\T^*)}{\X}.\label{noise}
$

\begin{theorem}\label{main:thm}
Suppose Assumption \ref{assump:lowrank} holds with $\S^{\ast}={\bf 0}$ and $\BB_2^{\ast}=\{\T: \|\T-\T^{\ast}\|_{\rm F}\leq c_{0,m}\usigma, \T\in\MM_{\br}\}$ for a small constant $c_{0,m}>0$ depending on $m$ only, also suppose $1.5b_lb_u^{-2}\leq 1$ and $0.75b_lb_u^{-1} \geq \delta^{1/2}$ for some $\delta\in (0,1]$ and the stepsize $\beta\in[0.4b_lb_u^{-2},1.5b_lb_u^{-2}]$ in Algorithm \ref{algo:lowrank}. Assume
\vspace{-0.2cm}
\begin{enumerate}[label=(\alph*)]
\item {\it Initialization}: $\|\hat\T_0-\T^*\|_{\rm F} \leq \usigma\cdot c_{1,m} \delta\rmax^{-1/2}$
\vspace{-0.3cm}
\item {\it Signal-to-noise ratio}: $\errrank/\usigma \leq c_{2,m}\delta^2\rmax^{-1/2}$
\end{enumerate}
\vspace{-0.2cm}
where $c_{1,m},c_{2,m} > 0$ are small constants depending only on $m$. 
Then for all $l=1,\cdots,l_{\max}$,
\begin{align*}
	\fro{\hat\T_{l} - \T^*}^2 \leq (1-\delta^2)^l \fro{\hat\T_0-\T^*}^2 +C_{\delta}\textsf{Err}_{2\r}^2
\end{align*}
where $C_{\delta}>0$ is a constant depending only on $\delta$. Then after at most $l_{\max} \asymp \log \big(\usigma/\errrank\big)$ iterations (also depends on $b_l, b_u, m, \rmax$ and $\beta$), we get
$$
\fro{\hat\T_{l_{\max}} - \T^*} \leq C\cdot \textsf{Err}_{2\br},
$$
where the constant $C>0$ depends on only $b_l, b_u, m, \rmax$ and $\beta$.
\end{theorem}

Note that Theorem~\ref{main:thm} holds without spikiness condition in contrast with Theorem~\ref{thm:lowrank+sparse}. 
It makes sense for the model has no missing values or sparse corruptions. The assumptions on loss function are also weaker (e.g., no need to be an entry-wise loss or entry-wisely smooth) than those in Theorem~\ref{thm:lowrank+sparse}. As a result, Theorem~\ref{main:thm} is also applicable to the low-rank tensor regression model among others. See \citep{han2020optimal, chen2019non,xia2020inference} and references therein. 
The initialization and signal-to-noise conditions are similar to those in Theorem~\ref{thm:lowrank+sparse}, e.g., by setting $|\Omega^{\ast}|=\alpha=0$ there. In addition, the error of final estimate depends only on $\textsf{Err}_{2\br}$. Interestingly, the contraction rate does not depend on the condition number $\kappa_0$. 

\textit{Comparison with existing literature}
In \citep{han2020optimal}, the authors proposed a general framework for exact low-rank tensor estimation based on regularized jointly gradient descent on the core tensor and associated low-rank factors. Their method is fast and achieves statistical optimality in various models. In contrast, our algorithm is based on Riemannian gradient descent, requires no regularization and also runs fast. An iterative tensor projection algorithm was studied in \citep{yu2016learning}. But their method only applies to tensor regression. Other notable works focusing only on tensor regression include \citep{zhang2020islet, zhou2013tensor, hao2020sparse, sun2017provable,li2018tucker,pan2018covariate}. 
A general projected gradient descent algorithm was proposed in \citep{chen2019non} for generalized low-rank tensor estimation. For Tucker low-rank tensors, their algorithm is similar to our Algorithm~\ref{algo:lowrank} except that they use vanilla gradient $\G_l$ while we use the Riemannian gradient $\calP_{\TT_l}\G_l$. As explained in Section~\ref{sec:method}, using the vanilla gradient can cause heavy computation burdens in the subsequent steps.  
Riemannian gradient descent algorithm for tensor completion was initially proposed by \citep{kressner2014low}. They focused only on tensor completion model and did not investigate its theoretical guarantees. 
Recently in \citep{cai2020provable}, the Riemannian gradient descent algorithm is applied for noiseless tensor regression and its convergence analysis is proved.

\section{Numerical Comparisons with Existing Methods}\label{sec:numerical}

We test the performances of our algorithms on synthetic datasets, specifically for the four applications studied in Section~\ref{sec:app}. Due to page limit, here we only present the comparative simulation results with competing methods on  SG-RPCA and binary tensor learning. The comprehensive simulation results and the performance of proposed BIC-type criterion are collected in the supplementary file. 

\textit{Choice of parameters}.  
First of all, for the stepsize, we choose $\beta$ that lies in the range we provide in the theories. For the spikiness parameter $\mu_1$, our theorem only requires it to be larger than the truth,  we can initially set $\mu_1 =2^m+ \log(\dmax)$ and gradually increase it by a factor of $2$ if the algorithm fails to converge. For the rank $\br$ and sparsity $\alpha$, we treat them as given or select them by the BIC-type criterion (\ref{eq:BIC}). Note that $\gamma$ {\it only} plays a role in the term $\gamma\alpha=:\alpha'$, i.e., the desired sparsity level. If the true $\alpha$ is unknown, then the BIC-type criterion actually searches for $\alpha'$ in which case $\gamma$ is irrelevant and we simply set it to $1$; if the true $\alpha$ is known, then we initially set $\gamma=1.1$ or $2$ and gradually increase it by a factor of $1.5$ if the algorithm fails to converge. 
 In the case of SG-RPCA and tensor PCA with heavy tailed noise, $\kprune = \infty$; and in the case of binary tensor, we follow the choice in \cite{wang2020learning} and set $\kprune = 1$.



For the first experiment, we compare RGrad and PGD \citep{chen2019non} on noisy tensor decomposition without outliers, i.e., no $\alpha$ or $\gamma$. The low-rank tensor $\T^*\in\RR^{d\times d\times d}$ with $d = 300$ and Tucker rank $\br = (2,2,2)^\top$ is generated from the HOSVD of a trimmed standard normal tensor. The noise tensor $\Z$ has i.i.d. entries sampled from ${\rm N}(0,\sigma_z^2)$.  Both algorithm terminate either when the relative error $\|\hat \T_l-\hat\T_{l-1}\|_{\rm F}/\|\hat\T_{l}\|_{\rm F} < 0.001$ or the maximum iteration (100) is reached.  
The noise level $\sigma_z$ ranges from 0.01 to 0.05.  For each fixed $\sigma_z$,  10 random instances for both algorithms. are conducted. The result is displayed in the left plot of Figure \ref{fig:rgradvspgd}.  We see that the statistical performance of RGrad and PGD, when there is no outliers,  are similar.  However,  the right panel of Figure \ref{fig:rgradvspgd} shows that the per-step runtime using RGrad is only roughly 1/4 of the per-step runtime using PGD. This illustrates the computational efficiency of using RGrad over PGD.
 
For the second experiment,  we set $d=100$ and $\sigma_z = 0.01$.  Given a sparsity level $\alpha\in(0,1)$, the entries of sparse tensor $\S^{\ast}$ are i.i.d. sampled from $\textsf{S}_{\textsf{amp}}\times{\rm Be}(\alpha)\times {\rm N}(0,1)$, which ensures $\S^{\ast}\in \SS_{O(\alpha)}$ with high probability.  Here the constant $\textsf{S}_{\textsf{amp}}$ is set as $0.1$ or $1$ modeling the two cases of small magnitude and large magnitude, respectively.  The sparsity $\alpha$ is varied between 0.025 and 0.1, and $\gamma=1.1$ for RGrad.  We refer to \cite{gu2014robust}'s method as convex relaxation and \cite{lu2016tensor}'s method as tubal-tRPCA. The results are displayed in Figure \ref{fig:rgradvsothers}. Here RGrad (BIC) means that $\alpha$ is treated as unknown and selected by BIC-type criterion (\ref{eq:BIC}). It shows that the proposed BIC-type criterion works nicely in SG-RPCA. We can see the tubal-tRPCA performs poorly due to the ignorance of the low rank structure along the third direction. When the magnitude of the outliers is small, the performance of PGD and RGrad are similar. However, when the magnitude of the outliers is large, PGD performs poorly since it cannot deal with the outliers. The performance of convex relaxation is also worse than RGrad since it unfolds a tensor into an unbalanced matrix, and is statistically sub-optimal. 

For the third experiment, we compare RGrad and PGD on binary tensor learning.  Here $d=100$ and $\br=(2,2,2)^{\top}$.  The incoherent $\T^{\ast}$ is generated such that $\|\T^{\ast}\|_{\ell_\infty}\approx 5$.  The entries of $\S^*$ are i.i.d. sampled from $\textsf{S}_{\textsf{amp}}\times{\rm Be}(\alpha)$. Here the constant $\textsf{S}_{\textsf{amp}}$ is set as $1$ or $10$ modeling the two cases of small magnitude and large magnitude, respectively. The sparsity $\alpha$ is varied between $0.005$ and $0.02$, and $\gamma=1.1$ for RGrad.  Initialization is obtained by the algorithm in the supplement.  The link function is set to $p(x) = (1+e^{-x/5})^{-1}$. Here RGrad (BIC) has a similar meaning as above. 
The results are displayed in Figure \ref{fig:binary_compare}.  When the magnitude of outliers is small ($\textsf{S}_{\textsf{amp}} = 1$), the performance of RGrad and PGD are comparable.  However, when the magnitude of the outliers become large ($\textsf{S}_{\textsf{amp}} = 10$), PGD cannot handle them well while our proposed algorithm has a much better performance. 

\begin{figure}
\centering
	\begin{subfigure}[b]{.98\linewidth}
		\includegraphics[width=0.45\textwidth]{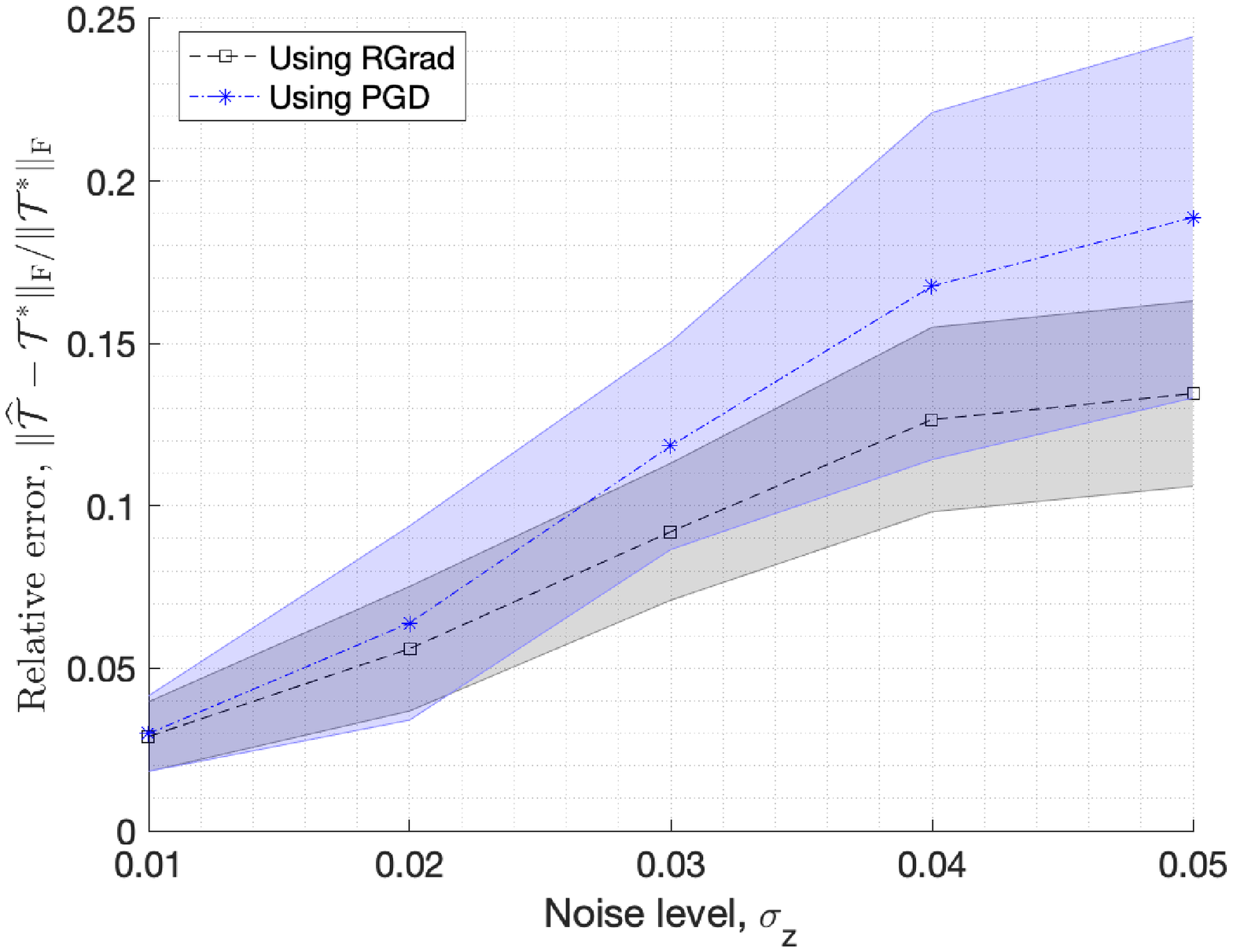}
		\includegraphics[width=0.45\textwidth]{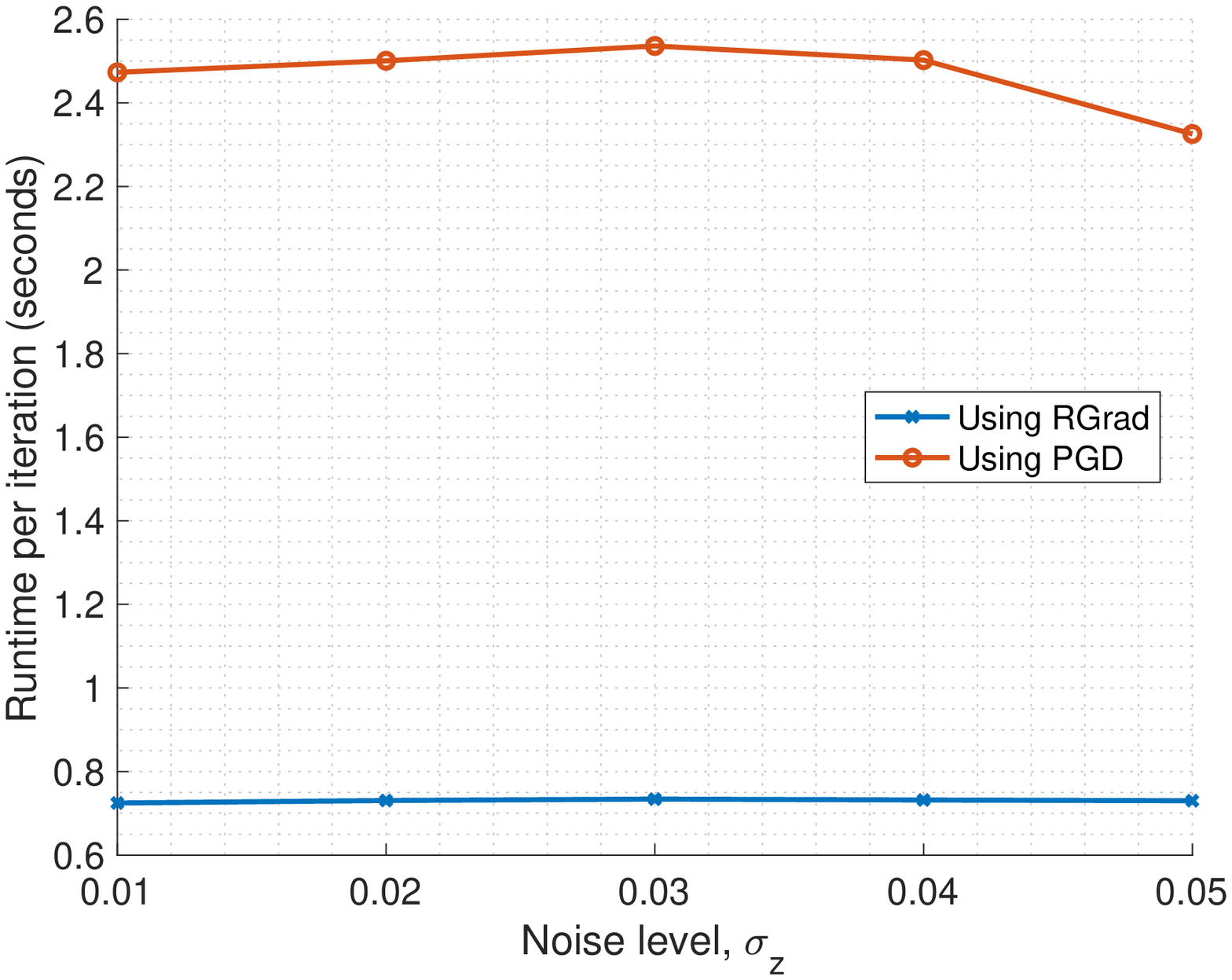}
		\caption{Tensor PCA without outliers.  Left: Error bar of RGrad and PGD for 10 random instances.  Right: Per-step runtime of RGrad and PGD.}
		\label{fig:rgradvspgd}
	\end{subfigure}
	\begin{subfigure}[b]{.98\linewidth}
		\includegraphics[width=0.45\textwidth]{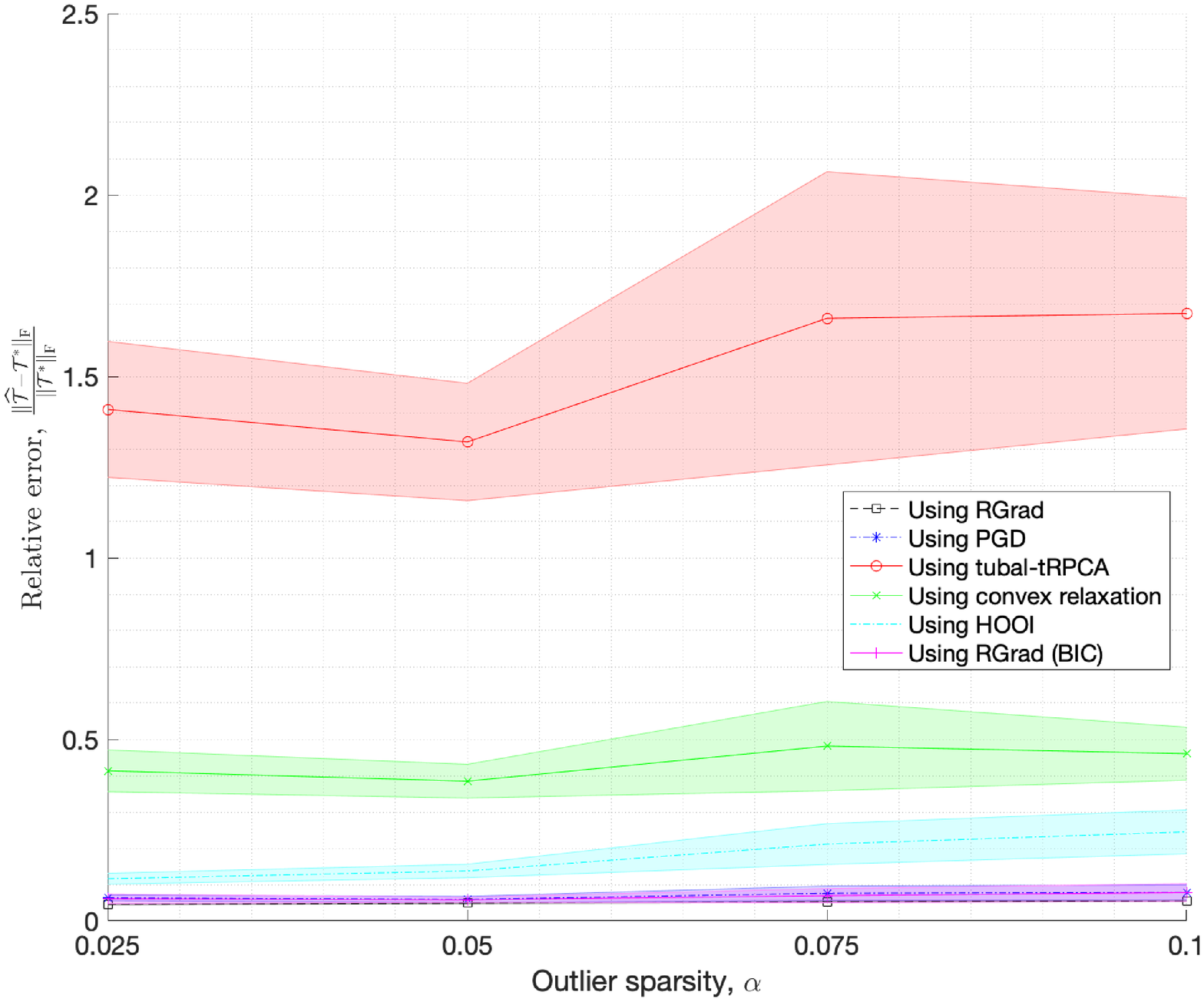}
		\includegraphics[width=0.45\textwidth]{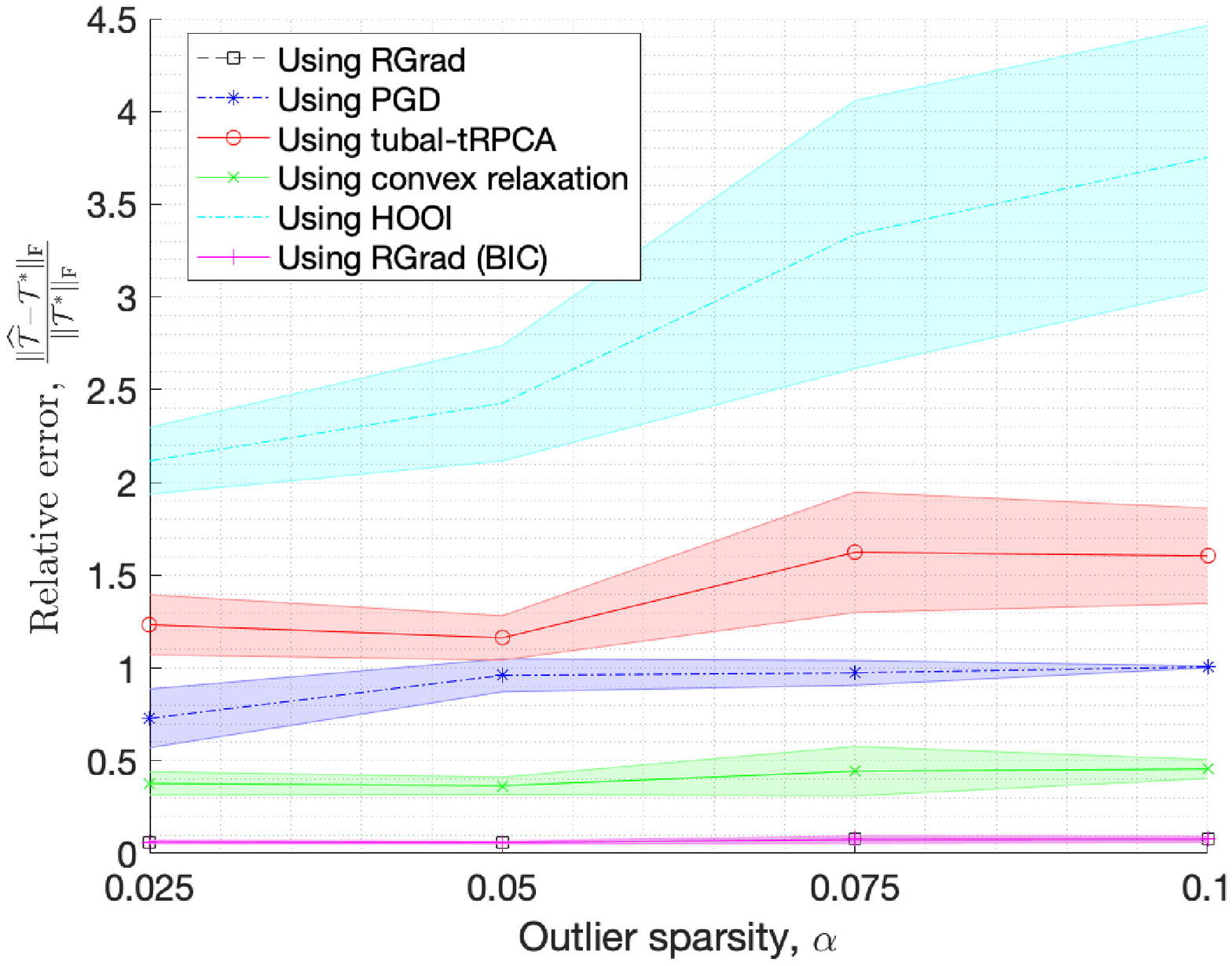}
		\caption{SG-RPCA.  Left: Small amplitude of outlier with $\textsf{S}_{\textsf{amp}} = 0.1$; BIC suggested $\alpha$: $\{0.025,0.04,0.065,0.09\}$. Right: Large amplitude of outlier with $\textsf{S}_{\textsf{amp}} = 1$; BIC suggested $\alpha$: $\{0.025,0.05,0.075,0.1\}$.}
		\label{fig:rgradvsothers}
	\end{subfigure}
		\begin{subfigure}[b]{.98\linewidth}
		\includegraphics[width=0.45\textwidth]{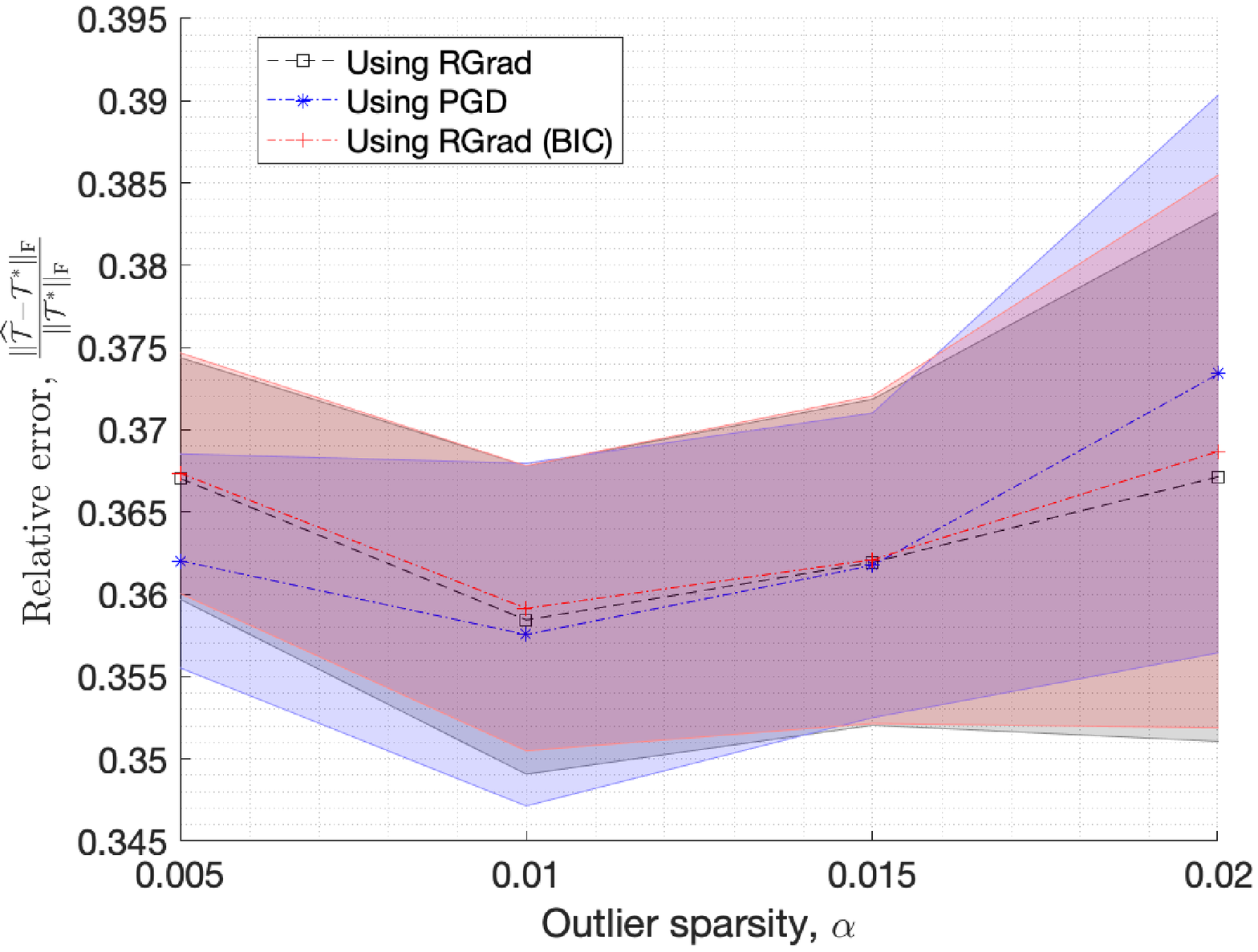}
		\includegraphics[width=0.45\textwidth]{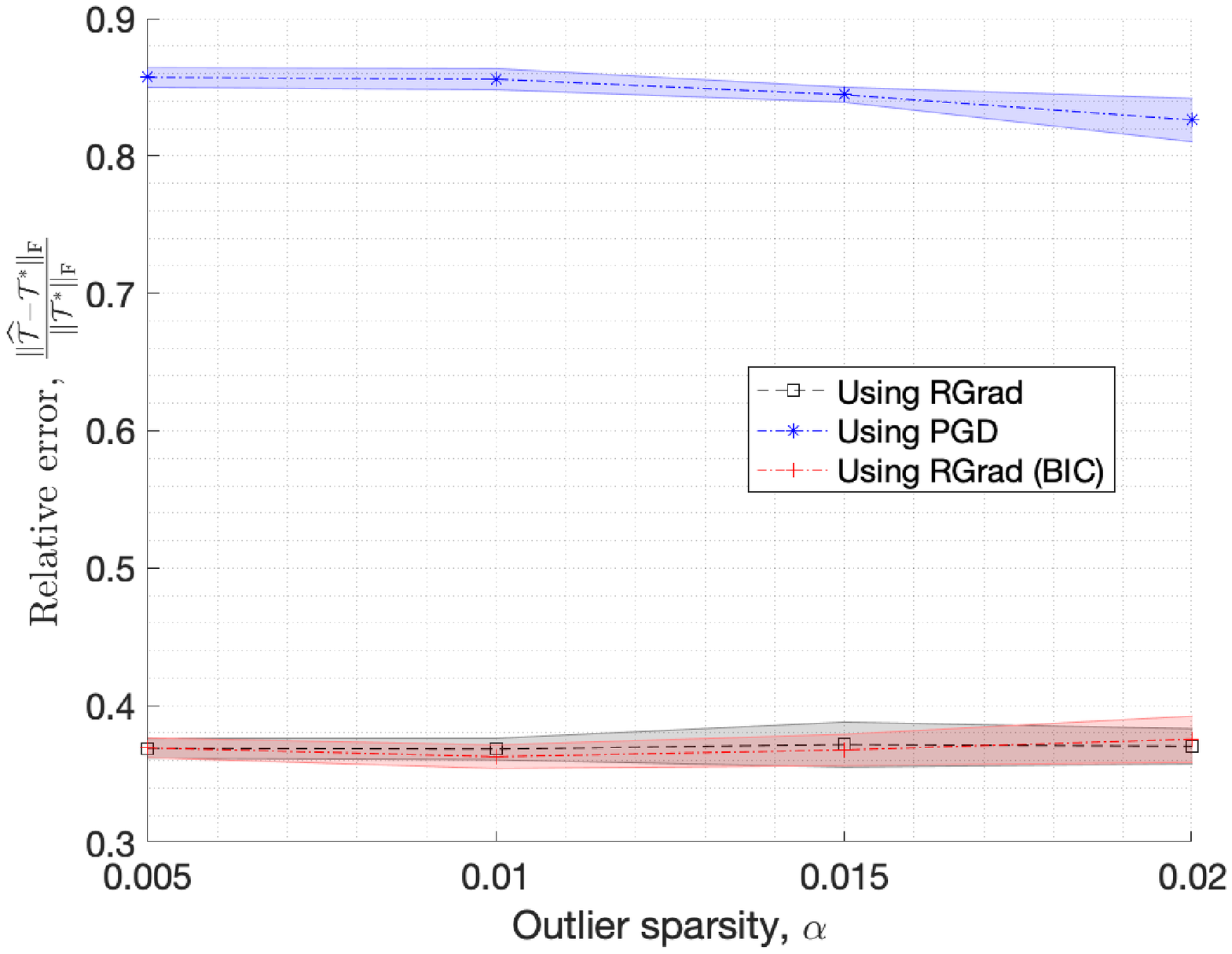}
		\caption{On binary tensor learning with outliers. Left: Small amplitude of outlier with $\textsf{S}_{\textsf{amp}} = 1$; BIC suggested $\alpha$: $\{0.003,0.008,0.013,0.015\}$.  Right: Large amplitude of outlier with $\textsf{S}_{\textsf{amp}} = 10$; BIC suggested $\alpha$: $\{0.003,0.007,0.012,0.016\}$.}
		\label{fig:binary_compare}
	\end{subfigure}
\caption{Comparison of RGrad, PGD \citep{chen2019non},  convex \citep{gu2014robust}, tubal-tRPCA \citep{lu2016tensor} and HOOI \citep{zhang2018tensor}. }
\end{figure}

\section{Real Data: International Commodity Trade Flows}\label{sec:real_app}
We collected the international commodity trade data from the API provided by UN website \textit{https://comtrade.un.org}. The dataset contains the monthly information of imported commodities by countries from Jan. 2010 to Dec. 2016 ($84$ months in total). For simplicity, we focus on $50$ countries among which $35$ are from Europe, $9$ from America, $5$ from Asia\footnote{Egypt is at the cross of Eastern Africa and Western Asia. For simplicity, we treat it as an Asian country. In addition, Turkey is treated as an Eastern European country rather than a Western Asian country. } and $1$ from Africa. All the commodities are classified into $100$ categories based on the $2$-digit HS code (\textit{https://www.foreign-trade.com/reference/hscode.htm}). Thus, the raw data is a $50\times 50\times 100\times 84$ tensor. At any month and for any category of commodity,  there is a directed and weighted graph of size $50\times 50$ depicting the trade flow between countries. The international trade has cyclic pattern annually. Since we are less interested in the time domain, we eliminate the fourth dimension by simply adding up the entries. Finally, we end up with a tensor $\A$ of size $50\times 50\times 100$. 

In Figure~\ref{fig:trade_flow}, circular plots are presented for illustrating the special trade patterns of some commodities. The countries are grouped and coloured by continent, i.e., Europe by red, Asia by green, America by blue and Africa by black. The links represent the directional trade flow between nations and are coloured based on the starting end of the link. The position of starting end of the link is shorter than the other end to give users the feeling that the link is moving out. The thickness of link indicates the volume of trade. From the top-left plot, we observe that Japan imports a large volume of tobacco related commodities; Germany is the largest exporter; Poland and Brazil are the second and third largest exporter; USA both import and export a large quantity of tobacco commodity. The top-right plot shows that USA and Canada import and export large volumes of mineral fuels; Malaysia exports lots of mineral fuels to Japan; Algeria exports a large quantity of miner fuels which plays the major role of international trade of this Africa country. The middle-left plot shows that Portugal is the largest exporter of Cork, and European countries are the major exporter and importer of this commodity. From the middle-right plot, we observe that Pakistan is the major exporter of Cotton in Asia; the European countries Turkey, Italy and Germany all export and import large volumes of cotton; USA exports a great deal of cotton to Mexico,Turkey and Philippines. The bottom-left plot shows that Malaysia and Belgium are the largest exporter of Tin and USA is the major importer. Finally, the bottom-right plot shows that Switzerland is the single largest exporter of clocks and watches; USA is the major importer; France and Germany both export and import large quantities of clocks and watches. 

\begin{figure}
\centering
	\begin{subfigure}[b]{.98\linewidth}
		\includegraphics[width=0.45\textwidth]{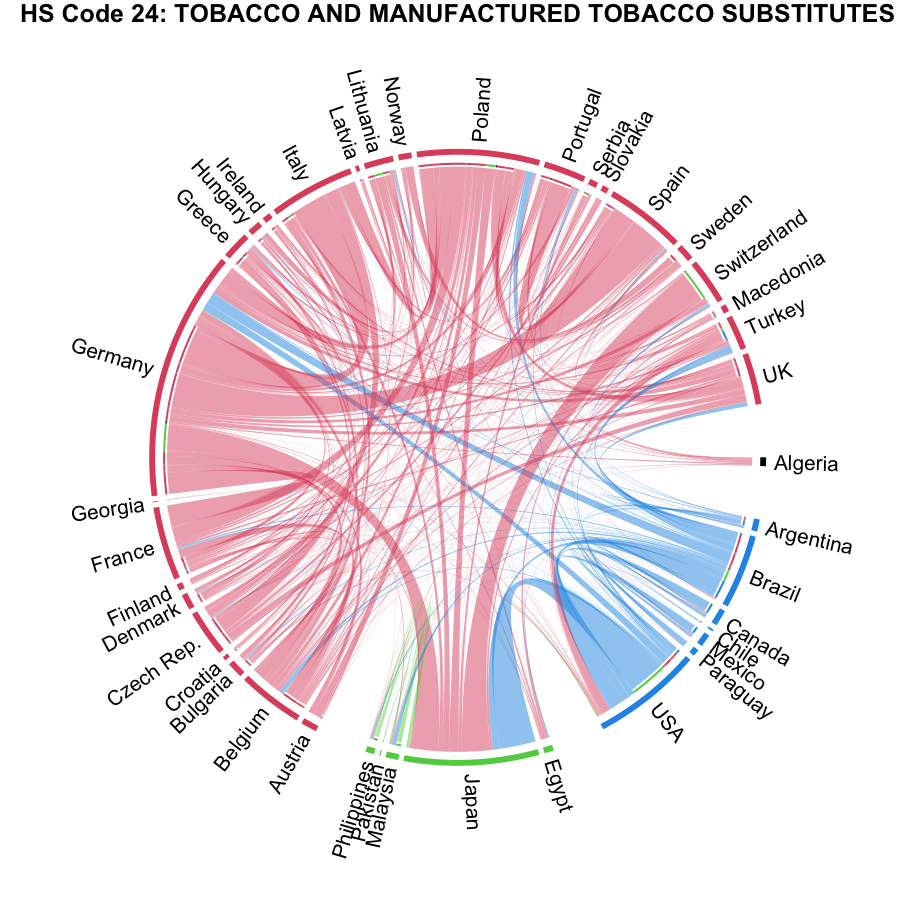}
		\includegraphics[width=0.45\textwidth]{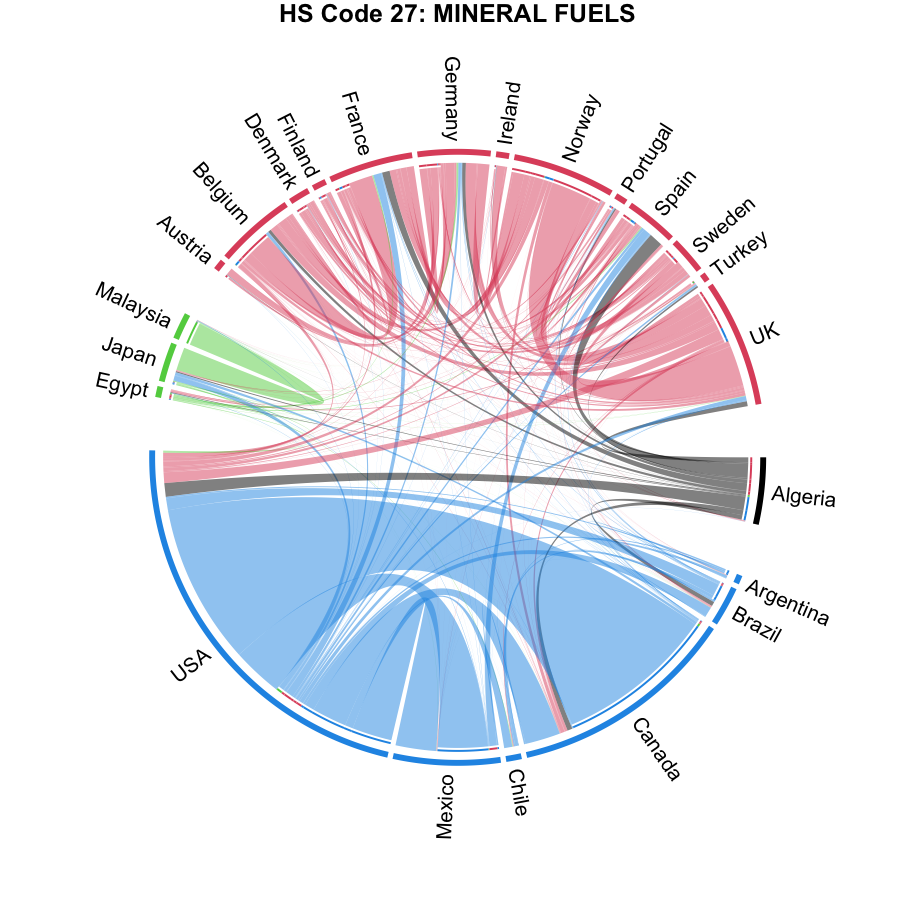}
	\end{subfigure}
	\begin{subfigure}[b]{.98\linewidth}
		\includegraphics[width=0.45\textwidth]{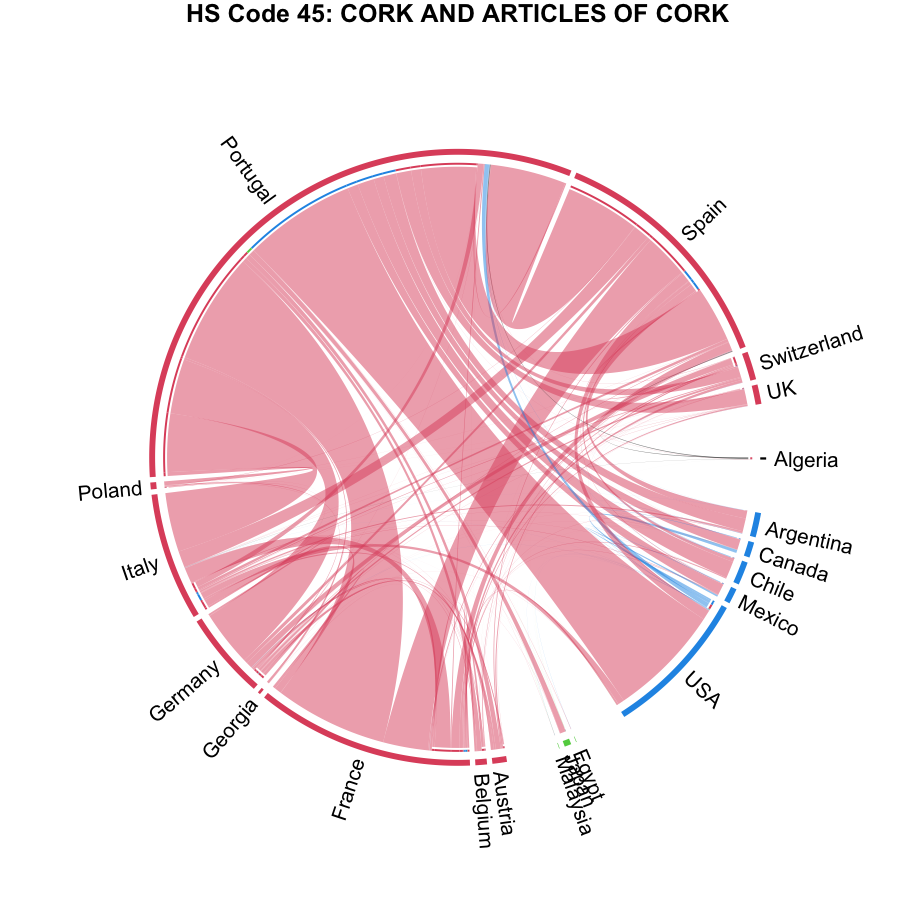}
		\includegraphics[width=0.45\textwidth]{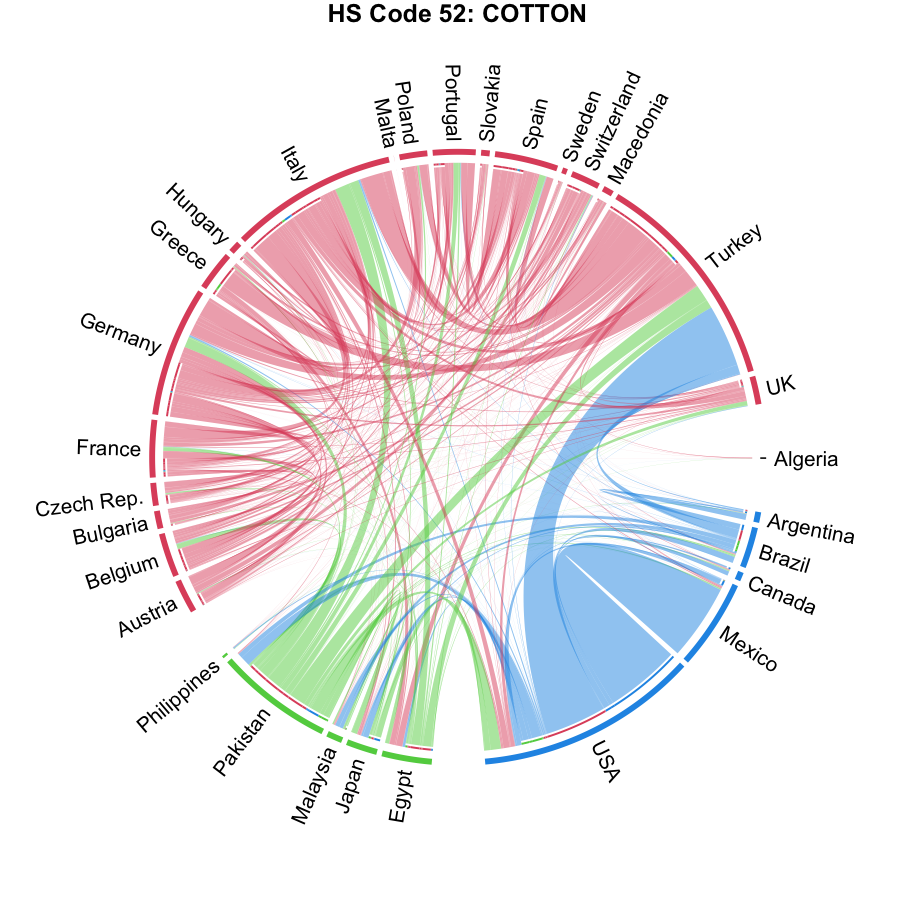}
	\end{subfigure}
	\begin{subfigure}[b]{.98\linewidth}
		\includegraphics[width=0.45\textwidth]{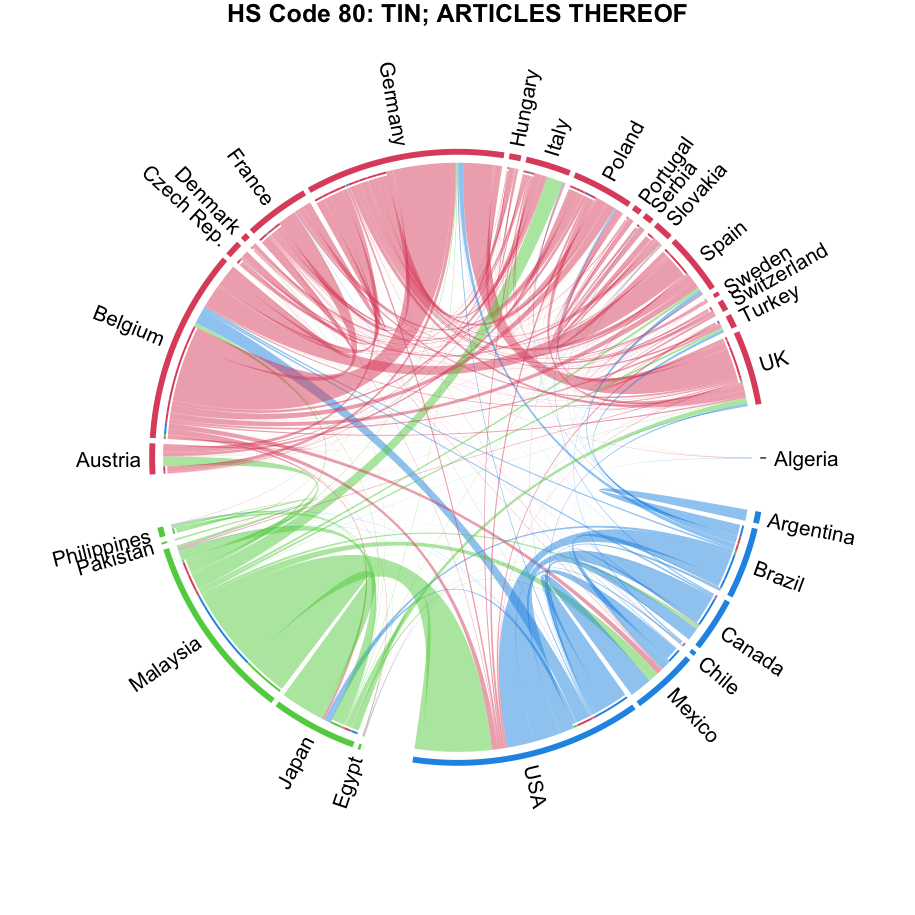}
		\includegraphics[width=0.45\textwidth]{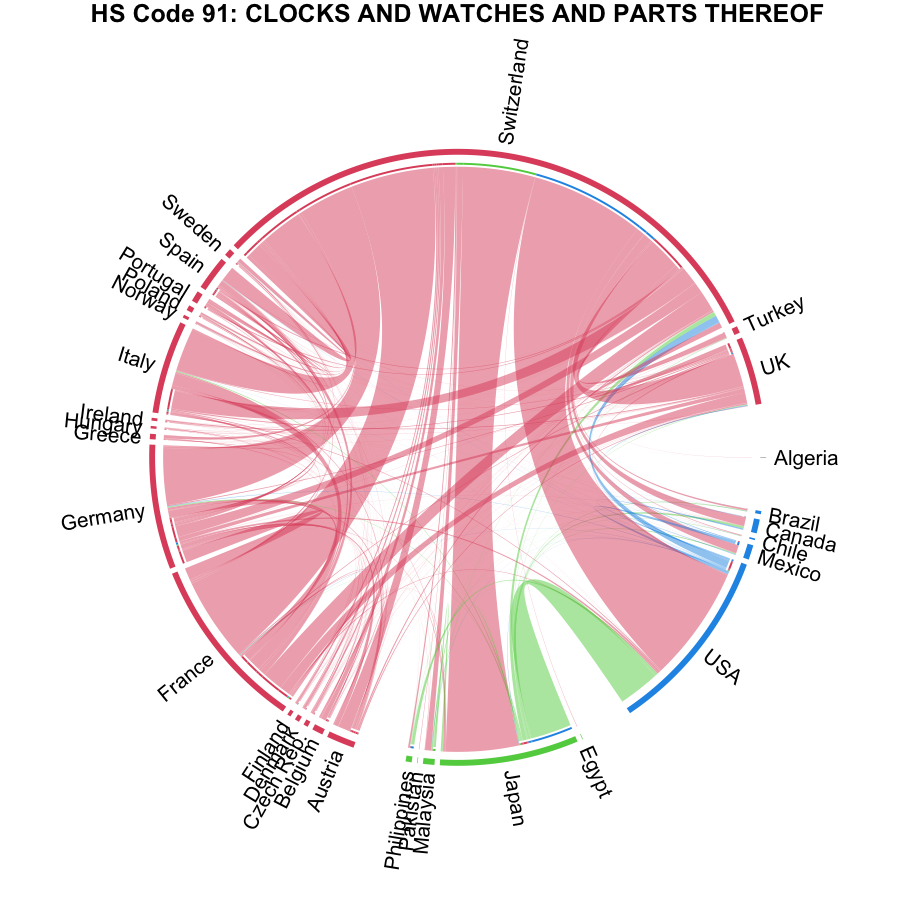}
	\end{subfigure}
\caption{International Commodity Trade Flow. Import and export flow of some commodities. }
\label{fig:trade_flow}
\end{figure}

We implement the  SG-RPCA framework as Section~\ref{sec:rpca} to analyze the tensor $\log(1+\A)$, where a logarithmic transformation helps shrink the extremely large entries. 
The Tucker ranks are set as $(3,3,3)$, although most of the results seem insensitive to the ranks so long as they are bounded by $5$ (BIC values within a $1\%$ discrepancy). We apply the BIC-type criterion (\ref{eq:BIC}) 
and the detailed BIC result is postponed to the supplement. It suggests that any $\alpha$ between $0.002$ and $0.03$ yield similar (within 0.1\% discrepancy) BIC values. See also top left and right panel in Figure~\ref{fig:trade_flow_nations}. The algorithm is initialized by a HOSVD, which finally produces a low-rank $\hat\T$ and a sparse $\hat \S$. We shall use $\hat\T$ to uncover underlying relations among countries, and $\hat\S$ to examine distinctive trading patterns of certain commodities. 

In particular, the singular vectors of $\hat\T$ are utilized to illustrate the community structure of nations. Note that the $1$st-dim and $2$nd-dim singular vectors of $\hat\T$ are distinct because the trading flows are directed. We observe that the $2$nd-dim singular vectors often render better results. Then, a procedure of multi-dimensional scaling is adopted to visualize the rows of these singular vectors. We note that, though the BIC-type criterion suggests an $\alpha\in[0.002, 0.03]$, intriguing phenomenons are observed for larger values of $\alpha$. 
The results are presented in Figure~\ref{fig:trade_flow_nations} for four choices of $\alpha\in\{0.003, 0.03, 0.1, 0.3\}$. All the plots in Figure~\ref{fig:trade_flow_nations} reveal certain degrees of the geographical relations among countries. It is reasonable since regional trade partnerships generally dominate the inter-continental trade relations. The European countries (coloured in blue) are mostly separated from the others. Overall, countries from America (coloured in red) and Asia (coloured in magenta) are less separable especially when $\alpha$ is large. For small $\alpha$ like $0.003$ or $0.03$, the $5$ Asian countries are clustered together and the major $8$ American nations lie on the top-right corner of the plot. The two geographically close African countries Algeria and Egypt are also placed together in the top-left plot of Figure~\ref{fig:trade_flow_nations}, as is the case with the Western European nations such as United Kingdom, Spain, France, Germany and Italy. 

Figure~\ref{fig:trade_flow_nations} show that the low-rank estimate is sensitive to the larger sparsity ratio. Interesting shifts appear as $\alpha$ increases. Indeed, the geographical relations become a less important factor but the economic similarity plays the dominating role. 
For instance, some Asian and American nations split and merge into two clusters. The three large economies US, Canada and Japan are merged into one cluster, while the other small and less-developed Asian and American countries are merged into another cluster. It may be caused by that these three large economies are better at advanced technology and share similar structures in exporting high end commodities. Moreover, as $\alpha$ increases, the African country Algeria moves closer to the less-developed American and Asian nations. All these nations including Algeria rely heavily on exporting natural resources even if Algeria is geographically far from the others. Another significant shift is that the European countries split into two clusters as $\alpha$ increases. Moreover, one cluster comprising those wealthy and advanced Western European countries move closer to the group of US, Japan and Canada. These countries have close ties in trading high end products and components, although they belong to distinct continents. The other cluster includes mostly the Central and Eastern European countries, among which regional trade flows are particularly intense. Interestingly, there are two outlier countries Ireland (north-western Europe) and Antigua and Barbuda (a small island country in middle America). They do not merge into any clusters. The magnitudes of coordinates of these two points suggest that their international trade is not active.  

\begin{figure}
\centering
	\begin{subfigure}[b]{.98\linewidth}
		\includegraphics[width=0.5\textwidth]{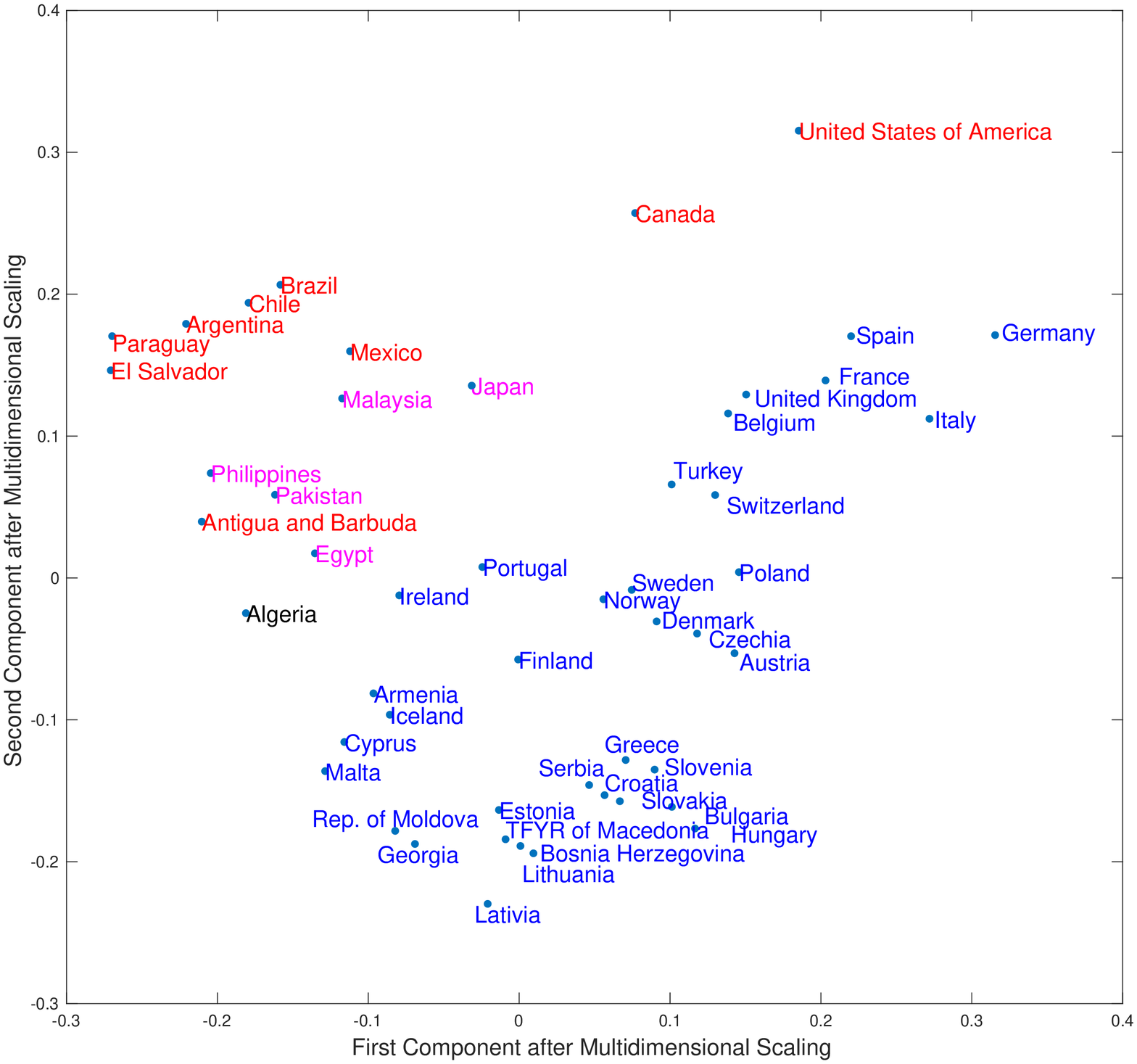}
		\includegraphics[width=0.5\textwidth]{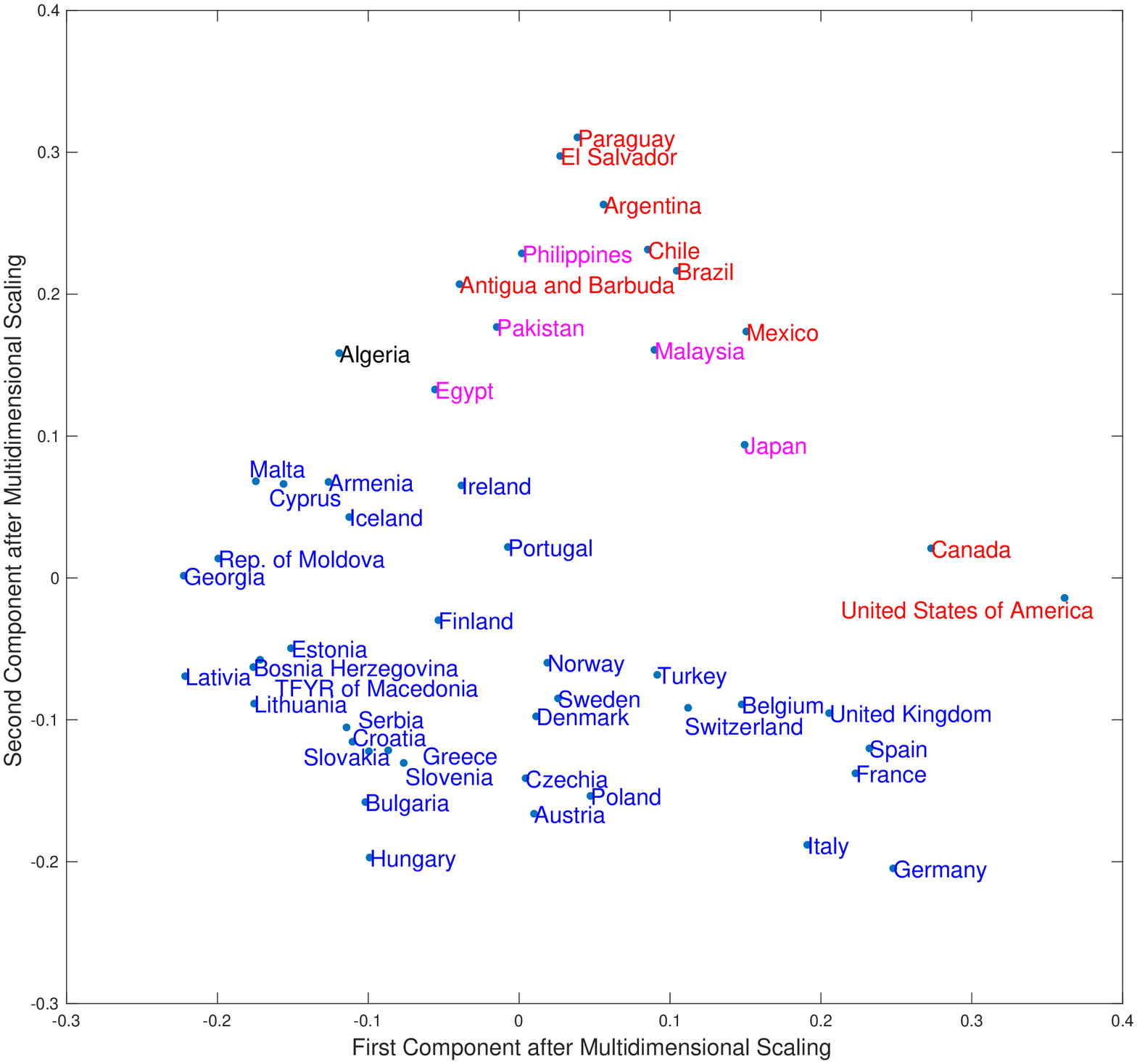}
		\includegraphics[width=0.5\textwidth]{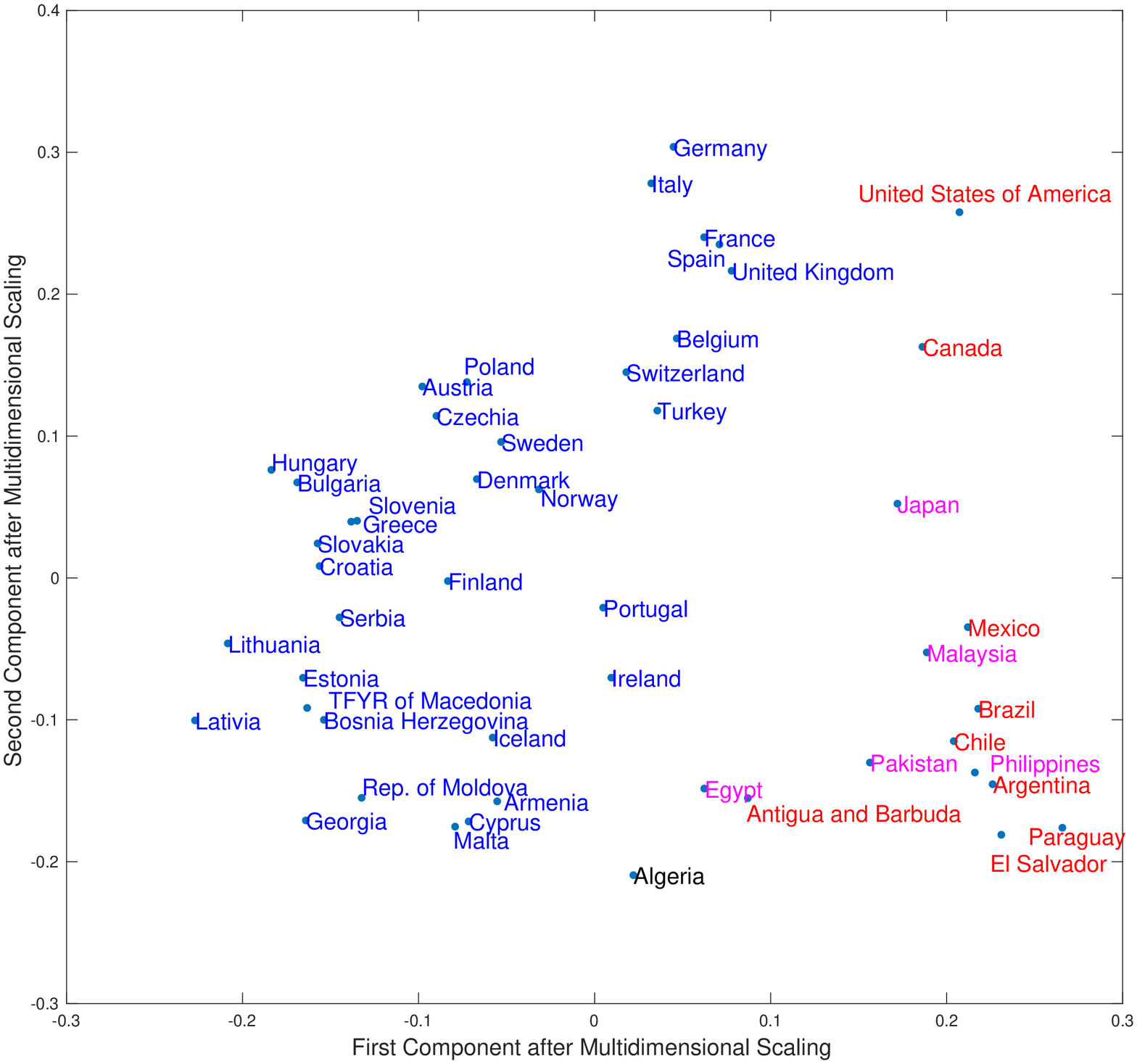}
		\includegraphics[width=0.5\textwidth]{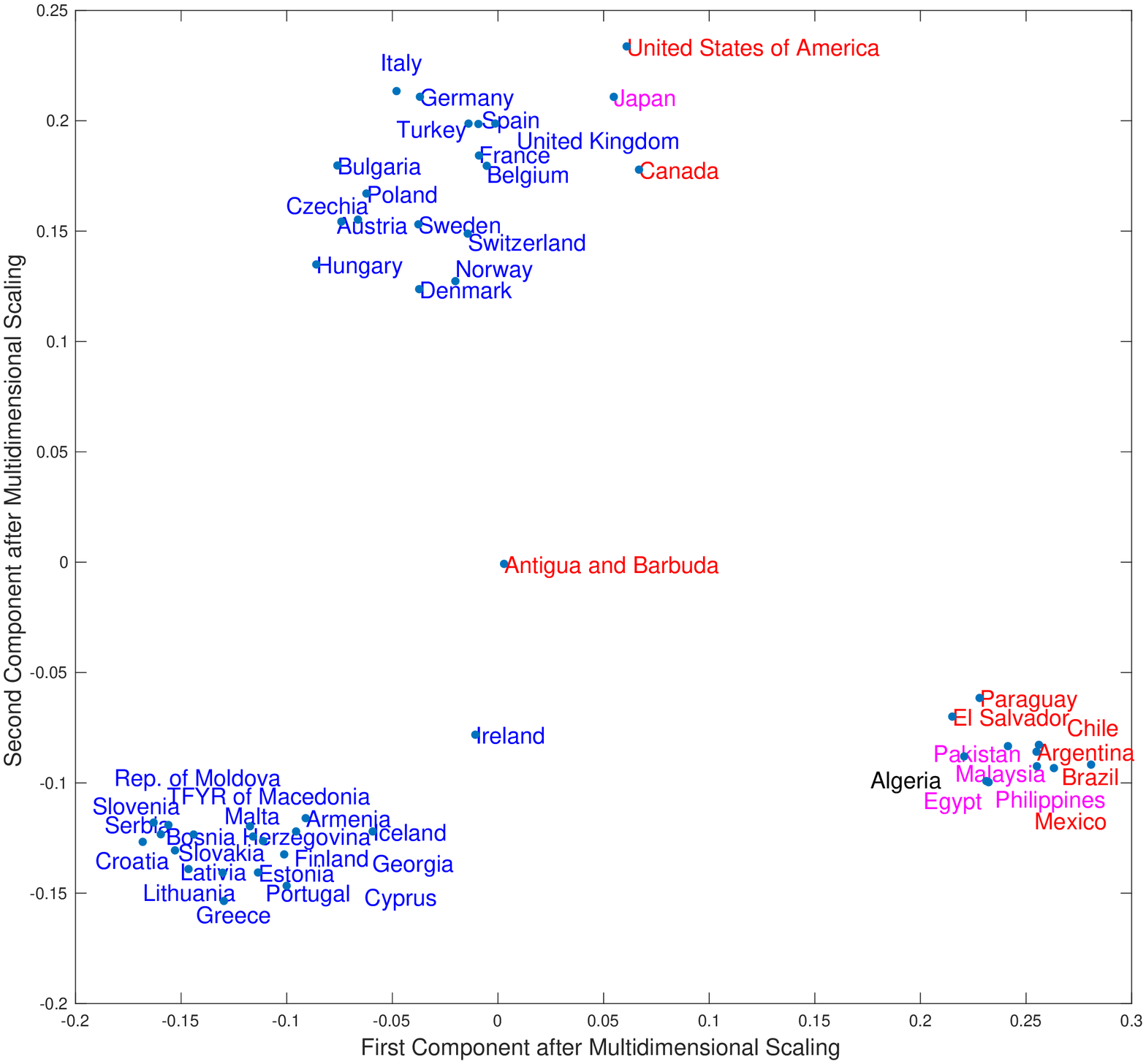}
		\caption{Sparsity level: top-left, $\alpha=0.003$; top-right, $\alpha=0.03$; bottom-left, $\alpha=0.1$; bottom-right, $\alpha=0.3$.}
	\end{subfigure}
\caption{Visualize the relations between countries by spectral estimates. Tucker ranks are set as $(3,3,3)$ and the low-rank tensor is estimated by SG-RPCA as in Section~\ref{sec:rpca} with HOSVD initialization. Blue coloured countries are from Europe, red from America, magenta from Asia and black from Africa. When $\alpha$ is small, clusters have strong implications on the geographical closeness between nations; As $\alpha$ becomes large, economic structures become the major factor in that large and advanced economies tend to merge, and so do low end economies and natural resource reliant economies. But our BIC-type criterion suggests to choose $\alpha\in[0.002, 0.03]$.}
\label{fig:trade_flow_nations}
\end{figure}

We now look into the slices of the sparse estimate $\hat\S$ and investigate the distinctive trading patterns of certain commodities. As the sparsity ratio $\alpha$ grows, the slices of $\hat\S$ become denser whose patterns are more difficult to examine. For better exposition, we mainly focus on small values of $\alpha$ like $0.03$. The results are presented in Figure~\ref{fig:trade_flow_comm}.  The top-left plot shows that Algeria exports  exceptionally large volumes of mineral fuels that can not be explained by the low-rank tensor estimate. Similarly, based on the top-right plot, we observe that the exact low-rank tensor PCA fails to explain the trading export of cork by Portugal. Fortunately, these interesting and significant trading patterns can be easily captured by the additional sparse tensor estimate. The bottom-left and bottom-right plots of Figure~\ref{fig:trade_flow_comm} showcase the unusual exports of cotton by Pakistan and exports of Tin by Malaysia, respectively. These findings echo some trading patterns displayed in Figure~\ref{fig:trade_flow}. 

\begin{figure}
\centering
	\begin{subfigure}[b]{0.98\linewidth}
		\includegraphics[width=0.5\textwidth]{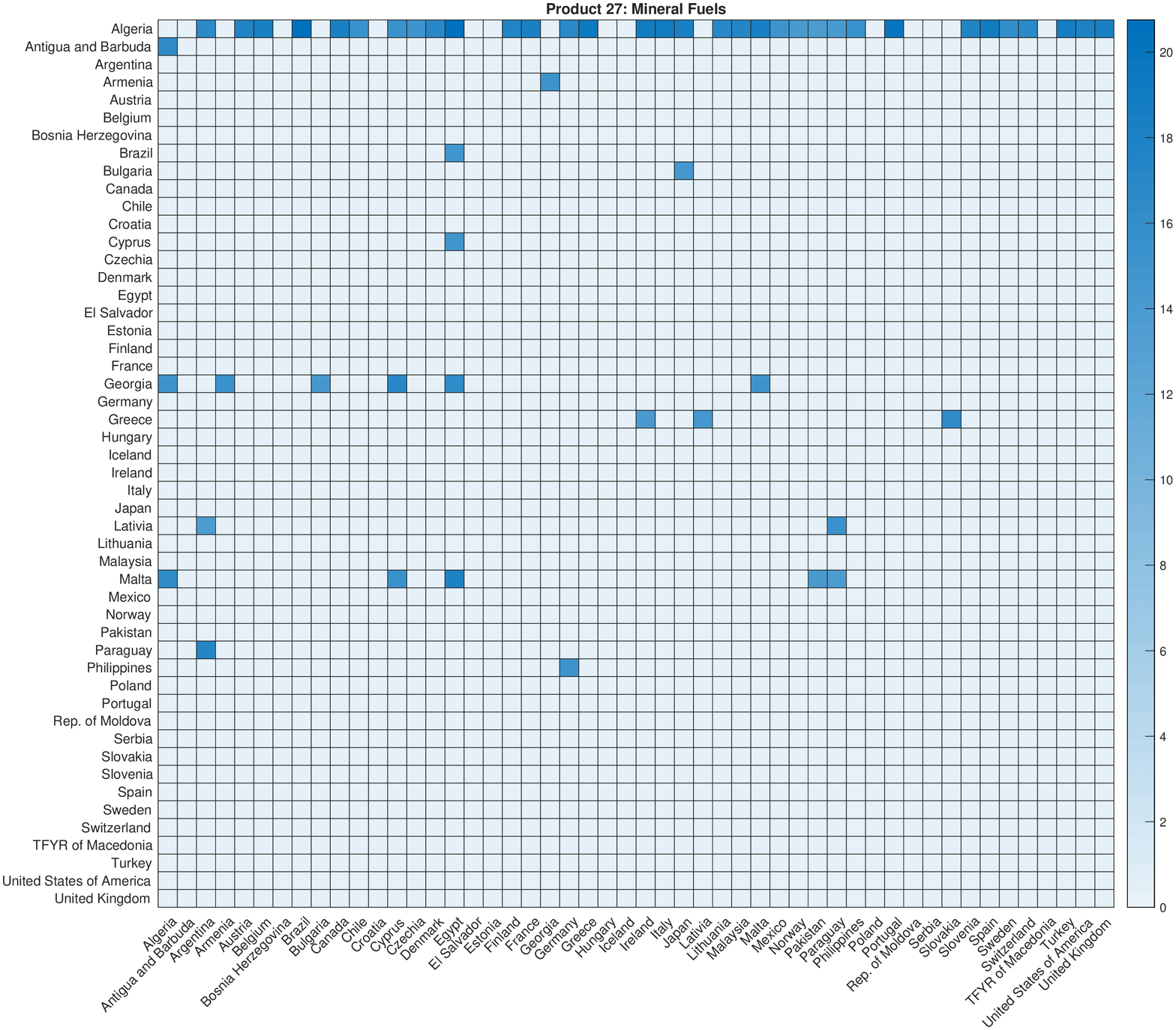}
		\includegraphics[width=0.5\textwidth]{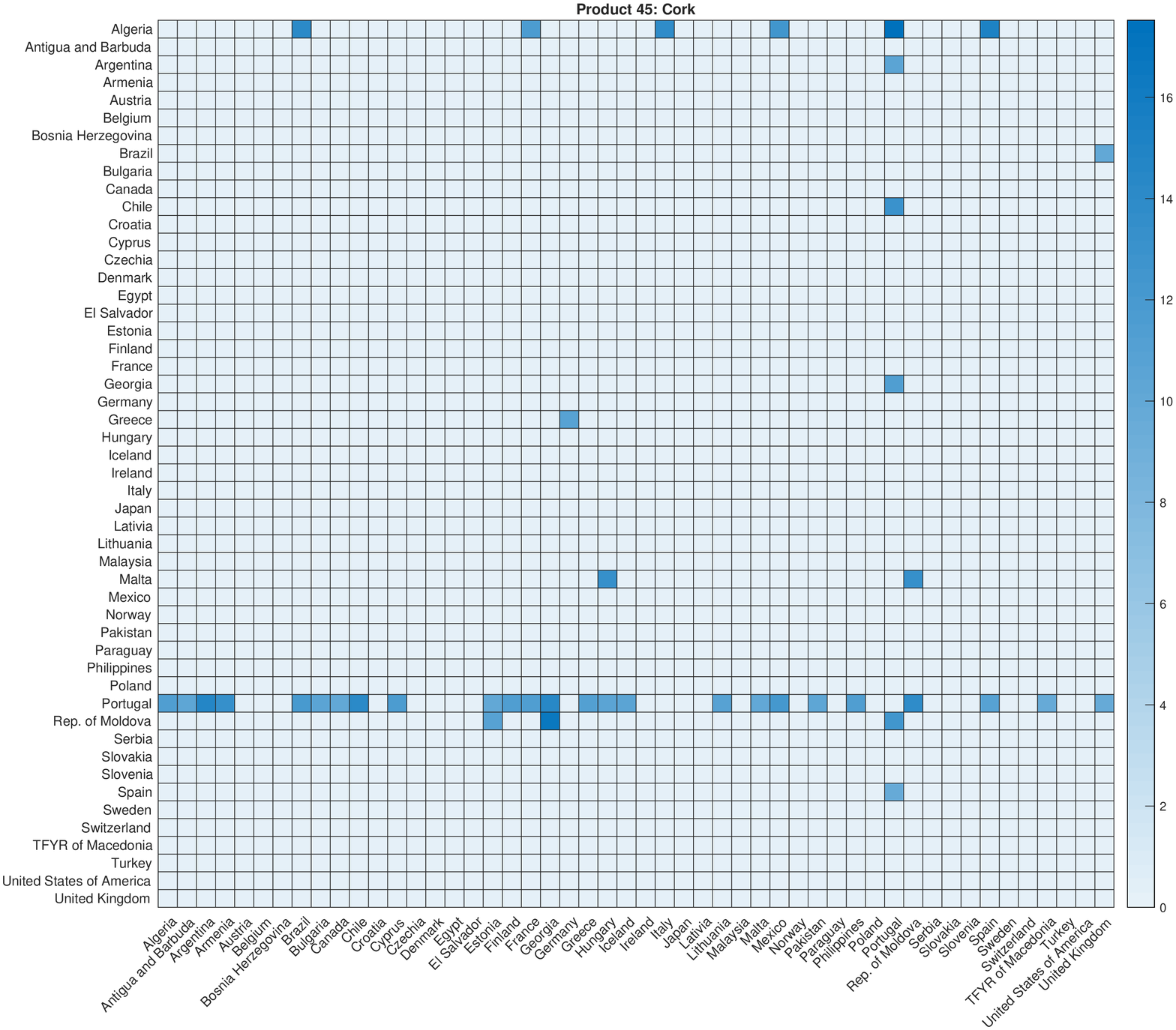}
		\includegraphics[width=0.5\textwidth]{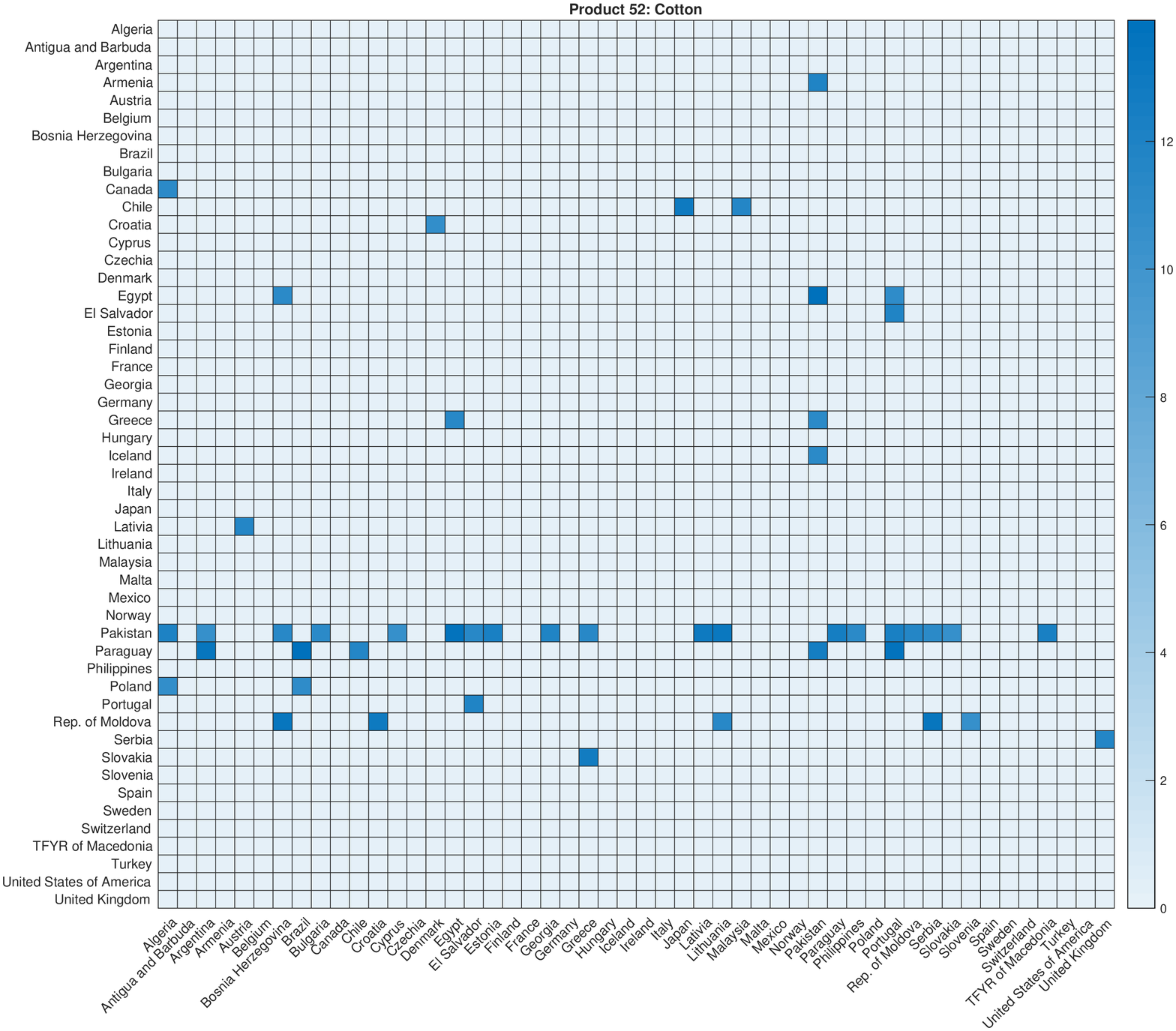}
		\includegraphics[width=0.5\textwidth]{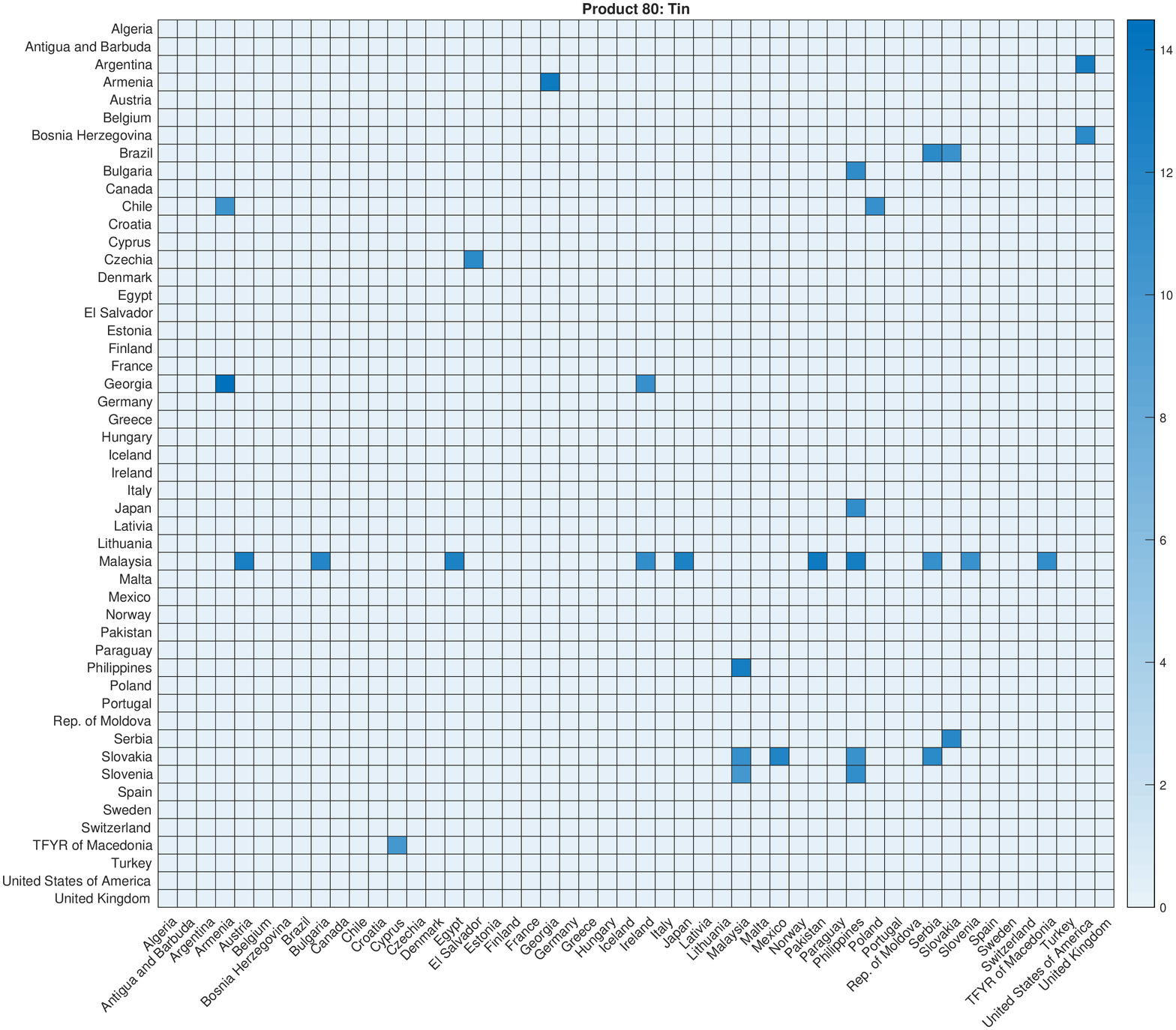}
		\caption{Commodity: top-left, Mineral Fuels; top-right, Cork; bottom-left, Cotton; bottom-right, Tin.}
	\end{subfigure}
\caption{Heatmaps of the slices of $\hat \S$. They reveal distinctive trading patterns of certain commodities. The sparsity ratio is set at $\alpha=0.03$. }
\label{fig:trade_flow_comm}
\end{figure}

We now compare \cite{gu2014robust}'s method (convex relaxation) and \cite{lu2016tensor}'s method (tubal-tRPCA) with our method in terms of prediction error on the international trade flow dataset. As in Section~\ref{sec:real_app}, we analyze the tensor $\log(1 + \A)$. We split $\log(1 + \A)$ into two parts, namely $\log(1 + \A) =: \A_{\text{train}} + \A_{\text{test}}$, where $\A_{\text{test}}$ is generated by randomly taking 10\% of the non-zero entries of $\log(1 + \A)$. We then apply tubal-tRPCA with the default parameter the authors provide \footnote{Their codes are available at https://github.com/canyilu/tensor-completion-under-linear-transform.} and convex relaxation with carefully tuned parameters \footnote{The codes in \cite{gu2014robust} is not publicly released so we have to tune the parameters by ourselves.}. We use the proposed BIC-type criterion to select the rank and sparsity.  As the left panel of Figure \ref{fig:bic:realdata} suggests, we choose $\br =(3,3,3)^{\top}$, and the right panel of Figure \ref{fig:bic:realdata} shows the BIC is less sensitive to $\alpha$ for a small range.
\begin{figure}
	\includegraphics[width=0.5\textwidth]{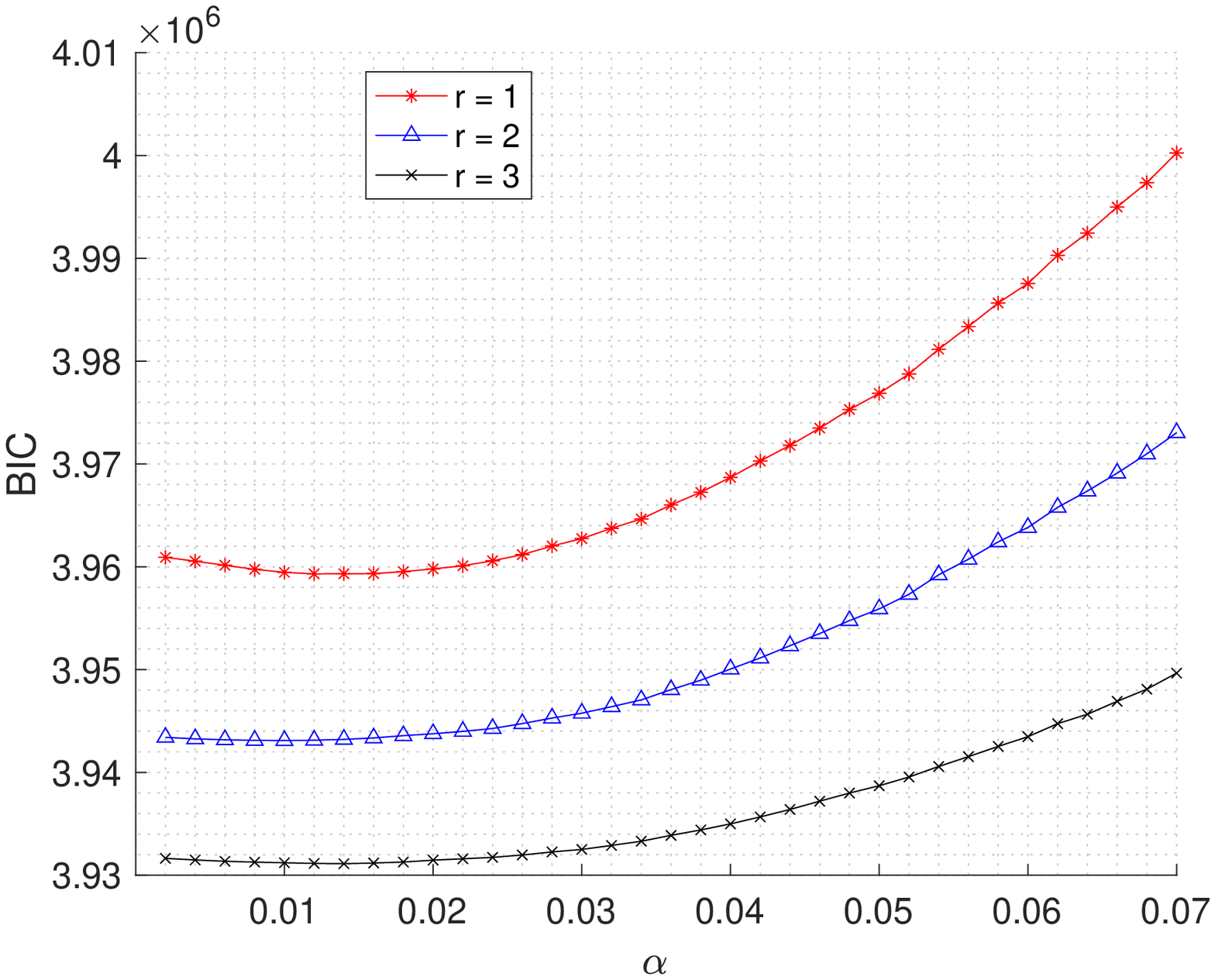}
	\includegraphics[width=0.5\textwidth]{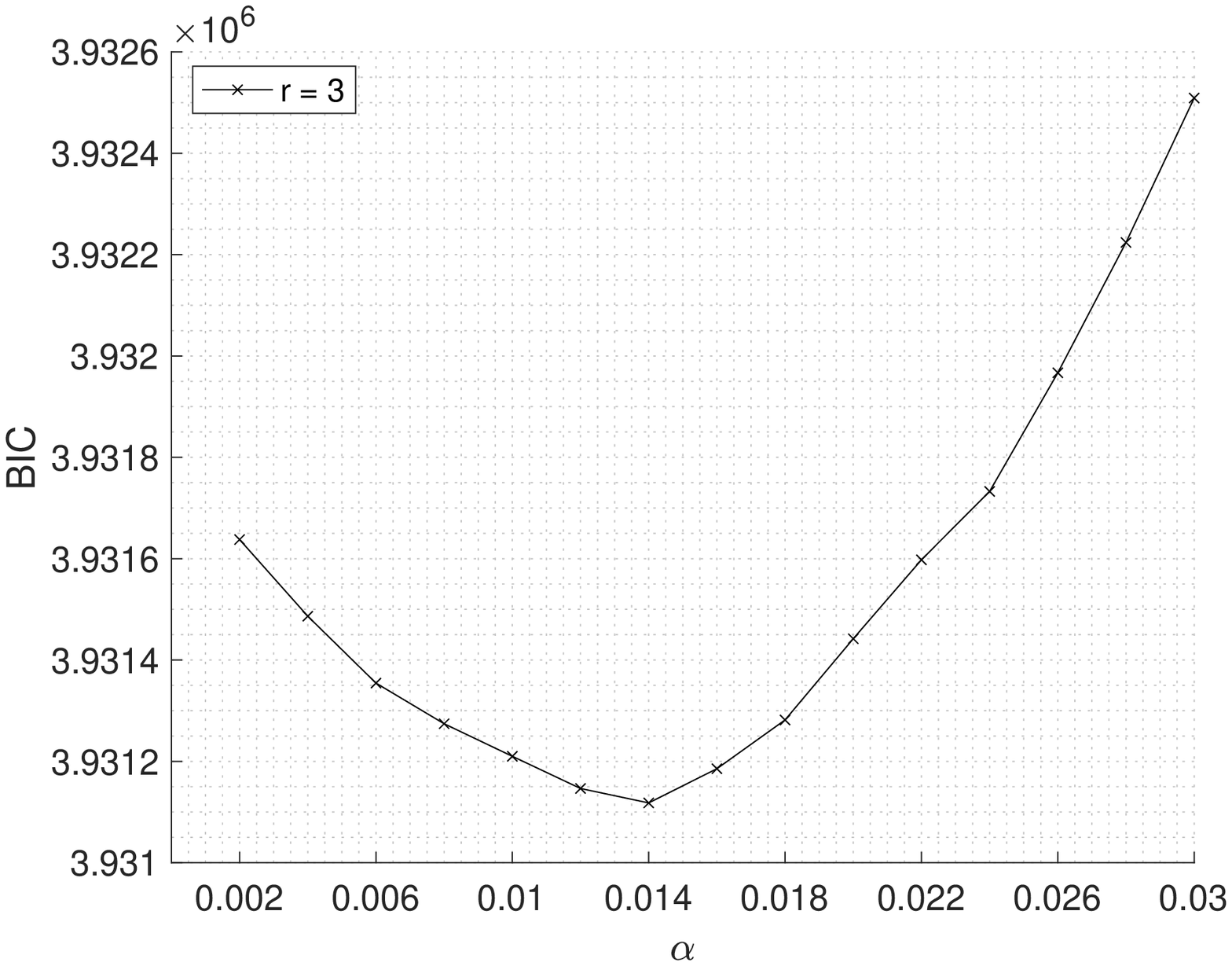}
	\caption{BIC values on the International Trade Flow Data. Left: BIC values for different rank and sparsity; Right: Zoom in on the case $r = 3$.}
	\label{fig:bic:realdata}
\end{figure}
Therefore we set the rank as $\br = (3,3,3)^{\top}$ and try $\alpha = 0.01,0.02,0.03$. The error is measured in terms of the test error $\fro{[\hat\T +\hat\S]_{\Omega_{\text{test}}} - \A_{\text{test}}}$ and the results are presented in Table \ref{table:realdata}.

\begin{table}[h!]
\centering
\begin{tabular}{ |p{2.5cm}||p{2cm}|p{2.5cm}|p{2cm}|p{2cm}|p{2cm}|p{2cm}|}
	\hline
	Method& \thead{Convex \\\cite{gu2014robust}} &\thead{tubal-tRPCA \\\cite{lu2016tensor}}&\thead{Our method \\$(\alpha = 0)$}&\thead{Our method \\$(\alpha = 0.01)$}& \thead{Our method \\$(\alpha = 0.02)$}&\thead{Our method \\$(\alpha = 0.03)$}\\
	\hline
	Pred. Error   & 1892.3    &1894.2  &1891.2  & \textbf{693.5}&800.5&980.1\\
	\hline
\end{tabular}
\caption{Comparison of our method with convex relaxation \cite{gu2014robust} and tubal-tRPCA \cite{lu2016tensor} in terms of prediction error on the international trade flow data. Our BIC criterion suggests any $\alpha$ between $0.002$ and $0.03$. We note that our method with $\alpha=0.003$ yields a prediction error $566.0$.}
\label{table:realdata}
\end{table}

When $\alpha=0$, all methods perform poorly because the existence of outliers distort the low-rank estimate making it ineffective in prediction. Meanwhile, if $\alpha$ is too large, say $0.1$, the sparse component might incorrectly absorb useful information from the low-rank component which, as a result, sabotages its prediction accuracy. Fortunately, our method with the BIC suggested $\alpha$ indeed significantly outperforms other methods.

\bibliographystyle{plainnat}
\bibliography{References.bib}

\begin{thebibliography}{51}
\providecommand{\natexlab}[1]{#1}
\providecommand{\url}[1]{\texttt{#1}}
\expandafter\ifx\csname urlstyle\endcsname\relax
  \providecommand{\doi}[1]{doi: #1}\else
  \providecommand{\doi}{doi: \begingroup \urlstyle{rm}\Url}\fi

\bibitem[Anandkumar et~al.(2014)Anandkumar, Ge, Hsu, Kakade, and
  Telgarsky]{anandkumar2014tensor}
Animashree Anandkumar, Rong Ge, Daniel Hsu, Sham~M Kakade, and Matus Telgarsky.
\newblock Tensor decompositions for learning latent variable models.
\newblock \emph{Journal of Machine Learning Research}, 15:\penalty0 2773--2832,
  2014.

\bibitem[Bi et~al.(2018)Bi, Qu, and Shen]{bi2018multilayer}
Xuan Bi, Annie Qu, and Xiaotong Shen.
\newblock Multilayer tensor factorization with applications to recommender
  systems.
\newblock \emph{Annals of Statistics}, 46\penalty0 (6B):\penalty0 3308--3333,
  2018.

\bibitem[Bi et~al.(2020)Bi, Tang, Yuan, Zhang, and Qu]{bi2020tensors}
Xuan Bi, Xiwei Tang, Yubai Yuan, Yanqing Zhang, and Annie Qu.
\newblock Tensors in statistics.
\newblock \emph{Annual Review of Statistics and Its Application}, 8, 2020.

\bibitem[Cai et~al.(2019)Cai, Li, Poor, and Chen]{cai2019nonconvex}
Changxiao Cai, Gen Li, H~Vincent Poor, and Yuxin Chen.
\newblock Nonconvex low-rank tensor completion from noisy data.
\newblock In \emph{Advances in Neural Information Processing Systems}, pages
  1863--1874, 2019.

\bibitem[Cai et~al.(2020)Cai, Miao, Wang, and Xian]{cai2020provable}
Jian-Feng Cai, Lizhang Miao, Yang Wang, and Yin Xian.
\newblock Provable near-optimal low-multilinear-rank tensor recovery.
\newblock \emph{arXiv preprint arXiv:2007.08904}, 2020.

\bibitem[Cand{\`e}s et~al.(2011)Cand{\`e}s, Li, Ma, and
  Wright]{candes2011robust}
Emmanuel~J Cand{\`e}s, Xiaodong Li, Yi~Ma, and John Wright.
\newblock Robust principal component analysis?
\newblock \emph{Journal of the ACM (JACM)}, 58\penalty0 (3):\penalty0 1--37,
  2011.

\bibitem[Cattell(1966)]{cattell1966scree}
Raymond~B Cattell.
\newblock The scree test for the number of factors.
\newblock \emph{Multivariate behavioral research}, 1\penalty0 (2):\penalty0
  245--276, 1966.

\bibitem[Chen et~al.(2019)Chen, Raskutti, and Yuan]{chen2019non}
Han Chen, Garvesh Raskutti, and Ming Yuan.
\newblock Non-convex projected gradient descent for generalized low-rank tensor
  regression.
\newblock \emph{The Journal of Machine Learning Research}, 20\penalty0
  (1):\penalty0 172--208, 2019.

\bibitem[Chen et~al.(2020)Chen, Fan, Ma, and Yan]{chen2020bridging}
Yuxin Chen, Jianqing Fan, Cong Ma, and Yuling Yan.
\newblock Bridging convex and nonconvex optimization in robust pca: Noise,
  outliers, and missing data.
\newblock \emph{arXiv preprint arXiv:2001.05484}, 2020.

\bibitem[Davenport et~al.(2014)Davenport, Plan, Van Den~Berg, and
  Wootters]{davenport20141}
Mark~A Davenport, Yaniv Plan, Ewout Van Den~Berg, and Mary Wootters.
\newblock 1-bit matrix completion.
\newblock \emph{Information and Inference: A Journal of the IMA}, 3\penalty0
  (3):\penalty0 189--223, 2014.

\bibitem[Edelman et~al.(1998)Edelman, Arias, and Smith]{edelman1998geometry}
Alan Edelman, Tom{\'a}s~A Arias, and Steven~T Smith.
\newblock The geometry of algorithms with orthogonality constraints.
\newblock \emph{SIAM journal on Matrix Analysis and Applications}, 20\penalty0
  (2):\penalty0 303--353, 1998.

\bibitem[Fan et~al.(2021)Fan, Pensky, Yu, and Zhang]{fan2021alma}
Xing Fan, Marianna Pensky, Feng Yu, and Teng Zhang.
\newblock Alma: Alternating minimization algorithm for clustering mixture
  multilayer network.
\newblock \emph{arXiv preprint arXiv:2102.10226}, 2021.

\bibitem[Gu et~al.(2014)Gu, Gui, and Han]{gu2014robust}
Quanquan Gu, Huan Gui, and Jiawei Han.
\newblock Robust tensor decomposition with gross corruption.
\newblock \emph{Advances in Neural Information Processing Systems},
  27:\penalty0 1422--1430, 2014.

\bibitem[Han et~al.(2020)Han, Willett, and Zhang]{han2020optimal}
Rungang Han, Rebecca Willett, and Anru Zhang.
\newblock An optimal statistical and computational framework for generalized
  tensor estimation.
\newblock \emph{arXiv preprint arXiv:2002.11255}, 2020.

\bibitem[Hao et~al.(2020)Hao, Zhang, and Cheng]{hao2020sparse}
Botao Hao, Anru Zhang, and Guang Cheng.
\newblock Sparse and low-rank tensor estimation via cubic sketchings.
\newblock \emph{IEEE Transactions on Information Theory}, 2020.

\bibitem[Ji and Jin(2016)]{ji2016coauthorship}
Pengsheng Ji and Jiashun Jin.
\newblock Coauthorship and citation networks for statisticians.
\newblock \emph{The Annals of Applied Statistics}, 10\penalty0 (4):\penalty0
  1779--1812, 2016.

\bibitem[Jin(2015)]{jin2015fast}
Jiashun Jin.
\newblock Fast community detection by score.
\newblock \emph{Annals of Statistics}, 43\penalty0 (1):\penalty0 57--89, 2015.

\bibitem[Jing et~al.(2020)Jing, Li, Lyu, and Xia]{jing2020community}
Bing-Yi Jing, Ting Li, Zhongyuan Lyu, and Dong Xia.
\newblock Community detection on mixture multi-layer networks via regularized
  tensor decomposition.
\newblock \emph{arXiv preprint arXiv:2002.04457}, 2020.

\bibitem[Ke et~al.(2019)Ke, Shi, and Xia]{ke2019community}
Zheng~Tracy Ke, Feng Shi, and Dong Xia.
\newblock Community detection for hypergraph networks via regularized tensor
  power iteration.
\newblock \emph{arXiv preprint arXiv:1909.06503}, 2019.

\bibitem[Kolda and Bader(2009)]{kolda2009tensor}
Tamara~G Kolda and Brett~W Bader.
\newblock Tensor decompositions and applications.
\newblock \emph{SIAM review}, 51\penalty0 (3):\penalty0 455--500, 2009.

\bibitem[Kressner et~al.(2014)Kressner, Steinlechner, and
  Vandereycken]{kressner2014low}
Daniel Kressner, Michael Steinlechner, and Bart Vandereycken.
\newblock Low-rank tensor completion by riemannian optimization.
\newblock \emph{BIT Numerical Mathematics}, 54\penalty0 (2):\penalty0 447--468,
  2014.

\bibitem[Li et~al.(2018)Li, Xu, Zhou, and Li]{li2018tucker}
Xiaoshan Li, Da~Xu, Hua Zhou, and Lexin Li.
\newblock Tucker tensor regression and neuroimaging analysis.
\newblock \emph{Statistics in Biosciences}, 10\penalty0 (3):\penalty0 520--545,
  2018.

\bibitem[Liu et~al.(2017)Liu, Yuan, and Zhao]{liu2017characterizing}
Tianqi Liu, Ming Yuan, and Hongyu Zhao.
\newblock Characterizing spatiotemporal transcriptome of human brain via low
  rank tensor decomposition.
\newblock \emph{arXiv preprint arXiv:1702.07449}, 2017.

\bibitem[Lu et~al.(2016)Lu, Feng, Chen, Liu, Lin, and Yan]{lu2016tensor}
Canyi Lu, Jiashi Feng, Yudong Chen, Wei Liu, Zhouchen Lin, and Shuicheng Yan.
\newblock Tensor robust principal component analysis: Exact recovery of
  corrupted low-rank tensors via convex optimization.
\newblock In \emph{Proceedings of the IEEE conference on computer vision and
  pattern recognition}, pages 5249--5257, 2016.

\bibitem[Luo and Zhang(2020)]{luo2020tensor}
Yuetian Luo and Anru~R Zhang.
\newblock Tensor clustering with planted structures: Statistical optimality and
  computational limits.
\newblock \emph{arXiv preprint arXiv:2005.10743}, 2020.

\bibitem[Pan et~al.(2018)Pan, Mai, and Zhang]{pan2018covariate}
Yuqing Pan, Qing Mai, and Xin Zhang.
\newblock Covariate-adjusted tensor classification in high dimensions.
\newblock \emph{Journal of the American Statistical Association}, 2018.

\bibitem[Paul and Chen(2020)]{paul2020spectral}
Subhadeep Paul and Yuguo Chen.
\newblock Spectral and matrix factorization methods for consistent community
  detection in multi-layer networks.
\newblock \emph{The Annals of Statistics}, 48\penalty0 (1):\penalty0 230--250,
  2020.

\bibitem[Pensky and Zhang(2019)]{pensky2019spectral}
Marianna Pensky and Teng Zhang.
\newblock Spectral clustering in the dynamic stochastic block model.
\newblock \emph{Electronic Journal of Statistics}, 13\penalty0 (1):\penalty0
  678--709, 2019.

\bibitem[Raskutti et~al.(2019)Raskutti, Yuan, and Chen]{raskutti2019convex}
Garvesh Raskutti, Ming Yuan, and Han Chen.
\newblock Convex regularization for high-dimensional multiresponse tensor
  regression.
\newblock \emph{The Annals of Statistics}, 47\penalty0 (3):\penalty0
  1554--1584, 2019.

\bibitem[Richard and Montanari(2014)]{richard2014statistical}
Emile Richard and Andrea Montanari.
\newblock A statistical model for tensor pca.
\newblock \emph{Advances in Neural Information Processing Systems},
  27:\penalty0 2897--2905, 2014.

\bibitem[Robin et~al.(2020)Robin, Klopp, Josse, Moulines, and
  Tibshirani]{robin2020main}
Genevi{\`e}ve Robin, Olga Klopp, Julie Josse, {\'E}ric Moulines, and Robert
  Tibshirani.
\newblock Main effects and interactions in mixed and incomplete data frames.
\newblock \emph{Journal of the American Statistical Association}, 115\penalty0
  (531):\penalty0 1292--1303, 2020.

\bibitem[Sun and Li(2019)]{sun2019dynamic}
Will~Wei Sun and Lexin Li.
\newblock Dynamic tensor clustering.
\newblock \emph{Journal of the American Statistical Association}, 114\penalty0
  (528):\penalty0 1894--1907, 2019.

\bibitem[Sun et~al.(2017)Sun, Lu, Liu, and Cheng]{sun2017provable}
Will~Wei Sun, Junwei Lu, Han Liu, and Guang Cheng.
\newblock Provable sparse tensor decomposition.
\newblock \emph{Journal of the Royal Statistical Society: Series B (Statistical
  Methodology)}, 79\penalty0 (3):\penalty0 899--916, 2017.

\bibitem[Vershynin(2011)]{vershynin2011spectral}
Roman Vershynin.
\newblock Spectral norm of products of random and deterministic matrices.
\newblock \emph{Probability theory and related fields}, 150\penalty0
  (3):\penalty0 471--509, 2011.

\bibitem[Vershynin(2018)]{vershynin2018high}
Roman Vershynin.
\newblock \emph{High-dimensional probability: An introduction with applications
  in data science}, volume~47.
\newblock Cambridge university press, 2018.

\bibitem[Wang et~al.(2019)Wang, Zhang, and Dunson]{wang2019common}
Lu~Wang, Zhengwu Zhang, and David Dunson.
\newblock Common and individual structure of brain networks.
\newblock \emph{Ann. Appl. Stat.}, 13\penalty0 (1):\penalty0 85--112, 03 2019.
\newblock \doi{10.1214/18-AOAS1193}.
\newblock URL \url{https://doi.org/10.1214/18-AOAS1193}.

\bibitem[Wang and Li(2020)]{wang2020learning}
Miaoyan Wang and Lexin Li.
\newblock Learning from binary multiway data: Probabilistic tensor
  decomposition and its statistical optimality.
\newblock \emph{Journal of Machine Learning Research}, 21\penalty0
  (154):\penalty0 1--38, 2020.

\bibitem[Wang and Zeng(2019)]{wang2019multiway}
Miaoyan Wang and Yuchen Zeng.
\newblock Multiway clustering via tensor block models.
\newblock \emph{arXiv preprint arXiv:1906.03807}, 2019.

\bibitem[Xia(2019)]{xia2019normal}
Dong Xia.
\newblock Normal approximation and confidence region of singular subspaces,
  2019.

\bibitem[Xia and Yuan(2019)]{xia2019polynomial}
Dong Xia and Ming Yuan.
\newblock On polynomial time methods for exact low-rank tensor completion.
\newblock \emph{Foundations of Computational Mathematics}, 19\penalty0
  (6):\penalty0 1265--1313, 2019.

\bibitem[Xia and Zhou(2019)]{xia2019sup}
Dong Xia and Fan Zhou.
\newblock The sup-norm perturbation of hosvd and low rank tensor denoising.
\newblock \emph{J. Mach. Learn. Res.}, 20:\penalty0 61--1, 2019.

\bibitem[Xia et~al.(2020)Xia, Zhang, and Zhou]{xia2020inference}
Dong Xia, Anru~R Zhang, and Yuchen Zhou.
\newblock Inference for low-rank tensors--no need to debias.
\newblock \emph{arXiv preprint arXiv:2012.14844}, 2020.

\bibitem[Xia et~al.(2021)Xia, Yuan, and Zhang]{xia2017statistically}
Dong Xia, Ming Yuan, and Cun-Hui Zhang.
\newblock Statistically optimal and computationally efficient low rank tensor
  completion from noisy entries.
\newblock \emph{Annals of Statistics}, 49\penalty0 (1):\penalty0 76--99, 2021.

\bibitem[Yu and Liu(2016)]{yu2016learning}
Rose Yu and Yan Liu.
\newblock Learning from multiway data: Simple and efficient tensor regression.
\newblock In \emph{International Conference on Machine Learning}, pages
  373--381, 2016.

\bibitem[Yuan and Zhang(2017)]{yuan2017incoherent}
Ming Yuan and Cun-Hui Zhang.
\newblock Incoherent tensor norms and their applications in higher order tensor
  completion.
\newblock \emph{IEEE Transactions on Information Theory}, 63\penalty0
  (10):\penalty0 6753--6766, 2017.

\bibitem[Zhang and Xia(2018)]{zhang2018tensor}
Anru Zhang and Dong Xia.
\newblock Tensor svd: Statistical and computational limits.
\newblock \emph{IEEE Transactions on Information Theory}, 64\penalty0
  (11):\penalty0 7311--7338, 2018.

\bibitem[Zhang et~al.(2020{\natexlab{a}})Zhang, Luo, Raskutti, and
  Yuan]{zhang2020islet}
Anru~R Zhang, Yuetian Luo, Garvesh Raskutti, and Ming Yuan.
\newblock Islet: Fast and optimal low-rank tensor regression via importance
  sketching.
\newblock \emph{SIAM Journal on Mathematics of Data Science}, 2\penalty0
  (2):\penalty0 444--479, 2020{\natexlab{a}}.

\bibitem[Zhang et~al.(2020{\natexlab{b}})Zhang, Han, Zhang, and
  Voyles]{zhang2020denoising}
Chenyu Zhang, Rungang Han, Anru~R Zhang, and Paul~M Voyles.
\newblock Denoising atomic resolution 4d scanning transmission electron
  microscopy data with tensor singular value decomposition.
\newblock \emph{Ultramicroscopy}, 219:\penalty0 113123, 2020{\natexlab{b}}.

\bibitem[Zhang et~al.(2018)Zhang, Wang, and Gu]{zhang2018unified}
Xiao Zhang, Lingxiao Wang, and Quanquan Gu.
\newblock A unified framework for nonconvex low-rank plus sparse matrix
  recovery.
\newblock In \emph{International Conference on Artificial Intelligence and
  Statistics}, pages 1097--1107. PMLR, 2018.

\bibitem[Zhou et~al.(2013)Zhou, Li, and Zhu]{zhou2013tensor}
Hua Zhou, Lexin Li, and Hongtu Zhu.
\newblock Tensor regression with applications in neuroimaging data analysis.
\newblock \emph{Journal of the American Statistical Association}, 108\penalty0
  (502):\penalty0 540--552, 2013.

\bibitem[Zhou and Feng(2017)]{zhou2017outlier}
Pan Zhou and Jiashi Feng.
\newblock Outlier-robust tensor pca.
\newblock In \emph{Proceedings of the IEEE Conference on Computer Vision and
  Pattern Recognition}, pages 2263--2271, 2017.

\end{thebibliography}

\newpage
\begin{center}
{\large\bf SUPPLEMENTARY MATERIAL for ``Generalized Low-rank plus Sparse Tensor Estimation by Fast Riemannian Optimization"}
\end{center}

\section{Higher Order Orthogonal Iteration Algorithm}
The HOOI algorithm is summarized as follows which is applied for the initialization in Section~\ref{sec:rpca} and \ref{sec:heavy_tail}.  
\begin{algorithm}
	\caption{HOOI}\label{alg:hooi}
	\begin{algorithmic}
		\STATE{Input: $\Y\in\RR^{d_1\times\cdots\times d_m}, \br = (r_1,\cdots,r_m)$, maximum  number of iteration: $t_{\max}$.}
		\STATE{Let $t = 0$, initiate $\hat\bU_i^{0} = \textrm{SVD}_{r_i}(\opM_i(\Y))$, $i\in[m]$.}
		\FOR{$t = 1,\ldots,t_{\max}$}
		\FOR{$i=1,\ldots,m$}
		\STATE{$\hat\bU_{i}^{t} = \textrm{SVD}_{r_i}\big(\opM_{i}(\Y)(\hat\bU_m^{t-1}\otimes\cdots\otimes\hat\bU_{i+1}^{t-1}\otimes\hat\bU_{i-1}^{t}\otimes\cdots\otimes\hat\bU_{1}^{t-1})\big)$}
		\ENDFOR
		\ENDFOR
		\STATE{Output: $\hat\bU_i = \hat\bU_i^{t_{max}}$, $\hat\T = \Y\times_{i=1}^m\hat\bU_i\hat\bU_i^T$.}
	\end{algorithmic}
\end{algorithm}

\section{More Numerical Simulations}\label{sec:num_more}
In Section~\ref{sec:num_BIC}, we apply the proposed BIC-type criterion for SG-RPCA and binary tensor learning and demonstrate its effectiveness on synthetic data.  Through Section~\ref{sec:num_rpca}-\ref{sec:num_poisson}, we treat $\br$ and $\alpha$ as given and test the performance of our estimator with respect to different choices of $\gamma$. Other algorithmic parameters like $\mu_1$ and $\kprune$ are decided as explained in Section~\ref{sec:numerical}.

\subsection{Performance of BIC-type Criterion}\label{sec:num_BIC}
We test the performance of BIC-type criterion (\ref{eq:BIC}) for SG-RPCA and binary tensor learning. As explained in Section~\ref{sec:numerical}, $\gamma$ is set to $1$ and $\mu_1=2^m+\log(\bar d)$. 
More exactly, the BIC-type criterion for SG-RPCA (assuming Gaussian noise with equal but unknown variances) is 
$$
{\rm BIC}(\br,\alpha):=\big(\|\hat \S_{\br,\alpha}\|_{\ell_0}+\sum\nolimits_{i=1}^m r_id_i\big)\cdot \log(d^{\ast})+d^{\ast}\log(\fro{\A-\hat\T_{\br,\alpha}-\hat\S_{\br,\alpha}}^2).
$$
The true tensor $\T^*\in\RR^{d\times d\times d}$ with $d = 100$ and $\br = (3,3,3)^{\top}$. We test two true sparsity levels $\alpha\in\{0.05,0.1\}$. 
The true tensor $\T^*$ satisfies $\|\T^*\|_{\ell_{\infty}} = 0.1$, and $\S^*$ is generated as above satisfying $\|\S^*\|_{\ell_{\infty}} = 4$, and all entries of $\Z^*$ satisfy i.i.d. ${\rm N}(0,\sigma_z^2)$ with $\sigma_z = 0.01$. For each $\alpha\in\{0.05,0.1\}$, we test, in our algorithm, $\br \in\{(1,1,1),(2,2,2),(3,3,3),(4,4,4),(5,5,5)\}$ and $\alpha\in(0.02, 0.2)$. The results are displayed in Figure \ref{fig:bic:rpca}. The BIC-values are sensitive to both $\br$ and $\alpha$. We note that the BIC-values for $\br>3$ are strictly larger than that of $\br=(3,3,3)$, but the difference is too small to be spotted in the figures.
\begin{figure}
		\includegraphics[width=0.5\textwidth]{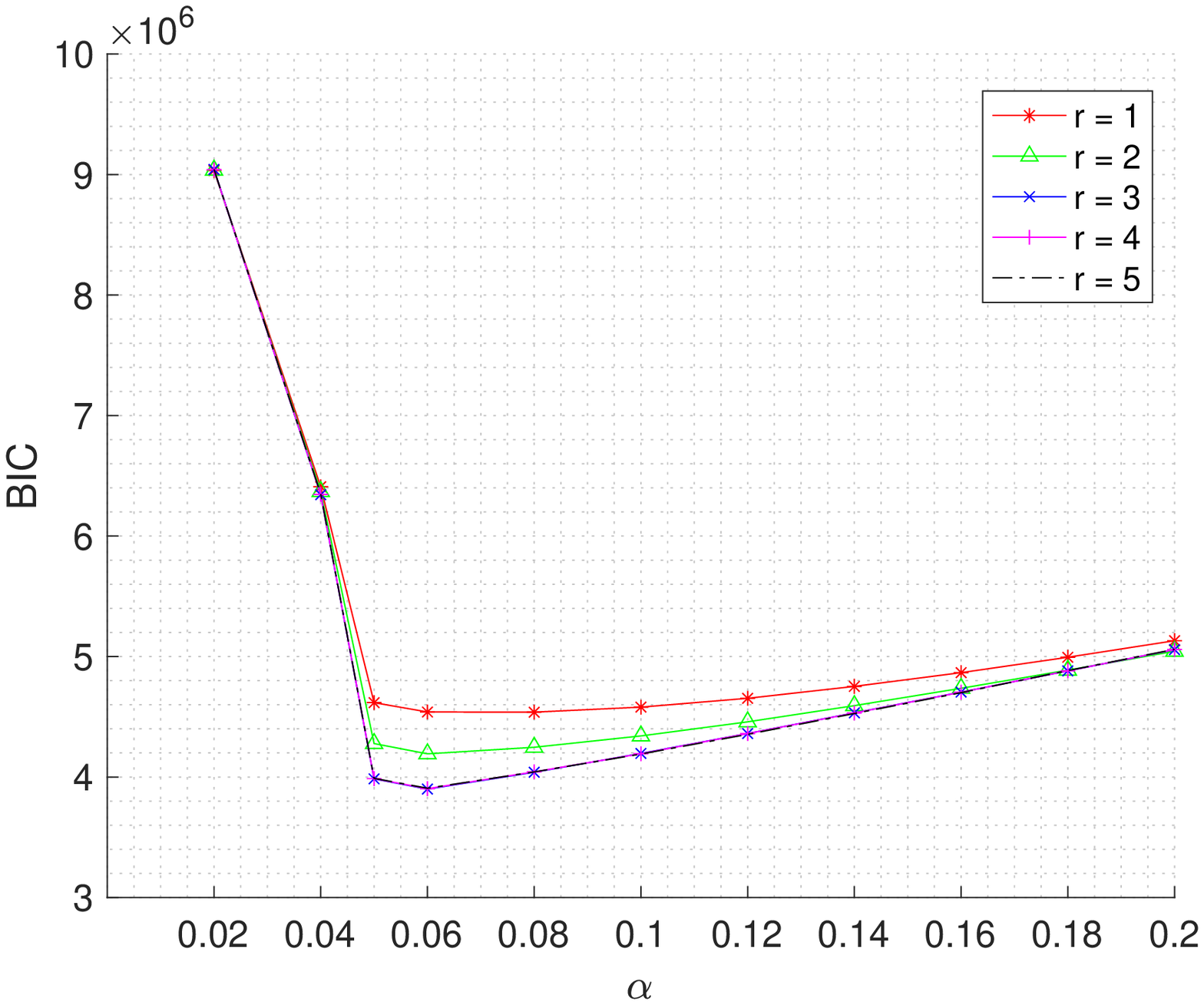}
		\includegraphics[width=0.5\textwidth]{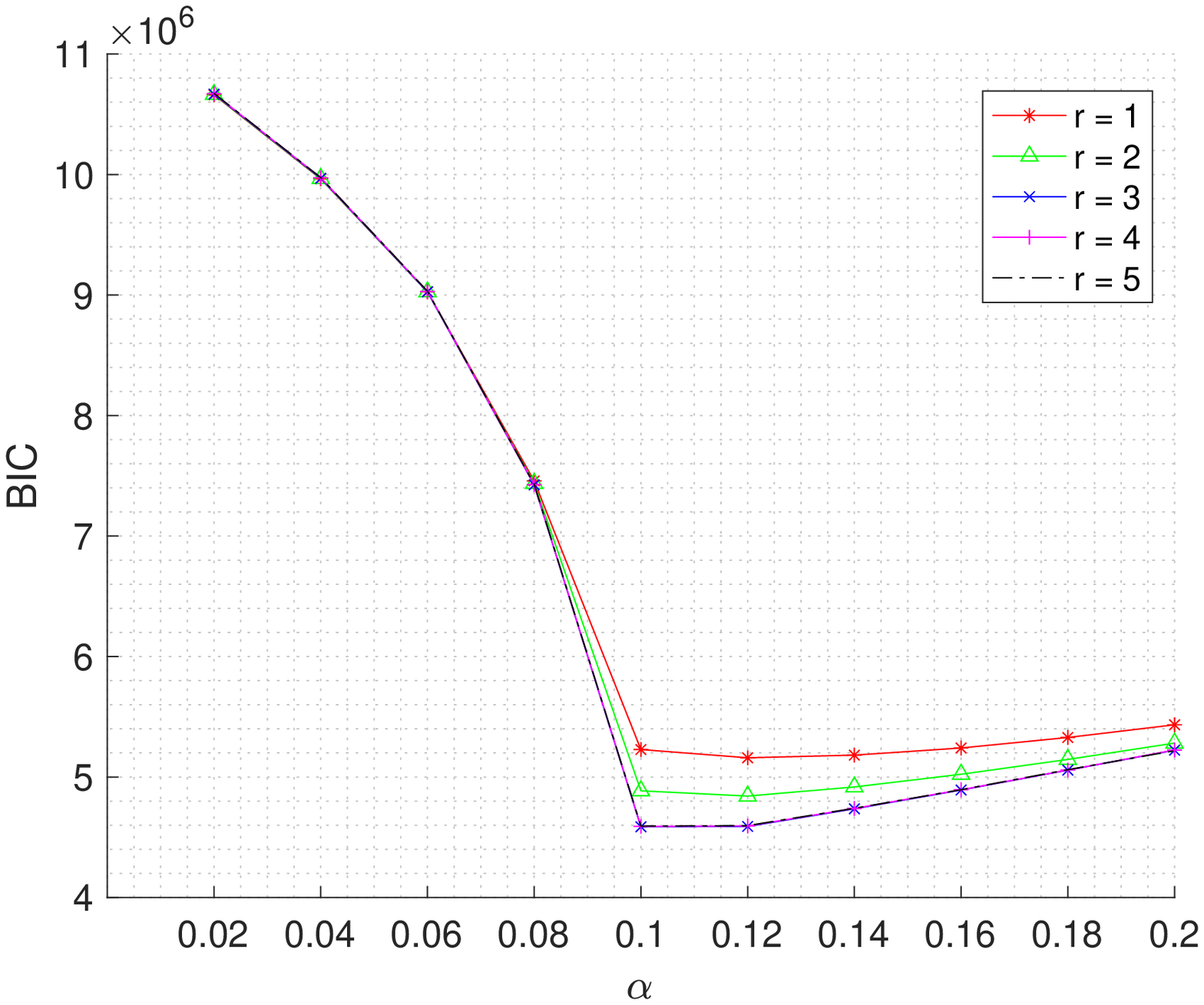}
	\caption{BIC values for SG-RPCA; the true rank $\br = (3,3,3)^{\top}$. Left: true $\alpha = 0.05$; Right: true $\alpha = 0.1$.}
	\label{fig:bic:rpca}
\end{figure}

For robust binary tensor learning, we also set $\br = (3,3,3)^{\top}$ and $\T^*\in\RR^{d\times d\times d}$ with $d = 100$, and $\S^*$ is generated as above.
The true $\alpha\in\{0.005,0.01\}$. We fix $p(x) = (1+e^{-10x})^{-1}$.
The BIC criterion for the binary case is: 
$${\rm BIC}(\br,\alpha):=\big(\|\hat \S\|_{\ell_0}+\sum\nolimits_{i=1}^m r_id_i\big)\cdot\log(d^{\ast})-2\sum_{\omega}\big([\A]_{\omega}\log p([\hat\T+\hat\S]_{\omega})+\big(1-[\A]_{\omega}\big)\log \big(1-p([\hat\T+\hat\S]_{\omega})\big)\big).$$
For each true $\alpha\in\{0.005,0.01\}$, we test BIC for $\br \in\{(1,1,1),(2,2,2),(3,3,3),(4,4,4),(5,5,5)\}$ and $\alpha$ varying from 0.001 to 0.015. 
The results are displayed in Figure \ref{fig:bic:binary} showing that BIC is more sensitive to $\br$ and less sensitive to $\alpha$ for a small range. After the true $\br$ is identified, the BIC criterion works reasonably well for selecting $\alpha$. 

\begin{figure}
		\includegraphics[width=0.5\textwidth]{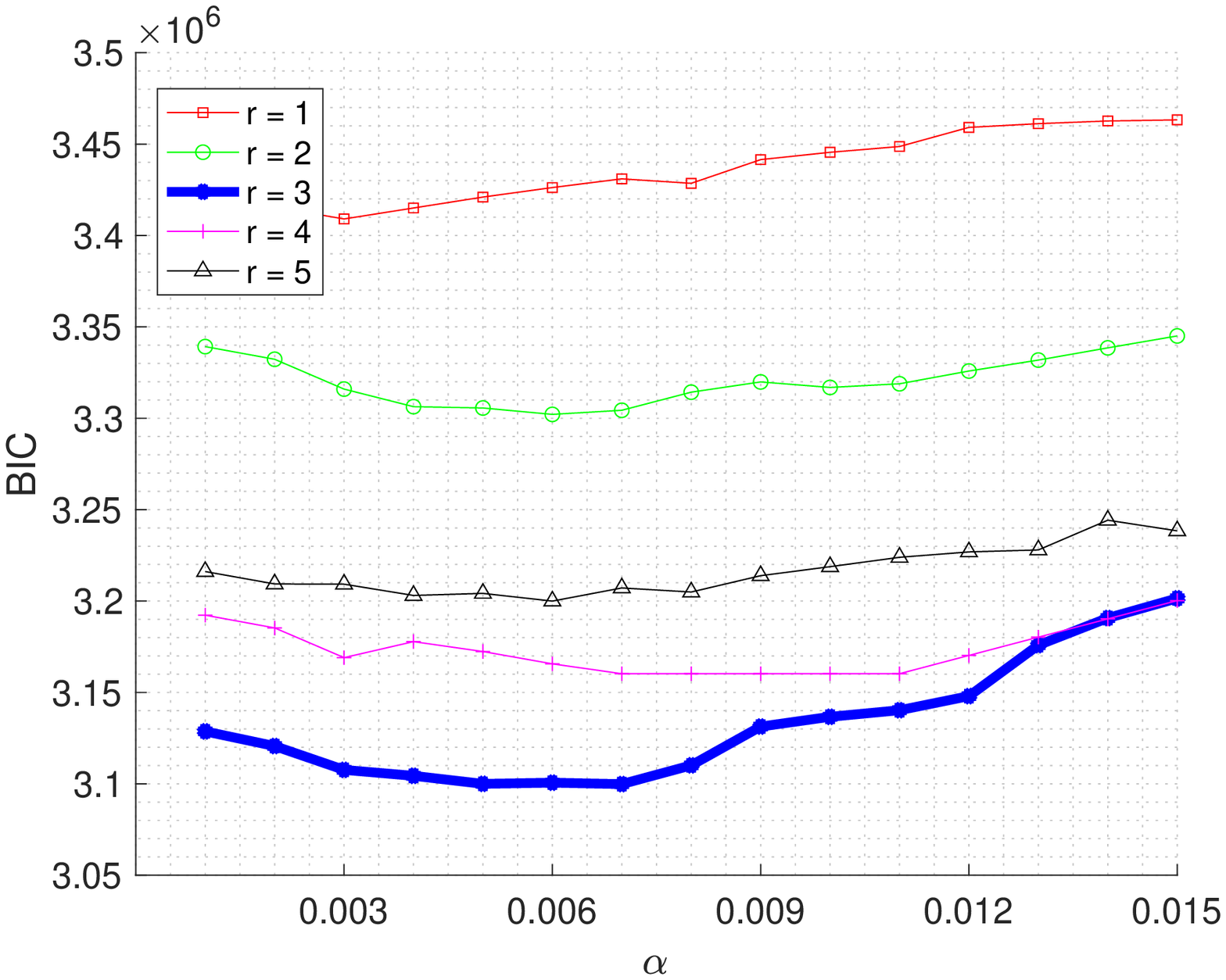}
		\includegraphics[width=0.5\textwidth]{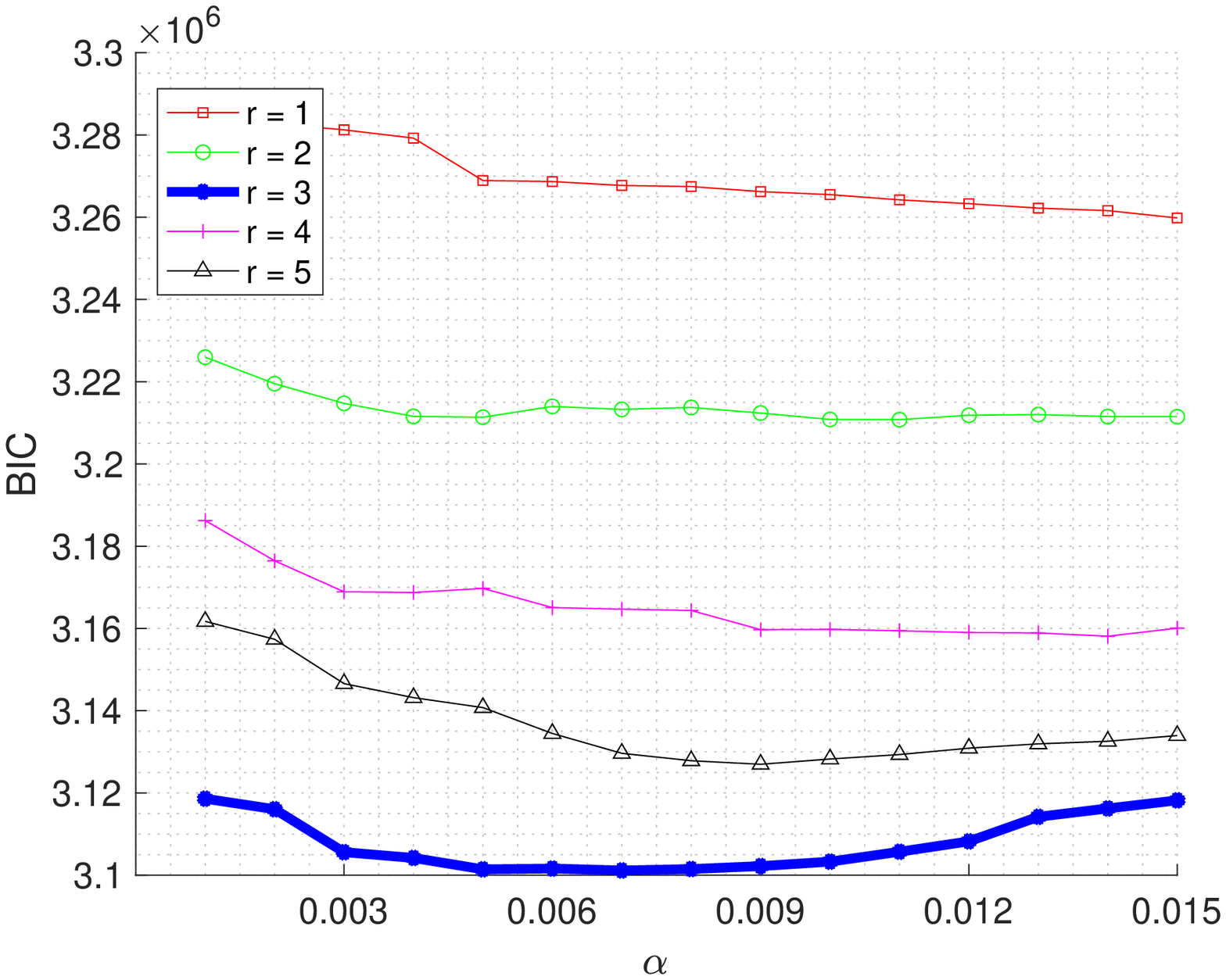}
	\caption{BIC values for binary tensor learning; the true rank $\br = (3,3,3)^{\top}$. Left: true $\alpha = 0.005$; Right: true $\alpha = 0.01$. The BIC curve for $\br = (3,3,3)^{\top}$ is highlighted.}
	\label{fig:bic:binary}
\end{figure}

\subsection{Tensor Sub-Gaussian Robust PCA}\label{sec:num_rpca}
The low-rank tensor $\T^{\ast}\in\RR^{d\times d\times d} $ with $d=100$ and Tucker ranks $\br=(2,2,2)^{\top}$ is generated from the HOSVD of a trimmed standard normal tensor.  It satisfies the spikiness condition, with high probability, and has singular values $\bsigma\approx 3$ and $\usigma\approx 1$. Given a sparsity level $\alpha\in(0,1)$, the entries of sparse tensor $\S^{\ast}$ are i.i.d. sampled from ${\rm Be}(\alpha)\times {\rm N}(0,1)$, which ensures $\S^{\ast}\in \SS_{O(\alpha)}$ with high probability. This ensures that the non-zero entries of $\S^{\ast}$ have typically much larger magnitudes than the entries of $\T^{\ast}$. 
The noise tensor $\Z$ has i.i.d. entries sampled from ${\rm N}(0,\sigma_z^2)$. The default choice of $\gamma$ is 2, $\kprune=\infty$ and $\mu_1$ is set as previously. The convergence performances of $\log(\|\hat \T_l-\T^{\ast}\|_{\rm F}/\|\T^{\ast}\|_{\rm F})$ by Algorithm~\ref{algo:lowrank+sparse} are examined and presented in the left panels of Figure~\ref{fig:rpca}. 

The top-left plot in Figure~\ref{fig:rpca} displays the effects of $\alpha$ on the convergence of Algorithm~\ref{algo:lowrank+sparse}. It shows that the convergence speed of Algorithm~\ref{algo:lowrank+sparse} is insensitive to $\alpha$, while the error of final estimates $\hat\T_{l_{\max}}$ is related to $\alpha$. This is consistent with the claims of Theorem~\ref{thm:rpca}. In the middle-left plot of Figure~\ref{fig:rpca}, we observe that, for a fixed sparsity level $\alpha$, the error of final estimates grows as the tuning parameter $\gamma$ becomes larger. The bottom-left plot of Figure~\ref{fig:rpca} shows the convergence of Algorithm~\ref{algo:lowrank+sparse} for different noise levels. All these plots confirm the fast convergence of our Riemannian gradient descent algorithm. In particular, there are stages during which the log relative error decreases linearly w.r.t. the number of iterations, as proved in Theorem~\ref{thm:lowrank+sparse}. 

The statistical stability of the final estimates by Algorithm~\ref{algo:lowrank+sparse} is demonstrated in the right panels of Figure~\ref{fig:rpca}. Each curve represents the average relative error of $\hat \T_{l_{\max}}$ based on $10$ simulations, and the error bar shows the confidence region by one empirical standard deviation. Based on these plots, we observe that the standard deviations of $\|\hat\T_{l_{\max}}-\T^{\ast}\|_{\rm F}$ grow as the noise level $\sigma_z$, the sparsity level $\alpha$  or the tuning parameter $\gamma$ increases. 

\begin{figure}
\centering
	\begin{subfigure}[b]{.98\linewidth}
		\includegraphics[width=0.5\textwidth]{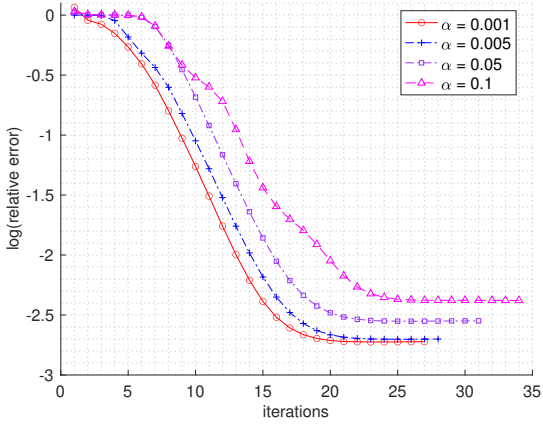}
		\includegraphics[width=0.5\textwidth]{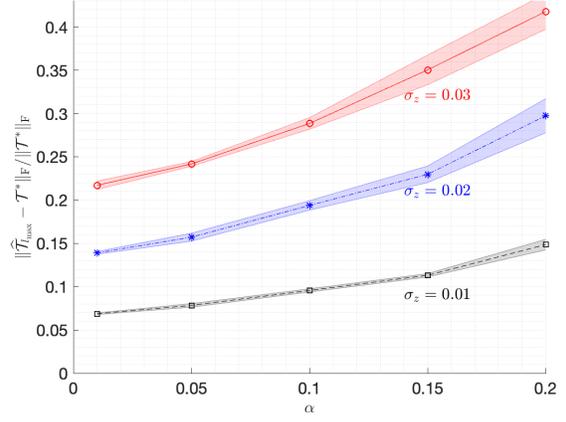}
		\caption{Change of sparsity $\alpha$. Left: $\sigma_z=0.01, \gamma=2$; Right: $\gamma=2$}
	\end{subfigure}
	\begin{subfigure}[b]{.98\linewidth}
		\includegraphics[width=0.5\textwidth]{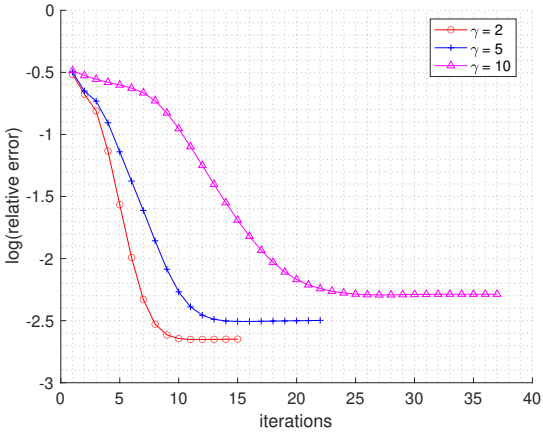}
		\includegraphics[width=0.5\textwidth]{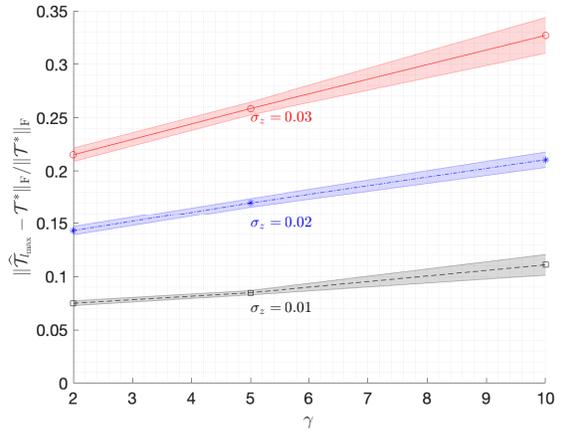}
		\caption{Change of $\gamma$. Left: $\alpha=0.02, \sigma_z=0.01$; Right: $\alpha=0.02$}
	\end{subfigure}
	\begin{subfigure}[b]{.98\linewidth}
		\includegraphics[width=0.5\textwidth]{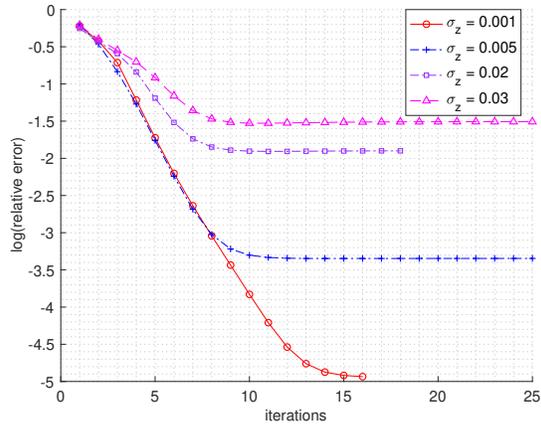}
		\includegraphics[width=0.5\textwidth]{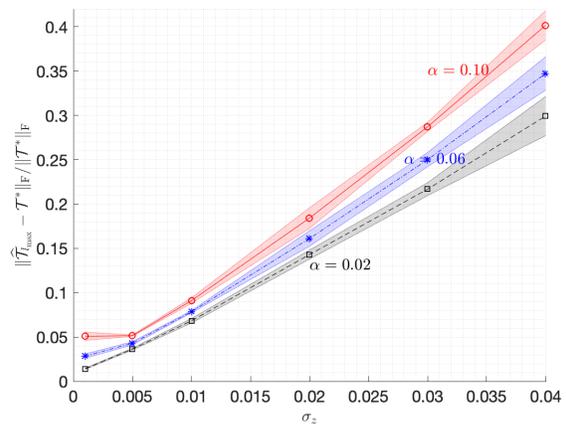}
		\caption{Change of noise size $\sigma_z$. Left: $\alpha=0.02, \gamma=2$; Right: $\gamma=2$}
	\end{subfigure}
\caption{Performances of Algorithm~\ref{algo:lowrank+sparse} for SG-RPCA. The low-rank $\T^{\ast}$ has size $d\times d\times d$ with $d=100$ and has Tucker ranks $\br=(2,2,2)^{\top}$. The relative error on left panels is defined by $\|\hat\T_l-\T^{\ast}\|_{\rm F}/\|\T^{\ast}\|_{\rm F}$. The error bars on the right panels are based on $1$ standard deviation from 10 replications. Here the default $\gamma$ is $2$.}
\label{fig:rpca}
\end{figure}

\subsection{Tensor PCA with Heavy-tailed Noise}
The low-rank tensor $\T^{\ast}\in\RR^{d\times d\times d} $ with $d=100$ and Tucker ranks $\br=(2,2,2)^{\top}$ is generated from the HOSVD of a trimmed standard normal tensor, as in Section~\ref{sec:num_rpca}. Given a parameter $\theta$, we generate the noisy tensor whose entries are i.i.d. and satisfy the Student-t distribution with degree of freedom $\theta$. But notice here we also apply a global scaling to better control the noise standard deviation. We denote the noisy tensor after scaling by $\Z$. This generated tensor $\Z$ satisfies Assumption~\ref{assump:heavy_tail} with the same parameter $\theta$. Once the parameter $\theta$ and global scaling are given, we are able to calculate the variance $\sigma_z^2$. The convergence performances of $\log(\|\hat \T_l-\T^{\ast}\|_{\rm F}/\|\T^{\ast}\|_{\rm F})$ by Algorithm~\ref{algo:lowrank+sparse} are examined and presented in the upper panels of Figure~\ref{fig:heavytailpca}. 

In this experiment, we set $\gamma = 2, \kprune=\infty$ and $\mu_1$ as previously. The top-left plot in Figure \ref{fig:heavytailpca} displays the effects of $\alpha$ on the convergence of Algorithm~\ref{algo:lowrank+sparse}. The case $\alpha = 0$ reduces to the normal Riemannian gradient descent, which cannot output a satisfiable result due to the heavy-tailed noise, even if a warm initialization is provided. This shows the importance of gradient pruning in Algorithm~\ref{algo:lowrank+sparse}. When $\alpha > 0$, the convergence speed of the algorithm is insensitive to $\alpha$, but the final estimates $\hat\T_{l_{\max}}$ is related to $\alpha$. In the top-right plot of Figure~\ref{fig:heavytailpca}, we observe the error becomes larger as $\theta$ decreases (or equivalently, as $\sigma_z^2$ increases). All these results match the claim of Theorem~\ref{thm:heavy_tail} and confirm the fast convergence of Riemannian gradient descent. And there are indeed stages where the log relative error decreases linearly w.r.t. the number of iterations.

The statistical stability of the final estimates by Algorithm~\ref{algo:lowrank+sparse} applied to tensor PCA with heavy-tailed noise is demonstrated in the bottom panel of Figure~\ref{fig:heavytailpca}. Each curve represents the average relative error of $\hat \T_{l_{\max}}$ based on $5$ simulations, and the error bar shows the confidence region by one empirical standard deviation. Based on these plots, we observe that for each fixed $\theta$ (or $\sigma_z^2$, equivalently), we need to choose $\alpha$ carefully to achieve the best performance. This is reasonable since in the heavy-tail noise setting, we do not know the sparsity of outliers. Also, the figure shows that Algorithm~\ref{algo:lowrank+sparse} is stable for different $\alpha$ and $\theta$.

\begin{figure}
	\centering
	\begin{subfigure}[b]{.98\linewidth}
		\includegraphics[width=0.5\textwidth]{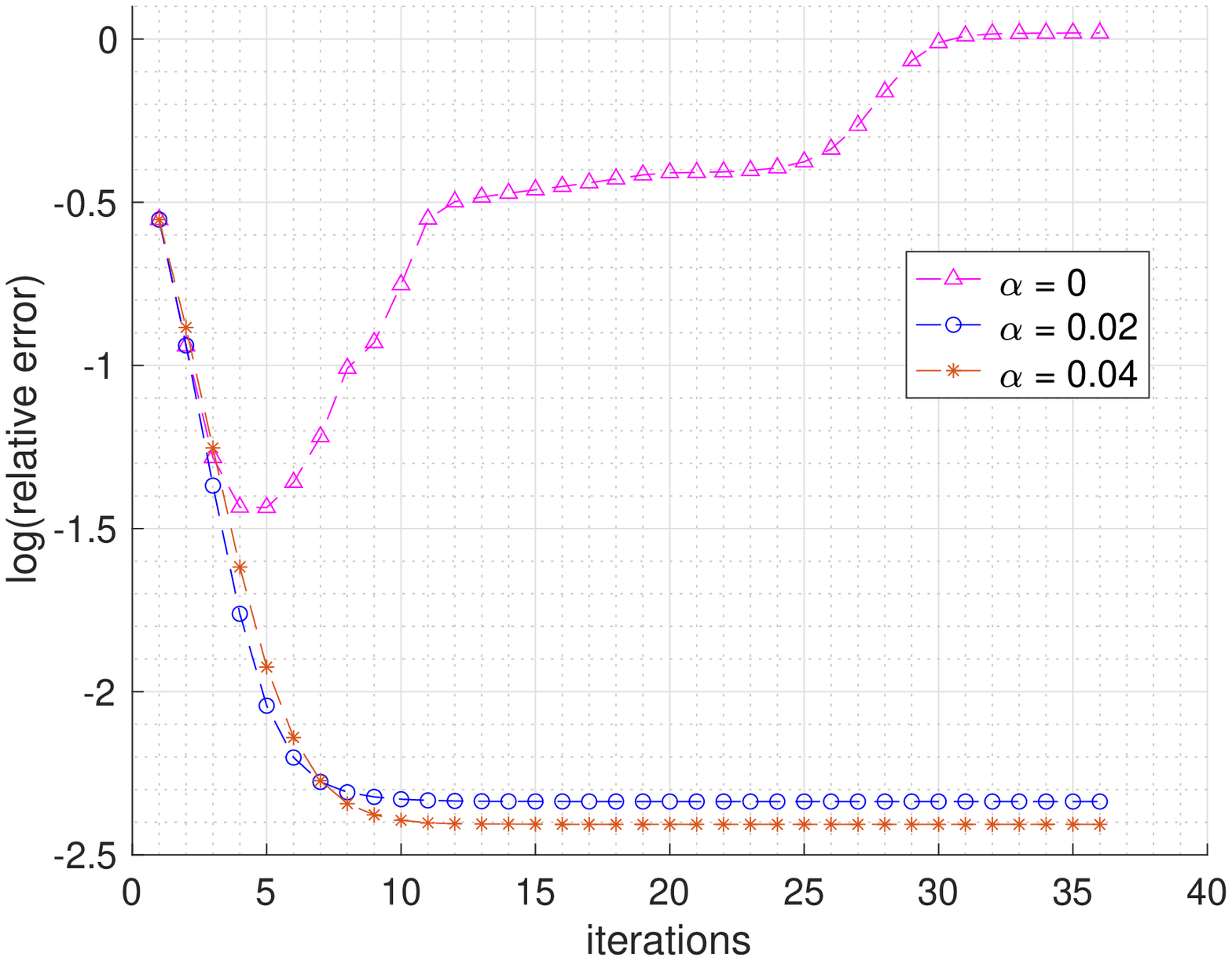}
		\includegraphics[width=0.5\textwidth]{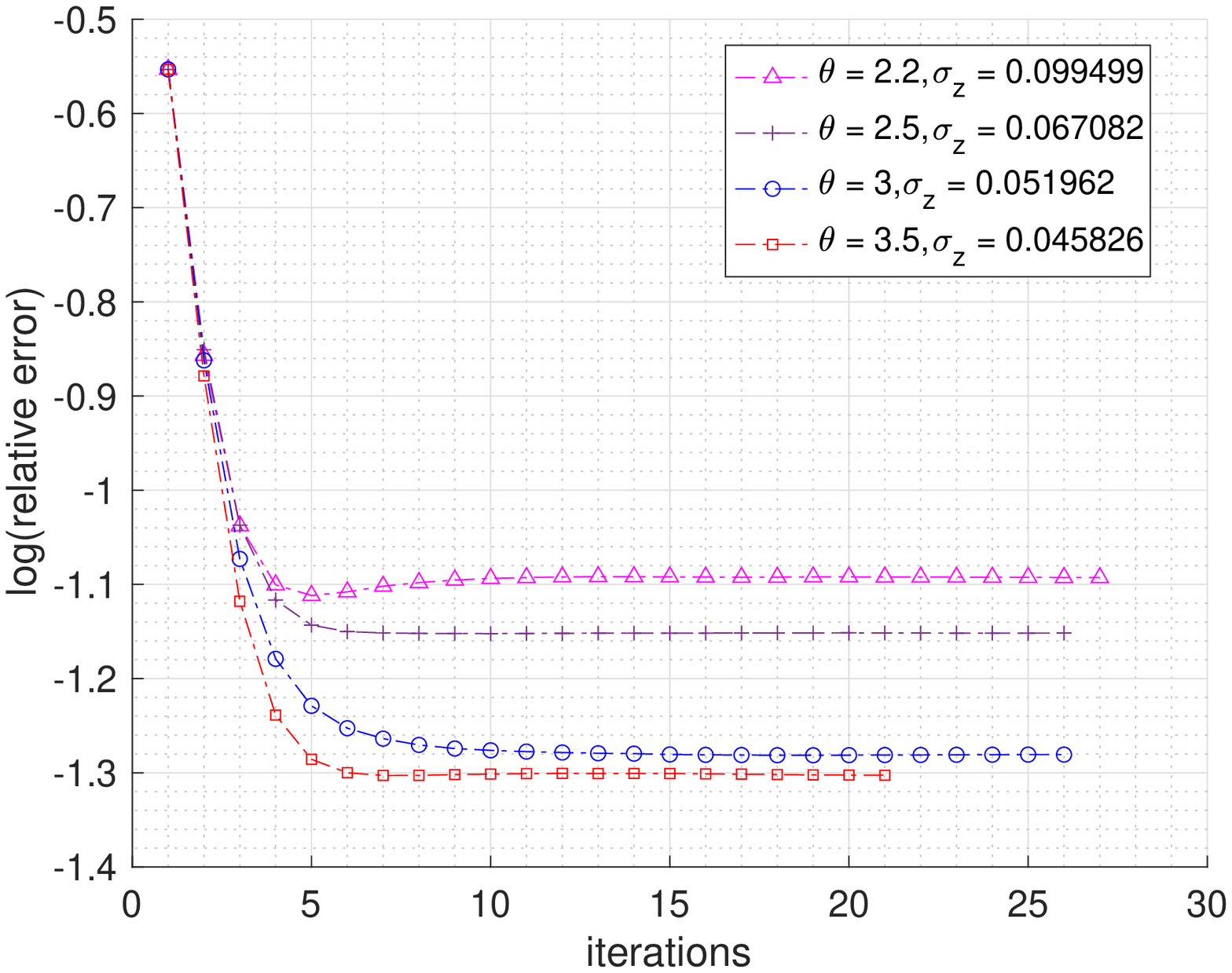}
		\caption{Left: Change of $\alpha$, $\theta = 2.2 (\sigma_z = 0.332)$; Right: Change of $\theta$, $\alpha = 0.01$}
	\end{subfigure}
	\begin{subfigure}[b]{.98\linewidth}
		\centering
		\includegraphics[width=0.5\textwidth]{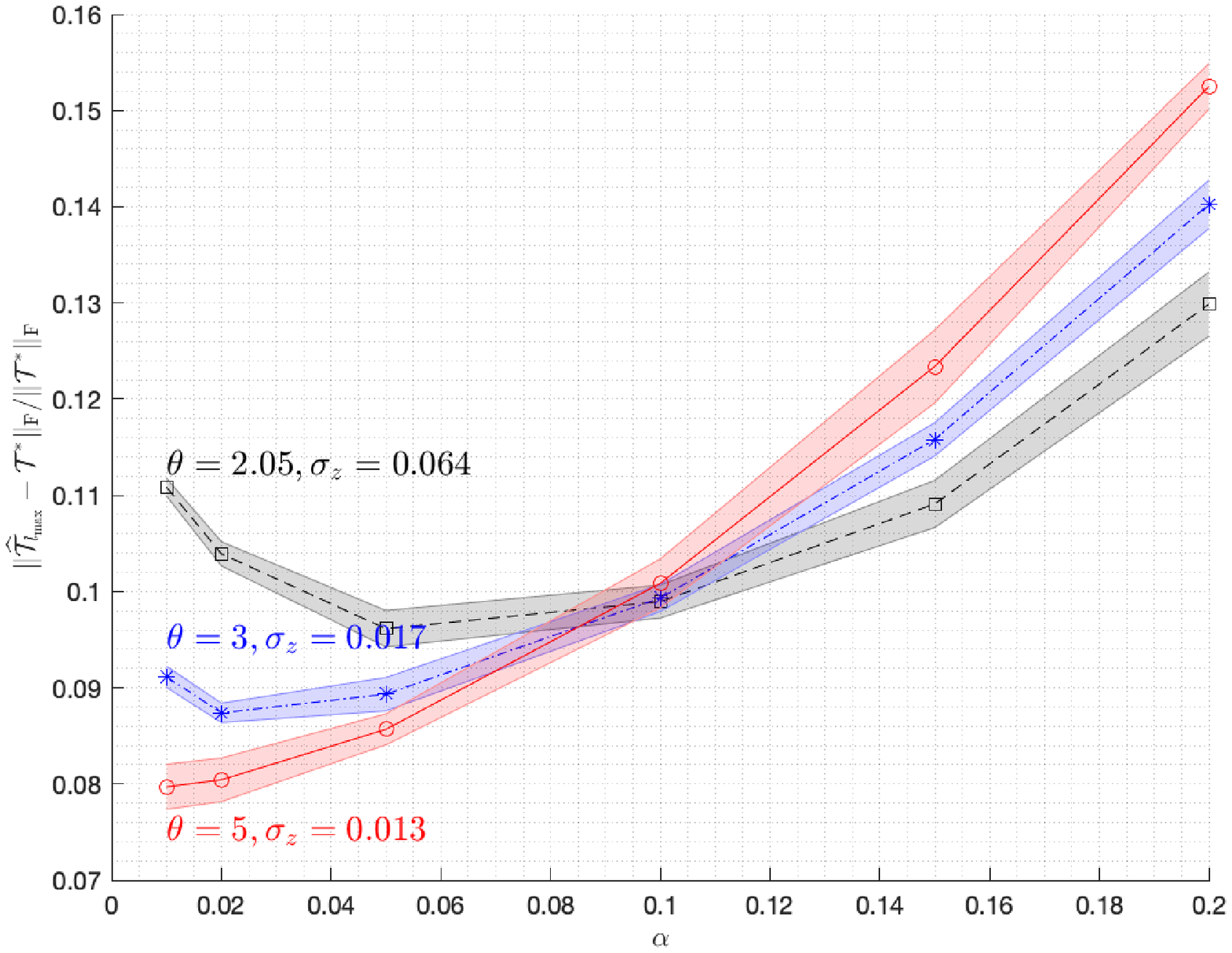}
		\caption{Change of $\theta$}
	\end{subfigure}
	\caption{Performances of Algorithm~\ref{algo:lowrank+sparse} for tensor PCA with heavy-tailed noise. The low-rank $\T^{\ast}$ has size $d\times d\times d$ with $d=100$ and has Tucker ranks $\br=(2,2,2)^{\top}$. The relative error on upper panels is defined by $\|\hat\T_l-\T^{\ast}\|_{\rm F}/\|\T^{\ast}\|_{\rm F}$. The error bars on the lower panels are based on $1$ standard deviation from 5 replications. Here the default choice of $\gamma$ is 2.}
	\label{fig:heavytailpca}
\end{figure}

\subsection{Binary Tensor Learning}\label{sec:num_binary}
In the binary tensor setting, we generate the low-rank tensor $\T^{\ast}\in\RR^{d\times d\times d}$ with $d=100$ and Tucker ranks $\br=(2,2,2)^{\top}$  from the HOSVD of a trimmed standard normal tensor. But here we did a scaling to $\T^*$ so that the singular value $\bsigma\approx 300$ and $\usigma\approx 100$. Given a sparsity level $\alpha\in(0,1)$, the entries of sparse tensor $\S^{\ast}$ are i.i.d. sampled from ${\rm Be}(\alpha)\times {\rm N}(0,1)$, which ensures $\S^{\ast}\in \SS_{O(\alpha)}$ with high probability. We generate the tensor $\T^*$ and  $\S^*$ in this way in order to meet the requirements of Assumption~\ref{assump:binary_tensor}. In the following experiments, we are considering the logistic link function with the scaling parameter $\sigma$, i.e., $p(x) = (1+e^{-x/\sigma})^{-1}$. The default choice of $\gamma$ is $1.1$, $\kprune=1$ and $\mu_1$ is set as previously. The convergence performances of $\log(\|\hat \T_l-\T^{\ast}\|_{\rm F}/\|\T^{\ast}\|_{\rm F})$ by Algorithm~\ref{algo:lowrank+sparse} are examined and presented in the top two panels of Figure~\ref{fig:binary}. 

The top-left plot in Figure~\ref{fig:binary} shows the effect of $\alpha$ on the convergence of Algorithm~\ref{algo:lowrank+sparse}. From the figure, it is clear that the error of final estimates $\hat\T_{l_{\max}}$ is related to $\alpha$. This again verifies the results in Theorem~\ref{thm:binary_tensor}. In the top-right plot in Figure~\ref{fig:binary}, we can see the error of the final estimates increases as the parameter $\gamma$ becomes larger.
All these experiments show that Riemannian gradient descent converges fast and there are stages when the log relative error decreases linearly w.r.t. the number of iterations.

The statistical stability of the final estimates by Algorithm~\ref{algo:lowrank+sparse} is demonstrated in the bottom panel of Figure~\ref{fig:binary}. Each curve represents the average relative error of $\hat \T_{l_{\max}}$ based on $5$ simulations, and the error bar shows the confidence region by one empirical standard deviation. From these plots, we observe that the standard deviations of $\|\hat\T_{l_{\max}}-\T^{\ast}\|_{\rm F}$ grow as the noise level, the sparsity level $\alpha$  or the tuning parameter $\gamma$ increases. 

\begin{figure}
	\centering
	\begin{subfigure}[b]{.98\linewidth}
		\includegraphics[width=0.5\textwidth]{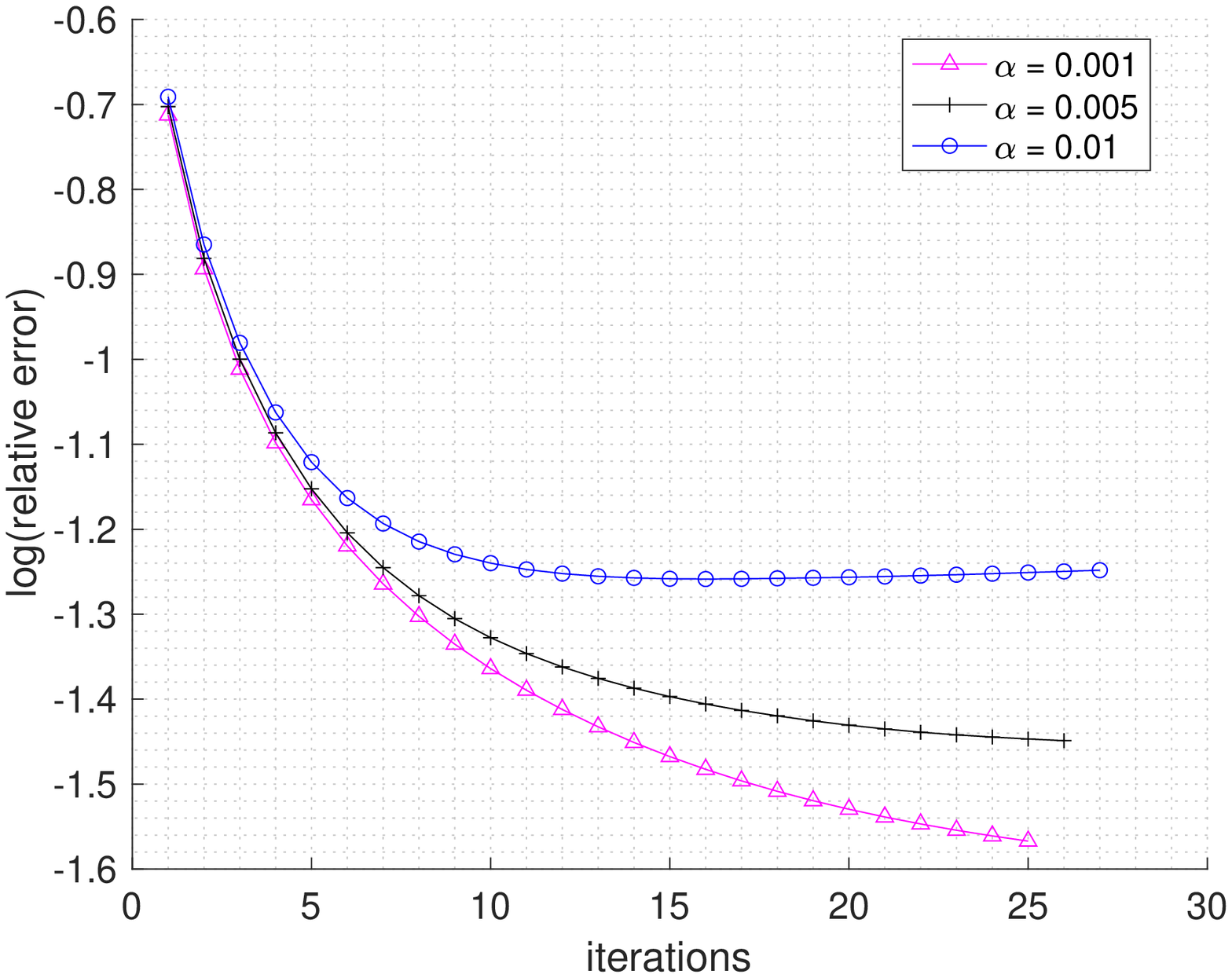}
		\includegraphics[width=0.5\textwidth]{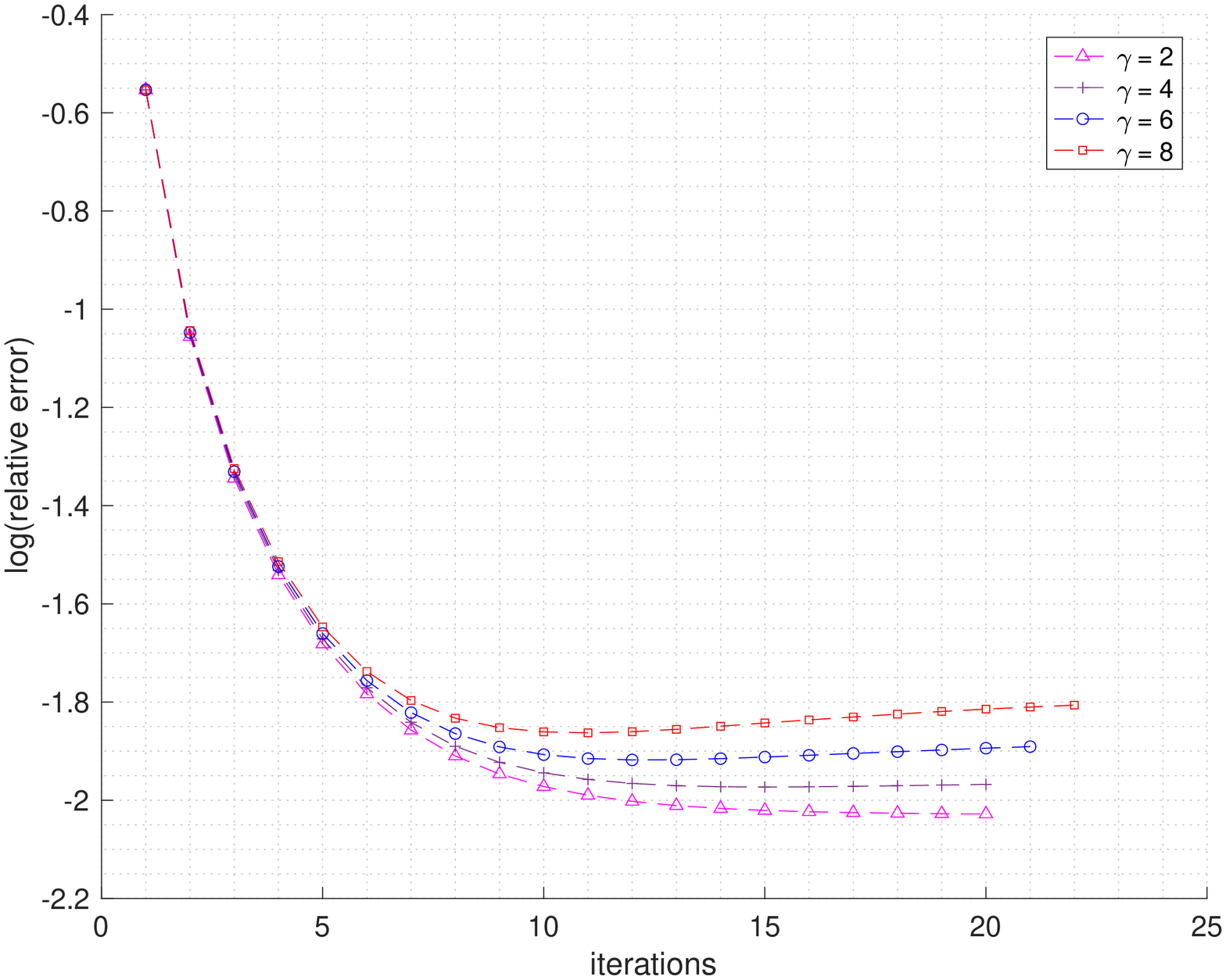}
		\caption{Left: Change of sparsity $\alpha$, $\sigma=1, \gamma=1.1$; Right: Change of $\gamma$, $\alpha=0.001, \sigma = 1$}
	\end{subfigure}
	\begin{subfigure}[b]{.98\linewidth}
		\includegraphics[width=0.5\textwidth]{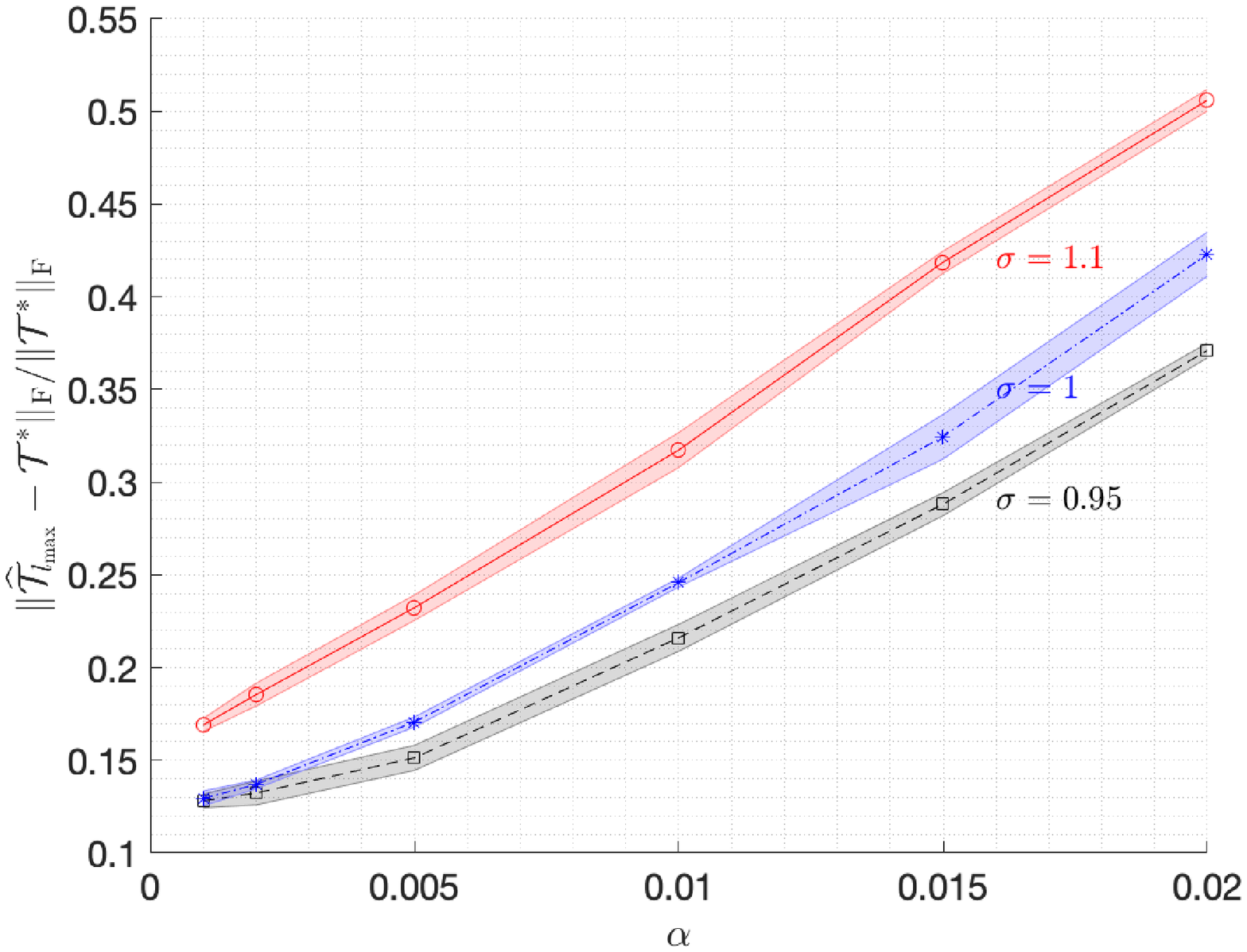}
		\includegraphics[width=0.5\textwidth]{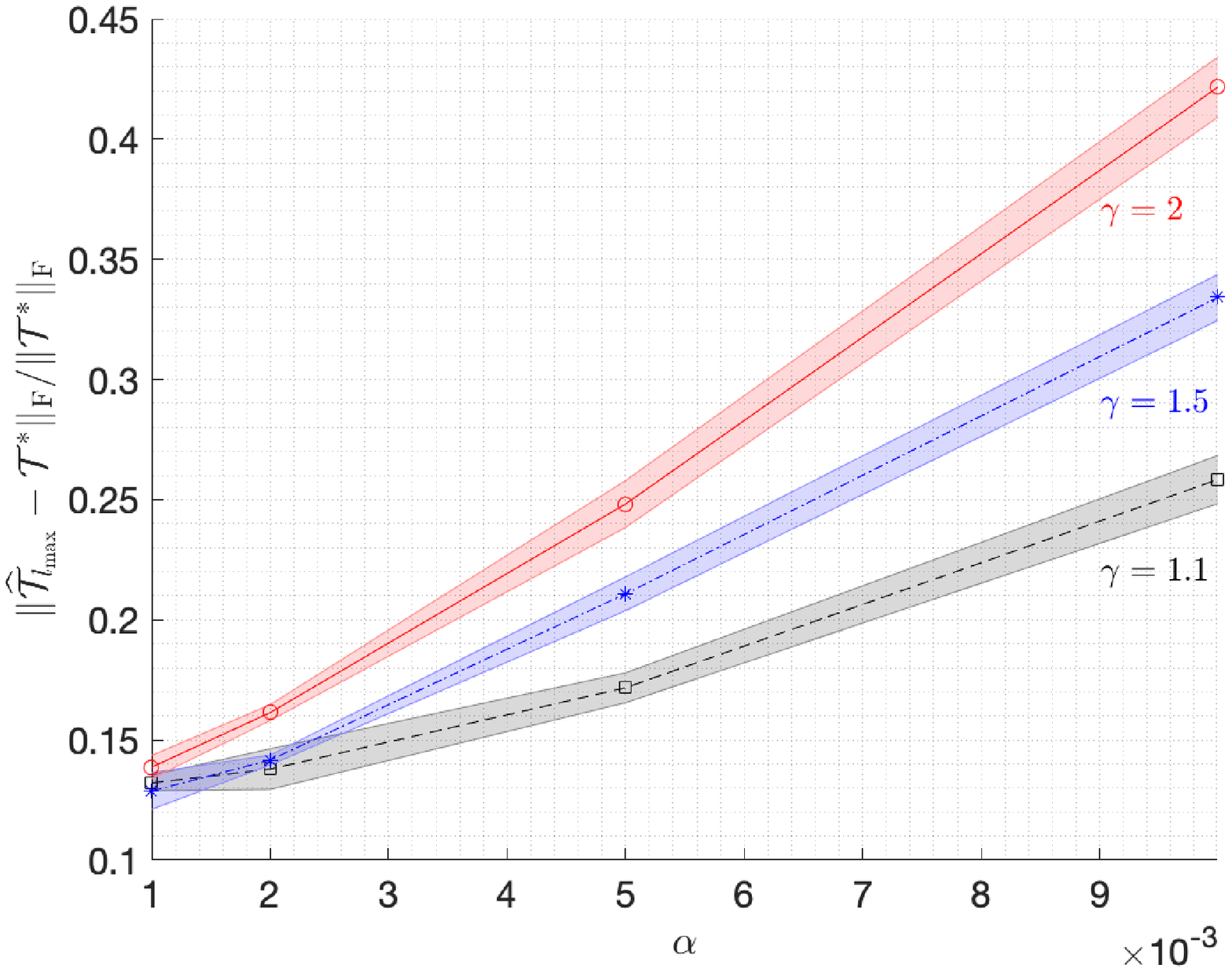}
		\caption{Left: Change of $\sigma$, $\gamma=1.1$; Right: Change of $\gamma$, $\alpha=0.001$}
	\end{subfigure}
	\caption{Performances of Algorithm~\ref{algo:lowrank+sparse} for binary tensor learning. The low-rank $\T^{\ast}$ has size $d\times d\times d$ with $d=100$ and has Tucker ranks $\br=(2,2,2)^{\top}$. The relative error on left panels is defined by $\|\hat\T_l-\T^{\ast}\|_{\rm F}/\|\T^{\ast}\|_{\rm F}$. The error bars on the bottom panels are based on $1$ standard deviation from 5 replications. The default choice of $\gamma$ is $1.1$.}
	\label{fig:binary}
\end{figure}

\subsection{Tensor Poisson Robust PCA}\label{sec:num_poisson}
In the Poisson tensor RPCA case, we generate $\T^*\in\RR^{d\times d\times d}$ with $d = 100$ and Tucker rank $\br = (2,2,2)^{\top}$ such that $\|\T^*\|_{\ell_{\infty}} = 0.5$. Meanwhile, the sparse outliers $\S^*$ is generated such that all its entries are i.i.d. sampled from ${\rm Be}(\alpha)\times {\rm N}(0,1)$ and scaled such that $\|\S^*\|_{\ell_{\infty}} = 0.5$. Throughout the experiments, both $\zeta$ and $\kprune$ is set to 0.5, and the default choice of $\gamma$ is $1.1$.

In the first experiment, we fix the intensity $I = 10$ and change the sparsity level. 
The convergence performances of $\log(\fro{\hat\T_l - \T^*}/\fro{\T^*})$ by Algorithm \ref{algo:lowrank+sparse} is displayed in the left panel of Figure \ref{fig:poisson}. In the second experiment, for different values of $\alpha$ and $I$, we conduct 5 i.i.d. instances and plot the error bar. The results are displayed in the right panel of Figure \ref{fig:poisson}. 

\begin{figure}
	\includegraphics[width=0.5\textwidth]{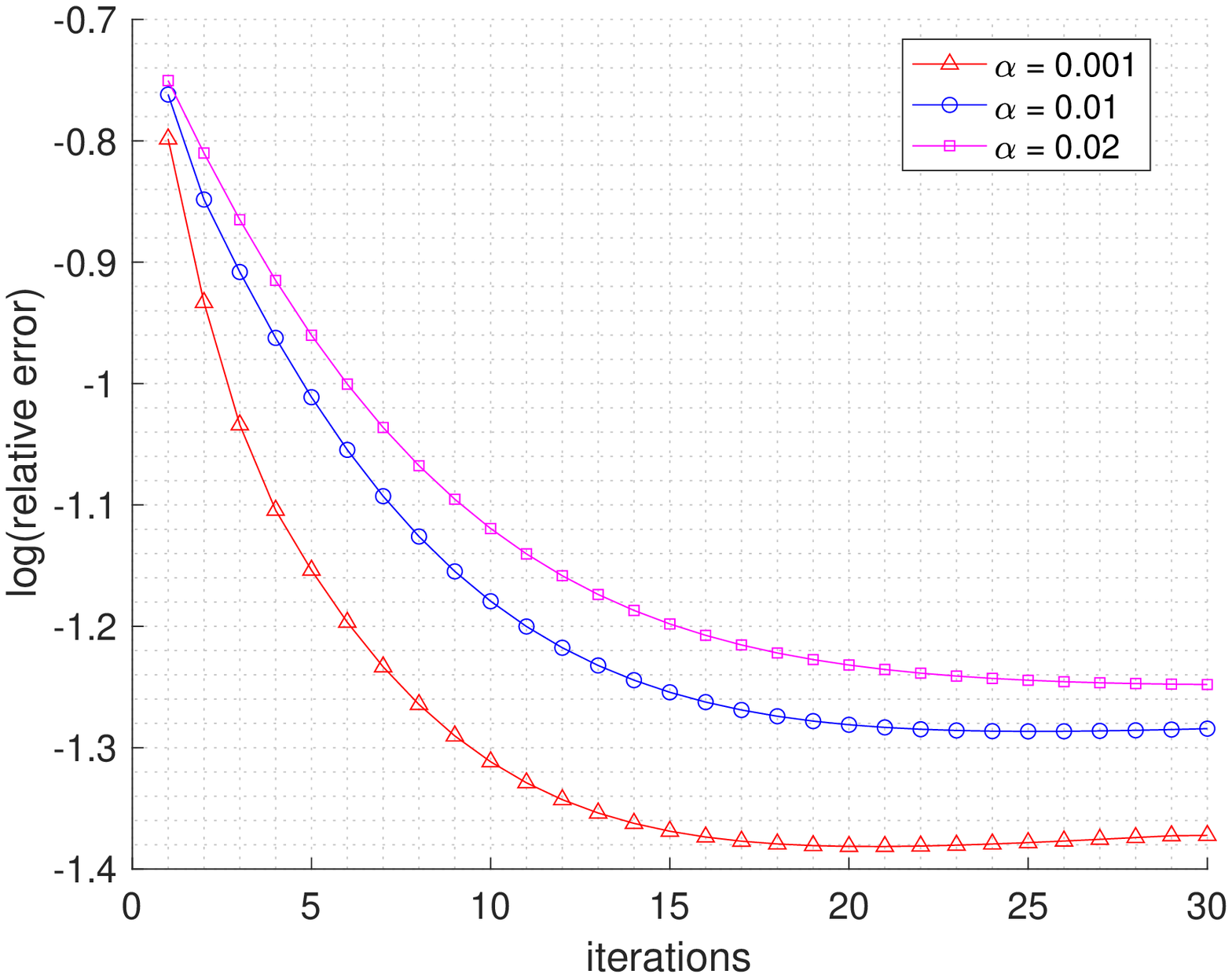}
	\includegraphics[width=0.5\textwidth]{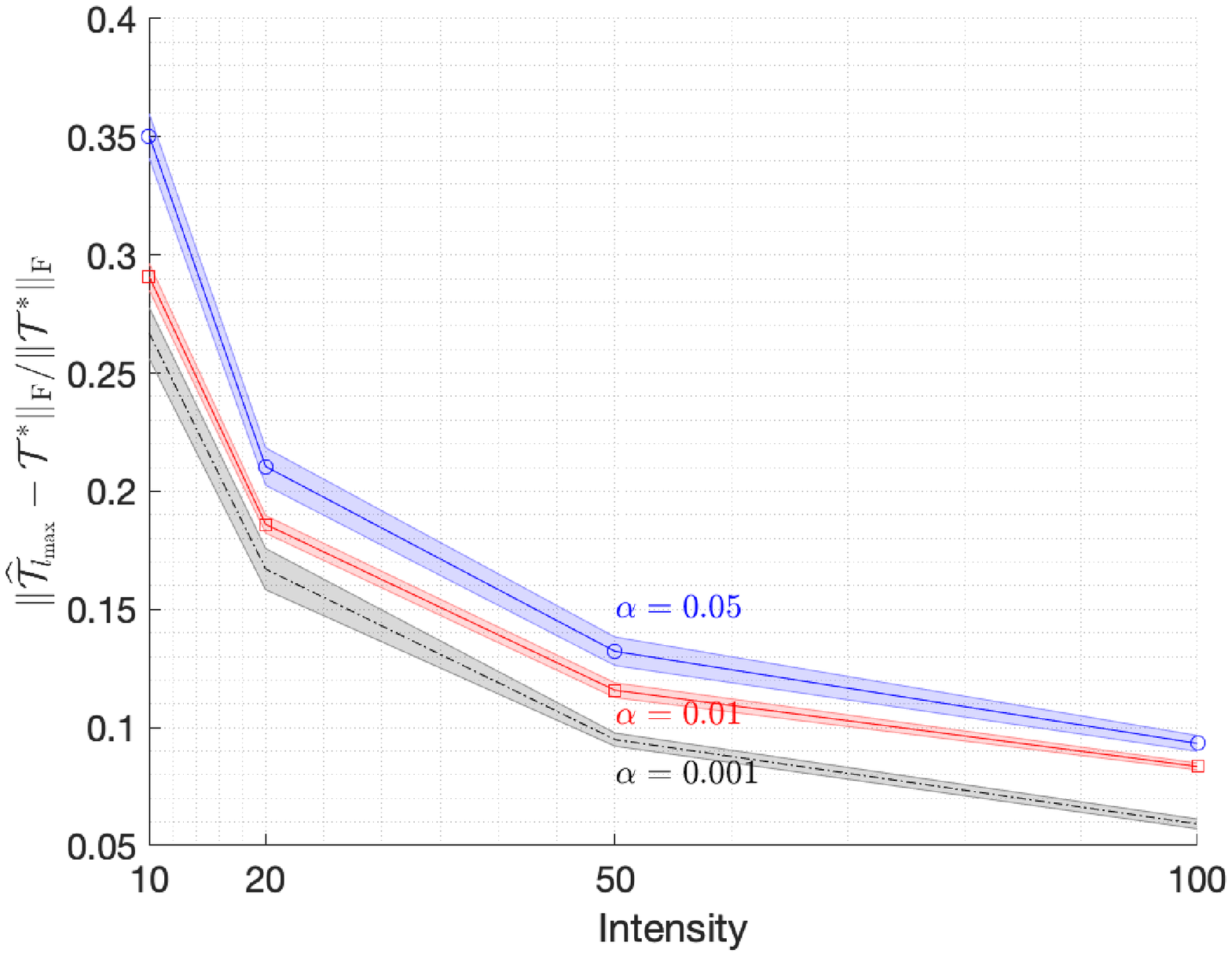}
	\caption{Performance of Algorithm \ref{algo:lowrank+sparse}  for tensor Poisson RPCA. The Tucker rank of $\T^*$ is $\br = (2,2,2)^{\top}$.
		Left: Convergence behaviors with different $\alpha$ and $I$ is fixed with $I=10$; Right: Error bar with each setting repeated 5 i.i.d. times. Here $\gamma=1.1$ and $\kprune=0.5$.}
	\label{fig:poisson}
\end{figure}

\section{Real Data: Statisticians Hypergraph Co-authorship Network}
This dataset \citep{ji2016coauthorship} contains the co-authorship relations of $3607$ statisticians based on $3248$ papers published in four prestigious statistics journals during 2003-2012. The co-authorship network thus has $3607$ nodes and two nodes are connected by an edge if they collaborated on at least one paper. A giant connected component of this network consisting of $236$ nodes is seen to be the ``High-Dimensional Data Analysis" community. They also carried out community detection  analysis to discover substructures in this giant component. See more details in \citep{ji2016coauthorship}.

We analyze the substructures of the giant component by treating it as a hypergraph co-authorship network. These $236$ statisticians co-authored $542$ papers\footnote{There are $328$ single-authored papers. They provide no information to co-authorship relations, and are left out in our analysis.}, among which $356$ papers have two co-authors, $162$ papers have three co-authors and $24$ papers have four co-authors. A $3$-uniform hypergraph co-authorship network is constructed by, for $i\neq j\neq k$, adding the hyperedge $(i,j,k)$ if the authors $i,j,k$ co-authored at least one paper, and adding the hyperedges $(i,i,j)$ and $(i,j,j)$ if the authors $i,j$ co-authored at least one paper. The hyperedges are {\it undirected} resulting into a symmetric adjacency tensor $\A$. We adopt the framework from Section~\ref{sec:rpca} to learn the latent low-rank tensor $\hat\T$ in $\A$, which is used to detect communities in the giant component. We emphasize that our primary goal is to present the new findings by taking into consideration of higher-order interactions among co-authors and applying novel robust tensor methods. It is not our intention to label an author with a certain community.

The Tucker ranks are set as $(4,4,4)$ and sparsity ratio $\alpha$ is varied at $\{0, 10^{-4}, 5\times 10^{-4}\}$. The number of communities is set at $K=3$ and the algorithm is initialized by the HOSVD of $\A$. To uncover community structures, we apply spectral clustering to the singular vectors of $\hat \T$. The node degrees are severely heterogeneous with Peter Hall, Jianqing Fan and Raymond Carroll being the top-$3$ statisticians in terms of $\#$ of co-authors. The naive spectral clustering often performs poorly in the existence of heterogeneity, skewing to the high-degree nodes. Indeed, the top-left plot in Figure~\ref{fig:SN_author} shows that the naive spectral clustering identifies these three statisticians as the corners in a triangle, and puts Peter Hall in a single community. To mitigate the influence of node heterogeneity, we apply SCORE \citep{jin2015fast} for community detection, which uses the leading singular vector of $\hat\T$ as normalization. 

\begin{figure}
\centering
	\begin{subfigure}[b]{0.98\linewidth}
		\includegraphics[width=0.5\textwidth]{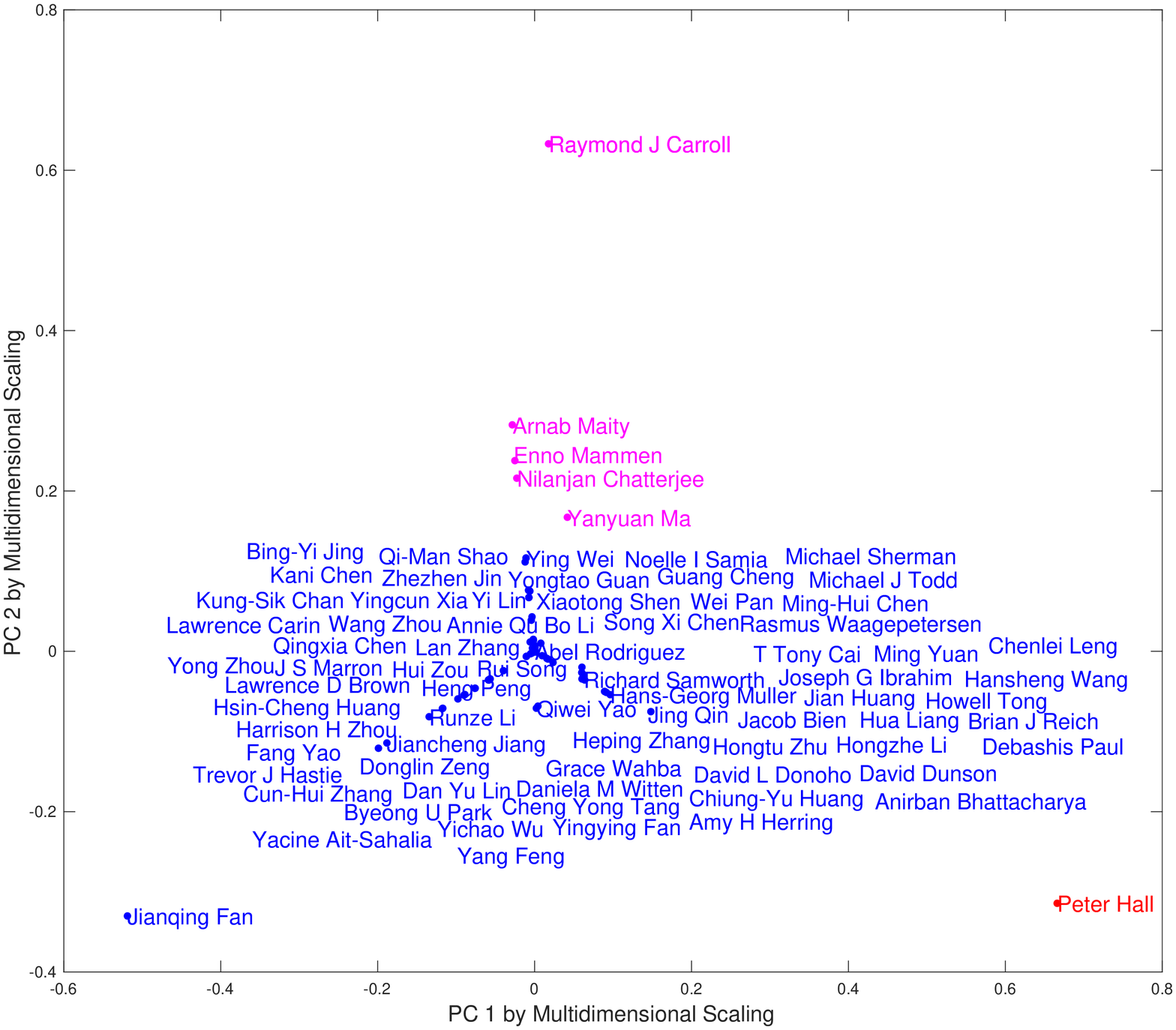}
		\includegraphics[width=0.5\textwidth]{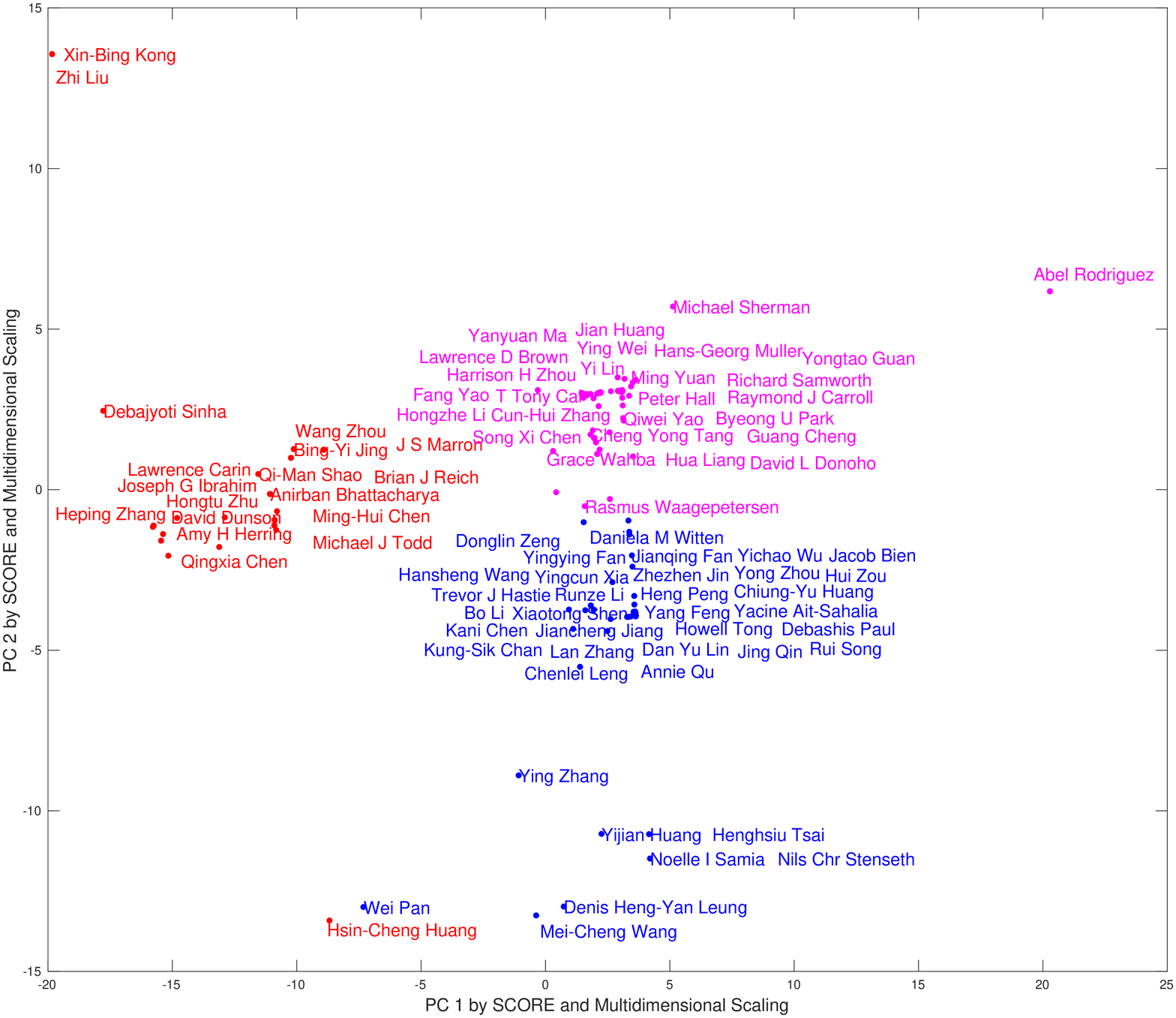}
		\includegraphics[width=0.5\textwidth]{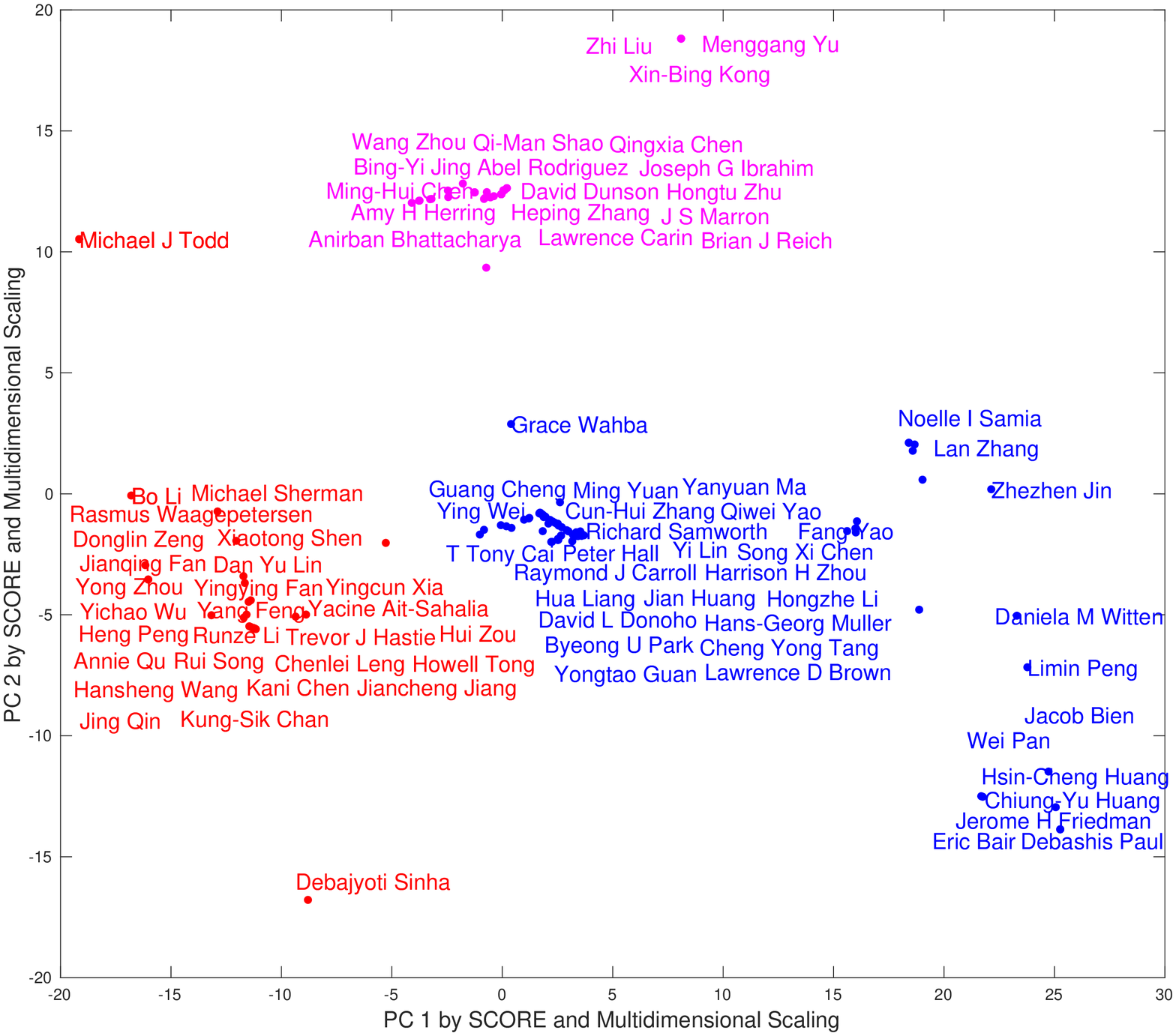}
		\includegraphics[width=0.5\textwidth]{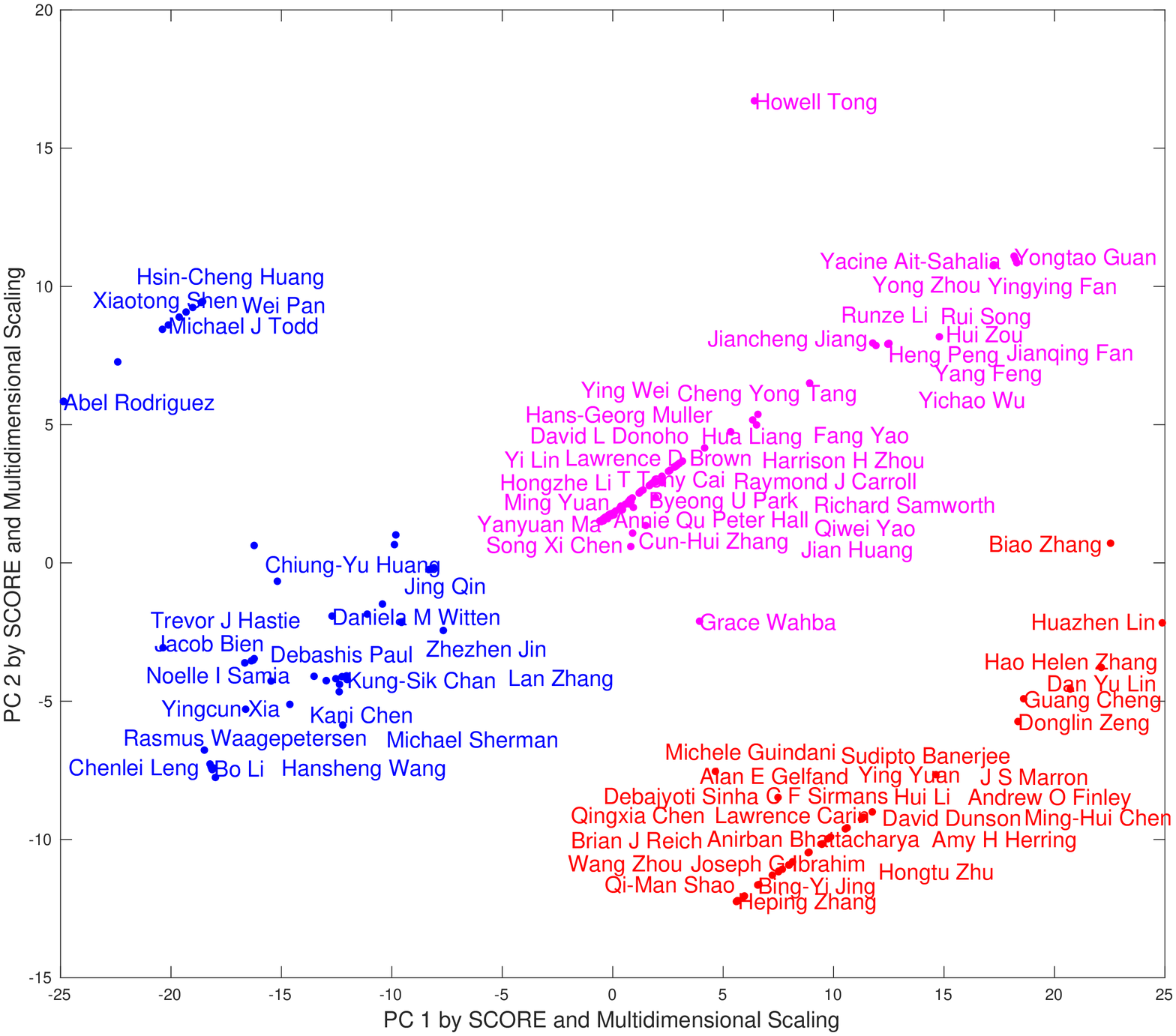}
		\caption{Top-left: $\alpha=10^{-4}$ and naive spectral clustering; top-right: $\alpha=0$ and SCORE; bottom-left: $\alpha=10^{-4}$ and SCORE; bottom-right: $\alpha=5\times 10^{-4}$ and SCORE. }
	\end{subfigure}
\caption{Sub-structures detected in the  ``High-Dimensional Data Analysis" community based on the hypergraph co-authorship network. The Tucker ranks are set as $(4,4,4)$ with varied sparsity ratio at $\{0, 10^{-4}, 5\times 10^{-4}\}$ and the algorithm is initialized by the HOSVD of adjacency tensor $\A$.}
\label{fig:SN_author}
\end{figure}

The community structures found by SCORE are displayed in Figure~\ref{fig:SN_author}. The top-right plot shows the three clusters identified by SCORE when the sparsity ratio is zero. The three communities are: 1). ``North Carolina" group including researchers from Duke University, University of North Carolina and North Carolina State University, together with their close collaborators such as Debajyoti Sinha, Qi-Man Shao, Bing-Yi Jing, Michael J Todd and etc.; 2). ``Carroll-Hall" group including researchers in non-parametric and semi-parametric statistics, functional estimation and high-dimensional statistics, together with collaborators; 3). ``Fan and Others" group\footnote{We name it the ``Fan and Others" group simply because many researchers in this group are the co-authors of Jianqing Fan. It is not our intention to rank/label the authors.} including {\it primarily} the researchers collaborating closely with Jianqing Fan or his co-authors, and other researchers who do not {\it obviously} belong to the first two groups. We note that the fields of researchers in ``Fan and Others" group are quite diverse, some of which overlap with those in ``Carroll-Hall" group and ``North Carolina" group. 
However, unlike the results in \citep{ji2016coauthorship}, the top-right plot in Figure~\ref{fig:SN_author} does not cluster the ``Fan and Others" group into either the ``North Carolina" group or ``Carroll-Hall" group. 

We then set the sparsity ratio of $\hat\S$ by $\alpha=10^{-4}$. The communities identified by SCORE based on the singular vectors of $\hat\T$ are illustrated in the bottom-left plot of Figure~\ref{fig:SN_author}. Compared with the top-right plot ($\alpha=0$), the three communities displayed in the bottom-left plot largely remain the same. But the group memberships of some authors do change. Notably, Debajyoti Sinha and Michael J Todd move from the ``North Carolina" group to ``Fan and Others" group; Abel Rodriguez moves from the ``Carroll-Hall" group to ``North-Carolina" group; several authors (e.g. Daniela M Witten, Jacob Bien, Pan Wei, Chiung-Yu Huang, Debashis Paul, Zhezhen Jin, Lan Zhang and etc.) move from the ``Fan and Others" group to ``Carroll-Hall" group; Hsin-Cheng Huang moves from the ``North Carolina" group to ``Carroll-Hall" group; Rasmus Waggepetersen moves from the ``Carroll-Hall" group to ``Fan and Others" group. These changes of memberships suggest that these authors may not have strong ties to the ``North Carolina", ``Carroll-Hall" group or be the co-authors of Jianqing Fan. It may be more reasonable that these authors constitute a separate group. 

This indeed happens when the sparsity ratio $\alpha$ increases to a certain level. The bottom-right plot of Figure~\ref{fig:SN_author} shows the clustering result of SCORE when $\alpha=5\times 10^{-4}$. Compared with the top-right ($\alpha=0$) and bottom-left ($\alpha=10^{-4}$) plots, the community structure has a significant change. Indeed, the ``Fan and Others" group now splits into a ``Fan" group including Jianqing Fan and his co-authors, and an ``Others" group including the researchers who do not have obvious ties with ``Fan" group. 
Moreover, the ``Fan" group merges into the ``Carroll-Hall" group, which coincides with the clustering result of SCORE when applied onto the graph co-authorship network (Fig. 6 in \citep{ji2016coauthorship}). Consequently, we name the three communities in the top-right plot by the ``North Carolina", ``Carroll-Fan-Hall" and ``Others" group. Interestingly, many of the authors in the ``Others" group are those whose memberships change when the sparsity ratio $\alpha$ increases from $0$ to $10^{-4}$. See the top-right and bottom-left plots of Figure~\ref{fig:SN_author}. 
In addition, we observe that, as $\alpha$ increases from $10^{-4}$ to $5\times 10^{-4}$, Donglin Zeng and Dan Yu Lin in the ``Fan and Others" group moves to ``North Carolina" group. This might be more reasonable since they both work at the University of North Carolina.

\section{Proofs of theorems}
\subsection{Proof of Theorem \ref{thm:lowrank+sparse}}
We prove the theorem by induction on $\fro{\hat\T_l - \T^*}$ and $\fro{\hat\S_l -\S^*}$ alternatively. From the initialization condition we have $\fro{\hat\T_0 - \T^*} \leq c_{1,m}\min\{\frac{\delta^2}{\sqrt{\rmax}}, (\kappa_0^{2m}\sqrt{\rmax})^{-1}\}\cdot\usigma$ and $\hat\T_0\in\BB_{\infty}^{\ast}$ is $(2\mu_1\kappa_0)^2$-incoherent.
\paragraph*{Step 1: Bounding $\|\hat \S_l-\S^{\ast}\|_{\rm F}$ for all $l\geq 0$.} 
Suppose we have $\hat\T_l\in\BB_{\infty}^{\ast}$ is $(2\mu_1\kappa_0)^2$-incoherent and $\fro{\hat\T_l - \T^*} \leq c_{1,m}\min\{\frac{\delta^2}{\sqrt{\rmax}}, (\kappa_0^{2m}\sqrt{\rmax})^{-1}\}\cdot\usigma$.

Now we estimate $\fro{\hat\S_l - \S^*}$. Denote $\Omega_l  = \text{supp}(\hat\S_l)$ and $\Omega^* = \text{supp}(\S^*)$.  
For $\forall\omega \in \Omega_l$, from the construction of $\hat\S_l$ in Algorithm \ref{algo:hd_thre}, we have by the definition of $\textsf{Err}_{\infty}$,
\begin{align}\label{est:init:0000}
	|\subw{\nabla\frakL(\hat\T_l + \hat\S_l)}| \leq \min\nolimits_{\|\X\|_{\ell_{\infty}}\leq \kprune} \|\nabla \frakL(\X)\|_{\ell_\infty}\leq \textsf{Err}_{\infty}
\end{align}
From Assumption \ref{assump:sparse}, we get 
\begin{align}\label{eq:init:proof_sparse_Somega_1}
	|\subw{\nabla\frakL(\hat\T_l +\hat \S_l)} - \subw{\nabla\frakL(\hat\T_l+\S^*)}| \geq b_l|\subw{\hat\S_l-\S^*}|.
\end{align}
Note that to use (\ref{eq:init:proof_sparse_Somega_1}), we shall verify the neighborhood condition.
From the upper bound of $\fro{\hat\T_l - \T^*}$ we have $\fro{\hat\T_l - \T^*} \leq \usigma/8$, and $\hat\T_l$ is $(2\mu_1\kappa_0)^2$-incoherent.
Therefore, from Lemma~\ref{lemma:entrywise}, we have:
$$|\subw{\hat\T_l - \T^*}|^2 \leq C_{1,m}\rmax^m\dmin^{-(m-1)}(\mu_1\kappa_0)^{4m}\fro{\hat\T_l - \T^*}^2.$$
So we have 
$$|\subw{\hat\T_l - \T^*}| \leq C_{1,m}\sqrt{\frac{\rmax^m}{\dmin^{m-1}}}(\mu_1\kappa_0)^{2m}\fro{\hat\T_l - \T^*}\leq C_{1,m}\mu_1^{2m}\sqrt{\frac{\rmax^{m-1}}{\dmin^{m-1}}}\underline{\lambda},$$
where the last inequality is from the upper bound of $\fro{\hat\T_l - \T^*}$.
As a result, we have
$$
\big| [\hat\T_l+\hat\S_l-\T^{\ast}-\S^{\ast}]_{\omega}\big|\leq \big| [\hat\T_l-\T^{\ast}]_{\omega}\big|+|[\hat\S_l]_{\omega}|+|[\S^{\ast}]_{\omega}|\leq C_{1,m}\mu_1^{2m}\sqrt{\frac{\rmax^{m-1}}{\dmin^{m-1}}}\usigma+\kprune+\|\S^{\ast}\|_{\ell_\infty}. 
$$
Thus, both $\hat\T_l+\hat\S_l$ and $\hat\T_l+\S^{\ast}$ belong to the ball $\BB_{\infty}^*$ and thus (\ref{eq:init:proof_sparse_Somega_1}) holds. 

As a result of \eqref{est:init:0000} and \eqref{eq:init:proof_sparse_Somega_1}, we get for any $\omega\in\Omega_l$
$$
b_l\big|[\hat\S_l-\S^{\ast}]_{\omega} \big|\leq \big|[\nabla \frakL(\hat\T_l+\S^{\ast})]_{\omega}\big|+\errinf.
$$
Therefore,
\begin{align}\label{est:init:6}
	\fro{\calP_{\Omega_l}&(\hat\S_l - \S^*)}^2 \leq \frac{2}{b_l^2}\fro{\calP_{\Omega_l}(\nabla\frakL(\hat\T_l+\S^*))}^2+\frac{2|\Omega_l|}{b_l^2}\errinf^2\notag\\
	&= \frac{2}{b_l^2}\fro{\calP_{\Omega_l}(\nabla\frakL(\hat\T_l+\S^*)) -\calP_{\Omega_l}(\nabla\frakL(\T^*+\S^*)) +\calP_{\Omega_l}(\nabla\frakL(\T^*+\S^*))}^2+\frac{2|\Omega_l|}{b_l^2}\errinf^2\notag\\
	&\leq \frac{4}{b_l^2}\fro{\calP_{\Omega_l}(\nabla\frakL(\hat\T_l+\S^*)) -\calP_{\Omega_l}(\nabla\frakL(\T^*+\S^*))}^2 + \frac{4}{b_l^2}\fro{\calP_{\Omega_l}(\nabla\frakL(\T^*+\S^*))}^2+\frac{2|\Omega_l|}{b_l^2}\errinf^2\notag\\
	&\leq \frac{4b_u^2}{b_l^2}\fro{\calP_{\Omega_l}(\hat\T_l - \T^*)}^2 + \frac{6|\Omega_l|}{b_l^2}\errinf^2,
\end{align}
where the last inequality is due to $\fro{\calP_{\Omega_l}(\nabla\frakL(\T^*+\S^*))}^2 \leq |\Omega_l| \errinf^2$ and Assumption~\ref{assump:sparse} since $\hat\T_l+\S^{\ast}\in\BB_{\infty}^*$.

From \eqref{est:init:6}, Lemma~\ref{lemma:est:proTl}, we have
\begin{align}\label{est:init:S:1}
	\fro{\calP_{\Omega_l}( \hat\S_l - \S^*)}^2 
	&\leq \frac{C_{2,m}b_u^2}{b_l^2}(\mu_1\kappa_0)^{4m}\rmax^m\alpha \fro{\hat\T_l - \T^*}^2 + \frac{6|\Omega_l|}{b_l^2}\errinf^2
\end{align}
here $C_{2,m}>0$ is an absolute constant depending only on $m$.

For $\forall \omega = (\omega_1,\ldots,\omega_m)\in \Omega^*\backslash\Omega_l$, we have $|\ijk{\hat\S_l - \S^*}| = |\ijk{\S^*}|$. Since the loss function is entry-wise by Assumption~\ref{assump:sparse}, we have $\subw{\nabla\frakL(\hat\T_l)} = \subw{\nabla\frakL(\hat\T_l+\hat\S_l)}$. Clearly, $\hat\T_l$ and $\hat\T_l+\S^{\ast}$ both belong to $\BB_{\infty}^*$, by Assumption \ref{assump:sparse} we get
$$
|\subw{\nabla\frakL(\hat\T_l)} - \subw{\nabla\frakL(\hat\T_l + \S^*)}| \geq b_l|\subw{\S^*}|.
$$
Now we bound $|\ijk{\hat\S_l - \S^*}|$ as follows. For any $\omega\in \Omega^{\ast}\backslash \Omega_l$, 
\begin{align*}
	|\ijk{\hat\S_l - \S^*}| &= |\ijk{\S^*}| \leq \frac{1}{b_l}|\ijk{\nabla\frakL(\hat\T_l)} - \ijk{\nabla\frakL(\hat\T_l + \S^*)}| \\
	&\leq \frac{1}{b_l}\left(|\ijk{\nabla\frakL(\hat\T_l)}|  + |\ijk{\nabla\frakL(\hat\T_l + \S^*)}|\right)\\
	&\leq \frac{1}{b_l}\left(|\ijk{\nabla\frakL(\hat\T_l)}|  + |\ijk{\nabla\frakL(\hat\T_l + \S^*) -\nabla\frakL(\T^* + \S^*)}| + |\ijk{\nabla\frakL(\T^* + \S^*)}|\right)\\
	&\leq \frac{1}{b_l}|\ijk{\nabla\frakL(\hat\T_l)}|  + \frac{b_u}{b_l}|\ijk{\hat\T_l - \T^*}| + 
	\frac{1}{b_l}\errinf,
\end{align*}
where the last inequality is again due to Assumption~\ref{assump:sparse} since $\hat\T_l + \S^* \in \BB_{\infty}^*$. Therefore we have
\begin{align}\label{est:init:44}
	\fro{\calP_{\Omega^*\backslash\Omega_l}(\hat\S_l - \S^*)}^2\leq \frac{2}{b_l^2}\fro{\calP_{\Omega^*\backslash\Omega_l}(\nabla\frakL(\hat\T_l))}^2 + \frac{4b_u^2}{b_l^2}\fro{\calP_{\Omega^*\backslash\Omega_l}(\hat\T_l - \T^*)}^2 + \frac{4}{b_l^2}|\Omega^{\ast}\backslash\Omega_l|\errinf^2
\end{align}
Since $\omega\in \Omega^*\backslash\Omega_l$, we have 
\begin{align}
	|\ijk{\nabla\frakL(\hat\T_l)}| \leq \max\nolimits_{i=1}^m |\be_{\omega_i}^{\top} \opM_i(\nabla\frakL(\hat\T_l))|^{(\gamma\alpha d_{i}^{-})}
\end{align}
Now since we have $\S^*\in \SS_{\alpha}$, we have
\begin{align}
	|\subw{\nabla\frakL(\hat\T_l)}| &\leq  \max\nolimits_{i=1}^m |\be_{\omega_i}^T \opM_i(\nabla\frakL(\hat\T_l+\S^*))|^{\left((\gamma-1)\alpha d_{i}^{-}\right)}\no\\
	&\leq \max\nolimits_{i=1}^m \Big|\be_{\omega_i}^{\top} \left(\opM_i(\nabla\frakL(\hat\T_l+\S^*)) - \opM_i(\nabla\frakL(\T^*+\S^*))\right)\Big|^{\left((\gamma-1)\alpha d_{i}^{-}\right)} + \errinf
\end{align}
Using AM-GM inequality, we have:
\begin{align}\label{est:init:5}
	|\subw{\nabla\frakL(\hat\T_l)}|^2 
	&\leq 2\max\nolimits_{i=1}^m \frac{\Big\|\be_{\omega_i}^{\top} \left(\opM_i(\nabla\frakL(\hat\T_l+\S^*)) - \opM_i(\nabla\frakL(\T^*+\S^*))\right)\Big\|_{\rm F}^2}{(\gamma-1)\alpha d_i^-}+2\errinf^2\notag\\
	&\leq 2\sum_{i=1}^m \frac{\Big\|\be_{\omega_i}^{\top} \left(\opM_i(\nabla\frakL(\hat\T_l+\S^*)) - \opM_i(\nabla\frakL(\T^*+\S^*))\right)\Big\|_{\rm F}^2}{(\gamma-1)\alpha d_i^-}+2\errinf^2
\end{align}
Now for all fixed $i\in[m]$, for all $\omega_i\in[d_i]$, $\omega_i$ appears at most $\alpha d_i^-$ times since $\Omega^*\backslash\Omega_l$ is an $\alpha$-fraction set. This observation together with \eqref{est:init:5} lead to the following:
\begin{align}\label{est:init:12}
	\fro{\calP_{\Omega^*\backslash\Omega_l}(\nabla\frakL(\hat\T_l))}^2 &\leq 2\sum_{i=1}^m\frac{\fro{\nabla\frakL(\hat\T_l+\S^*) - \nabla\frakL(\T^*+\S^*)}^2}{\gamma-1}+ 2|\Omega^*\backslash\Omega_l| \errinf^2\notag\\
	&\leq \frac{2m b_u^2}{\gamma-1}\fro{\hat\T_l - \T^*}^2 + 2|\Omega^*\backslash\Omega_l| \errinf^2.
\end{align} 
Therefore together with \eqref{est:init:44} and \eqref{est:init:12} and Lemma \ref{lemma:est:proTl}, we have 
\begin{align}\label{est:init:S:2}
	\fro{\calP_{\Omega^*\backslash\Omega_l}(\hat\S_l - \S^*)}^2 &\leq \left(\frac{4mb_u^2}{b_l^2}\frac{1}{\gamma-1} + C_{4,m}\frac{b_u^2}{b_l^2}(\mu_1\kappa_0)^{4m}\rmax^m \alpha \right)\fro{\hat\T_l-\T^*}^2 + \frac{16}{b_l^2}|\Omega^*\backslash\Omega_l| \errinf^2
\end{align}
where $C_{4,m} > 0$ are constants depending only on $m$.
Now we combine \eqref{est:init:S:1} and \eqref{est:init:S:2} and we get
\begin{align}
	\fro{\hat\S_l-\S^*}^2 &\leq \left(\frac{4mb_u^2}{b_l^2}\frac{1}{\gamma-1} + C_{5,m}(\mu_1\kappa_0)^{4m}\rmax^m\frac{b_u^2}{b_l^2}\alpha \right)\fro{\hat\T_l-\T^*}^2 + \frac{C_1}{b_l^2}|\Omega^{\ast}\cup\Omega_l|\errinf^2
\end{align}
where $C_{5,m}>0$ depending only on $m$ and $C_1>0$ an absolute constant.

Now if we choose $\alpha\leq (C_{5,m}\kappa_0^{4m}\mu_0^{4m}\rmax^m\frac{b_u^4}{b_l^4})^{-1}$ and $\gamma - 1 \geq 4m\frac{b_u^4}{b_l^4}$ for some sufficient large constants $C_{5,m}>0$ depending only on $m$, then we have  
\begin{align}\label{est:init:14}
	\fro{\hat\S_l - \S^*}^2 \leq \frac{b_l^2}{b_u^2}\fro{\hat\T_l-\T^*}^2 + \frac{C_1}{b_l^2}|\Omega^*\cup\Omega_l| \errinf^2
\end{align}
and 
\begin{align}\label{est:init:15}
	\fro{\hat\S_l - \S^*} \leq \frac{b_l}{b_u}\fro{\hat\T_l-\T^*} +  \frac{C_1}{b_l}\sqrt{|\Omega^*\cup\Omega_l|} \errinf
\end{align}
In addition, from the upper bound of $\fro{\T_l - \T^{\ast}}$, \eqref{est:init:15} implies that $\|\hat\S_l-\S^{\ast}\|_{\rm F}\leq c_0\usigma$ for a small $c_0>0$.
This fact is helpful later since it implies that $\hat\T_l+\hat\S_l$ belongs to the ball $\BB_{2}^*$ and thus activates the conditions in Assumption~\ref{assump:lowrank}.

\paragraph*{Step 2: bounding $\|\hat\T_{l}-\T^{\ast}\|_{\rm F}^2$ for all $l\geq 1$.} 
From previous step, we have verified 
\begin{align}\label{eq:prevS}
	\|\hat\S_{l-1}-\S^{\ast}\|_{\rm F}\leq \frac{b_l}{b_u}\fro{\hat\T_l-\T^*} +  \frac{C_1}{b_l}\sqrt{|\Omega^*\cup\Omega_l|} \errinf \leq c_0\usigma.
\end{align}
And from the Algorithm~\ref{algo:lowrank+sparse}, $\hat\T_l = \textsf{Trim}_{\zeta_{l},\br}(\W_{l-1})$. 
Now from Lemma~\ref{lemma:nbhd_spikiness_incoherence}, we get,
\begin{align}\label{eq:est:main}
	\fro{\hat\T_l -& \T^*}^2 = \fro{\textsf{Trim}_{\zeta_l,\br}(\W_{l-1}) - \T^*}^2 \no\\
	&\leq \fro{\W_{l-1}-\T^*}^2 +C_{m}\frac{\sqrt{\rmax}}{\usigma}\fro{\W_{l-1} - \T^*}^3\no\\
	&\leq (1+\frac{\delta}{4})\fro{\W_{l-1} - \T^*}^2\no\\
	&\leq (1-\delta^2)\fro{\hat\T_{l-1} - \T^*}^2 + 6\delta^{-1}\errrank + C_1\left(1+ b_u + b_u^2\right)b_l^{-2}\left(|\Omega^*| + \gamma\alpha d^*\right)\errinf^2 
\end{align}
Notice to use Lemma \ref{lemma:nbhd_spikiness_incoherence}, we need to verify $\fro{\W_{l-1} - \T^*} \leq \usigma/8$, which we will check momentarily.
Also, from \eqref{eq:est:main} and the signal-to-noise ration condition, we get $$\fro{\hat\T_l - \T^*} \leq c_1\min\{\delta^2\rmax^{-1/2},\kappa_0^{-2m}\rmax^{-1/2}\}\cdot \usigma.$$
On the other hand, from lemma \ref{lemma:nbhd_spikiness_incoherence}, we have $\hat\T_l$ is $(2\mu_1\kappa_0)^2$-incoherent. Further, from Lemma \ref{lemma:entrywise} and the definition of $\kdinf$ we have $\hat\T_l\in\BB_{\infty}^{\ast}$.
This finishes the induction for the error $\fro{\hat\T_l - \T^*}$. Now the only thing we need to check is the upper bound for $\fro{\W_{l-1} - \T^*}$.

\paragraph*{Step 2.1: bounding $\fro{\W_{l-1} - \T^*}$.}
From the Algorithm~\ref{algo:lowrank+sparse}, we have for arbitrary $1\geq \delta>0$,
\begin{align}
	\fro{\W_{l-1} - \T^*}^2 &
	= \fro{\hat\T_{l-1} - \T^* - \beta \calP_{\TT_{l-1}}(\G_{l-1} - \G^*) - \beta \pro_{\TT_{l-1}} \G^*}^2\notag\\
	&\leq (1+\frac{\delta}{2})\fro{\hat\T_{l-1} - \T^* - \beta \pro_{\TT_{l-1}}(\G_{l-1} - \G^*)}^2 + (1+\frac{2}{\delta}) \beta^2\fro{\pro_{\TT_{l-1}}( \G^*)}^2
\end{align}
Now we consider the bound for $\fro{\hat\T_{l-1} - \T^* - \beta \pro_{\TT_{l-1}}(\G_{l-1} - \G^*)}^2$,
\begin{align}\label{est:main}
	\fro{\hat\T_{l-1} - \T^* - \beta \pro_{\TT_{l-1}}(\G_{l-1} - \G^*)}^2 &= \fro{\hat\T_{l-1} - \T^*}^2 - 2\beta\inp{\hat\T_{l-1} - \T^*}{\pro_{\TT_{l-1}}(\G_{l-1} - \G^*)}\no\\
	&~~~~ + \beta^2\fro{\pro_{\TT_{l-1}}(\G_{l-1} - \G^*)}^2
\end{align}
The upper bound of $\fro{\hat\S_{l-1}-\S^{\ast}}$ ensures that $\hat\T_{l-1}+\hat\S_{l-1}\in \BB_{2}^*$.
Using the smoothness condition in Assumption~\ref{assump:lowrank}, we get
\begin{align}\label{est:1}
	\beta^2\fro{\pro_{\TT_{l-1}}(\G_{l-1} - \G^*)}^2 \leq \beta^2 b_u^2 \fro{\hat\T_{l-1}+\hat\S_{l-1} - \T^*-\S^*}^2
\end{align}  
Now we consider the bound for $|\inp{\hat\T_{l-1} - \T^*}{\pro_{\TT_{l-1}}(\G_{l-1} - \G^*)}|$. First we have:
$$\inp{\hat\T_{l-1} - \T^*}{\pro_{\TT_{l-1}}(\G_{l-1} - \G^*)} =  \inp{\hat \T_{l-1} -\T^*}{\G_{l-1} - \G^*} - \inp{\hat \T_{l-1} -\T^*}{\pro_{\TT_{l-1}}^{\perp}(\G_{l-1} - \G^*)}.$$
The estimation of $\inp{\hat \T_{l-1} -\T^*}{\G_{l-1} - \G^*}$ is as follows:
\begin{align}\label{eq:000001}
	\inp{\hat \T_{l-1} -\T^*}{\G_{l-1} - \G^*} &= \inp{\hat \T_{l-1} -\T^* +\hat\S_{l-1} - \S^*}{\G_{l-1} - \G^*} - \inp{\hat\S_{l-1} - \S^*}{\G_{l-1} - \G^*}\no\\
	&\geq b_l \fro{\hat \T_{l-1} -\T^* +\hat\S_{l-1} - \S^*}^2 -  \inp{\hat\S_{l-1} - \S^*}{\G_{l-1} - \G^*},
\end{align}
where the last inequality follows from Assumption~\ref{assump:lowrank}. And the estimation of $\inp{\hat \T_{l-1} -\T^*}{\pro_{\TT_{l-1}}^{\perp}(\G_{l-1} - \G^*)}$ is as follows:
\begin{align}\label{eq:000002}
	|\inp{\hat \T_{l-1} -\T^*}{\pro_{\TT_{l-1}}^{\perp}(\G_{l-1} - \G^*)}| &\leq \fro{\pro_{\TT_{l-1}}^{\perp}(\hat\T_{l-1} - \T^*)}\fro{\G_{l-1} - \G^*}\no\\
	&\leq \frac{C_{1,m}b_u}{\usigma} \fro{\hat \T_{l-1} -\T^*}^2\fro{\hat \T_{l-1} -\T^* +\hat\S_{l-1} - \S^*}
\end{align}
where the last inequality follows from Lemma~\ref{lem:ref:01}. Together with \eqref{eq:000001} and \eqref{eq:000002}, we get,
\begin{align}\label{est:2}
	\inp{\hat\T_{l-1} - \T^*}{\pro_{\TT_{l-1}}(\G_{l-1} - \G^*)} &\geq b_l \fro{\hat \T_{l-1} -\T^* +\hat\S_{l-1} - \S^*}^2 -  \inp{\hat\S_{l-1} - \S^*}{\G_{l-1} - \G^*}\no\\
	&~~~~- \frac{C_{1,m}b_u}{\usigma} \fro{\hat \T_{l-1} -\T^*}^2\fro{\hat \T_{l-1} -\T^* +\hat\S_{l-1} - \S^*}
\end{align}


Together with \eqref{est:1} and \eqref{est:2}, we get 
\begin{align}\label{est:33}
	\fro{\hat\T_{l-1} - \T^* - \beta \pro_{\TT_{l-1}}(\G_{l-1} - \G^*)}^2 &\leq  \left(1+ 2\beta b_u \frac{C_{1,m}}{\usigma}\fro{\hat \T_{l-1} -\T^* +\hat\S_{l-1} - \S^*}\right) \fro{\hat \T_{l-1} -\T^*}^2\no\\
	&~~~~+ (\beta^2 b_u^2 - 2\beta b_l)\fro{\hat \T_{l-1} -\T^* +\hat\S_{l-1} - \S^*}^2\no\\
	&~~~~+2\beta|\inp{\hat\S_{l-1} - \S^*}{\G_{l-1} - \G^*}|
\end{align}
In order to bound \eqref{est:33}, we derive separately the bound for each terms.
\paragraph*{Bounding $\fro{\hat \T_{l-1} -\T^* +\hat\S_{l-1} - \S^*}^2$.}
From the bound for $\fro{\hat\S_{l-1}-\S^*}$ in \eqref{eq:prevS}, we get,
\begin{align}\label{est:16}
	 \fro{\hat \T_{l-1} -\T^* +\hat\S_{l-1} - \S^*}^2&\leq 2\fro{\hat\T_{l-1} - \T^*}^2 + 2\fro{\hat\S_{l-1}-\S^*}^2\notag\\
	&\leq 4\fro{\hat\T_{l-1}-\T^*}^2 + \frac{C_1}{b_l^2}|\Omega^*\cup\Omega_{l-1}| \errinf^2
\end{align}
Thus,
\begin{align}\label{est:17}
	\fro{\hat \T_{l-1} -\T^* +\hat\S_{l-1} - \S^*} \leq 2\fro{\hat\T_{l-1}-\T^*} +  \frac{C_1}{b_l}\sqrt{|\Omega^*\cup\Omega_{l-1}|} \errinf
\end{align}
\paragraph*{Bounding $|\inp{\G_{l-1}-\G^*}{\hat\S_{l-1}-\S^*}|$.}
We first bound $\fro{\G_{l-1} - \G^*}$ by \eqref{est:17}:
\begin{align}\label{est:19}
	\fro{\G_{l-1} - \G^*} &\leq b_u\fro{\hat\T_{l-1}-\T^* + \hat\S_{l-1} -\S^*}\no\\
	&\leq 2b_u\fro{\hat\T_{l-1} -\T^*} + \frac{C_1b_u}{b_l}\sqrt{|\Omega^*\cup\Omega_{l-1}|} \errinf
\end{align}
Now we estimate $|\inp{\G_{l-1}-\G^*}{\hat\S_{l-1}-\S^*}|$ from \eqref{eq:prevS} and \eqref{est:19} as follows,
\begin{align}\label{inp:gs}
	&~~~~|\inp{\G_{l-1}-\G^*}{\hat\S_{l-1}-\S^*}| \leq \fro{\G_{l-1}-\G^*}\fro{\hat\S_{l-1}-\S^*}\no\\
	&\leq (0.02b_l + 0.01\beta b_u^2)\fro{\hat\T_{l-1}-\T^*}^2 + \frac{1}{\beta}\frac{C_1}{b_l^2}|\Omega^*\cup\Omega_{l-1}|\errinf^2 + \frac{C_1b_u}{b_l^2}|\Omega^*\cup\Omega_{l-1}|\errinf^2
\end{align}
\paragraph*{Bounding $|\inp{\hat\T_{l-1}-\T^*}{\hat\S_{l-1}-\S^*}|$.}
From \eqref{eq:prevS}, we have
\begin{align}\label{inp:ts}
	|\inp{\hat\T_{l-1}-\T^*}{\hat\S_{l-1}-\S^*}| &\leq \fro{\hat\T_{l-1} - \T^*}\fro{\hat\S_{l-1} - \S^*}\no\\
	&\leq (0.01 \frac{b_l}{b_u}\fro{\hat\T_{l-1} - \T^*} + \frac{C_1}{b_l}\sqrt{|\Omega^*\cup\Omega_{l-1}|}\errinf)\fro{\hat\T_{l-1} - \T^*}\no\\
	&\leq 0.02\fro{\hat\T_{l-1} - \T^*}^2 + \frac{C_1}{b_l^2}|\Omega^*\cup\Omega_{l-1}|\errinf^2
\end{align}

Now we go back to \eqref{est:33} and from \eqref{est:16} - \eqref{inp:ts}, we get:
\begin{align}\label{est:20}
	&~~~~\fro{\hat\T_{l-1} - \T^* - \beta \pro_{\TT_{l-1}}(\G_{l-1} - \G^*)}^2 \notag\\
	&\leq \left(1-1.84\beta b_l + 5\beta^2 b_u^2\right)\fro{\hat\T_{l-1}- \T^*}^2 + C_1(1+b_u+b_u^2)b_l^{-2}|\Omega^*\cup\Omega_{l-1}|\errinf^2
\end{align}
where the condition $\usigma \geq C_{1,m}  \frac{b_u}{b_l}\fro{\hat\T_{l-1} - \T^*}$ is used in the last step. 

By combining \eqref{est:main} and \eqref{est:20}, we get
\begin{align}\label{est:22}
	\fro{\W_{l-1} - \T^*}^2
	&=\fro{\hat\T_{l-1} - \T^* - \beta \pro_{\TT_{l-1}}\G_{l-1}}^2 \no\\
	&\leq (1+\frac{\delta}{2})\fro{\hat\T_{l-1} - \T^* - \beta \pro_{\TT_{l-1}}(\G_{l-1} - \G^*)}^2 + (1+\frac{2}{\delta}) \beta^2\fro{\pro_{\TT_{l-1}}( \G^*)}^2\notag\\
	&\leq (1+\frac{\delta}{2})\left(1-1.84\beta b_l + 5\beta^2 b_u^2\right)\fro{\hat\T_{l-1} - \T^*}^2 + (1+\frac{2}{\delta})\beta^2\errrank^2\notag\\
	&~~~~+ C_1\left(1+ \beta b_u + \beta^2 b_u^2\right)b_u^{-2}|\Omega^*\cup\Omega_{l-1}| \errinf^2\notag\\
	&\leq (1+\frac{\delta}{2})\left(1-1.84\beta b_l + 5\beta^2 b_u^2\right)\fro{\hat\T_{l-1} - \T^*}^2 + (1+\frac{2}{\delta})\beta^2\errrank^2\notag\\
	&~~~~+ C_1\left(1 + \beta b_u + \beta^2 b_u^2\right)\frac{1}{b_u^2}\left(|\Omega^*|+\gamma\alpha d^*\right) \errinf^2
\end{align}
where in the second inequality we used 
\begin{align}\label{eq:proof_errrank}
	\fro{\calP_{\TT_{l-1}}(\G^*)} = \sup_{\fro{\Y}=1}\inp{\calP_{\TT_{l-1}}(\G^*)}{\Y} = \sup_{\fro{\Y}=1}\inp{\G^*}{\calP_{\TT_{l-1}}(\Y)} \leq \errrank
\end{align}
since $\calP_{\TT_{l-1}}(\Y)\in \MM_{2\br}$ and in the last inequality we use $|\Omega^*\cup\Omega_{l-1}| \leq |\Omega^*|+ |\Omega_{l-1}| \leq |\Omega^*| + \gamma\alpha d^*$. 

Now we choose proper $\beta\in [0.005b_l/(b_u^2), 0.36b_l/(b_u^2)]$ so $1-1.84\beta b_l + 5\beta^2 b_u^2 \leq 1-\delta$, and we get
\begin{align}\label{eq:est:W-ptG}
	\fro{\W_{l-1} - \T^*} 
	&\leq  (1-\delta)(1+\delta/2)\fro{\hat\T_{l-1} -\T^*} + 3\delta^{-1}\errrank+ C_1(b_u + 1)b_l^{-1}\sqrt{|\Omega^*| + \alpha\gamma d^*}\errinf
\end{align}
where we use the fact that $\beta\leq 1$.
From the signal-to-noise ratio condition,  we have $3\delta^{-1}\errrank+ C_1(b_u + 1)b_l^{-1}\sqrt{|\Omega^*| + \alpha\gamma d^*}\errinf \leq \frac{\delta}{4}\frac{\usigma}{C_m\sqrt{\rmax}}$. This implies that $\fro{\W_{l-1}-\T^*} \leq \usigma/8$ holds.

\subsection{Proof of Theorem~\ref{thm:hatS_infty}}
Let $\hat\Omega$ and $\Omega^{\ast}$ denote the support of $\hat \S_{l_{\max}}$ and $\S^{\ast}$, respectively. 
By the proof of Theorem~\ref{thm:lowrank+sparse}, we have 
\begin{align*}
\big|[\hat \S_{l_{\max}}-\S^{\ast}]_{\omega} \big|\leq \begin{cases}
\frac{b_u}{b_l}\big| [\hat \T_{l_{\max}}-\T^{\ast}]_{\omega}\big|+\frac{2\textsf{Err}_{\infty}}{b_l}&, \textrm{ if } \omega\in \hat\Omega\\
\frac{2b_u}{b_l} \|\hat\T_{l_{\max}} - \T^*\|_{\ell_{\infty}} +\frac{2\textsf{Err}_{\infty}}{b_l}&, \textrm{ if } \omega\in \Omega^{\ast}\setminus \hat\Omega
\end{cases}
\end{align*}
Therefore, we conclude that 
\begin{align}\label{eq:hatSl_max_infty}
\|\hat\S_{l_{\max}}-\S^{\ast}\|_{\ell_\infty}\leq \frac{2b_u}{b_l}\big\| \hat \T_{l_{\max}}-\T^{\ast}\big\|_{\ell_\infty}+\frac{2\textsf{Err}_{\infty}}{b_l}.
\end{align}
Now, we can apply Lemma~\ref{lemma:entrywise} and we obtain 
\begin{align}\label{eq:wtTl_max_infty}
	\|\hat \T_{l_{\max}}-\T^{\ast}\|_{\ell_\infty} \leq C_{1,m}\rmax^{m/2}\dmin^{-(m-1)/2}\mu_1^{2m}\kappa_0^{2m}\fro{\hat\T_{l_{\max}} - \T^*}
\end{align}
Now, by putting together (\ref{eq:hatSl_max_infty}), (\ref{eq:wtTl_max_infty}) and (\ref{eq:hatTlmax}), we get 
$$
\|\hat\S_{l_{\max}}-\S^{\ast}\|_{\ell_\infty}\leq C_{2,m}\kappa_0^{2m}\mu_1^{2m}\Big(\frac{\rmax^m}{\dmin^{m-1}}\Big)^{1/2}\cdot \big(\textsf{Err}_{2\br}+(|\Omega^{\ast}|+\gamma\alpha d^{\ast})^{1/2}\textsf{Err}_{\infty}\big)+\frac{2\textsf{Err}_{\infty}}{b_l},
$$
where $C_{1,m}$ and $C_{2,m}$ are constants depending only on $m$. Now since we assume $b_l, b_u=O(1)$, we finish the proof of Theorem~\ref{thm:hatS_infty}.


\subsection{Proof of Theorem~\ref{thm:rpca}}
We first estimate the probability of the following two events.
\begin{align}
	\textsf{Err}_{2\br}&\leq C_{0,m}\sigma_z\cdot (\dmax\rmax+r^{\ast})^{1/2}\label{prob:errrank}\\
	\textsf{Err}_{\infty}&\leq C_{0,m}' \sigma_z \log^{1/2}\dmax\label{prob:errinf}
\end{align}
for some constants $C_{0,m}, C_{0,m}'>0$ depending only on $m$. Notice here the first event \eqref{prob:errrank} holds with probability at least $1 - \exp(-c_m \rmax\dmax)$ by Lemma \ref{lem:8}. And for the second event \eqref{prob:errinf}, we have from the definition, 
\begin{align}
	\textsf{Err}_{\infty} = \max\Big\{\|\nabla \frakL(\T^{\ast}+\S^{\ast})\|_{\ell_\infty}, \min\nolimits_{\|\X\|_{\ell_{\infty}}\leq \infty} \|\nabla \frakL(\X)\|_{\ell_\infty}\Big\} = \|\Z\|_{\ell_\infty}
\end{align}
So we have \eqref{prob:errinf} holds with probability at least $1 - 0.5\dmax^{-2}$ from Lemma \ref{lem:maxofsubg}. Taking union bounds and we get both \eqref{prob:errinf} and \eqref{prob:errrank} hold with probability at least $1-\dmax^{-2}$. And finally applying Theorem \ref{thm:lowrank+sparse} and Theorem \ref{thm:hatS_infty} gives the desired result.


\subsection{Proof of Lemma \ref{lemma:init:rpca}}
Denote the event $\calE_1 = \{\|\Z\|_{\ell_{\infty}}\leq 2\sqrt{m}\sigma_z\sqrt{\log(\dmax)}\}$, then from Lemma \ref{lem:maxofsubg}, we have $\calE_1$ holds with probability at least $1-2(d^*)^{-1}$. 
Now we set $\tau_l = 2\sqrt{m}\sigma_z\sqrt{\log(\dmax)} + (d^*)^{-1/2}\mu_1\fro{\T^*}$, then under $\calE_1$, we have $\|\T^*+\Z\|_{\ell_{\infty}}\leq \tau_l$. 
From the definition of $\tau_0$, we have $|\tau_0|\leq |\T^*+\Z|^{(\lfloor pd^* - |\Omega^*|\rfloor)} \leq \tau_l$.
Denote $\Omega_1 = \{\omega:|[\A]_{\omega}|\leq \tau_0\}$
From the definition of $\A_0$, we have
\begin{align}\label{A0}
	\fro{\A_0}^2 &= \sum_{\omega\in\Omega_1}[\T^*+\S^*+\Z]_{\omega}^2\no\\
	&\geq \sum_{\omega\in\Omega_1}[\T^*+\Z]_{\omega}^2 + 2\sum_{\omega\in\Omega_1\cap\Omega^*}[\S^*]_{\omega}[\T^*+\Z]_{\omega}\no\\
	&\geq \sum_{\omega\in\Omega_1}[\T^*+\Z]_{\omega}^2 - 4|\Omega^*|\tau_l^2\no\\
	&= \fro{\T^*+\Z}^2 - \sum_{\omega\in\Omega_1^c}[\T^*+\Z]_{\omega}^2- 4|\Omega^*|\tau_l^2\no\\
	&\geq \fro{\T^*+\Z}^2  - (pd^* + 4|\Omega^*|)\tau_l^2,
\end{align}
where the penultimate inequality holds since for all $\omega$, $|[\T^*+\Z]_{\omega}|\leq \tau_l$ and for all $\omega\in\Omega_1$, we have $|[\S^*]_{\omega}| \leq |[\T^*+\Z]_{\omega}| + \tau_0 \leq 2\tau_l$. Now we estimate the lower bound for $\fro{\T^*+\Z}^2$. Since $\Z$ has i.i.d. subgaussian entries, we have $\fro{\Z}^2 \geq \frac{1}{2}d^*\sigma_z^2$ with probability at least $1-2\exp(-cd^*)$ for some absolute constant $c> 0$, and $2\inp{\T^*}{\Z}\leq \frac{1}{2}\fro{\T^*}^2+2\sigma_z^2\log(\dmax)$ with probability at least $1-2(d^*)^{-1}$. Put these altogether, we see 
\begin{align}\label{T+Z}
	\fro{\T^*+\Z}^2 = \fro{\T^*}^2 + \fro{\Z}^2 +2\inp{\T^*}{\Z} \geq \frac{1}{2}\fro{\T^*}^2 + \frac{1}{4}\sigma_z^2 d^*.
\end{align}
Combine \eqref{A0} and \eqref{T+Z}, we have 
\begin{align}
	\fro{\A_0}^2\geq \frac{1}{2}\fro{\T^*}^2 + \frac{1}{4}\sigma_z^2 d^* -  (pd^* + 4|\Omega^*|)\tau_l^2.
\end{align}
Therefore with the choice $\tau = 10\sqrt{m}\sqrt{\log(\dmax)}\mu_1\frac{\fro{\A_0}}{\sqrt{d^*}}$, we see that 
$\tau \geq \tau_l$ and $\tau_u := 10\sqrt{m}\sqrt{\log(\dmax)}\mu_1\tau_l\geq\tau$.
With such a choice of $\tau$, since for $\omega\in(\Omega^*)^c$, we have $|[\T^*]_{\omega} + [\Z]_{\omega}|\leq \tau_l\leq \tau$, so we obtain
\begin{align*}
	\wt\A &= \calP_{(\Omega^*)^c}(\A) + \calP_{\Omega^*}(\wt\A) = \calP_{(\Omega^*)^c}(\T^* + \Z) + \calP_{\Omega^*}(\trunc_{\tau}(\A))\\ 
	&= \T^*+\Z + \calP_{\Omega^*}(\trunc_{\tau}(\A)-\T^*-\Z)\\
	&=: \T^*+\Z + \E,
\end{align*}
where $\errorE = \calP_{\Omega^*}(\trunc_{\tau}(\A)-\T^*-\Z)$
and 
the first equality holds since for $\omega\in(\Omega^*)^c$, $|[\A]_{\omega}|\leq |[\T^*]|_{\omega} + |[\Z_{\omega}]|\leq \tau_u$. Under event $\calE_1$, we have $\fro{\errorE}\leq 2|\Omega^*|^{1/2}\tau_u$.

Now we use bold-face capital letters as shorthand notation for the unfolding of corresponding calligraphic-font bold-face letters, for example, $\bT_i^* = \opM_i(\T^*), i\in[m]$. We denote $\X = \T^*+\errorE$. We also denote $\bU_i^*$ be the top $r_i$ left singular vectors of $\bT_i^*$, $\bV_i$ be the top $r_i$ left singular vectors of $\bX_i$ and $\hat\bU_i^0$ be the top $r_i$ left singular vectors of $\wt\bA_i$.

From Wedin's sin$\Theta$ theorem, we have from condition $(a)$,
\begin{align}\label{dist:UV}
	d_c(\bU_i^*,\bV_i) \leq \frac{C|\Omega^*|^{1/2}\tau_u}{\minl},
\end{align}
where $d_c(\bU,\bV) = \min_{\bR\in\OO_r}\op{\bU\bR - \bV}$. 
Meanwhile, from $\fro{\bX_i - \bT_i^*} = \fro{\errorE}\leq |\Omega^*|^{1/2}\tau_u$,
we also have $\sigma_{r_i}(\bX_i)\geq \frac{3\minl}{4}$, $\sigma_{r_i+1}(\bX_i)\leq \frac{\minl}{4}$ and $\op{\bX_i}\leq \frac{5\maxl}{4}$.

Since subtracting a multiple of identity matrix does not change the top eigenvectors, 
in order to bound the distance $d_c(\bV_i,\hat\bU_i^0)$, 
we consider $\op{\wt\bA_i\wt\bA_i^T - \bX_i\bX_i^T - \sigma_v^2d_i^-\bI_{d_i}}$, where $\sigma_v^2$ is the variance of the entry of $\Z$ and $d_i^- = d^*/d_i$. 
In fact, we have
$$\wt\bA_i\wt\bA_i^T - \bX_i\bX_i^T - \sigma_v^2d_i^-\bI_{d_i}= \bX_i\bZ_i^T + \bZ_i\bX_i^T + \bZ_i\bZ_i^T- \sigma_v^2d_i^-\bI_{d_i}.$$ 
Now we first consider the operator norm of $\bX_i\bZ_i^T$ under the event $\calE_1$. From Talagrand's concentration inequality, we have 
$$\PP\bigg(\big|\op{\bX_i\bZ_i^T} - \EE\op{\bX_i\bZ_i^T}\big|\leq C_m\sqrt{\log(\dmax)}\sigma_z\op{\bX_i}\cdot t\bigg| \calE_1\bigg)\geq 1-2\exp(-ct^2).$$
Since $\PP(\calE_1)\geq 1/2$ and from \cite[Theorem 1.1]{vershynin2011spectral}, we have $\EE[\op{\bX_i\bZ_i^T}|\calE_1] \leq 2 \EE\op{\bX_i\bZ_i^T} \leq C\sqrt{d_i}\sigma_z\op{\bX_i}$. Therefore setting $t = \sqrt{\log(\dmax)}$ and the event 
$$\calE_2^i = \{\op{\bX_i\bZ_i^T}\leq C_m\sqrt{d_i}\op{\bX_i}\sigma_z\}, \quad\calE_2 = \cap_{i=1}^m\calE_2^i,$$
we know that $\PP(\calE_2|\calE_1)\geq 1 -2m\dmax^{-1}$ and thus $\PP(\calE_2)\geq (1 -2m\dmax^{-1})(1-2(d^*)^{-1})$.

Now we turn to bounding $\op{\bZ_i\bZ_i^T - d_i^{-}\sigma_z^2\bI_{d_i}}$. From \cite[Theorem 4.6.1]{vershynin2018high}, we have with probability exceeding $1-2\exp(-d_i)$,
\begin{align*}
	\op{\bZ_i\bZ_i^T - d_i^{-}\sigma_z^2\bI_{d_i}} \leq C(d^*)^{1/2}\sigma_z^2.
\end{align*}
Denote the event $\calE_3^i = \{\op{\bZ_i\bZ_i^T - d_i^{-}\sigma_v^2\bI_{d_i}} \leq C(d^*)^{1/2}\sigma_z^2\}$ and $\calE_3 = \cap_{i=1}^m\calE_3^i$ and we have $\PP(\calE_3)\geq 1 - 2\sum_{i=1}^m\exp(-d_i)$.
Therefore under the event $\calE_2,\calE_3$, 
and from condition $(b)$,
we have 
$$d_c(\bV_i,\hat\bU_i^0) \leq \frac{C_m\sqrt{\dmax}\sigma_z\maxl + C(d^*)^{1/2}\sigma_z^2}{\minl^2}.$$
Together with \eqref{dist:UV}, we have 
\begin{align}\label{dist:UU0}
	d_c(\bU_i^*,\hat\bU_i^0) \leq \frac{C_m\sqrt{\dmax}\sigma_z\maxl + C(d^*)^{1/2}\sigma_z^2}{\minl^2} + \frac{C|\Omega^*|^{1/2}\tau_u}{\minl}.
\end{align}
Denote the event 
$$\calE_4 = \left\{\max_{i=1}^m\max_{\op{\bV_j}\leq 1,j\neq i}
\op{\bZ_i(\bV_{i+1}\otimes\cdots\otimes\bV_m\otimes\bV_1\otimes\cdots\otimes\bV_{i-1})}\leq C_m(\sqrt{\dmax\rmax}+\rmax^{\frac{m-1}{2}})\sigma_z\right\}.$$
And from \cite[Lemma 5]{zhang2018tensor}, we have $\PP(\calE_4)\geq 1 - Cm\exp(-c\dmax)$.
For the following we denote 
\begin{align*}
	\bX_1^t &= \bT_1^*(\hat\bU_2^t\otimes \cdots\otimes \hat\bU_m^t) =  \bT_1^*(\calP_{\bU_2^*}\hat\bU_2^t\otimes \cdots\otimes \calP_{\bU_m^*}\hat\bU_m^t)\\
	\bZ_1^t &= \bZ_1(\hat\bU_2^t\otimes \cdots\otimes \hat\bU_m^t)\\
	\wt\bA_1^t &= \wt\bA_1(\hat\bU_2^t\otimes \cdots\otimes \hat\bU_m^t),
\end{align*}
where $\calP_{\bU} = \bU\bU^T$. 
We shall denote $L_t = \max_{i=1}^m d_c(\hat\bU_i^t,\bU_i^*)$. For the base case, from \eqref{dist:UU0} and condition $(b)$, we see $L_0\leq \frac{1}{2}$.
Now suppose we have $L_t\leq \frac{1}{2}$.

From the process of HOOI, we have $\hat\bU_1^{t+1} = \textrm{SVD}_{r_1}(\wt\bA_1(\hat\bU_2^t\otimes \cdots\otimes \hat\bU_m^t))$.
And thus we obtain
\begin{align}
	\sigma_{r_1}(\bX_1^t)&\geq \sigma_{r_1}(\bU_2^*\otimes \cdots\otimes \bU_m^*)\cdot\prod_{i=2}^m\sigma_{\min}(\bU_i^{*T}\hat\bU_i^t)\no\\
	&\geq \sigma_{r_1}(\bU_2^*\otimes \cdots\otimes \bU_m^*)(1-L_t^2)^{(m-1)/2}\no\\
	&\geq c_m(1-L_t)^2\minl,
\end{align}
for some small constant $c_m>0$ depending only on $m$, and the last inequality holds since $1-L_t^2\geq \frac{3}{4}$. We bound $\op{\bZ_1^t}$ under the event $\calE_4$.
\begin{align}\label{Z1t}
	\op{\bZ_1^t} &= \op{\bZ_1(\hat\bU_2^t\otimes \cdots\otimes \hat\bU_m^t)}\no\\
	&= \op{\bZ_1\big((\calP_{\bU_2^*}+\calP_{\bU_2^*}^{\perp})\otimes\cdots\otimes
		(\calP_{\bU_m^*}+\calP_{\bU_m^*}^{\perp})\big) (\hat\bU_2^t\otimes \cdots\otimes \hat\bU_m^t)}\no\\
	&\leq C_m[(\dmax)^{1/2}+\rmax^{(m-1)/2}]\sigma_z + C_m[(\dmax\rmax)^{1/2}+\rmax^{(m-1)/2}]\sigma_zL_t,
\end{align}
where the last inequality holds since $\calE_4$ holds and $\op{\hat\bU_i^{t T}\bU_{i\perp}^*}\leq L_t$.
Now since $\hat\bU_1^{t+1}$ is the top $r_1$ left singular vectors of $\wt\bA_1^t$ and $\bU_1$ is the top $r_1$ left singular vectors of $\bX_1^t$, from Wedin's sin$\Theta$ Theorem, we have 
\begin{align*}
	d_c(\hat\bU_1^{t+1}, \bU_1) &\leq \frac{C\op{\wt\bA_1^t - \bX_1^t}}{\minl}
	\leq \frac{C(\fro{\bE_1} + \op{\bZ_1^t})}{\minl}\\
	&\overset{\eqref{Z1t}}{\leq}\frac{C|\Omega^*|^{1/2}\tau_u + C_m[(\dmax)^{1/2}+\rmax^{(m-1)/2}]\sigma_z + C_m[(\dmax\rmax)^{1/2}+\rmax^{(m-1)/2}]\sigma_zL_t}{\minl}.
\end{align*}
The derivation for $d_c(\hat\bU_i^{t+1}, \bU_i)$ when $i\geq 2$ is similar to this case and hence 
\begin{align*}
	L_{t+1} \leq \frac{C|\Omega^*|^{1/2}\tau_u + C_m[(\dmax)^{1/2}+\rmax^{(m-1)/2}]\sigma_z}{\minl} +\frac {C_m[(\dmax\rmax)^{1/2}+\rmax^{(m-1)/2}]\sigma_z}{\minl}L_t.
\end{align*}
From condition $(b)$, we have $C_m[(\dmax\rmax)^{1/2}+\rmax^{(m-1)/2}]\sigma_z/\minl \leq 1/2$, so the above inequality implies
\begin{align*}
	L_{t_{\max}} \leq (\frac{1}{2})^{t_{\max}}\cdot L_0  + \frac{C|\Omega^*|^{1/2}\tau_u + C_m[(\dmax)^{1/2}+\rmax^{(m-1)/2}]\sigma_z}{\minl}.
\end{align*}
If we choose $t_{\max} \geq (C_m\log(\dmax\kappa_0)\vee 1)$, then 
\begin{align}
	L_{t_{\max}} \leq \frac{C|\Omega^*|^{1/2}\tau_u}{\minl} + \frac{C_m[(\dmax)^{1/2}+\rmax^{(m-1)/2}]\sigma_z}{\minl}.
\end{align}

Set the event $\calE_5 = \{\fro{\Z\times_{i=1}^m\calP_{\hat\bU_i}} \leq C(r^* + \sum_{i=1}^md_ir_i)\sigma_z^2\}$. 
Then from \cite[Lemma 5]{zhang2018tensor}, $\PP(\calE_5)\geq 1- \exp(-C\dmax\rmax)$. 
And we also consider $\fro{\T^*\times_i\hat\bU_{i\perp}^T}$, we consider $i=1$ for simplicity.
\begin{align}\label{eq:TUperp}
	\fro{\T^*\times_1\hat\bU_{1\perp}^T} &= \fro{\hat\bU_{1\perp}^T\bT_1^*} \leq \fro{\calP_{\hat\bU_{1\perp}}\bT_1^*(\hat\bU_2^{t_{\max}-1}\otimes\cdots\otimes\hat\bU_m^{t_{\max}-1})}\cdot\prod_{i=2}^m\sigma_{\min}^{-1}(\bU_i^{*T}\hat\bU_i^{t_{\max}-1})\no\\
	&\leq C_m(\fro{\errorE} + \sqrt{r_1}\op{\bZ_1^{t_{\max}-1}})\no\\
	&\leq C_m|\Omega^*|^{1/2}\tau_u + C_m(\sqrt{d_1r_1}+\sqrt{r^*})\sigma_z,
\end{align}
where the second inequality holds from \cite[Lemma 6]{zhang2018tensor} and the last inequality holds from \eqref{Z1t}.

Now we are in the right position to bound $\fro{\hat\T- \T^*}$ under $\calE_5$.
\begin{align}
	\fro{\hat\T - \T^*} &= \fro{\wt\A \times_{i=1}^m \calP_{\hat\bU_i} - \T^*}\no\\
	&\leq \fro{(\wt\A - \T^*)\times_{i=1}^m \calP_{\hat\bU_i}} + \fro{\T^* - \T^*\times_{i=1}^m \calP_{\hat\bU_i}}\no\\
	&\leq \fro{\errorE} + \fro{\Z\times_{i=1}^m\calP_{\hat\bU_i}} + \sum_{i=1}^m\fro{\T^*\times_i\hat\bU_{i\perp}^T}\no\\
	&\leq C_m|\Omega^*|^{1/2}\tau_u + C_m(\sqrt{r^*} + \sqrt{\dmax\rmax})\sigma_z,
\end{align}
where the last inequality follows from \eqref{eq:TUperp}. 
Finally applying Lemma \ref{lemma:nbhd_spikiness_incoherence} and we get $\hat\T_0$ is $(2\mu_1\kappa_0)^2$-incoherent and $\fro{\hat\T_0 - \T^*}\leq 2\fro{\hat\T - \T^*}$. Therefore from condition $(a),(b)$ in Lemma \ref{lemma:init:rpca}, the initialization condition $(a)$ in Theorem \ref{thm:rpca} holds.

\subsection{Proof of Lemma~\ref{lem:heavy_tail}}
For each $j\in[m]$ and $i\in [d_j]$, we have 
$$
\|\be_i^{\top}\calM_j(\S_{\alpha})\|_{\ell_0}=\sum\nolimits_{\omega: \omega_j=i} \mathbbm{1}\big(|[\Z]_{\omega}|> \alpha \sigma_z\big)=\sum\nolimits_{\omega: \omega_j=i} [\Y]_{\omega}
$$
where $\Y\in \{0,1\}^{d_1\times\cdots\times d_m}$ having $i.i.d.$ Bernoulli entries and $q:=\PP([\Y]_{\omega}=1)=\PP(|[\Z]_{\omega}|>\alpha \sigma_z)\leq \alpha^{-\theta}$. 

Denote $X_{ij}=\sum\nolimits_{\omega: \omega_j=i} [\Y]_{\omega}$. By Chernoff bound, if $d_j^{-}q\geq 3\log (m\dmax^3)$, we get
\begin{align*}
\PP\Big(X_{ij}- d_j^{-}q \geq d_j^{-}q\Big)\leq \exp\big\{-d_j^{-1}q/3 \big\}\leq (m\dmax^3)^{-1}
\end{align*}
implying that 
\begin{align}\label{eq:heavy_tail_eq1}
\PP\Big(\bigcap\nolimits_{i,j}\big\{X_{ij}\leq 2d_j^-q\big\}\Big)\geq 1-m\dmax (m\dmax^3)^{-1}=1-\dmax^{-2}. 
\end{align}
On the other hand, if $d_j^{-}q\leq 3\log (m\dmax^3)$, by Chernoff bound, we get 
\begin{align*}
\PP\Big(X_{ij} \geq 10 \log (m\dmax^3)\Big)\leq (m\dmax^3)^{-1}
\end{align*}
implying that 
\begin{align}\label{eq:heavy_tail_eq2}
\PP\Big(\bigcap\nolimits_{i,j}\big\{X_{ij}\leq10 \log (m\dmax^3)\big\}\Big)\geq 1-m\dmax (m\dmax^3)^{-1}=1-\dmax^{-2}. 
\end{align}
Putting (\ref{eq:heavy_tail_eq1}) and  (\ref{eq:heavy_tail_eq2}), since $q\leq \alpha^{-\theta}$, we get 
$$
\PP\Big(\bigcap\nolimits_{i,j}\Big\{X_{ij}\leq \max\big\{10 \log (m\dmax^3), 2d_j^- \alpha^{-\theta}\big\}\Big\}\Big)\geq 1-\dmax^{-2},
$$
which completes the proof. 


\subsection{Proof of Theorem~\ref{thm:heavy_tail}}
Conditioned on $\frakE_1$ defined in Lemma~\ref{lem:heavy_tail}, Theorem~\ref{thm:heavy_tail} is a special case of Theorem~\ref{thm:rpca}. Indeed, in Theorem~\ref{thm:rpca}, we replace $\sigma_z$ with $\alpha \sigma_z$, and $|\Omega^{\ast}|\log\dmax$ with $\alpha' d^{\ast}\asymp \dmax \log(m\dmax)$, then we get Theorem~\ref{thm:heavy_tail}.


\subsection{Proof of Lemma \ref{lemma:init:heavytail}}
	From the choice of $\alpha$ in Theorem \ref{thm:heavy_tail}, we see that the sparsity of $\S_{\alpha}$ is bounded by $\alpha'\asymp\frac{\dmax}{d^*}\log(m\dmax^3)$. Therefore the condition $(a)$ in Lemma \ref{lemma:init:rpca} is satisfied. Now applying Lemma \ref{lemma:init:rpca} and we get the desired result.

\subsection{Proof of Lemma~\ref{lem:trim2}}
From Lemma~\ref{lemma:nbhd_spikiness_incoherence}, we have $\textsf{Trim}_{\eta,\br}(\W)$ is $2\mu_1\kappa_0$-incoherent. Now for all $j\in[m]$, 
$$\op{\opM_j(\opH(\wt\W))} \leq \op{\opM_j(\wt\W)} \leq \op{\opM_j(\T^*)} + \fro{\W - \T^*} \leq \frac{9}{8}\bsigma.$$
So we conclude
$$\|\textsf{Trim}_{\eta,\br}(\W)\|_{\ell_{\infty}} \leq \frac{9}{8}\bsigma\prod_{i=1}^m (2\mu_1\kappa_0)\sqrt{\frac{r_j}{d_j}} \leq (9\zeta/16)\cdot (\mu_1\kappa_0)^m.$$
where the last inequality follows from the upper bound for $\bsigma$. This finishes the proof of the lemma.


\subsection{Proof of Theorem~\ref{thm:binary_tensor}}
From the choice of $\zeta'$ and Lemma \ref{lem:trim2}, we know Assumption \ref{assump:lowrank} and \ref{assump:sparse} hold with parameters $b_{l,\zeta'}$ and $b_{u,\zeta'}$ with respect to the set $\BB_{2}^{\ast}=\BB_{\infty}^{\ast}=\{\T+\S: \|\T+\S\|_{\ell_\infty}\leq \zeta', \T\in\MM_{\br}, \S\in\SS_{\gamma\alpha}\}$.
Now the proof follows the proof of Theorem \ref{thm:lowrank+sparse} with slight modification. Since we can now guarantee in each iteration $\hat\T_l + \hat\S_l \in \BB_{2}^{\ast}=\BB_{\infty}^{\ast}$ from Lemma \ref{lem:trim2} and the choice of $\kprune$, we can use Assumption \ref{assump:sparse} instead of Assumption \ref{assump:lowrank} when estimating the low rank part.
So we only need to estimate $\errinf$ and $\errrank$. From \eqref{est:bionomial:errinf}, we have $\errinf \leq L_{\zeta}$. Now we estimate $\errrank$. In fact, from the definition of $\errrank$, we have
$$\errrank = \sup_{\M\in\MM_{2\br},\|\M\|_{\rm F}\leq 1}\inp{\nabla\frakL(\T^* + \S^*)}{\M}.$$
Since for all $\omega\in[d_1]\times\ldots\times[d_m]$, we have $\subw{\nabla\frakL(\T^*+\S^*)}$ is bounded random variable with the upper bound given by $L_{\zeta}$. So apply Lemma~\ref{lem:8}, we have $\errrank \leq CL_{\zeta}\cdot (\dmax\rmax + r^*)^{1/2}$ with with probability at least $1-\dmax^{-2}$. Now we plug in the bounds for $\errinf$ and $\errrank$ to Theorem~\ref{thm:lowrank+sparse} and we get the first part of the theorem. For the $\ell_{\infty}$ bound, we apply Theorem~\ref{thm:hatS_infty} and Lemma~\ref{lemma:entrywise}. And we finish the proof of the theorem.

\subsection{Proof of Lemma \ref{lemma:init:binary}}

\begin{algorithm}
	\caption{Initialization for binary tensor}\label{alg:init:binary}
	\begin{algorithmic}
		\STATE{Let $\bA = \A\lr{m_0}:= \reshape(\A,[d_1\ldots d_{m_0},d_{m_0+1}\ldots d_m])$ with $m_0 = \lfloor \frac{m}{2}\rfloor$ and let $\hat\bM$ be the minimizer to \eqref{prob:binary:cvx}.}
		\STATE{$\hat\T = \reshape(\hat\bM,[d_1,\ldots,d_m])$.}
		\STATE{$\hat\T_0 = \textsf{Trim}_{\eta,\br}(\hat\T)$ with $\eta = 16\mu_1\fro{\hat\T}/(7\sqrt{d^*})$.}
		\STATE{Output: $\hat\T_0$.}
	\end{algorithmic}
\end{algorithm}

We first introduce some notations. Let $m_0 = \lfloor \frac{m}{2}\rfloor$, and denote $\bT^* = (\T^*)\lr{m_0}$, $\bS^* = (\S^*)\lr{m_0}$ and $\bA = \A\lr{m_0}$, then $\bT^*,\bS^*,\bA$ are matrices of size $d_1\ldots d_{m_0} \times d_{m_0+1}\ldots d_m =: d_1^* \times d_2^*$. Since $\T^{\ast}$ admits the decomposition $\T^{\ast}=\C^{\ast}\cdot\llbracket \bU_1^{\ast},\cdots,\bU_m^{\ast}\rrbracket$, we have $\T^* = (\bU_{m_0}\otimes \cdots \otimes \bU_1)\C\lr{m_0}(\bU_m\otimes \cdots \otimes \bU_{m_0+1})^T$ and hence the rank of $\bT^*$ is $r = \min\{r_1\cdots r_{m_0}, r_{m_0+1}\cdots r_m\}$. We denote $\bM = \bT^* + \bS^*$.

Under Assumption \ref{assump:binary_tensor}, we have $\|\bT^*\|_{\ell_{\infty}},\|\bS^*\|_{\ell_{\infty}}\leq \frac{\zeta}{2}$ and thus $\|\bM\|_{\ell_{\infty}}\leq \zeta$. Now we bound the nuclear norm of $\bM$. Using triangle inequality and we have 
\begin{align*}
	\nuc{\bM} &\leq \nuc{\bT^*} + \nuc{\bS^*} \\
	&\leq \frac{\zeta}{2}(rd^*)^{1/2}  + \frac{\zeta}{2}|\Omega^*|^{1/2}\min(d_1^*,d_2^*)^{1/2}\\
	&= \big(\frac{\zeta}{2} + \frac{\zeta}{2}\cdot\frac{\min(d_1^*,d_2^*)^{1/2}}{(rd^*)^{1/2}}|\Omega^*|^{1/2}\big)(rd^*)^{1/2}\\
	&\leq \zeta (rd^*)^{1/2},
\end{align*}
where the last inequality holds since condition $(a)$ holds.
Now with a little bit abuse of notation, we consider the following convex program,
\begin{align}\label{prob:binary:cvx}
	\min\frakL(\bX) = -\inp{\bA}{\log(p(X))} - \inp{1-\bA}{\log(1-p(\bX))}, \text{~s.t.~} \nuc{\bX}\leq \zeta\sqrt{d^*r} \text{~and~} \|\bX\|_{\ell_{\infty}}\leq \zeta,
\end{align}
where the notation $1-\bA$ is the entrywise subtraction, and $p(\bX)$ is applying $p$ entrywisely to $\bX$.
Denote $\hat\bM$ be the minimizer to \eqref{prob:binary:cvx} and apply the Theorem 1 in \cite{davenport20141} with the sample size $d^*$ and we get with probability at least $1- \frac{C}{d_1^*+d_2^*}$,
\begin{align*}
	\fro{\hat\bM - \bM}^2 \leq C_{\zeta} [r(d_1^*+d_2^*)d^*]^{1/2}
\end{align*}
with $C_\zeta = C\cdot\zeta L_{\zeta}\beta_{\zeta}$ and $\beta_{\zeta} = \sup_{|x|\leq \zeta}\frac{p(x)(1-p(x))}{(p'(x))^2}$. 

Now we reshape $\hat\bM$ back to a tensor, and denote $\hat\T = \reshape(\hat\bM,[d_1,\ldots,d_m])$. Since $\reshape$ keeps the Frobenius norm unchanged, we have 
$$\fro{\hat\T - \T^*} = \fro{\hat\bT - \bT^*} \leq \fro{\hat\bM - \bM} + \fro{\bS^*}\leq C_{\zeta}^{1/2} [r(d_1^*+d_2^*)d^*]^{1/4} + |\Omega^*|^{1/2}\frac{\zeta}{2}.$$

Finally we output $\T_0 = \textsf{Trim}_{\eta,\br}(\hat\T)$ with $\eta = 16\mu_1\fro{\hat\T}/(7\sqrt{d^*})$, and from Lemma \ref{lem:trim2} and Lemma \ref{lemma:nbhd_spikiness_incoherence}, since condition $(b)$ and $(c)$ hold, we get
$$(1)~\mu(\hat\T_0)\leq 2\kappa_0\mu_1;~(2)~\fro{\hat\T_0 - \T^*}\leq 2\fro{\hat\T - \T^*};~(3)~\|\hat\T\|_{\ell_{\infty}}\leq C_m(\mu_1\kappa_0)^m\frac{\sqrt{r^*}}{\sqrt{d^*}}\bsigma.$$
And together with the upper bound for $\usigma$ in Assumption \ref{assump:binary_tensor}, the initialization condition in Theorem \ref{thm:binary_tensor} is satisfied.

\subsection{Proof of Theorem \ref{thm:poisson}}
The proof of this theorem is similar to that of Theorem \ref{thm:binary_tensor}. From the choice of $\zeta'$ and Lemma \ref{lem:trim2}, we know Assumption \ref{assump:lowrank} and \ref{assump:sparse} hold with parameters $b_{l,\zeta'} = e^{-\zeta'}$ and $b_{u,\zeta'} = e^{\zeta'}$ with respect to the set $\BB_{2}^{\ast}=\BB_{\infty}^{\ast}=\{\T+\S: \|\T+\S\|_{\ell_\infty}\leq \zeta', \T\in\MM_{\br}, \S\in\SS_{\gamma\alpha}\}$.
Now the proof follows the proof of Theorem \ref{thm:lowrank+sparse} with slight modification. Since we can now guarantee in each iteration $\hat\T_l + \hat\S_l \in \BB_{2}^{\ast}=\BB_{\infty}^{\ast}$ from Lemma \ref{lem:trim2} and the choice of $\kprune$, we can use Assumption \ref{assump:sparse} instead of Assumption \ref{assump:lowrank} when estimating the low rank part.
So we only need to estimate $\errinf$ and $\errrank$. From \eqref{est:bionomial:errinf}, we have $\errinf \leq \|\nabla\frakL(\T^*+\S^*)\|_{\ell_{\infty}}$. Simple calculation shows 
$$\nabla\frakL(\T^*+\S^*) = -\frac{1}{I}\Y + \exp(\T^*+\S^*),$$
and notice using a union bound and Poisson's tail bound, when $I\geq Ce^{\zeta}\log(d^*)$, we have with probability exceeding $1-\frac{1}{d^*}$, $\|\Y\|_{\ell_{\infty}}\leq 10 I e^{\zeta}$. Therefore we have $\errinf \leq 11e^{\zeta}$.

The estimation for $\errrank$ is given in Theorem 4.3 \cite{han2020optimal}, which states 
$$\errrank\leq C\sqrt{\frac{r^*+m\dmax\rmax}{I/e^{\zeta}}}$$
with probability exceeding $1-\frac{1}{d^*}$.

Now we plug in the bounds for $\errinf$ and $\errrank$ to Theorem~\ref{thm:lowrank+sparse} and we get the first part of the theorem. For the $\ell_{\infty}$ bound, we apply Theorem~\ref{thm:hatS_infty} and Lemma~\ref{lemma:entrywise}. And we finish the proof of the theorem.

\subsection{Proof of Lemma \ref{lemma:init:poisson}}
With slight modification of the proof of Theorem 4.3 in \cite{han2020optimal}, we have 
$$\fro{\wt\T_0 - \T^*}\leq C\sqrt{\frac{e^{\zeta}}{I}}(\sum_{i=1}^m\sqrt{d_ir_i} + \sqrt{d^{-}_ir_i})+\fro{\S^*}$$
under the condition $I\geq Ce^{\zeta}\dmax$ with probability exceeding $1-1/d^*$. Therefore since we assume $I\geq C_1\sum_{i=1}^m(d_ir_i+d_i^{-}r_i)\rmax\minl^{-2}$ and $|\Omega^*|\leq C\zeta^{-2}\minl^2\rmax^{-1}$, we have $\fro{\wt\T_0 - \T^*} \leq c_{1,m}\usigma\cdot\min\big\{\delta^2\rmax^{-1/2}, (\kappa_0^{2m}\rmax^{1/2})^{-1}\big\}\leq \minl/8$. Now we apply Lemma \ref{lem:trim2} and Lemma \ref{lemma:nbhd_spikiness_incoherence} we see 
$$(1)~\mu(\hat\T_0)\leq 2\kappa_0\mu_1;~(2)~\fro{\hat\T_0 - \T^*}\leq 2\fro{\hat\T - \T^*};~(3)~\|\hat\T\|_{\ell_{\infty}}\leq C_m(\mu_1\kappa_0)^m\frac{\sqrt{r^*}}{\sqrt{d^*}}\bsigma.$$
From Assumption \ref{assump:poisson}, we see the initialization requirements in \ref{thm:poisson} is satisfied.

\subsection{Proof of Theorem \ref{main:thm}}\label{pf:main:thm}
We use induction to prove this theorem.
\paragraph*{Step 0: Base case.}
From the initialization, we have $\fro{\hat\T_0 - \T^*} \leq c_{1,m}\delta\rmax^{-1/2} \cdot \usigma$.
\paragraph*{Step 1: Estimating $\fro{\hat\T_{l+1} - \T^*}$.}
We prove this case assuming 
\begin{align}\label{eq:prev:T}
	\fro{\hat\T_l - \T^*} \leq c_{1,m}\delta\rmax^{-1/2}\cdot \usigma.
\end{align}
We point out that this also implies $\fro{\hat\T_l - \T^*} \leq c_{1,m}b_lb_u^{-1}\rmax^{-1/2}\cdot \usigma$ since $\delta \lesssim b_l^2b_u^{-2}$.
 In order to use Lemma~\ref{lem:tensorest}, we need to derive an upper bound for $\fro{\hat\T_l - \T^* - \beta\calP_{\TT_{l}}\G_l}$.
\paragraph{Step 1.1: Estimating $\fro{\hat\T_l - \T^* - \beta\calP_{\TT_{l}}\G_l}$.}
	For arbitrary $1\geq \delta > 0$, we have,
	\begin{align}\label{lowrank:est:tl-t-g}
		\fro{\hat\T_l - \T^* - \beta\calP_{\TT_l}\G_l}^2 \leq (1+\delta/2)\fro{\hat\T_l - \T^* - \beta\calP_{\TT_l}(\G_l-\G^*)}^2 + (1+2/\delta)\beta^2\fro{\calP_{\TT_l}\G^*}^2
	\end{align}
	Now we consider the bound for $\fro{\hat\T_l - \T^* - \beta\calP_{\TT_l}(\G_l-\G^*)}^2$.
	\begin{align}\label{lowrank:est:norm:tl-tgl-g}
		\fro{\hat\T_l - \T^* - \beta \pro_{\TT_l}(\G_l - \G^*)}^2 &= \fro{\hat\T_l - \T^*}^2 - 2\beta\inp{\hat\T_l - \T^*}{\pro_{\TT_l}(\G_l - \G^*)} + \beta^2\fro{\pro_{\TT_l}(\G_l - \G^*)}^2\no\\
		&\leq (1+\beta^2 b_u^2)\fro{\hat\T_l - \T^*}^2 - 2\beta\inp{\hat\T_l - \T^*}{\pro_{\TT_l}(\G_l - \G^*)}
	\end{align}
	where the last inequality holds from the Assumption~\ref{assump:lowrank} since $\hat\T_{l}\in \BB_{2}^*$ from \eqref{eq:prev:T}.
	Also, 
	\begin{align}\label{lowrank:est:inp:tl-tgl-g}
		\inp{\hat\T_l - \T^*}{\pro_{\TT_l}(\G_l - \G^*)} &= \inp{\hat\T_l - \T^*}{\G_l - \G^*} - \inp{\calP_{\TT_l}^{\perp}(\hat\T_l - \T^*)}{\G_l-\G^*}\no\\
		&\geq b_l \fro{\hat\T_l - \T^*}^2  - \frac{C_{1,m}b_u}{\usigma}\fro{\hat\T_l - \T^*}^3
	\end{align}
	where the last inequality is from Assumption~\ref{assump:lowrank}, Lemma~\ref{lem:ref:01} and Cauchy-Schwartz inequality and $C_{1,m} = 2^m-1$. Together with \eqref{lowrank:est:norm:tl-tgl-g} and \eqref{lowrank:est:inp:tl-tgl-g}, and since we have $\fro{\hat\T_l - \T^*} \leq \frac{0.1b_l}{2b_uC_{1,m}}\cdot \usigma$, we get,
	\begin{align}\label{lowrank:est:1}
		\fro{\hat\T_l - \T^* - \beta \pro_{\TT_l}(\G_l - \G^*)}^2 &\leq (1-2\beta b_l + \beta^2b_u^2)\fro{\hat\T_l - \T^*}^2 + \frac{2\beta C_{1,m} b_u}{\usigma}\fro{\hat\T_l - \T^*}^3\no\\
		&\leq (1 - 1.9\beta b_l + \beta^2 b_u^2)\fro{\hat\T_l - \T^*}^2.
	\end{align}
	Since we have $0.75b_lb_u^{-1} \geq \delta^{1/2}$, if we choose $\beta\in [0.4b_l b_u^{-2}, 1.5b_l b_u^{-2}]$, we have $1 - 1.9\beta b_l + \beta^2 b_u^2 \leq 1 - \delta$.
	
	So from \eqref{lowrank:est:tl-t-g} and \eqref{lowrank:est:1}, we get 
	\begin{align}\label{lowrank:est:2}
		\fro{\hat\T_l - \T^* - \beta\calP_{\TT_l}\G_l}^2\leq (1+\frac{\delta}{2})(1-\delta)\fro{\hat\T_l - \T^*}^2 + (1+\frac{2}{\delta})\errrank^2
	\end{align}
where in the inequality we use the definition of $\errrank$ and that $\beta \leq 1$. Now from the upper bound for $\fro{\hat\T_l - \T^*}$ and the signal-to-noise ratio, we verified that $\fro{\hat\T_l - \T^* - \beta\calP_{\TT_l}\G_l} \leq \usigma/8$ and thus $\sigma_{\max}(\hat\T_l - \T^* - \beta\calP_{\TT_l}\G_l) \leq \usigma/8$.
\paragraph*{Step 1.2: Estimating $\fro{\hat\T_{l+1} - \T^*}$.}
	Now that we verified the condition of Lemma~\ref{lem:tensorest}, from the Algorithm~\ref{algo:lowrank}, we have,
	\begin{align}\label{lowrank:est:main}
		\fro{\hat\T_{l+1} - \T^*}^2 \leq \fro{\hat\T_l - \T^* - \beta\calP_{\TT_l}\G_l}^2 + C_m \frac{\sqrt{\rmax}}{\usigma}\fro{\hat\T_l - \T^* - \beta\calP_{\TT_l}\G_l}^3
	\end{align}
where $C_m>0$ is the constant depending only on $m$ as in Lemma~\ref{lem:tensorest}. From \eqref{lowrank:est:2} and the assumption that $\fro{\hat\T_l - \T^*} \lesssim_{m} \frac{\delta}{\sqrt{\rmax}}\cdot \usigma$ and $\errrank \lesssim_{m} \frac{\delta^2}{\sqrt{\rmax}}\cdot\usigma$, we get
\begin{align}\label{lowrank:est:3}
	C_m \frac{\sqrt{\rmax}}{\usigma}\fro{\hat\T_l - \T^* - \beta\calP_{\TT_l}\G_l} \leq \frac{\delta}{4}
\end{align}
From \eqref{lowrank:est:main}, \eqref{lowrank:est:2} and \eqref{lowrank:est:3}, we get
\begin{align}
	\fro{\hat\T_{l+1} - \T^*}^2 \leq (1+\frac{\delta}{4})\fro{\hat\T_l - \T^* - \beta\calP_{\TT_l}\G_l}^2\leq (1-\delta^2)\fro{\hat\T_l - \T^*}^2 + \frac{4}{\delta}\errrank^2
\end{align}
Together with the assumption $\fro{\hat\T_l - \T^*} \lesssim_{m} \frac{\delta}{\sqrt{\rmax}}\cdot \usigma$ and $\errrank \lesssim_{m} \frac{\delta^2}{\sqrt{\rmax}}\cdot\usigma$, we get
\begin{align}
	\fro{\hat\T_{l+1} - \T^*} \leq c_{1,m}\frac{\delta}{\sqrt{\rmax}}\cdot \usigma,
\end{align}
which completes the induction and completes the proof.

\section{Technical Lemmas}
\begin{lemma}\label{lem:ref:01}
	Suppose $\TT_{l}$ is the tangent space at the point $\hat\T_l$, then we have
	\begin{align*}
		\fro{\calP_{\TT_l}^{\perp} \T^*} \leq \frac{2^m-1}{\usigma}\fro{\T^*-\hat\T_l}^2.
	\end{align*}
\end{lemma}
\begin{proof}
	See (\cite{cai2020provable}, Lemma 5.2).
\end{proof}


\begin{lemma}\label{lem:tensorest}
	Let $\T^* = \S^*\cdot(\bV_1^*,\ldots,\bV_m^*)$ be the tensor with Tucker rank $\br = (r_1,\ldots,r_m)$. Let $\D\in \RR^{d_1\times\ldots\times d_m}$ be a perturbation tensor such that $\usigma \geq 8\sigma_{\max}(\D)$, where $\sigma_{\max}(\D) = \max_{i=1}^m \op{\opM_i(\D)}$. Then we have
	$$
	\fro{\opH(\T^* +  \D) - \T^*} \leq \fro{\D} + C_m\frac{\sqrt{\rmax}\fro{\D}^2}{\usigma}
	$$
	where $C_m>0$ is an absolute constant depending only on $m$.
\end{lemma}
\begin{proof}
Without loss of generality, we only prove the Lemma in the case $m=3$. 
	First notice that 
	$$
	\opH(\T^* +  \D) = (\T^* +\D)\cdot \llbracket \calP_{\bU_1},\calP_{\bU_2},\calP_{\bU_3} \rrbracket,
	$$
	 where $\bU_i$ are leading $r_i$ left singular vectors of $\opM_i(\T^* +  \D)$ and $\calP_{\bU_i}=\bU_i\bU_i^{\top}$.
	
	First from (\cite{xia2019normal}, Theorem 1), we have for all $i\in[m]$
	$$
		\calP_{\bU_i} - \calP_{\bV_i^*}= \cS_{i,1} + \sum_{j\geq 2}\cS_{i,j},
	$$
	where $\cS_{i,j} = \cS_{\opM_i(\T^*),j}(\opM_i(\D))$ and specially $\cS_{i,1} = \ps{(\opM_i(\T^*)^{\top})}(\opM_i(\D))^{\top}\calP_{\bV_i^*}^{\perp}
	+\calP_{\bV_i^*}^{\perp}\opM_i(\D)\ps{(\opM_i(\T^*))}$. The explicit form of $S_{i,j}$ can be found in \cite[Theorem~1]{xia2019normal}. 
	Here, we denote $\bA^{\dagger}$ the pseudo-inverse of $\bA$, i.e., $\bA^{\dagger}=\bR\bSigma^{-1}\bL^{\top}$ if $\bA$ has a thin-SVD as $\bA=\bL\bSigma \bR^{\top}$. With a little abuse of notations, we write $(\bA^{\dagger})^{k}=\bR \bSigma^{-k}\bL^{\top}$ for any positive integer $k\geq 1$. 
	
For the sake of brevity, we denote $\bS_i = \sum_{j \geq 1}\calS_{i,j}$. By the definition of $\calS_{i,j}$, we have the bound $\op{\cS_{i,j}} \leq \left(\frac{4\sigma_{\max}(\D)}{\usigma}\right)^j$. We get the upper bound for $\op{\bS_i}$ as follows,
	\begin{align}\label{upper:S}
		\op{\bS_i} = \op{\sum_{j\geq 1}\cS_{i,j}} \leq \frac{4{\sigma_{\max}(\D)}}{\usigma - 4\sigma_{\max}(\D)}\leq \frac{8\sigma_{\max}(\D)}{\usigma} 
	\end{align}
	So we have,
	\begin{align}
		\T^*\cdot&\llbracket\calP_{\bU_1},\calP_{\bU_2},\calP_{\bU_3}\rrbracket = \T^*\cdot \llbracket\calP_{\bV_1^*}+\bS_1,\calP_{\bV_2^{\ast}}+\bS_2,\calP_{\bV_3^*}+ \bS_3\rrbracket\no\\
		=& \T^*\cdot \llbracket \calP_{\bV_1^*},\calP_{\bV_2^*},\calP_{\bV_3^*}\rrbracket\\
		& + \T^*\cdot \llbracket \bS_1,\calP_{\bV_2^*},\calP_{\bV_3^*}\rrbracket+  \T^*\cdot \llbracket \calP_{\bV_1^*},\bS_2,\calP_{\bV_3^*}\rrbracket+ \T^*\cdot \llbracket \calP_{\bV_1^*},\calP_{\bV_2^*},\bS_3\rrbracket\no\\
		&+\T^*\cdot \llbracket \bS_1,\bS_2,\calP_{\bV_3^*}\rrbracket+  \T^*\cdot \llbracket \calP_{\bV_1^*},\bS_2,\bS_3\rrbracket+ \T^*\cdot \llbracket \bS_1,\calP_{\bV_2^*},\bS_3\rrbracket\no\\
		&~~~~ +\T^*\cdot\llbracket\bS_1,\bS_2,\bS_3\rrbracket
	\end{align}
We now bound each of  $\fro{\T^*\cdot \llbracket \bS_1,\bS_2,\calP_{\bV_3^*}\rrbracket}$, $\fro{\T^*\cdot \llbracket \calP_{\bV_1^*},\bS_2,\bS_3\rrbracket}$ and $\fro{\T^*\cdot \llbracket \bS_1,\calP_{\bV_2^*},\bS_3\rrbracket}$. Without loss of generality, we only prove the bound of the first term. 
	\begin{align}\label{eq:summand}
		\calM_1\big(\T^*\cdot \llbracket \bS_1,\bS_2,\calP_{\bV_3^*}\rrbracket\big) = \bS_{1}\opM_{1}(\T^*)\left(\calP_{\bV_3}^*\otimes \bS_2\right)^{\top}
	\end{align}
Write
	\begin{align}\label{eq:bS}
		\bS_{1}\opM_{1}(\T^*) &=  \left(\calS_{1,1} + \sum_{j \geq 2}\calS_{1,j}\right)\opM_{1}(\T^*)\no\\
		&= \calP_{\bV_1^*}^{\perp}\opM_1(\D)\ps{\left(\opM_1(\T^*)\right)}\opM_1(\T^*) + \sum_{j \geq 2}\calS_{1,j}\opM_{1}(\T^*)\no\\
		&= \calM_1\big(\D\cdot \llbracket \calP_{\bV_1^*}^{\perp}, \calP_{\bV_2^*},  \calP_{\bV_3^*}\rrbracket\big)+ \sum_{j \geq 2}\calS_{1,j}\opM_{1}(\T^*)
	\end{align}
where we used the fact $\calP_{\bV_1^*}^{\perp}\calM_1(\T^*)={\bf 0}$. 

Thus we obtain an upper bound for $\op{\bS_{1}\opM_{1}(\T^*)}$ as follows
	\begin{align}\label{upper:SMT}
		\op{\bS_{1}\opM_{1}(\T^*)} \leq \sigma_{\max}(\D)+ \usigma\sum_{j\geq 2}\left(\frac{4\sigma_{\max}(\D)}{\usigma}\right)^{j}\leq 4\sigma_{\max}(\D),
	\end{align}
where the first inequality is due to the explicit form of $\calS_{1,j}$. See \cite[Theorem~1]{xia2019normal}. 

So from \eqref{eq:summand} and \eqref{upper:SMT}, we get
	\begin{align}\label{summand:T}
		\fro{\T^*\cdot \llbracket \bS_1,\bS_2,\calP_{\bV_3^*}\rrbracket} &\leq \fro{\bS_{1}\opM_{1}(\T^*)}\cdot \op{\calP_{\bV_3}^*\otimes \bS_2}
		\leq C_1 \sqrt{\rmax}\frac{\sigma_{\max}(\D)^2}{\usigma}
	\end{align}
	where $C_1 > 0$ is an absolute constant.
	
Now we consider the linear terms $\T^*\cdot \llbracket \bS_1,\calP_{\bV_2^*},\calP_{\bV_3^*}\rrbracket$, $\T^*\cdot \llbracket \calP_{\bV_1^*},\bS_2,\calP_{\bV_3^*}\rrbracket$ and $\T^*\cdot \llbracket \calP_{\bV_1^*},\calP_{\bV_2^*},\bS_3\rrbracket$. Clearly, we have
	\begin{align}
	\opM_1\left(\T^*\cdot \llbracket \bS_1,\calP_{\bV_2^*},\calP_{\bV_3^*}\rrbracket\right) &= 	\bS_1\opM_1(\T^*)\notag\\
	\opM_2\left(\T^*\cdot \llbracket \calP_{\bV_1^*},\bS_2,\calP_{\bV_3^*}\rrbracket\right)&=\bS_2\opM_2(\T^{\ast})\notag\\
	\opM_3\left(\T^*\cdot \llbracket \calP_{\bV_1^*},\calP_{\bV_2^*},\bS_3\rrbracket\right)&=\bS_3\opM_3(\T^{\ast}),
	\end{align}
whose explicit representations are already studied in eq. \eqref{eq:bS}. As a result, we can write
\begin{align}\label{eq:hosvd-1st-rep1}
\T^*\cdot \llbracket \bS_1,&\calP_{\bV_2^*},\calP_{\bV_3^*}\rrbracket+\T^*\cdot \llbracket \calP_{\bV_1^*},\bS_2,\calP_{\bV_3^*}\rrbracket+\T^*\cdot \llbracket \calP_{\bV_1^*},\calP_{\bV_2^*},\bS_3\rrbracket\notag\\
=&  \D\cdot \llbracket \calP_{\bV_1^*}^{\perp}, \calP_{\bV_2^*},  \calP_{\bV_3^*}\rrbracket+\D\cdot \llbracket \calP_{\bV_1^*}, \calP_{\bV_2^*}^{\perp},  \calP_{\bV_3^*}\rrbracket+\D\cdot \llbracket \calP_{\bV_1^*}, \calP_{\bV_2^*},  \calP_{\bV_3^*}^{\perp}\rrbracket\notag\\
&+ \sum_{j \geq 2}\Big(\opM_{1}(\T^*)\cdot \llbracket \calS_{1,j}, \calP_{\bV_2^{\ast}}, \calP_{\bV_3^{\ast}}\rrbracket+\opM_{2}(\T^*)\cdot \llbracket \calP_{\bV_1^{\ast}}, \calS_{2,j}, \calP_{\bV_3^{\ast}}\rrbracket+\opM_{3}(\T^*)\cdot \llbracket \calP_{\bV_1^{\ast}}, \calP_{\bV_2^{\ast}}, \calS_{3,j}\rrbracket\Big).
\end{align}
Now we bound $\D\cdot\llbracket\calP_{\bU_1},\calP_{\bU_2},\calP_{\bU_3}\rrbracket$ as follows
	\begin{align}
		\D\cdot\llbracket&\calP_{\bU_1},\calP_{\bU_2},\calP_{\bU_3}\rrbracket = \D\cdot \llbracket\calP_{\bV_1^*}+\bS_1,\calP_{\bV_2^*}+\bS_2,\calP_{\bV_3^*}+ \bS_3\rrbracket\no\\
		=& \D\cdot \llbracket\calP_{\bV_1^*},\calP_{\bV_2^*},\calP_{\bV_3^*}\rrbracket \no\\
		&+ \D\cdot \llbracket\bS_1,\calP_{\bV_2^*},\calP_{\bV_3^*}\rrbracket+\D\cdot \llbracket\calP_{\bV_1^*},\bS_2,\calP_{\bV_3^*}\rrbracket+\D\cdot \llbracket\calP_{\bV_1^*},\calP_{\bV_2^*},\bS_3\rrbracket\no\\
		&+\D\cdot \llbracket\bS_1,\bS_2,\calP_{\bV_3^*}\rrbracket+\D\cdot \llbracket\calP_{\bV_1^*},\bS_2,\bS_3\rrbracket+\D\cdot \llbracket\bS_1,\calP_{\bV_2^*},\bS_3\rrbracket\no\\
		&~~~~ +\D\cdot\llbracket\bS_1,\bS_2,\bS_3\rrbracket
	\end{align}
Similarly as proving the bound \eqref{summand:T}, we can show 
	\begin{align}\label{summand:D}
		\max\left\{\fro{ \D\cdot \llbracket\bS_1,\calP_{\bV_2^*},\calP_{\bV_3^*}\rrbracket}, \fro{\D\cdot \llbracket\bS_1,\bS_2,\calP_{\bV_3^*}\rrbracket}, \fro{\D\cdot\llbracket\bS_1,\bS_2,\bS_3\rrbracket}\right\} \leq C_1 \sqrt{\rmax}\frac{\sigma_{\max}(\D)^2}{\usigma}
	\end{align}
where $C_1 > 0$ is an absolute constant.
	
Finally, by \eqref{eq:bS}, \eqref{summand:T}, \eqref{eq:hosvd-1st-rep1} and \eqref{summand:D}, we have 
	\begin{align}
		\big\|(\T^*& +\D)\cdot \llbracket\calP_{\bU_1},\calP_{\bU_2},\calP_{\bU_3}\rrbracket - \T^*\big\|_{\rm F} \notag\\
		&\leq \big\|\D\cdot \llbracket \calP_{\bV_1^*}^{\perp}, \calP_{\bV_2^*},  \calP_{\bV_3^*}\rrbracket+ \D\cdot\llbracket\calP_{\bV^*_1},\calP_{\bV_2^*}^{\perp},\calP_{\bV_3^*}\rrbracket+ \D\cdot\llbracket\calP_{\bV^*_1},\calP_{\bV_2^*},\calP_{\bV_3^*}^{\perp}\rrbracket+ \D\cdot\llbracket\calP_{\bV^*_1},\calP_{\bV_2^*},\calP_{\bV_3^*}\rrbracket\big\|_{\rm F}\no\\
		&~~~~ + C_{1}\frac{\sqrt{\rmax}\sigma_{\max}(\D)^2}{\usigma}\no\\
		&\leq \fro{\D} + C_{2}\frac{\sqrt{\rmax}\sigma_{\max}(\D)^2}{\usigma}
	\end{align}
	where $C_1,C_2>0$ are absolute constants ($C_{2,m} = 16m+2^{m+1}$ in the case of general $m$). This finishes the proof of the lemma.
\end{proof}


\begin{lemma}\label{lem:8}
	Assume all the entries of $\Z\in \R^{d_1\times\ldots\times d_m}$ are independent mean-zero random variables with bounded Orlicz-$\psi_2$ norm:
	$$\subg{\subw{\Z}} = \sup_{q \geq 1} (\E |\subw{\Z}|^q)^{1/q}/q^{1/2} \leq \sigma_z$$
	Then there exists some constants $C_m,c_m>0$ depending only on $m$ such that 
	$$ \sup_{\M\in\MM_{2\br},\|\M\|_F\leq 1} \inp{\Z}{\M} \leq C_m\sigma_z\left( r^* + \sum_{i=1}^m d_ir_i\right)^{1/2} $$
	with probability at least $1- \exp(-c_m\sum_{i=1}^m d_ir_i)$, where $r^* = r_1\ldots r_m$.
\end{lemma}
\begin{proof}
	See the proof of (\cite{han2020optimal}, Lemma D.5).
\end{proof}

\begin{lemma}[Maximum of sub-Gaussian]\label{lem:maxofsubg}
	Let $Z_1,\ldots,Z_N$ be $N$ random variables such that $\EE\exp\{tZ_i\}\leq \exp\{t^2\sigma_z^2/2\}$ for all $i\in[N]$. Then 
	$$
	\PP(\max_{1\leq i\leq N}|Z_i| > t) \leq 2N\exp(-\frac{t^2}{2\sigma_z^2}).
	$$
\end{lemma}
\begin{proof}
The claim follows from the following two facts:
$$
\PP(\max_{1\leq i\leq N} Z_i > t) \leq \PP(\cup_{1\leq i\leq N}\{Z_i > t\}) \leq N\PP(Z_i>t) \leq N\exp(-\frac{t^2}{2\sigma_z^2}),
$$
and
$$
\max_{1\leq i\leq N} |Z_i| = \max_{1\leq i\leq 2N} Z_i 
$$
with $Z_{N+i} = -Z_i$ for $i\in [N]$.
\end{proof}


\begin{lemma}[Spikiness implies incoherence]\label{lemma:spikiness_incoherence}
	Let $\T^*\in\MM_{\br}$ satisfies Assumption~\ref{assump:spikiness} with parameter $\mu_1$. Then we have:
	$$\mu(\T^*) \leq \mu_1 \kappa_0.$$
	where $\mu(\T^*)$ is the incoherence parameter of $\T^*$ and $\kappa_0$ is the condition number of $\T^*$.
\end{lemma}
\begin{proof}
Denote $\T^* = \C^*\cdot\llbracket\bU_1,\ldots,\bU_m\rrbracket$.
Now we check the incoherence condition of $\T^*$. For all $i\in[d_j]$ and $j\in[m]$,
$$\|\be_i^{\top}\opM_{j}(\T^*)\|_{\ell_2} = \|\be_i\bU_j\opM_{j}(\C^*)\|_{\ell_2}\geq \|\be_i^{\top}\bU_j\|_{\ell_2}\cdot\usigma \geq \|\be_i^{\top}\bU_j\|_{\ell_2}\frac{\fro{\T^*}}{\sqrt{r_j}\kappa_0}.$$
On the other hand, we have 
$$\|\be_i^{\top}\opM_{j}(\T^*)\|_{\ell_2} \leq \sqrt{d_j^-}\|\T^*\|_{\ell_{\infty}}\leq \mu_1\fro{\T^*}\frac{1}{\sqrt{d_j}},$$
where the last inequality is due to the spikiness condition $\T^*$ satisfies. Together with these two inequalities, we have
$$\|\be_i^{\top}\bU_j\|_{\ell_2} \leq \sqrt{\frac{r_j}{d_j}}\mu_1\kappa_0.$$
And this finishes the proof of the lemma.
\end{proof}


\begin{lemma}\label{lemma:nbhd_spikiness_incoherence}
	Let $\T^*\in\MM_{\br}$ satisfies Assumption~\ref{assump:spikiness} with parameter $\mu_1$. Suppose that $\W$ satisfies $\fro{\W - \T^*} \leq \frac{\underline{\lambda}}{8}$, then we have $\textsf{Trim}_{\zeta, \br}(\W)$ is $(2\mu_1 \kappa_0)^2$-incoherent if we choose $\zeta = \frac{16}{7}\mu_1\frac{\fro{\W}}{\sqrt{d^*}}$. Also, it satisfies
	$$\fro{\textsf{Trim}_{\zeta,\br}(\W) - \T^*} \leq \fro{\W - \T^*} + \frac{C_{m}\sqrt{\rmax}\fro{\W -\T^*}^2}{\usigma},$$
	where $C_m>0$ depends only on $m$.
\end{lemma}
\begin{proof}
	Notice $\textsf{Trim}_{\zeta, \br}(\W) = \opH(\wt\W)$, where $\wt\W$ is the entrywise truncation of $\W$ with the thresholding $\zeta/2$. To check the incoherence of $\opH(\wt\W)$, denote $\wt \bU_j$ the top-$r_j$ left singular vectors of $\calM_j(\wt\W)$, and $\wt\bLambda_j$ the $r_j\times r_j$ diagonal matrix containing the top-$r_j$ singular values of $\calM_j(\wt\W)$. Then, there exist a $\wt\bV_j\in\RR^{d_j^-\times r_j}$ satisfying $\wt\bV_j^{\top}\wt\bV_j=\bI_{r_j}$ such that 
	$$
	\wt\bU_j\wt\bLambda_j=\calM_j(\wt\W)\wt\bV_j.
	$$
	Now we can also bound the $\ell_{\infty}$-norm of $\T^*$:
	$$\|\T^*\|_{\ell_{\infty}} \leq \mu_1\frac{\fro{\T^*}}{\sqrt{d^*}} \leq \mu_1\frac{\fro{\W} + \fro{\T^* - \W}}{\sqrt{d^*}} \leq \mu_1\frac{\fro{\W} + \fro{\T^*}/8}{\sqrt{d^*}}.$$
	This together with the definition of $\zeta$, we have:
	$$\mu_1\frac{\fro{\T^*}}{\sqrt{d^*}} \leq 8/7\cdot \mu_1 \frac{\fro{\W}}{\sqrt{d^*}} = \zeta/2.$$
	And thus $\|\T^*\|_{\ell_{\infty}} \leq \zeta/2$.
	Then for all $i\in [d_j]$,
	$$\|\be_i^{\top}\wt\bU_j\|_{\ell_2} = \|\be_i^{\top} \opM_j(\wt\W) \wt\bV_j \wt\bLambda_j^{-1}\|_{\ell_2}\leq \frac{\|\be_i^{\top}\opM_j(\wt\W)\|_{\ell_2}}{\lambda_{r_j}(\wt\bLambda_j)}\leq \frac{\zeta/2 \cdot (d_j^-)^{1/2}}{7/8\cdot \lambda_{r_j}(\opM_j(\T^*))}.$$	
where the last inequality is due to $\fro{\wt\W - \T^*} \leq \fro{\W - \T^*}\leq \underline{\lambda}/8$ since $\|\T^*\|_{\ell_{\infty}}\leq \zeta/2$ and $\|\wt\W\|_{\ell_{\infty}}\leq \zeta/2$. Meanwhile, 
$$\fro{\T^*} \leq \sqrt{r_j}\kappa_0 \lambda_{r_j}(\opM_j(\T^*)).$$
There for the $\zeta = \frac{16}{7}\mu_1\frac{\fro{\W}}{\sqrt{d^*}}$, we have for all $j\in[m]$
$$\max_{i\in[d_j]}\|\be_i\wt\bU_j\|_{\ell_2} \leq \frac{64}{49}\mu_1\kappa_0\frac{\fro{\T^*}+\usigma/8}{\fro{\T^*}}\sqrt{\frac{r_j}{d_j}}\leq2\mu_1\kappa_0 \sqrt{\frac{r_j}{d_j}}.$$
where the second last inequality is from $\fro{\W} \leq \fro{\T^*} + \fro{\W - \T^*}$ and the last inequality is from $\fro{\T^*} \geq \usigma$. 

The second claim follows from the fact that $\fro{\wt\W - \T^*} \leq \fro{\W - \T^*}\leq \usigma/8$, and from Lemma~\ref{lem:tensorest},
\begin{align*}
	\fro{\textsf{Trim}_{\zeta,\br}(\W) - \T^*}  = 	\fro{\wt\W - \T^*} &\leq \fro{\wt\W - \T^*} + C_{m}\frac{\sqrt{\rmax}\fro{\wt\W - \T^*}^2}{\usigma}\no\\
	&\leq \fro{\W - \T^*} + C_{m}\frac{\sqrt{\rmax}\fro{\W - \T^*}^2}{\usigma}
\end{align*}
This finishes the proof of the lemma.
\end{proof}


We introduce some notations for the following lemmas. Denote by $\hat\T_l = \C_l\cdot(\bU_1,\ldots,\bU_m)$, $\T^* =\C^*\cdot(\bU_1^*,\ldots,\bU_m^*)$.
\begin{align}\label{eq:Ri}
	\bR_i = \arg\min_{\bR\in \OO_{r_i}} \fro{\bU_i - \bU_i^*\bR}, i \in [m]
\end{align}
If we let $\bU_i^{*T}\bU_i = \bL_i\bS_i\bW_i^{\top}$ be the SVD of $\bU_i^{*T}\bU_i$, then the closed form of $\bR_i$ is given by $\bR_i = \bL_i\bW_i^{\top}$. And we rewrite $$\T^* = \S^*\cdot (\bV_1^*,\cdots, \bV_m^*)$$ where $\S^* = \C^*\cdot (\bR_1^{\top},\cdots, \bR_m^{\top})$ and $\bV_i^* = \bU_i^*\bR_i, i \in [m]$. So $\bV_i^*$ is also $\mu_0$-incoherent.


\begin{lemma}[Entry-wise estimation of $|\subw{\hat\T_l - \T^*}|$]\label{lemma:entrywise}
	Suppose $\T^*$ satisfies Assumption~\ref{assump:spikiness}. Under the assumptions that $\hat\T_l$ is $(2\mu_1\kappa_0)^2$-incoherent and $\fro{\hat\T_l - \T^*} \leq 
	\frac{\usigma}{16m\rmax^{1/2}\kappa_0}$, then we have
	$$|\subw{\hat\T_l - \T^*}|^2 \leq C_m\rmax^m\dmin^{-(m-1)}(\mu_1\kappa_0)^{4m}\fro{\hat\T_l - \T^*}^2,$$
	where $C_m = 2^{4m+1}(m+1)$.
\end{lemma}
\begin{proof}
	First we have 
	\begin{align}
		\hat\T_l - \T^* &= (\C_l-\S^*)\cdot (\bU_1,\cdots, \bU_m) + \sum_{i=1}^m \S^*\cdot (\bV_1^*,\ldots,\bV_{i-1}^*,\bU_i - \bV_i^*,\bU_{i+1},\ldots,\bU_m)
	\end{align}
	From Lemma~\ref{lemma:spikiness_incoherence}, we get $\T^*$ is $\mu_1^2 \kappa_0^2$-incoherent. So we have for all $\omega = (\omega_1,\ldots,\omega_m)\in[d_1]\times\ldots\times[d_m]$
	\begin{align*}
		|\subw{\hat\T_l - \T^*}| &\leq \fro{\C_l - \S^*}\prod_{i = 1}^m \op{(\bU_i)_{\omega_i :}} + \sum_{i = 1}^m\fro{\S^*}\op{(\bU_i-\bV_i^*)_{\omega_i :}}\prod_{k = 1}^{i-1}\op{(\bV_k^*)_{\omega_k :}}\prod_{k = i+1}^{m}\op{(\bU_k)_{\omega_k :}}\\
		&\leq \sqrt{\frac{r^*}{d^*}}(2\mu_1\kappa_0)^{2m}\fro{\C_l - \S^*} + (2\mu_1\kappa_0)^{2m-2}\sqrt{\frac{\rmax^{m-1}}{\dmin^{m-1}}}\fro{\S^*}\sum_{i=1}^m \op{(\bU_i - \bV_i^*)_{\omega_i:}}
	\end{align*}
	where $r^* = \prod_{i=1}^m r_i, d^* = \prod_{i=1}^m d_i$ and $\rmax = \max_{i=1}^m r_i, \dmin = \min_{i=1}^m d_i$.
	From AG–GM inequality, we have
	\begin{align}
		|\subw{\hat\T_l - \T^*}|^2 &\leq (m+1)(2\mu_1\kappa_0)^{4m}\frac{r^*}{d^*}\fro{\C_l-\S^*}^2 + (m+1)(2\mu_1\kappa_0)^{4m-4}\frac{\rmax^{m-1}}{\dmin^{m-1}}\fro{\S^*}^2\sum_{i=1}^m\op{(\bU_i - \bV_i^*)_{\omega_i:}}^2\label{eq:entrywise}\\
		&\leq (m+1)\rmax^m\dmin^{-(m-1)}(2\mu_1\kappa_0)^{4m}\left(\fro{\C_l-\S^*}^2 + \underline{\lambda}^2\sum_{i=1}^m\fro{\bU_i - \bV_i^*}^2\right)\notag\\
		&\leq 2(m+1)\rmax^m\dmin^{-(m-1)}(2\mu_1\kappa_0)^{4m}\fro{\hat\T_l - \T^*}^2\notag
	\end{align}
where the last inequality is from Lemma~\ref{lem:est:tl-t}, and this finishes the proof of the lemma.
\end{proof}


\begin{lemma}[Estimation of $\fro{\calP_{\Omega}(\hat\T_l-\T^*)}^2$]\label{lemma:est:proTl}
	Let $\Omega$ be the $\alpha$-fraction set. Suppose $\T^*$ satisfies Assumption~\ref{assump:spikiness}. Under the assumptions that $\hat\T_l$ is $(2\mu_1\kappa_0)^2$-incoherent and $\fro{\hat\T_l - \T^*} \leq 
	\frac{\usigma}{16m\rmax^{1/2}\kappa_0}$, we have 
	$$\fro{\calP_{\Omega}(\hat\T_l-\T^*)}^2 \leq C_m(\mu_1\kappa_0)^{4m}\rmax^m \alpha\fro{\hat\T_l - \T^*}^2,$$
	where $C_m = 2^{4m+1}(m+1)$.
\end{lemma}
\begin{proof}
First from \eqref{eq:entrywise} in Lemma~\ref{lemma:entrywise}, we have
$$|\subw{\hat\T_l - \T^*}|^2 \leq (m+1)(2\mu_1\kappa_0)^{4m}\frac{r^*}{d^*}\fro{\C_l-\S^*}^2 + (m+1)(2\mu_1\kappa_0)^{4m-4}\frac{\rmax^{m-1}}{\dmin^{m-1}}\fro{\S^*}^2\sum_{i=1}^m\op{(\bU_i - \bV_i^*)_{\omega_i:}}^2.$$
Since $\Omega$ is an $\alpha$-fraction set, we have
\begin{align}
	\fro{\calP_{\Omega}(\hat\T_l-\T^*)}^2 &= \sum_{\omega\in\Omega}\subw{\hat\T_l-\T^*}^2\notag\\
	&\leq (m+1)(2\mu_1\kappa_0)^{4m}\alpha r^*\fro{\C_l-\S^*}^2 + (m+1)(2\mu_1\kappa_0)^{4m-4}\alpha\rmax^{m-1}\fro{\S^*}^2\sum_{i=1}^m\fro{\bU_i - \bV_i^*}^2\no\\
	&\leq (m+1)(2\mu_1\kappa_0)^{4m}\alpha r^*\fro{\C_l-\S^*}^2 + (m+1)(2\mu_1\kappa_0)^{4m-4}\alpha\rmax^{m}\overline{\lambda}^2\sum_{i=1}^m\fro{\bU_i - \bV_i^*}^2\no\\
	&\leq (m+1)(2\mu_1\kappa_0)^{4m}\rmax^m\alpha \left(\fro{\C_l-\S^*}^2 +\usigma^2\sum_{i=1}^m\fro{\bU_i - \bV_i^*}^2\right)
\end{align}
Now we invoke Lemma \ref{lem:est:tl-t}, and we get
$$	\fro{\calP_{\Omega}(\hat\T_l-\T^*)}^2\leq 2(m+1)(2\mu_1\kappa_0)^{4m}\rmax^m\alpha\fro{\hat\T_l - \T^*}^2,$$
which finishes the proof of the lemma.
\end{proof}


\begin{lemma}[Estimation of $\fro{\hat\T_l - \T^*}^2$]\label{lem:est:tl-t}
	Let $\hat\T_l = \C_l\cdot (\bU_1,\cdots,\bU_m)$ be the $l$-th step value in Algorithm \ref{algo:lowrank+sparse} and let $\T^* = \S^*\cdot (\bV_1^*,\cdots, \bV_m^*)$. Suppose $\hat\T_l$ satisfies $\fro{\hat\T_l - \T^*} \leq 
	\frac{\usigma}{16m\rmax^{1/2}\kappa_0}$. Then we have the following estimation for $\fro{\hat\T_l - \T^*}^2$:
	$$\fro{\hat\T_l - \T^*}^2 \geq 0.5\fro{\C_l - \S^*}^2 + 0.5\usigma^2\sum_{i=1}^m\fro{\bU_i -\bV_i^*}^2.$$
\end{lemma}
\begin{proof}
	First we have
	\begin{align}
		\hat\T_l - \T^* &= (\C_l-\S^*)\cdot (\bU_1,\cdots, \bU_m) + \sum_{i=1}^m \S^*\cdot (\bV_1^*,\ldots,\bV_{i-1}^*,\bU_i - \bV_i^*,\bU_{i+1},\ldots,\bU_m)
	\end{align}

	Notice that we have
	\begin{align}
		\fro{\S^*\cdot (\bV_1^*,\ldots,\bV_{i-1}^*,\bU_i - \bV_i^*,\bU_{i+1},\ldots,\bU_m)}^2 = \fro{(\bU_i - \bV_i^*)\opM_i(\S^*)}^2
	\end{align}
	Denote $\X_i = \S^*\cdot (\bV_1^*,\ldots,\bV_{i-1}^*,\bU_i - \bV_i^*,\bU_{i+1},\ldots,\bU_m)$, then we have
	\begin{align}\label{tl-t}
		\fro{\hat\T_l - \T^*}^2 &=\fro{\C_l - \S^*}^2 + \sum_{i=1}^m\fro{(\bU_i-\bV_i^*)\opM_i(\S^*)}^2+ 2\sum_{i<j} \inp{\X_i}{\X_j}+ 2\sum_{i=1}^m \inp{(\C_l-\S^*)\cdot (\bU_1,\cdots, \bU_m)}{\X_i} \no\\
		&\geq \fro{\C_l - \S^*}^2 + \sum_{i=1}^m
		\usigma^2\fro{\bU_i-\bV_i^*}^2+ 2\sum_{i<j} \inp{\X_i}{\X_j} + 2\sum_{i=1}^m \inp{(\C_l-\S^*)\cdot (\bU_1,\cdots, \bU_m)}{\X_i} \no\\
	\end{align}
	Notice that $\opM_i(\X_i) = (\bU_i-\bV_i^*)\opM_i(\S^*)\big(\bU_m\otimes\bU_{i+1}\otimes\bV_{i-1}\otimes\bV_{1}\big)^{\top}$. So we have the estimation of $|\inp{(\C_l-\S^*)\cdot (\bU_1,\cdots, \bU_m)}{\X_i}|$ is as follows:
	\begin{align}\label{inp1est}
		|\inp{(\C_l-\S^*)\cdot (\bU_1,\cdots, \bU_m)}{\X_i}|&= |\inp{\opM_{i}\big((\C_l-\S^*)\cdot (\bU_1,\cdots, \bU_m)\big)}{\opM_{i}(\X_i)}|\no\\
		&\leq \op{(\bU_i-\bV_i^*)^{\top}\bU_i}\fro{\C_l-\S^*}\fro{\S^*}\no\\
		&\leq \sqrt{\rmax}\bsigma \fro{\bU_i^{\top}(\bU_i-\bV_i^*)} \fro{\C_l-\S^*}
	\end{align}
	Now we estimate $\fro{\bU_i^{\top}(\bU_i-\bV_i^*)}$ by plugging in the closed form of $\bV_i^*$ as in \eqref{eq:Ri}
	\begin{align}\label{eq:ui-vi}
		\fro{\bU_i^{\top}(\bU_i-\bV_i^*)} 
		= \fro{\bI - \bS_i}\leq \fro{\bI - \bS_i^2}
		= \fro{\bU_{i\perp}^{*T}\bU_i}^2
		\leq \fro{\bU_i - \bU_i^*\bR_i}^2
	\end{align}
	From Wedin' sin$\Theta$ Theorem, we have for $i\in[m]$
	\begin{align}\label{uj-vj}
		\fro{\bU_i- \bV_i^*} \leq \fro{\bU_i- \bU_i^*}&\leq \frac{\sqrt{2}\fro{\hat\T_l-\T^*}}{\usigma - \op{\hat\T_l - \T^*}}\leq \frac{2\sqrt{2}\fro{\hat\T_l-\T^*}}{\usigma}\leq \frac{1}{4m\rmax^{1/2}\kappa_0}
	\end{align}
	where the second last inequality is from $\op{\hat\T_l - \T^*} \leq {\usigma/2}$ and the last inequality is from 
	$\fro{\hat\T_l - \T^*} \leq \frac{\usigma}{16m\rmax^{1/2}\kappa_0}$. Then from \eqref{inp1est} and \eqref{uj-vj}, we have
	\begin{align}\label{inp:1}
		|\inp{(\C_l-\S^*)\cdot (\bU_1,\cdots, \bU_m)}{\X_i}|\leq \frac{1}{8m^2}\fro{\C_l-\S^*}^2 + \frac{1}{8}\usigma^2\fro{\bU_i- \bV_i^*}^2
	\end{align}
	The estimation of $|\inp{\X_i}{\X_j}| (i<j)$ is as follows. From $\eqref{uj-vj}$, we have
	\begin{align}\label{inp:2}
		|\inp{\X_i}{\X_j}| &= |\inp{\opM_i(\S^*)\bM_{i,j}}{(\bU_i-\bV_i^*)^{\top}\bV_i^*\opM_i(\S^*)}|\no\\
		&\leq \bsigma\fro{\S^*}\op{\bM_{i,j}} \fro{(\bU_i-\bV_i^*)^{\top}\bV_i^*}\no\\
		&\leq \bsigma\fro{\S^*}\fro{(\bU_i-\bV_i^*)^{\top}\bV_i^*}\fro{(\bU_j-\bV_j^*)^{\top}\bV_j^*}\no\\
		&\overset{(a)}{\leq} \sqrt{\rmax}\bsigma^2 \fro{\bU_i-\bV_i^*}^2\fro{\bU_j-\bV_j^*}^2\no\\
		&\overset{(b)}{\leq} \frac{1}{16m^2}\usigma^2\fro{\bU_i-\bV_i^*}\fro{\bU_j-\bV_j^*}\no\\
		&\leq \frac{1}{32m^2}\usigma^2\fro{\bU_i-\bV_i^*}^2 + \frac{1}{32m^2}\usigma^2\fro{\bU_j-\bV_j^*}^2
	\end{align}
where $\bM_{i,j} = \bI\otimes\ldots\otimes\bI\otimes \bU_j^{\top}(\bU_j - \bV_j^*)\otimes\bU_{j-1}^{\top}\bV_{j-1}^*\otimes\ldots\bU_{i+1}^{\top}\bV_{i+1}^*\otimes\bI\otimes\ldots\otimes\bI$, $(a)$ holds because of \eqref{eq:ui-vi}, $(b)$ holds because of \eqref{uj-vj}.
	
	As a result of \eqref{tl-t}, \eqref{inp:1} and \eqref{inp:2}, we have
	\begin{align*}
		\fro{\hat\T_l - \T^*}^2 \geq 0.5\fro{\C_l - \S^*}^2 + 0.5\usigma^2\sum_{i=1}^m\fro{\bU_i -\bV_i^*}^2
	\end{align*}
	which finishes the proof of the lemma.
\end{proof}

\end{document}